\RequirePackage{fix-cm}
\documentclass[twocolumn]{svjour3}          
\usepackage{amssymb}
\usepackage{booktabs} 
\usepackage{times} 
\usepackage{color} 
\usepackage{balance} 
\usepackage{url}
\usepackage{tabularx}
\usepackage{siunitx} 
\usepackage{epstopdf} 
\usepackage{epsfig,tabularx,subfigure,multirow, graphicx}
\usepackage{enumitem}
\usepackage{caption}
\usepackage{bbding}
\usepackage{amsmath}

\newcolumntype{L}[1]{>{\raggedright\arraybackslash}p{#1}}
\newcolumntype{C}[1]{>{\centering\arraybackslash}p{#1}}
\newcolumntype{R}[1]{>{\raggedleft\arraybackslash}p{#1}}

\usepackage[linesnumbered,ruled,vlined]{algorithm2e}
\SetKwRepeat{Do}{do}{while}
\SetCommentSty{mycommfont}

\long\def\comment#1{}

\newcommand{\nop}[1]{}

\newcommand{\figureCaptionMargin}{\vspace{-1ex}}

\newcommand{\itemMargin}{\vspace{-1ex}}

\smartqed  
\newcommand{\outereqsize}[1]{{\scriptsize#1}}
\newcommand{\outereqsizelarge}[1]{{\footnotesize#1}}
\newcommand{\innereqsize}[1]{{\footnotesize#1}}

\newcommand{\entity}[1]{\mathcal{#1}}
\newcommand{\algvar}[1]{\mathcal{#1}}
\newcommand{\constvar}[1]{\mathbb{#1}}

\newcommand{\vectorfont}[1]{\boldsymbol{#1}}

\newcommand{\problemDefineSimpleName}{PWEPP-IDS}
\newcommand{\problemDefineTotalName}{Personalized $w$-Event Private Publishing for Infinite Data Streams}

\newcommand{\problemDefineEnhancedSimpleName}{DPWEPP-IDS}
\newcommand{\problemDefineEnhancedTotalName}{Dynamic Personalized $w$-Event Private Publishing for Infinite Data Streams}

\newcommand{\constantPrivacyLevelSimpleName}{($\vectorfont{w}$, $\vectorfont{\algvar{E}}$)-EPDP}
\newcommand{\constantPrivacyLevelTotalName}{$\vectorfont{w}$-Event $\vectorfont{\algvar{E}}$ Personalized Differential Privacy}
\newcommand{\dynamicPrivacyLevelSimpleName}{($\tau$, $\vectorfont{w}_B$, $\vectorfont{w}_F$, $\vectorfont{\algvar{E}}_B$, $\vectorfont{\algvar{E}}_F$)-EPDP}
\newcommand{\dynamicPrivacyLevelTotalName}{($\tau$, $\vectorfont{w}_B$, $\vectorfont{w}_F$)-Event ($\vectorfont{\algvar{E}}_B$, $\vectorfont{\algvar{E}}_F$)-Personalized Differential Privacy}
\newcommand{\dynamicPrivacyLevelSimpleNameT}{($t$, $\vectorfont{w}_B$, $\vectorfont{w}_F$, $\vectorfont{\algvar{E}}_B$, $\vectorfont{\algvar{E}}_F$)-EPDP}

\newcommand{\solutionA}{PWSM}
\newcommand{\solutionATotalName}{Personalized Window Size Mechanism}
\newcommand{\solutionB}{DPWSM}
\newcommand{\solutionBTotalName}{Dynamic Personalized Window Size Mechanism}

\newcommand{\solutionMethodA}{PBD}
\newcommand{\solutionMethodATotalName}{Personalized Budget Distribution}
\newcommand{\solutionMethodB}{PBA}
\newcommand{\solutionMethodBTotalName}{Personalized Budget Absorption}
\newcommand{\solutionMethodD}{DPBD}
\newcommand{\solutionMethodDTotalName}{Dynamic Personalized Budget Distribution}
\newcommand{\solutionMethodE}{DPBA}
\newcommand{\solutionMethodETotalName}{Dynamic Personalized Budget Absorption}

\newcommand{\solutionCMPPLDPU}{PLBU}
\newcommand{\solutionCMPPLDPUTotalName}{Personalized LDP Budget Uniform}

\newcommand{\solutionCMPA}{BD}
\newcommand{\solutionCMPATotalName}{Budget Distribution}
\newcommand{\solutionCMPB}{BA}
\newcommand{\solutionCMPBTotalName}{Budget Absorption}
\newcommand{\solutionCMPLDPU}{LBU}
\newcommand{\solutionCMPLDPUTotalName}{LDP Budget Uniform}
\newcommand{\solutionCMPSPAS}{SPAS}

\newcommand{\checkInDatasetName}{Foursquare}
\newcommand{\trajectoryDatasetName}{Taxi}
\newcommand{\tlnsDatasetName}{TLNS}
\newcommand{\sinDatasetName}{Sin}
\newcommand{\logDatasetName}{Log}

\newcounter{ruleCounter}
\begin{document}

\title{Personalized $w$-Event Privacy for Infinite Stream Estimation
}


\author{Leilei Du\textsuperscript{1}      	\and
		Xu Zhou\textsuperscript{1} 			\and
        Peng Cheng\textsuperscript{2} 		\and
        Lei Chen\textsuperscript{3,4} 			\and
        Xuemin Lin\textsuperscript{5} 		\and
        Wei Xi\textsuperscript{6} 			\and
        Kenli Li\textsuperscript{1} 						
}


\institute{
		\begin{itemize}
			\item[\Envelope] Xu Zhou \at zhxu@hnu.edu.cn \\
			\item[] Leilei Du \at leileidu@hnu.edu.cn \\
			\item[] Peng Cheng \at cspcheng@tongji.edu.cn \\
			\item[] Lei Chen \at leichen@cse.ust.hk \\
			\item[] Xuemin Lin \at xuemin.lin@gmail.com \\
			\item[] Wei Xi \at xiwei@xjtu.edu.cn \\
			\item[] Kenli Li \at lkl@hnu.edu.cn \\
			\item [\textsuperscript{1}] Hunan University, Changsha, China 
			\item [\textsuperscript{2}] Tongji University, Shanghai, China
			\item [\textsuperscript{3}] HKUST (GZ), Guangzhou, China
			\item [\textsuperscript{4}] HKUST, HK SAR, China
			\item [\textsuperscript{5}] Shanghai Jiaotong University, Shanghai, China
			\item [\textsuperscript{6}] Xi'an Jiaotong University, Xi'an, China
		\end{itemize}
}


\date{Received: date / Accepted: date}

\maketitle

\sloppy

\begin{abstract}
	In many real-life applications, such as event monitoring, log analysis and video querying, $w$-event privacy is widely used to protect individual privacy within a given time window while maintaining high accuracy in data collection.
	However, existing $w$-event privacy studies on infinite data stream typically focus only on homogeneous privacy requirements for all users. 
	In this paper, we propose personalized $w$-event privacy protection that enables users to set different privacy requirements in private data stream estimation. 
	Specifically, we first design a  \solutionATotalName{} (\solutionA{}) that allows users to maintain personalized privacy requirements at each time slot. 
	Then, we propose two solutions---\solutionMethodATotalName{} (\solutionMethodA{}) and \solutionMethodBTotalName{} (\solutionMethodB{})---to accurately estimate streaming data statistics while achieving \constantPrivacyLevelTotalName{} (\constantPrivacyLevelSimpleName{}). 
	\solutionMethodA{} ensures that the privacy budget for the next time step is at least equal to the amount consumed in the previous release.
	\solutionMethodB{} enhances the current time slot's privacy budget by combining the privacy budget from the previous $k$ time slots and borrowing from the next $k$ time slots.
	In addition, we design two additional solutions---\solutionMethodDTotalName{} (\solutionMethodD{}) and \solutionMethodETotalName{} (\solutionMethodE{})---that allow users to dynamically adjust their privacy requirements at each time slot while achieving \dynamicPrivacyLevelTotalName{} (\dynamicPrivacyLevelSimpleName{}).
	The proposed methods are all proven to achieve personalized differential privacy levels and establish error upper bounds for each method.
	Experimental results show that the proposed mechanisms improve utility over classical homogeneous $w$-event baselines in the heterogeneous personalized setting. 
	In particular, for real datasets, \solutionMethodD{} reduces \hbox{\em AMRE} by at least 62.7\% compared with \solutionCMPA{}. For synthetic datasets, \solutionMethodE{} reduces \hbox{\em AMRE} by at least 53.6\% compared with \solutionCMPB{}.
	
	\keywords{Differential privacy\and Stream data\and Event privacy\and Personalized privacy}
\end{abstract}

\section{Introduction}
With the widespread adoption of smart devices and wireless networks, more people are sharing and receiving data through various platforms, making real-time analysis (e.g., event monitoring~\cite{DBLP:journals/tpds/GuoJZZ12}, log analysis~\cite{DBLP:conf/www/XiePMJM23}, and video querying~\cite{DBLP:conf/cvpr/MoonHPPH23}) increasingly necessary. 
During these data collection and analysis processes, protecting users' privacy is crucial. 
To prevent user data leakage, Differential Privacy (DP) has emerged as a widely adopted solution in data publishing and statistical analysis. 

Existing $w$-event privacy mechanisms~\cite{DBLP:conf/infocom/WangZLWQR16,DBLP:conf/infocom/WangLPRLC20,DBLP:conf/sigmod/RenSYYZX22} based on DP protect user data privacy across consecutive events. 
However, these mechanisms use uniform privacy requirements for all users (applying the same privacy budget $\algvar{E}$ and window size $w$). 
This one-size-fits-all approach has significant limitations in real-world applications. 
For example, celebrities in the entertainment industry may require strong protection of their location data, while street artists may want to share their locations to gain visibility. 
The fixed privacy budget and window size thus provide inadequate privacy protection for some users while unnecessarily increasing data error rates (excessive privacy protection) for others.

\begin{figure}[t!]
	\centering
	\setlength{\abovecaptionskip}{0.2cm}
	\includegraphics[width=0.48\textwidth]{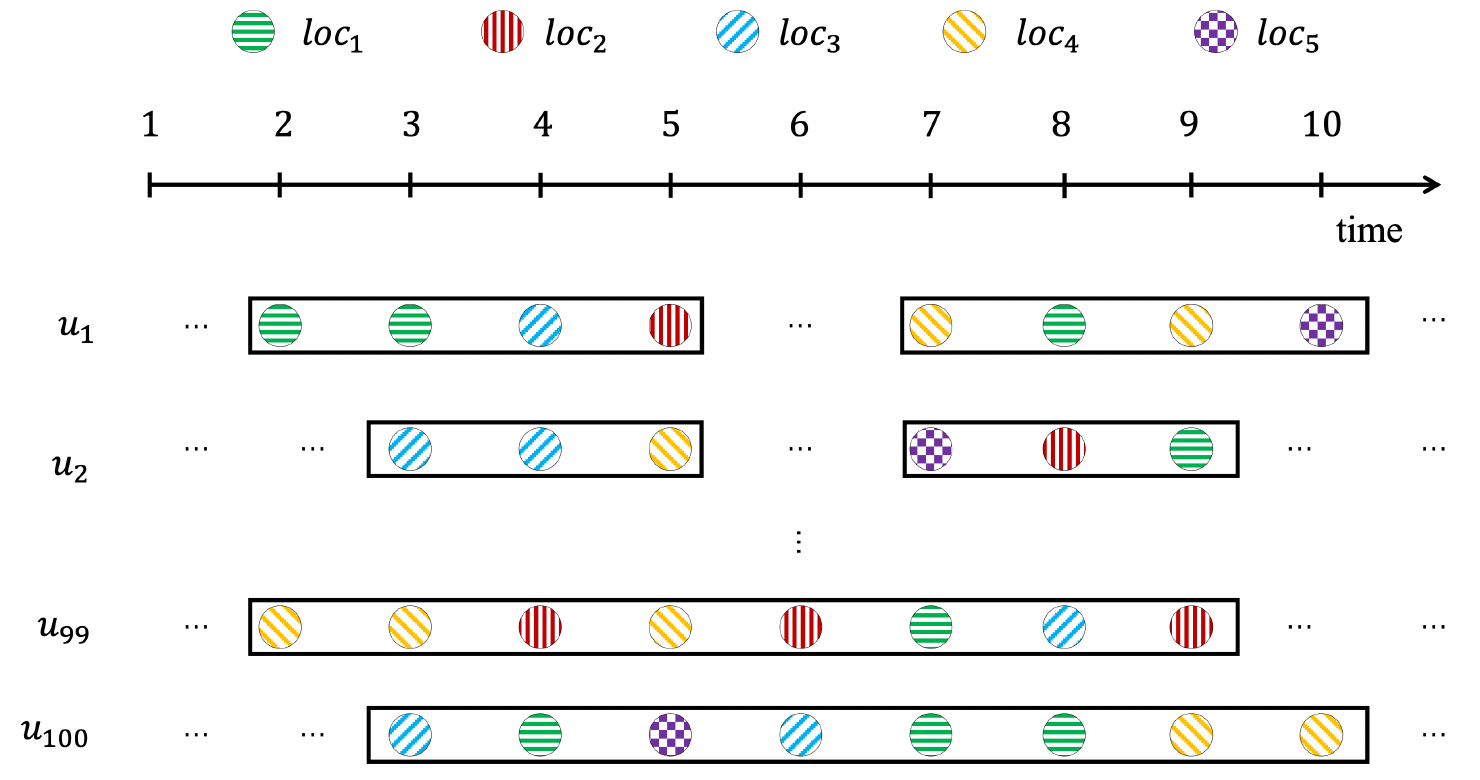}
	\caption{Different event window sizes for different time slots. }\label{fig:introduction_example}
\end{figure}

\begin{example}\itemMargin{}
	Consider an example of online car-hailing shown in Figure~\ref{fig:introduction_example}. 
	It has $100$ drivers $U=\{u_1,..., u_{100}\}$ who share their locations from $\{\hbox{\em loc}_1,...,\hbox{\em loc}_5\}$ at each time slot. 
	Each driver $u_i$ is protected by $w_i$-event privacy, meaning their location data is safeguarded through $\algvar{E}$-DP across at least $w_i$ consecutive time slots, where $\algvar{E}$ represents their required privacy protect strength.
	For example, $u_1$ requires location protection across any $4$ consecutive time slots, while $u_{99}$ and $u_{100}$ need protection across any $8$ consecutive time slots. 
	For the drivers $u_i \in U \backslash \{u_{99}, u_{100}\}$, the window sizes do not exceed $4$. 
	
	By traditional $w$-event privacy with $\algvar{E}=1$, it necessitates setting the event window sizes to the maximal value (i.e., $w=8$) for satisfying all drivers' privacy requirements.
	While this achieves $8$-event privacy---providing the strongest privacy protection---it also results in the lowest data utility.
	Let $AE_{\hbox{\scriptsize\em avg}}$ denote the average square error at each time slot.
	When using the Laplace mechanism, the error at any time slot equals the variance of added Laplace noise \Big(i.e., {\scriptsize$AE_{\hbox{\scriptsize\em avg}}=2\times\left(\frac{1}{\epsilon/w}\right)^2$}\Big).
	Using the \textit{Uniform method}~\cite{DBLP:journals/pvldb/KellarisPXP14}, the error is $AE_{\hbox{\scriptsize\em avg}} = 2\times(\frac{w}{\epsilon})^2 = 128$ under $8$-event privacy.
	While the first $98$ drivers only require $4$-event privacy. Setting the privacy level to $4$-event privacy would reduce the error to $AE_{\hbox{\scriptsize\em avg}} = 2\times\left(\frac{w}{\epsilon}\right)^2 = 32$.
	However, $4$-event privacy only protects data within a window size of $4$, which cannot meet the privacy requirements of the $99$-th and the $100$-th drivers who need protection within a window size of $8$, thus compromising their privacy.
\end{example}\vspace{-1.5ex}

\textit{Challenges.} From the example above, there are three main challenges:

(1) \textbf{Unified Privacy Budget.} Traditional DP requires a uniform privacy budget $\epsilon$ for all users to achieve $\epsilon$-DP. However, users have distinct privacy budgets (a type of privacy requirements).  While setting $\epsilon$ to the minimum value would satisfy everyone' privacy requirements, this approach significantly reduces data utility.
The challenge lies in transforming users' distinct privacy budgets into a valid system-level release decision while maximizing the utility of published data.

(2) \textbf{Personalized Privacy Budget Allocation.} The rate of change in streaming data fluctuates over time. Within a given window size, time slots with rapid changes contain more information compared to those with slower changes. Since the privacy budget serves as a privacy protection resource, it should be used more effectively at time slots where new releases are more beneficial. The challenge is determining how to allocate each user's personalized privacy budget across the corresponding privacy window while supporting accurate shared releases.

(3) \textbf{Dynamic Privacy Requirements.} Users have different privacy requirements at different time slots. These varying requirements can lead to privacy budget waste or privacy requirement conflicts between current and historical time slot. The challenge is how to allocate varying privacy budgets while maintaining high utility.

A further challenge is that the dynamic personalized setting is not a straightforward extension of the fixed one. In the fixed case, each user keeps the same privacy requirement over time. In contrast, in the dynamic case, privacy requirements may vary across time slots, so each release must remain consistent with previously consumed privacy budgets while preserving feasibility for future requirements. Moreover, although privacy constraints are specified at the user level, the system still publishes one shared aggregate result at each time slot. Therefore, the dynamic personalized setting introduces a new online feasibility problem under heterogeneous time-varying constraints.

\textit{Contributions.} 
This paper studies a more general problem than the fixed personalized setting, namely \problemDefineEnhancedTotalName{} (\problemDefineEnhancedSimpleName{}), where each user may specify time-varying backward and forward privacy requirements. This setting is practically important because privacy preferences in real systems may evolve over time, and technically challenging because each release must remain compatible with both historical budget consumption and future privacy feasibility.
To address this problem, we develop a unified view of personalized stream release. Our main observation is that heterogeneous personalized privacy requirements must ultimately be transformed into a valid system-level release decision, because only one aggregate statistic is published at each time slot. 
The main contribution is not replacing a global window or budget with personalized parameters, but constructing and maintaining a valid shared release budget under overlapping, heterogeneous, and time-varying personalized privacy requirements.
This heterogeneous-to-release unification challenge does not arise in classical homogeneous $w$-event privacy.
We summarize our contributions as follows:
\begin{itemize}[leftmargin=*]\itemMargin{}
	\item We  formulate \problemDefineEnhancedTotalName{} and define the corresponding privacy notation, namely \dynamicPrivacyLevelTotalName{} (\dynamicPrivacyLevelSimpleName{}), which generalizes the fixed personalized setting in Section~\ref{pro_def}.
	\item We identify a new online feasibility challenge under dynamic personalized privacy: at each time slot, the mechanism must reconcile heterogeneous user-specific requirements with both past budget consumption and future privacy feasibility, while still producing one shared aggregate release in Section~\ref{pro_def}.
	\item We propose a unified framework for personalized stream release. In the fixed setting, this framework is instantiated as \solutionA{} with two mechanisms, \solutionMethodA{} and \solutionMethodB{} in Section~\ref{basic_method}. In the dynamic setting, it is generalized to \solutionB{} with two mechanisms, \solutionMethodD{} and \solutionMethodE{} in Section~\ref{advanced_method}.
	\item We provide privacy guarantees and utility bounds for all mechanisms, analyze their computational properties in Section~\ref{basic_method} and~\ref{advanced_method}, and experimentally evaluate them on real and synthetic datasets in Section~\ref{experiment}.
\end{itemize}\itemMargin{}

Compared with the conference version~\cite{DBLP:journals/pvldb/Du05}, this paper studies a more general setting in which each user’s privacy requirement may vary over time. This generalization introduces a new online feasibility problem: each release must satisfy heterogeneous user-specific privacy constraints while remaining consistent with previously consumed budgets and feasible for future requirements. To address this problem, we develop the dynamic framework \solutionB{} together with two mechanisms, \solutionMethodD{} and \solutionMethodE{}, and provide corresponding privacy guarantees, utility analysis, and new experiments for this generalized setting.

\vspace{-2em}
\section{Related Work}\label{relatedwork}
\vspace{-1em}
We classify the related work in the area of data stream estimation under differential privacy and non-uniformity differential privacy.

\begin{table}[t!]\vspace{-2ex}
	\caption{Summary for related work.}
	\label{relatedwork_table}\vspace{1ex}
	\centering
		\resizebox{8.5cm}{!}{
			\begin{tabular}{|ccc|c|c|}
				\hline
				\multicolumn{2}{|c|}{\textbf{Model Types}}                                                                          & \textbf{Methods}                                                                                                                           & \textbf{\begin{tabular}[c]{@{}c@{}}Infinite \& \\  correlated\end{tabular}} & \textbf{\begin{tabular}[c]{@{}c@{}}Personalized \\ privacy\end{tabular}} \\ \hline
				\multicolumn{1}{|c|}{}                                 & \multicolumn{1}{c|}{}                                      & Finite B-tree~\cite{DBLP:conf/stoc/DworkNPR10}                                                                       & \XSolidBrush                                                                                      & \XSolidBrush                      \\ \cline{3-5} 
				\multicolumn{1}{|c|}{}                                 & \multicolumn{1}{c|}{}                                      & Infinite B-tree~\cite{DBLP:journals/tissec/ChanSS11}                                                                 & \XSolidBrush                                                                                      & \XSolidBrush                      \\ \cline{3-5} 
				\multicolumn{1}{|c|}{}                                 & \multicolumn{1}{c|}{}                                      & \begin{tabular}[c]{@{}c@{}}Adaptive-density \\ Counter~\cite{DBLP:conf/soda/Dwork10}\end{tabular}                    & \XSolidBrush                                                                                      & \XSolidBrush                      \\ \cline{3-5} 
				\multicolumn{1}{|c|}{}                                 & \multicolumn{1}{c|}{}                                      & Decayed Privacy~\cite{DBLP:conf/icdt/BolotFMNT13}                                                                    & \XSolidBrush                                                                                      & \XSolidBrush                      \\ \cline{3-5} 
				\multicolumn{1}{|c|}{}                                 & \multicolumn{1}{c|}{\multirow{-6}{*}{\begin{tabular}[c]{@{}c@{}}event-level \\ privacy\end{tabular}}} & PeGaSus~\cite{DBLP:conf/ccs/ChenMHM17}                                                                               & \XSolidBrush                                                                                      & \XSolidBrush                      \\ \cline{2-5} 
				\multicolumn{1}{|c|}{}                                 & \multicolumn{1}{c|}{}                                      & FAST~\cite{DBLP:journals/tkde/FanX14}                                                                                & \CheckmarkBold                                                                                    & \XSolidBrush                      \\ \cline{3-5} 
				\multicolumn{1}{|c|}{}                                 & \multicolumn{1}{c|}{}                                      & \begin{tabular}[c]{@{}c@{}}Private heterogeneous \\ mean estimation~\cite{DBLP:conf/nips/CummingsFMT22}\end{tabular}                                                                               & \CheckmarkBold                                                                                    & \XSolidBrush                      \\ \cline{3-5} 
				\multicolumn{1}{|c|}{}                                 & \multicolumn{1}{c|}{}                                      & Dynamic user-DP~\cite{DBLP:conf/sp/DongLY23}                                                                               & \CheckmarkBold                                                                                    & \XSolidBrush                      \\ \cline{3-5} 
				\multicolumn{1}{|c|}{}                                 & \multicolumn{1}{c|}{}                                      & DPI~\cite{DBLP:conf/sp/FengMWLQH24}                                                                         & \CheckmarkBold                                                                                    & \XSolidBrush                      \\ \cline{3-5} 
				\multicolumn{1}{|c|}{}                                 & \multicolumn{1}{c|}{\multirow{-6}{*}{\begin{tabular}[c]{@{}c@{}}user-level \\ privacy\end{tabular}}}  &SMM-TM, RBM~\cite{DBLP:conf/focs/DvijothamMP0T24}							 & \CheckmarkBold                                                                                    & \XSolidBrush                      \\ \cline{2-5} 
				\multicolumn{1}{|c|}{}                                 & \multicolumn{1}{c|}{}                                      & BD, BA~\cite{DBLP:journals/pvldb/KellarisPXP14}                                                                    & \CheckmarkBold                                                                                    & \XSolidBrush                      \\ \cline{3-5} 
				\multicolumn{1}{|c|}{}                                 & \multicolumn{1}{c|}{}                                      & RescueDP~\cite{DBLP:conf/infocom/WangZLWQR16}                                                           & \CheckmarkBold                                                                                    & \XSolidBrush                      \\ \cline{3-5} 
				\multicolumn{1}{|c|}{\multirow{-15}{*}{Centralized DP}} & \multicolumn{1}{c|}{\multirow{-3}{*}{\begin{tabular}[c]{@{}c@{}}$w$-event \\ privacy\end{tabular}}}     & SPAS~\cite{DBLP:journals/pacmmod/LiLCGRQW25}                                                                        & \CheckmarkBold                                                                                    & \XSolidBrush                      \\ \hline
				\multicolumn{1}{|c|}{}                                 & \multicolumn{1}{c|}{}                                      & RAPPOR~\cite{DBLP:conf/ccs/ErlingssonPK14}                                                                           & \XSolidBrush                                                                                      & \XSolidBrush                      \\ \cline{3-5} 
				\multicolumn{1}{|c|}{}                                 & \multicolumn{1}{c|}{\multirow{-2}{*}{\begin{tabular}[c]{@{}c@{}}event-level \\ privacy\end{tabular}}} & ToPL~\cite{DBLP:conf/ccs/0001C0SC0LJ21}                                                                              & \XSolidBrush                                                                                      & \XSolidBrush                      \\ \cline{2-5} 
				\multicolumn{1}{|c|}{}                                 & \multicolumn{1}{c|}{}                                      & CGM~\cite{DBLP:journals/pvldb/BaoYXD21}                                                                 & \CheckmarkBold                                                                                    & \XSolidBrush                      \\ \cline{3-5} 
				\multicolumn{1}{|c|}{}                                 & \multicolumn{1}{c|}{}                                      & DDRM~\cite{DBLP:journals/tkde/XueYHZW23}                                                                             & \CheckmarkBold                                                                                    & \XSolidBrush                      \\ \cline{3-5} 
				\multicolumn{1}{|c|}{}                                 & \multicolumn{1}{c|}{\multirow{-3}{*}{\begin{tabular}[c]{@{}c@{}}user-level \\ privacy\end{tabular}}}  & StaSwitch~\cite{DBLP:conf/infocom/YeHHAX23}                                                                              & \CheckmarkBold                                                                                    & \XSolidBrush                      \\ \cline{2-5} 
				\multicolumn{1}{|c|}{\multirow{-6}{*}{Local DP}}       & \multicolumn{1}{c|}{\begin{tabular}[c]{@{}c@{}}$w$-event \\ privacy\end{tabular}}                       & LDP-IDS~\cite{DBLP:conf/sigmod/RenSYYZX22}                                                                           & \CheckmarkBold                                                                                    & \XSolidBrush                      \\ \hline
				
				\multicolumn{1}{|c|}{}       & \multicolumn{1}{c|}{\begin{tabular}[c]{@{}c@{}}event-level\\ privacy\end{tabular}}                       & Concurrent-SDP~\cite{DBLP:conf/icml/TenenbaumKMS23}                                                                       & \CheckmarkBold                                                                                     & \XSolidBrush                      \\ \cline{2-5}
				\multicolumn{1}{|c|}{}                                 & \multicolumn{1}{c|}{}                                      & LPS-SS~\cite{DBLP:conf/icde/LiCY23}                                                                          & \XSolidBrush                                                                                      & \XSolidBrush                      \\ \cline{3-5} 
				\multicolumn{1}{|c|}{\multirow{-4}{*}{Shuffled DP}}                                 & \multicolumn{1}{c|}{\multirow{-2}{*}{\begin{tabular}[c]{@{}c@{}}user-level \\ privacy\end{tabular}}} & ExSub~\cite{DBLP:journals/tmc/WangLPCYJL25}                                                                             & \XSolidBrush                                                                                      & \XSolidBrush                      \\\hline 
				
				\multicolumn{2}{|c|}{Item heterogeneous}                                                                            & HDP~\cite{DBLP:journals/jpc/AlagganGK16}                                                                             & \XSolidBrush                                                                                      & \XSolidBrush                      \\ \hline
				\multicolumn{2}{|c|}{}                                                                                              & PDP~\cite{DBLP:conf/icde/JorgensenYC15}                                                                              & \XSolidBrush                                                                                      & \CheckmarkBold                    \\ \cline{3-5} 
				\multicolumn{2}{|c|}{}                                                                                              & OSDP~\cite{DBLP:conf/icde/KotsogiannisDHM20}                                                                         & \XSolidBrush                                                                                      & \CheckmarkBold                    \\ \cline{3-5} 
				\multicolumn{2}{|c|}{}                                                                                              & Geo-I~\cite{DBLP:conf/ccs/AndresBCP13}                                                                               & \XSolidBrush                                                                                      & \CheckmarkBold                    \\ \cline{3-5} 
				\multicolumn{2}{|c|}{}                                                                                              & PWSM, VPDM~\cite{DBLP:journals/tmc/WangHLWWYQ19}                                                                     & \XSolidBrush                                                                                      & \CheckmarkBold                    \\ \cline{3-5} 
				\multicolumn{2}{|c|}{}                                                                                              & PUCE, PGT~\cite{DBLP:conf/icde/DuCZX00F23}                                                                           & \XSolidBrush                                                                                      & \CheckmarkBold                    \\ \cline{3-5} 
				\multicolumn{2}{|c|}{\multirow{-6}{*}{Record heterogenous}}                                                         & PFA, PFA+~\cite{DBLP:journals/pvldb/LiuLXLM21}                                                                       & \XSolidBrush                                                                                      & \CheckmarkBold                    \\ \hline
				\multicolumn{3}{|c|}{ \textbf{Our mechanisms}}                                                                                                                                                                                                       & \CheckmarkBold                                                                                    & \CheckmarkBold                    \\ \hline
			\end{tabular}
		}
\end{table}

\vspace{-2em}
\subsection{Data Stream Estimation under Differential Privacy}
\vspace{-1em}
Based on the privacy model, there are three types of data stream estimation methods: centralized differential privacy~\cite{DBLP:conf/icalp/Dwork06} (CDP) based methods, local differential privacy~\cite{DBLP:conf/stoc/BassilyS15} (LDP) based methods and shuffled differential privacy~\cite{DBLP:journals/corr/abs-2107-11839,DBLP:conf/eurocrypt/CheuSUZZ19}.

\noindent\textbf{Data Stream Estimation under CDP.}
Dwork et al.~\cite{DBLP:conf/stoc/DworkNPR10} first address the problem of Differential Privacy (DP) on data streams.
They define two types of DP levels: \emph{event-level differential privacy} (event-DP) and \emph{user-level differential privacy} (user-DP).

In event-DP, each single event is hidden in statistic queries. 
Dwork et al.~\cite{DBLP:conf/stoc/DworkNPR10} focus on the finite event scenarios and propose a binary tree method to achieve high statistical utility while maintaining event-DP.
Chan et al.~\cite{DBLP:journals/tissec/ChanSS11} extend it to infinite cases, and produce partial summations for binary counting.
Dwork et al.~\cite{DBLP:conf/soda/Dwork10} introduce a cascade buffer counter that updates adaptively based on stream density.
Bolot et al.~\cite{DBLP:conf/icdt/BolotFMNT13} propose \textit{decayed privacy} which reduces the privacy costs for past data.
Chen et al.~\cite{DBLP:conf/ccs/ChenMHM17} develop PeGaSus, a perturb-group-smooth framework for multiple queries under event-DP.
However, event-DP assumes all elements in a stream are independent, making it unsuitable for correlated data stream publishing.

In user-DP, all events for each user are hidden in statistic queries.
Fan et al.~\cite{DBLP:journals/tkde/FanX14} propose the FAST algorithm, which uses a sampling-and-filtering framework to count finite stream data under user-DP.
Cummings et al.~\cite{DBLP:conf/nips/CummingsFMT22} address heterogeneous user data by estimating population-level means while achieving user-DP. However, they only consider finite data.
Dong et al.~\cite{DBLP:conf/sp/DongLY23} introduce continual observation mechanisms under user-DP for dynamic data streams, achieving utility guarantees without prior data restrictions and providing down-neighborhood optimality for count and sum functions. However, their approach assumes independence between different stages.
Feng et al.~\cite{DBLP:conf/sp/FengMWLQH24} develop the DPI framework with bidrectional reweighting, $0$-DP synopsis generation, and dynamic error control, ensuring that privacy preservation does not significantly degrade accuracy over time.
Dvijotham et al.~\cite{DBLP:conf/focs/DvijothamMP0T24} tackle cascading correlations in data through two methods: Streaming Matrix Multiplication for Toeplitz Matrices (SMM-TM) and Recursive Binary Tree Mechanism (RBM). 
These approaches reduce the impact of data dependencies on differential privacy in streaming continual counting tasks.
However, providing user-DP for infinite data requires infinite perturbation, resulting in poor long-term utility~\cite{DBLP:journals/pvldb/KellarisPXP14}.

To bridge the gap between event-DP and user-DP, Kellaris et al.~\cite{DBLP:journals/pvldb/KellarisPXP14} propose $w$-event DP for infinite streams.
This ensures $\epsilon$-DP for any group of events within a time window of size $w$. They introduce two methods, \textit{Budget Distribution} (\solutionCMPA{}) and \textit{Budget Absorption} (\solutionCMPB{}), to optimize privacy budget use and estimate statistics effectively. However, neither method handles stream data with significant changes. 
Wang et al.~\cite{DBLP:conf/infocom/WangZLWQR16} apply the $w$-event concept to the FAST method, proposing a multi-dimensional stream release mechanism called \textit{RescueDP}, which achieves accurate estimation for both rapid and slow data stream changes. 
Li et al.~\cite{DBLP:journals/pacmmod/LiLCGRQW25} propose \textit{SPAS} for the continuous release of infinite data streams under $w$-event differential privacy. It improves adaptability through data-dependent strategy prediction, adaptive sampling, and privacy budget allocation. However, SPAS assumes a single global privacy requirement and does not support heterogeneous user-specific privacy budgets or window sizes, nor dynamically changing personalized requirements over time.
Overall, existing centralized methods for $w$-event private stream release, including \solutionCMPA{}, \solutionCMPB{}, RescueDP, and SPAS, are designed for the classical homogeneous setting with a single global privacy requirement, and do not model heterogeneous user-specific privacy requirements or their time-varying extensions.

\noindent\textbf{Data Stream Estimation under LDP.}
To overcome the dependence on a trusted server, LDP~\cite{DBLP:conf/stoc/BassilyS15} has recently been proposed and adopted by many major companies such as Microsoft, Apple and Google. 
Similar to DP, data stream estimation under LDP can be classified into event-LDP, user-LDP and $w$-event LDP.

Erlingsson et al.~\cite{DBLP:conf/ccs/ErlingssonPK14} introduce RAPPOR to estimate finite streams under LDP. They design a two-layer randomized response mechanism (i.e., permanent randomized response and instantaneous randomized response) to protect each individual's data. Wang et al.~\cite{DBLP:conf/ccs/0001C0SC0LJ21} extend event-level privacy from CDP to LDP and design the efficient ToPL method under event LDP.
Nevertheless, both RAPPOR and ToPL focuses solely on event-level privacy, lacking privacy protection for correlated data in streams. 

To address the problem of correlated time series data, 
Bao et al.~\cite{DBLP:journals/pvldb/BaoYXD21} propose CGM, an $(\epsilon,\delta)$-LDP method that uses the analytic Gaussian mechanism for streaming data collection. However, CGM is limited to finite streaming data.
Xue et al.~\cite{DBLP:journals/tkde/XueYHZW23} introduce DDRM for continual frequency estimation under LDP. While it dynamically allocates privacy budgets and employs difference trees to reduce unnecessary consumption, DDRM suffers from eventual budget depletion, which compromises estimation accuracy.
Ye et al.~\cite{DBLP:conf/infocom/YeHHAX23} develop the StaSwitch mechanism, which employs a stateful switch operation for efficient privacy budget management. Though this allows flexible privacy parameter settings and improves data utility, the budget still accumulates over time.

Ren et al.~\cite{DBLP:conf/sigmod/RenSYYZX22} introduce LDP-IDS for infinite streaming data collection and analysis under $w$-event LDP.
They propose two budget allocation methods and two population allocation methods, bridging the gap between event LDP and user LDP while improving estimation accuracy.
However, all these methods cannot be adopted to support personalized event window sizes.

\noindent
\textbf{Data Stream Estimation under SDP.}
Tenenbaum et al.~\cite{DBLP:conf/icml/TenenbaumKMS23} propose a shuffle-based continual observation mechanism that supports concurrent streaming queries with provable accuracy guarantees. However, its privacy notion is limited to event-level and does not extend to user-level protection.
Li et al.~\cite{DBLP:conf/icde/LiCY23} propose a shuffle-based LDP streaming framework with subsampling that achieves double privacy amplification and improved utility.
However, it is only suitable for finite stream data.
Wang et al.~\cite{DBLP:journals/tmc/WangLPCYJL25} propose ExSub, a user-level differentially private streaming analytics framework under the local and shuffle models that achieves near-centralized accuracy for finite-length data streams. However, it relies on bounded user changes and a predefined time horizon.

\vspace{-2em}
\subsection{Personalized and Heterogeneous Differential Privacy}
\vspace{-1em}
Recently, some studies address the non-uniform privacy requirements among items (table columns) or records (table rows)~\cite{DBLP:conf/uss/Murakami019}. 

Alaggan et al.~\cite{DBLP:journals/jpc/AlagganGK16} first examine scenarios where each database instance comprises a single user's profile. They focus on varying privacy requirements for different items and formally define Heterogeneous Differential Privacy (HDP).
Jorgensen et al.~\cite{DBLP:conf/icde/JorgensenYC15} investigate the privacy preservation for individual rows, introducing Personalized Differential Privacy (PDP). They design two mechanisms leveraging non-uniform privacy requirements to achieve better utility than standard uniform DP.
Kotsogiannis et al.~\cite{DBLP:conf/icde/KotsogiannisDHM20} recognize that different data have different sensitivity, then define One-sided Differential Privacy (OSDP) and propose algorithms that truthfully release non-sensitive record samples to enhance accuracy in DP-solutions.
Andr\'{e}s et al.~\cite{DBLP:conf/ccs/AndresBCP13} introduce a novel non-uniform privacy concept called Geo-Indistinguishability (Geo-I), where the privacy level for any point increases as the distance to this point decreases.
Wang et al.~\cite{DBLP:journals/tmc/WangHLWWYQ19} and Du et al.~\cite{DBLP:conf/icde/DuCZX00F23} explore  PDP in spatial crowdsourcing, and develop highly effective private task assignment methods to satisfy diverse workers' privacy and utility requirements.
Liu et al.~\cite{DBLP:journals/pvldb/LiuLXLM21} investigate HDP in federated learning. 
They assume different clients hold DP budget and divide them into private and public parts, then propose two methods to project the ``public'' clients' models into ``private'' clients' models to improve the joint model's utility.
More recently, Sun et al.~\cite{SunDQY24} propose Personalized Truncation for personalized differential privacy (PDP) in count, sum, and SJA query processing. This line of work is related in spirit but does not consider infinite-stream continual release or $w$-event privacy.
However, all above studies are not suitable for stream data.

\vspace{-2em}
\section{Problem Settings}\label{pro_def}
In this section, we introduce key concepts, including data streams and differential privacy (DP). 
We then define two types of personalized privacy requirements that address different real-world scenarios. 
Finally, we provide the problem definition of \problemDefineEnhancedTotalName{} (\problemDefineEnhancedSimpleName{}). 
Table~\ref{tbl:notations} summarizes the notations used throughout this paper.

\begin{table}[!h]
	\caption{Notations.}
	\label{tbl:notations}
	\centering
	\scalebox{0.9}{
		\begin{tabular}{c|l}
			\hline
			\textbf{Notations}        & \textbf{Description}                                       \\ \hline
			$\entity{D}$             & the database domain                                        \\ 
			$D_t$             & a database at time slot $t$                                        \\ 
			$S$             & a data stream                                        \\ 
			$U$				& the user set				\\
			$u_i$			& the $i$-th user in $U$			\\ 
			$\vectorfont{x}_{i,t}$       & $u_i$'s data at time slot $t$                                        \\ 
			$\vectorfont{c}_t$             & a real statistical histogram at time slot $t$                                     \\ 
			$\vectorfont{r}_t$             & an estimation statistic histogram at time slot $t$                                      \\ 
			$\vectorfont{\epsilon}$             & all users' privacy budget requirement at any time slot                                      \\ 
			$\epsilon_i$             & $u_i$'s privacy budget requirement in at any time slot                                      \\ 
			$\vectorfont{w}$             & all users' fixed window size requirements                                     \\
			$w_{i}$             & $u_i$'s fixed window size requirement                                      \\ 
			$\vectorfont{\algvar{E}}$             & all users' fixed privacy budget requirements                                  \\
			$\algvar{E}_{i}$             & $u_i$'s fixed privacy budget requirement                                  \\
			$\vectorfont{w}_{B,t}$             & all users' backward window size requirements at time slot $t$                                    \\ 
			$w_{B,i,t}$             & $u_i$'s backward window size requirement at time slot $t$                                    \\ 
			$\vectorfont{\algvar{E}}_{B,t}$             & all users' backward privacy budget requirement at time slot $t$                                 \\
			$\algvar{E}_{B,i,t}$             & $u_i$'s backward privacy budget requirement at time slot $t$                                 \\
			$\vectorfont{w}_{F,t}$             & all users' forward window size requirement at time slot $t$                                    \\ 
			$w_{F,i,t}$             & $u_i$'s forward window size requirement at time slot $t$                                    \\ 
			$\vectorfont{\algvar{E}}_{F,t}$             & all users' forward privacy budget requirement at time slot $t$                                  \\
			$\algvar{E}_{F,i,t}$             & $u_i$'s forward privacy budget requirement at time slot $t$                                  \\ \hline
		\end{tabular}
	}
\end{table}

\subsection{Data Stream}
\begin{definition}(Data Stream~\cite{DBLP:journals/pvldb/KellarisPXP14}).
	Let $D_t\in \algvar{D}$ be a database with $d$ columns and $n$ rows (each row representing a user) at $t$-th time slot.
	The infinite database sequence $S=[D_1,D_2, \ldots]$ is called a data stream, where $S[t]$ is the $t$-th element in $S$ (i.e., $S[t]=D_t$).
\end{definition} 

For any data stream $S$, a substream between time slot $t_l$ and $t_r$ (where $t_l < t_r$) is denoted as $S_{t_l,t_r}=[D_{t_l},D_{t_l+1}, \ldots, D_{t_r}]$.
When $t_l=1$, we denote $S_t=[D_1,D_2, \ldots, D_t]$ as the \textit{stream prefix} of $S$.

\begin{definition}(Data Stream Count Publishing).\label{def:count_release}
	Let $Q: \algvar{D}\to \constvar{R}^d$ be a count query. 
	Then, $Q(S[t])=Q(D_t)=\vectorfont{c}_t$ is the count data to be published at time slot $t$, where $\vectorfont{c}_t(j)$ represents the count of the $j$-th column of $D_t$.
	The infinite count data series $[\vectorfont{c}_1, \vectorfont{c}_2,\ldots]$ is called a data stream count publishing.
\end{definition}

\begin{definition}($w$-neighboring stream prefixes~\cite{DBLP:journals/tissec/ChanSS11,DBLP:journals/pvldb/KellarisPXP14}).
	Let $w$ be a positive integer, two stream prefixes $S_t$, $S'_t$ are $w$-neighboring (i.e., $S_t \sim_w S'_t$), if 
	\begin{enumerate}
		\item for each $S_t[k]$, $S'_t[k]$ such that $k\leq t$ and $S_t[k]\neq S'_t[k]$, it holds that $S_t[k]$ and $S'_t[k]$ are neighboring~\cite{DBLP:journals/pvldb/KellarisPXP14} in centralized DP, and
		\item for each $S_t[k_1]$, $S_t[k_2]$, $S'_t[k_1]$, $S'_t[k_2]$ with $k_1<k_2$, $S_t[k_1]\neq S'_t[k_1]$ and $S_t[k_2]\neq S'_t[k_2]$, it holds that $k_2-k_1+1\leq w$.
	\end{enumerate}
\end{definition}

\begin{definition}($(\tau,w)$-backward neighboring stream prefixes).
	Let $w$ be a positive integer. Two stream prefixes $S_{t}$, $S'_{t}$ are $(\tau,w)$-backward neighboring (denoted as $S_{t}\sim_{B,\tau,w} S'_{t}$), if 
	\begin{enumerate}
		\item for each $S_{t}[k]$, $S'_{t}[k]$ such that $k\in[t]$ and $S_{t}[k]\neq S'_{t}[k]$, it holds that $S_{t}[k]$ and $S'_{t}[k]$ are neighboring, and
		\item for each $S_t[k]$, $S'_t[k]$, with $k<\tau$, $S_t[k]\neq S'_t[k]$, it holds that $\tau-k+1\leq w$.
	\end{enumerate}
\end{definition}

\begin{definition}($(\tau,w)$-forward neighboring stream prefixes).
	Let $w$ be a positive integer. Two stream prefixes $S_{t}$, $S'_{t}$ are $(\tau,w)$-forward neighboring (denoted as $S_{t}\sim_{F,\tau,w} S'_{t}$), if 
	\begin{enumerate}
		\item for each $S_{t}[k]$, $S'_{t}[k]$ such that $k\in[t]$ and $S_{t}[k]\neq S'_{t}[k]$, it holds that $S_{t}[k]$ and $S'_{t}[k]$ are neighboring, and
		\item for each $S_t[k]$, $S'_t[k]$, with $k>\tau$, $S_t[k]\neq S'_t[k]$, it holds that $k-\tau+1\leq w$.
	\end{enumerate}
\end{definition}

\begin{remark}(Equivalence of forward and shifted backward neighboring).
	For any $\tau$ and $w$, the $(\tau,w)$-forward neighboring relation is equivalent to the $(\tau+w-1,w)$-backward neighboring relation, since both restrict the changed events to the same window $[\tau, \tau+w-1]$. We nevertheless keep both notations because they correspond to two semantically different types of dynamic privacy requirements in our model: a backward requirement is anchored at the current time slot and constrains privacy loss accumulated from the recent past, whereas a forward requirement is declared at the current time slot and constrains feasible future releases in the upcoming window. This distinction is convenient for stating user requirements and for presenting the online feasibility rules in Section~\ref{advanced_method}.
\end{remark}

\subsection{Differential Privacy}
There are two parts in the differential privacy paradigm: a large number of respondents (data owners) and a trusted curator (server). 
The goal of differential privacy mechanisms is to publish statistic of $D$ without compromising the privacy of respondents.

\noindent
\textbf{Threat Model.} This work follows the centralized-DP paradigm. A trusted curator collects the stream data $x_{i,t}$ and the corresponding personalized privacy requirements from users, and releases only perturbed aggregate statistics $r_t$. The adversary is any party that observes the released sequence $(r_1,r_2,\dots)$. As is standard in differential privacy, we assume the adversary knows the mechanism, the neighboring relation, and all public system parameters, and may have arbitrary side information. Our privacy objective is therefore to bound the information leaked about any user’s contribution within the relevant event window through the released outputs. This trusted-curator assumption is consistent with prior centralized $w$-event stream release methods reviewed in Section~\ref{relatedwork}.

\begin{definition}($\epsilon$-differential privacy~\cite{DBLP:conf/icalp/Dwork06,DBLP:journals/pvldb/KellarisPXP14}).
	A mechanism $\entity{M}:\entity{D}\to\entity{O}$ satisfies $\epsilon$-differential privacy (or $\epsilon$-DP), where $\epsilon\geq 0$ if for all sets $O\subseteq\entity{O}$, and every pair of neighboring databases $D,D'\in\entity{D}$, it holds that
	\begin{equation}\notag
		\outereqsizelarge{
			\begin{aligned}
				\Pr[\entity{M}(D)\in O]\leq e^{\epsilon}\cdot\Pr[\entity{M}(D')\in O].
			\end{aligned}
		}
	\end{equation}
\end{definition}

\begin{definition}($\vectorfont{\epsilon}$-Personalized Differential Privacy~\cite{DBLP:conf/icde/JorgensenYC15}).
	Given a set of users $U=\{u_1,\ldots,u_n\}$ with privacy requirements (preferences) $\vectorfont{\epsilon}=\{\epsilon_1,\ldots,\epsilon_n\}$,
	a randomized mechanism $\entity{M}:\entity{D}\to\entity{O}$ satisfies $\vectorfont{\epsilon}$-personalized differential privacy (or $\vectorfont{\epsilon}$-PDP), if for every pair of neighboring datasets $D,D'\subseteq\entity{D}$ with $D$, $D\stackrel{tp_i}{\sim} D'$~\cite{DBLP:conf/icde/JorgensenYC15}, and for all sets $O\subseteq\entity{O}$ of possible outputs, it holds that
	\begin{equation}
		\outereqsizelarge{
			\begin{aligned}
				\Pr[\entity{M}(D)\in O]\leq e^{\epsilon_i}\cdot\Pr[\entity{M}(D')\in O],		
			\end{aligned}
		}
	\end{equation}
	where $tp_i$ is the tuple set associated with $u_i$. 
\end{definition}

\begin{definition}(\constantPrivacyLevelTotalName{})\label{Def_EPDP}.
	Let $\entity{M}$ be a mechanism that takes a stream prefix of arbitrary size as input.
	Let $\entity{O}$ be the set of all possible outputs of $\entity{M}$.
	Let $U$ be the set of all users.
	$\entity{M}$ is \constantPrivacyLevelTotalName{} (or \constantPrivacyLevelSimpleName{}) if $\forall O\subseteq\entity{O}$, $\forall i\in[|U|]$ with $w_i\in\vectorfont{w}$ and $\algvar{E}_i\in\vectorfont{\algvar{E}}$ and $\forall S_t, S'_t$ satisfying $S_t \sim_{w_i} S'_t$, it holds that
	\begin{equation}
		\outereqsizelarge{
			\begin{aligned}
				\Pr[M(S_t)\in O]\leq e^{\algvar{E}_i} \Pr[M(S'_t)\in O].
			\end{aligned}
		}
	\end{equation}
\end{definition}
When $\vectorfont{w}=\vectorfont{1}$, \constantPrivacyLevelSimpleName{} simplifies to $\vectorfont{\algvar{E}}$-PDP~\cite{DBLP:conf/icde/JorgensenYC15}.
\begin{definition}(\dynamicPrivacyLevelTotalName{})\label{Def_BEPDP}.
	Let $\entity{M}$ be a mechanism that takes a stream prefix of arbitrary size as input.
	Let $\entity{O}$ be the set of all possible outputs of $\entity{M}$.
	Let $U$ be the set of all users.
	Then $\entity{M}$ is \dynamicPrivacyLevelTotalName{} (or \dynamicPrivacyLevelSimpleName{}) if $\forall O\subseteq\entity{O}$, $\forall i\in[|U|]$ with $(w_{B,i}, w_{F,i}, \algvar{E}_{B,i}, \algvar{E}_{F,i})\in (\vectorfont{w}_B, \vectorfont{w}_F, \vectorfont{\algvar{E}}_B, \vectorfont{\algvar{E}}_F)$ and $\forall S_{t}, S'_{t}$ satisfying $S_{t}\sim_{B,\tau,w_{B,i}} S'_{t}$ and $S_{t}\sim_{F,\tau,w_{F,i}} S'_{t}$, it holds that
	\begin{equation}
		\outereqsizelarge{
			\begin{aligned}\notag
				\Pr[M(S_{t})\in O]\leq e^{\algvar{E}_{B,i}+\algvar{E}_{F,i}} \Pr[M(S'_{t})\in O].		
			\end{aligned}
		}
	\end{equation}
\end{definition}

\noindent
\textbf{Scope and Limitation of Protection.} The proposed notions \constantPrivacyLevelSimpleName{} and \dynamicPrivacyLevelSimpleName{} protect user data contributions, rather than the privacy preferences themselves. In particular, the personalized budgets and window sizes are treated as mechanism parameters. These notions inherit the bounded-window semantics of classical $w$-event privacy: they protect each user’s contribution only within the corresponding fixed or dynamic event window through the released outputs, but do not provide full trajectory-level privacy over an unbounded stream. Therefore, our goal is not to resolve the general trajectory-level limitation of $w$-event privacy, but to extend the $w$-event paradigm to heterogeneous and time-varying personalized privacy requirements. The trusted curator assumption is standard in centralized stream release; protecting against an untrusted curator is outside the scope of this paper and belongs to local/shuffled privacy settings.

\subsection{Personalized Privacy Requirement}

In this paper, we consider two kinds of personalized privacy requirements.

\noindent
\textbf{Fixed Personalized Privacy Requirement.} 
For any user $u_i$, they have a privacy level requirement $\algvar{E}_i$ within a specific window size $w_i$ meaning the privacy level in this window should achieve $\algvar{E}_i$-DP.
We define $w_i$ as the \textit{fixed personalized window size requirement} and $\algvar{E}_i$ as the \textit{fixed personalized privacy budget requirement}.
Together, the pair ($w_i$, $\algvar{E}_i$) constitutes the \textit{fixed personalized privacy requirement}. 

\noindent
\textbf{Dynamic Personalized Privacy Requirement.}
For any user $u_i$ at time slot $t$,  there are two privacy requirements: a dynamic backward requirement and a dynamic forward requirement.
The backward requirement specifies a privacy level $\algvar{E}_{B,i,t}$ within a dynamic backward window size  $w_{B,i,t}$, 
with the window ending at time slot $t$ to achieve $\algvar{E}_{B,i,t}$-DP.
The forward requirement specifies a privacy level $\algvar{E}_{F,i,t}$ within a dynamic forward window size $w_{F,i,t}$, 
with the window beginning at time slot $t$ to achieve $\algvar{E}_{F,i,t}$-DP.
These components constitute $u_i$'s \textit{backward window size requirement at time slot $t$}, \textit{backward privacy budget requirement at time slot $t$}, \textit{forward window size requirement at time slot $t$}, and \textit{forward privacy budget requirement at time slot $t$}.
Together, the pairs $(w_{B,i,t},\algvar{E}_{B,i,t})$ and $(w_{F,i,t},\algvar{E}_{F,i,t})$ form $u_i$'s \textit{backward privacy requirement at time slot $t$} and \textit{forward privacy requirement at time slot $t$}, respectively. 
Although the forward requirement can be equivalently rewritten as a shifted backward requirement, we keep the forward form because it directly captures the user’s requirement declared at time slot $t$ for the future window starting from $t$.

\subsection{\problemDefineEnhancedSimpleName{}}
Given a data stream $S$, the analyst aims to obtain the data stream count (i.e., original count) publishing as $\vectorfont{c}=[\vectorfont{c}_1, \vectorfont{c}_2, \ldots]$.
To protect user privacy, however, the analyst only receives the obfuscated data stream count (i.e., estimation count) $\vectorfont{r}=[\vectorfont{r}_1, \vectorfont{r}_2, \ldots]$. 
The goal of the problem is to minimize the difference between the estimation count and the original count while meeting the personalized privacy requirement.
We present our problem definition as follows.

\begin{definition} (\problemDefineEnhancedTotalName{})\label{Problem_Enhanced_Def}. 
	Given a user set $U=\{u_1,u_2,\ldots,u_n\}$ where each $u_i$ holds a data collection $(w_{B,i,t},\algvar{E}_{B,i,t},w_{F,i,t},\algvar{E}_{F,i,t},\vectorfont{x}_{i,t})$ at time slot $t$. All $\vectorfont{x}_{i,t}$ for $u_i\in U$ at time slot $t$ form $D_t$. All $D_t$ constitute an infinite data stream $D=[D_1, D_2, \ldots]$. \problemDefineEnhancedTotalName{} (or \problemDefineEnhancedSimpleName{}) is to release an obfuscated histogram $\vectorfont{r}=[\vectorfont{r}_1, \vectorfont{r}_2,\ldots]$ of $D$ in each time slot $t$ achieving $(t,\vectorfont{w}_{B,t},\vectorfont{w}_{F,t},\vectorfont{\algvar{E}}_{B,t},\vectorfont{\algvar{E}}_{F,t})$-EPDP with the difference between $\vectorfont{r}$ and $\vectorfont{c}$ minimized where $\vectorfont{w}_{B,t}=[w_{B,1,t},\ldots,w_{B,n,t}]$, $\vectorfont{\algvar{E}}_{B,t}$ $=[\algvar{E}_{B,1,t},\ldots,\algvar{E}_{B,n,t}],\vectorfont{w}_{F,t}=[w_{F,1,t},\ldots,w_{F,n,t}]$, $\vectorfont{\algvar{E}}_{F,t}=[\algvar{E}_{F,1,t},\ldots,\algvar{E}_{F,n,t}]$. Namely,
	
	\begin{equation}\notag
		\outereqsizelarge{
			\begin{aligned}
				\min_{\epsilon_\theta}\;\; & \sum\limits_{t\in[T]} \|\vectorfont{r}_t-\vectorfont{c}_t\|_2^2  &\\
				s.t.\;\; & \sum_{k=t-w_{B,i,t}+1}^{t}\epsilon_{i,k}\leq \algvar{E}_{B,i,t}, \;\;\;\; \forall u_i\in U & \\
				\;\;\;\; &\;\; \sum_{k=t}^{t+w_{F,i,t}-1}\epsilon_{i,k}\leq \algvar{E}_{F,i,t}, \;\;\;\; \forall u_i\in U & \\
			\end{aligned}
		}
	\end{equation}
	where $\epsilon_{i,k}$ indicates the privacy budget cost at time slot $k$.
\end{definition}

\section{\solutionATotalName{}}\label{basic_method}
In this section, we consider the fixed personalized setting, where each user maintains the same privacy requirement across all time slots. Although this setting is simpler than the dynamic case, it provides the core release components needed by our general framework. In particular, it allows us to introduce the basic mechanism for transforming heterogeneous budgets into a system-level release decision. We refer to this fixed problem as \problemDefineTotalName{} (\problemDefineSimpleName{}).

To maximize estimation accuracy at each time slot, we first analyze the reporting error and then introduce Optimal Budget Selection (OBS), a basic component for determining a release threshold under heterogeneous privacy budgets. Based on OBS, we develop the \solutionATotalName{} (\solutionA{}), which transforms heterogeneous personalized privacy requirements into a system-level release decision for the fixed personalized setting.

\subsection{Problem Simplifying}
\begin{definition} (\problemDefineSimpleName{}).\label{Problem_Def} 
	Given a user set $U=\{u_1,u_2$$,$ $\ldots$, $u_n\}$, each $u_i$ holds a privacy requirement $(w_{i}, \algvar{E}_{i})$ and a series data $\vectorfont{x}_{i,t}$ for $t\in\constvar{N}^+$. All the $\vectorfont{x}_{i,t}$ for $u_i\in U$ at time slot $t$ form $D_t$. All the $D_t$ form an infinite data stream $S=[D_1, D_2, \ldots]$. \problemDefineSimpleName{} is to publish an obfuscated histogram $\vectorfont{r}=[\vectorfont{r}_1, \vectorfont{r}_2, \ldots]$ of $S$ at each time slot $t$ achieving $(\vectorfont{w},\vectorfont{\algvar{E}})$-EPDP with the distance between $\vectorfont{r}$ and $\vectorfont{c}$ minimized, namely $\forall T\in\constvar{N}^+$:
	\begin{equation}\notag
		\outereqsizelarge{
			\begin{aligned}
				\min_{\epsilon_\theta}\;\; & \sum\limits_{t\in[T]} \|\vectorfont{r}_t-\vectorfont{c}_t\|_2^2 &\\
				s.t.\;\; & \sum_{\tau=\max{(t-w_{i}+1,1)}}^{t}\epsilon_{i,\tau}\leq \algvar{E}_{i}, \;\;\;\; \forall u_i\in U & 
			\end{aligned}
		}
	\end{equation}
	where $\epsilon_{i,\tau}$ indicates the privacy budget cost at time slot $\tau$.
\end{definition}

\begin{proposition}
	The fixed personalized setting is a special case of the dynamic personalized setting when, for every user $u_i$ and every time slot $t$, the privacy requirements remain constant over time. That is, $(w_{B,i,t}, \algvar{E}_{B,i,t},w_{F,i,t}, \algvar{E}_{F,i,t})$ reduce to fixed user-specific parameters $(w_i,\algvar{E}_i)$. Moreover, when all users share the same fixed privacy requirement, the model further reduces to the classical homogeneous $w$-event setting.
\end{proposition}
\begin{proof} (sketch).
	When privacy requirements do not vary over time, the dynamic feasibility constraints become fixed personalized budget constraints, and \problemDefineEnhancedSimpleName{} reduces to \problemDefineSimpleName{}. If all users further share the same privacy budget and event window, the personalized constraints collapse into a single homogeneous constraint, which recovers the classical $w$-event setting.
\end{proof}

\subsection{Reporting Errors}
Privacy budget allocation can be determined for any type of privacy requirement at each time slot.
For time slot $t$ with privacy budget allocation $\vectorfont{\epsilon}=\{\epsilon_1,\ldots\epsilon_n\}$, we use the Sampling Mechanism (SM)~\cite{DBLP:conf/icde/JorgensenYC15} to satisfy all users' privacy requirements (i.e., achieving $\vectorfont{\epsilon}$-PDP). SM operates in two steps: \textit{sample} ($\hbox{SM}_{s}$) and \textit{disturb} ($\hbox{SM}_{d}$). 
In $\hbox{SM}_s$, the server sets a privacy budget threshold $\epsilon_\theta$ and constructs a sampling subset $D_S$. 
Specifically, it adds items $x_i$ with $\epsilon_i \geq \epsilon_\theta$ directly to $D_S$, while sampling other items $x_i$ with $\epsilon_i < \epsilon_\theta$ at a probability of $p_i= \frac{e^{\epsilon_i} - 1}{e^{\epsilon_\theta} - 1}$. 
In $\hbox{SM}_d$, the server uses a DP mechanism (e.g., the Laplace Mechanism) to generate an obfuscated result that achieves $\epsilon_\theta$-DP.

SM introduces two types of errors: \textit{sampling error} ($\hbox{\em err}_{s}$) and \textit{noise error} ($\hbox{\em err}_{\hbox{\scriptsize\em dp}}$).
Given a privacy budget threshold $\epsilon_\theta$, $\hbox{\em err}_{s}(\epsilon_\theta)$ occurs when sampling users with privacy budgets below $\epsilon_\theta$, while $\hbox{\em err}_{\hbox{\scriptsize\em dp}}(\epsilon_\theta)$ results from adding noise to achieve $\epsilon_\theta$-DP.  The sum of these two errors constitutes the total reporting error.
Next, we introduce these sampling and noise errors in detail.

\begin{definition}(Sampling Error~\cite{DBLP:conf/icde/JorgensenYC15}).
	Given a privacy budget threshold $\epsilon_\theta$ and $m$ distinct privacy budgets $\tilde{\epsilon}_1, \tilde{\epsilon}_2,\ldots , \tilde{\epsilon}_m$ from $n$ users with $\tilde{\epsilon}_i<\tilde{\epsilon}_j$ for $i<j$ and $i, j\in[m]$ where $\tilde{\epsilon}_i$ is declared by $n_i$ users~$\left(\sum\limits_{i=1}^{m}n_i=n\right)$,
	the sampling error $\hbox{\em err}_{s}(\epsilon_\theta)$ is defined as
	\begin{equation}\notag
		\outereqsizelarge{
			\begin{aligned}
				\hbox{\em err}_{s}(\epsilon_\theta) &= \hbox{\em Var}(\hbox{\em count}(\vectorfont{r}_t)) + \hbox{\em bias}(\vectorfont{r}_t)^2\\
				&= \sum\limits_{\tilde{\epsilon}_i<\epsilon_\theta}n_i p_i(1-p_i) + \left(\sum\limits_{\tilde{\epsilon}_i<\epsilon_\theta}n_i(1-p_i)\right)^2,
			\end{aligned}
		}
	\end{equation}
	where $p_i=\frac{e^{\tilde{\epsilon}_i}-1}{e^{\epsilon_\theta}-1}$.
\end{definition}
\begin{definition}(Noise Error~\cite{DBLP:conf/icde/JorgensenYC15}).
	The noise error $\hbox{\em err}_{\hbox{\scriptsize\em dp}}(\epsilon_{\theta})$ is defined as the error of the Laplace mechanism, namely, 
	\begin{equation}\notag
		\outereqsizelarge{
			\begin{aligned}
				\hbox{\em err}_{\hbox{\scriptsize\em dp}}(\epsilon_{\theta}) &= \frac{2}{\epsilon_\theta^2}.
			\end{aligned}
		}
	\end{equation}
\end{definition}
Various metrics exist to measure the errors of Laplace mechanisms for noise error, including variance \cite{DBLP:conf/icde/JorgensenYC15,DBLP:conf/sigmod/RenSYYZX22}, scale \cite{DBLP:journals/pvldb/KellarisPXP14,DBLP:journals/fttcs/DworkR14}, and $(\alpha,\beta)$-usefulness \cite{DBLP:journals/fttcs/DworkR14,DBLP:conf/stoc/BlumLR08}. In this work, we employ variance as our metric.


\vspace{-2em}
\subsection{Optimal Budget Selection} 
Given the budget allocation $(\epsilon_{1,t}, \epsilon_{2,t}, \ldots, \epsilon_{n,t})$ of $n$ users, we can determine the frequency of each distinct budget and select the optimal $\epsilon_\theta$ that minimizes the data reporting error $\hbox{\em err}$. This process is detailed in Algorithm~\ref{alg:OPT_B_C}.

Taking $n$ privacy budgets as input, the Optimal Budget Selection (OBS) algorithm first counts the distinct privacy budgets (Lines \ref{e_set}-\ref{n_set}).
It then finds the minimum reporting error $\hbox{\em err}_{\hbox{\scriptsize\em min}}$ (lines \ref{start_search}-\ref{end_search}).
Specifically, it iterates through all distinct privacy budgets $\tilde{\epsilon}_k\in\tilde{\vectorfont{\epsilon}}$ and identifies the value $\tilde{\epsilon}_k$ that produces the smallest total error $\hbox{\em err}=\hbox{\em err}_{s}(\tilde{\epsilon}_k)+\hbox{\em err}_{\hbox{\scriptsize\em dp}}(\tilde{\epsilon}_k)$. This value and its error are returned as the optimal privacy budget $\epsilon_{\hbox{\scriptsize\em opt}}$ and the minimum error $\hbox{\em err}_{\hbox{\scriptsize\em min}}$ .
\begin{algorithm}[t!]\small
	\caption{Optimal Budget Selection (OBS)}
	\label{alg:OPT_B_C}
	\DontPrintSemicolon
	\KwIn{personalized privacy budget list $\vectorfont{\epsilon}=(\epsilon_1,\epsilon_2,\ldots , \epsilon_n)$}
	\KwOut{$\epsilon_{\hbox{\scriptsize\em opt}}$, $\hbox{\em err}_{\hbox{\scriptsize\em min}}$}
	Extract distinct budget set $\tilde{\vectorfont{\epsilon}}=(\tilde{\epsilon}_1,\tilde{\epsilon}_2,\ldots ,\tilde{\epsilon}_{\tilde{n}})$ from $\vectorfont{\epsilon}$;\\ \label{e_set}
	Count the frequency $n_k$ of all $\tilde{\epsilon}_k\in\tilde{\vectorfont{\epsilon}}$;\\ \label{n_set}
	Initialize $\hbox{\em err}_{\hbox{\scriptsize\em min}}$ as the upper bound of error value;\\
	\For{$\tilde{\epsilon}_k\in\tilde{\vectorfont{\epsilon}}$}{\label{start_search}
		$\hbox{\em err}\gets \hbox{\em err}_{s}(\tilde{\epsilon}_k)+\hbox{\em err}_{\hbox{\scriptsize\em dp}}(\tilde{\epsilon}_k)$;\\
		\If{$\hbox{err}<\hbox{err}_{\hbox{\scriptsize min}}$}{
			$\hbox{\em err}_{\hbox{\scriptsize\em min}}\gets \hbox{\em err}$;\\
			$\epsilon_{\hbox{\scriptsize\em opt}}\gets \tilde{\epsilon}_k$;\\ 
		}
	}\label{end_search}
	\Return{$\epsilon_{\hbox{\scriptsize opt}}$, $\hbox{\em err}_{\hbox{\scriptsize min}}$}\;
\end{algorithm}

\begin{example}[Running Example of the OBS Algorithm]
	Suppose we have $10$ privacy budgets as input: $\vectorfont{\epsilon} = ($0.1, 0.4, 0.4, 0.1, 0.4, 0.4, 0.8, 0.8, 0.8, 0.4$)$.
	OBS first determines $\tilde{\vectorfont{\epsilon}}=($0.1, 0.4, 0.8$)$, $\tilde{n}=|\tilde{\vectorfont{\epsilon}}|=3$, and $N=($2, 5, 3$)$.
	Based on these statistics, OBS iterates through the $3$ privacy budgets in $\tilde{\vectorfont{\epsilon}}$ and calculates their errors:
	\innereqsize{$err_1=0+\frac{2}{0.1^2}=200$}, \innereqsize{$err_2=2\times\frac{e^{0.1}-1}{e^{0.4}-1}\times\left(1-\frac{e^{0.1}-1}{e^{0.4}-1}\right)+\left(2\times\left(1-\frac{e^{0.1}-1}{e^{0.4}-1}\right)\right)^2+\frac{2}{0.4^2}=15.31$} and \innereqsize{$err_3=2\times\frac{e^{0.1}-1}{e^{0.8}-1}\times\left(1-\frac{e^{0.1}-1}{e^{0.8}-1}\right)+5\times\frac{e^{0.4}-1}{e^{0.8}-1}\times\left(1-\frac{e^{0.4}-1}{e^{0.8}-1}\right)+\left(2\times\left(1-\frac{e^{0.1}-1}{e^{0.8}-1}\right)+5\times\left(1-\frac{e^{0.4}-1}{e^{0.8}-1}\right)\right)^2+\frac{2}{0.8^2}=89.74$}.
	Finally, OBS returns $0.4$ with the minimum error $15.31$.
\end{example}

\subsection{\solutionATotalName{}}
In real applications,  to realize personalized privacy protection, the system needs to collect users' privacy requirements. 
To accomplish this, system administrators first define a discretized privacy budget range (e.g., \{$0.1, 0.5, 0.9$\}) and a window size range (e.g., \{$40, 80, 120$\}). 
Then, they map ascending privacy budget values to descending privacy budget levels (e.g., High, Medium, Low) and ascending window size values to ascending window size levels (e.g., Small, Medium, Large). 
Users can then select both a privacy budget level and a window size level based on their needs and past experience. 
Once users submit these selections, the server converts them into the corresponding values.

After receiving all users’ privacy requirements, the system must determine how to allocate privacy budgets within each user’s feasible window while maximizing estimation accuracy. Classical budget-division methods~\cite{DBLP:journals/pvldb/KellarisPXP14,DBLP:conf/sigmod/RenSYYZX22}  are designed for homogeneous $w$-event privacy, where all users share the same privacy budget and the same event window. In our setting, however, users may specify different window sizes and privacy budgets, whereas the system still publishes one shared aggregate result at each time slot.

Therefore, the key issue is no longer how to allocate a single global privacy budget sequence, but how to transform heterogeneous feasible budgets into a valid system-level release decision. To address this challenge, we propose the Personalized Window Size Mechanism (PWSM). The core idea is to use OBS together with the sampling mechanism to convert heterogeneous user-specific privacy budgets into a common release threshold, and then use that threshold to perform personalized dissimilarity estimation and adaptive release.

\noindent
\textbf{Personalized Private Dissimilarity Measure.} The personalized dissimilarity measure $\hbox{\em dis}^*$ is defined as the absolute error between the true statistic $\tilde{\vectorfont{c}}_t$ under $\hbox{SM}_s$ (i.e., the \textit{sample} step of SM) at current time slot $t$ and the last publication $\vectorfont{r}_l$, namely,
\begin{equation}
	\begin{aligned}\notag
		\outereqsizelarge{
			\hbox{\em dis}^* = \frac{1}{d}\sum\limits_{k=1}^{d}|\tilde{\vectorfont{c}}_t[k] - \vectorfont{r}_l[k]|.
		}
	\end{aligned}
\end{equation}
Our goal is to privately obtain the personalized dissimilarity $\hbox{\em dis}^*$ using the optimal privacy budget $\epsilon_{\hbox{\scriptsize\em opt}}$ calculated through OBS algorithm. 
The personalized private dissimilarity measure $\hbox{\em dis}$ is then defined as:
\begingroup
\footnotesize
\setlength{\abovedisplayskip}{3pt plus 1pt minus 1pt}
\setlength{\belowdisplayskip}{3pt plus 1pt minus 1pt}
\setlength{\abovedisplayshortskip}{2pt plus 1pt minus 1pt}
\setlength{\belowdisplayshortskip}{3pt plus 1pt minus 1pt}
\begin{equation}\notag
	\hbox{\em dis}	= 	\hbox{\em dis}^* 	+ 	Lap\left(\frac{1}{d\cdot\epsilon_{\hbox{\scriptsize\em opt}}}\right),
\end{equation}
\endgroup
where $Lap$ denotes the Laplace noise in the Laplace mechanism~\cite{DBLP:journals/fttcs/DworkR14}.

Based on the above observation, we introduce \solutionA{} as a framework for unifying heterogeneous personalized privacy budgets into a shared release decision.
As shown in Algorithm~\ref{alg:PWSM},
the \solutionA{} algorithm takes the historical estimation $\hbox{\em His}$ and the current personalized privacy requirement $(\vectorfont{w}_t,\vectorfont{\algvar{E}}_t)$ as the input. 
\solutionA{} first calculates all users' budget allocations $\vectorfont{\epsilon}_t$ at the current time slot $t$ on the premise of satisfying $(\vectorfont{w}_t,\vectorfont{\algvar{E}}_t)$-EPDP (line~\ref{currentPB}).
It then divides $\vectorfont{\epsilon}_t$ into two parts: calculation budget $\vectorfont{\epsilon}_t^{(1)}$ and publication budget $\vectorfont{\epsilon}_t^{(2)}$ (line~\ref{dividePB}).
Using $\vectorfont{\epsilon}_t^{(1)}$, \solutionA{} calculates the personalized private dissimilarity $\hbox{\em dis}$ between the current count value and the last reported one (line~\ref{calculate:dis}).
Next, it sets the change threshold as the reporting error $\hbox{\em err}$ calculated with $\vectorfont{\epsilon}_t^{(2)}$ (line~\ref{calculate:err}).
Finally, \solutionA{} adaptively decides whether to publish a new obfuscated estimation or skip (i.e., use the last published one to approximate) by comparing $\hbox{\em dis}$ to $\sqrt{\hbox{\em err}}$ (lines~\ref{begin:judge_report}-\ref{end:judge_report}). 

\begin{algorithm}[t!]\small
	\caption{\solutionATotalName{} (\solutionA{})}
	\label{alg:PWSM}
	\DontPrintSemicolon
	\KwIn{historical estimation $\hbox{\em His}$,  privacy requirement ($\vectorfont{w}_t$, $\vectorfont{\algvar{E}}_t$) at time slot $t$} 
	\KwOut{$\vectorfont{r}$}
	Calculate the current privacy budgets $\vectorfont{\epsilon}_t$ of all users according to $\vectorfont{\algvar{E}}_t$ and $\vectorfont{w}_t$;\\ \label{currentPB}
	Split $\vectorfont{\epsilon}_t$ into two components: the calculation budget $\vectorfont{\epsilon}^{(1)}_t$ and the publication budget $\vectorfont{\epsilon}^{(2)}_t$ satisfying $\vectorfont{\epsilon}_t=\vectorfont{\epsilon}^{(1)}_t+\vectorfont{\epsilon}^{(2)}_t$;\\\label{dividePB}
	Calculate dissimilarity $dis$ between current estimation and the last estimation by $\hbox{SM}(\vectorfont{\epsilon}^{(1)}_t)$;\label{calculate:dis}\\
	Calculate the reporting error $\hbox{\em err}$ of current estimation by $\hbox{OBS}\left(\vectorfont{\epsilon}^{(2)}_t\right)$;\\ \label{calculate:err}
	\If{$\hbox{\em dis}>\sqrt{\hbox{\em err}}$}{\label{begin:judge_report}
		Calculate current estimation $\vectorfont{r}$ by $\hbox{SM}\left(\vectorfont{\epsilon}^{(2)}_t\right)$;\\
	} \Else {
		Set current estimation $\vectorfont{r}$ as the last reporting value;\\
	}
	\Return{$\vectorfont{r}$}.\label{end:judge_report}
\end{algorithm}

Next, we present two methods based on \solutionA{}: \solutionMethodATotalName{} (\solutionMethodA{}) and \solutionMethodBTotalName{} (\solutionMethodB{}), each designed to handle different types of data streams.

\subsection{\solutionMethodATotalName{} and \solutionMethodBTotalName{}}

\noindent\textbf{Basic notations.} Before describing our personalized methods, we first need to declare some important notations specific to these methods.

\begin{definition}(Null/Non-null Publication).
	Given a sequence of publications $(r_1,r_2,\dots,r_t)$, a \textit{null publication} refers to approximating a historical value without consuming any privacy budget in Part$_{\hbox{\scriptsize NOP}}$, while a \textit{non-null publication} represents a new publication that consumes privacy budget in Part$_{\hbox{\scriptsize NOP}}$.
\end{definition}

For any time slot $2\leq\tau\leq t$, we refer to $r_{\tau-1}$ as the last reporting value (or last publication) of time slot $\tau$. 
In the sequence $(r_1,r_2,\dots,r_{\tau})$, we define the most recent non-null publication $r_l$ where $l<\tau$ as \textit{the last non-null publication}.

For example in Figure~\ref{fig:non-null-publication}, the publications at time slots $\tau, \tau+1, \tau+4$ are non-null publications, while those at $\tau+2$ and $\tau+3$ are null publications.
The last non-null publication at time slot $\tau+4$ is the publication at time slot $\tau+1$.

\begin{figure}[h]
	\centering
	\includegraphics[width=0.48\textwidth]{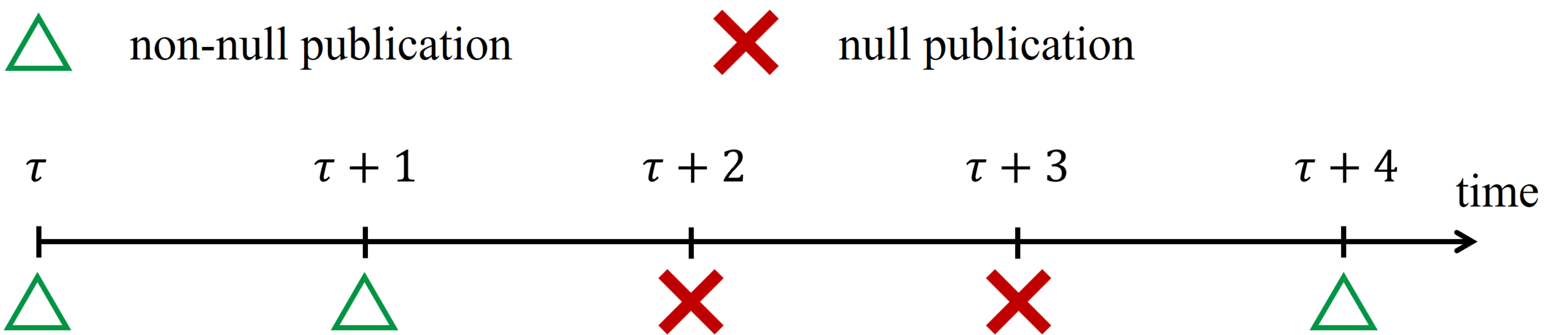}
	\vspace{3pt}
	\caption{A null/non-null publication example.}\label{fig:non-null-publication}
\end{figure}

\begin{definition}(Skipped/Nullified Publication).
	The skipped publications are those null publications with $\hbox{\em dis}\leq \sqrt{\hbox{\em err}}$.
	Given a privacy budget requirement $\algvar{E}$ and a window size $w$, a budget share $\bar{\epsilon}=\algvar{E}/w$ is defined as the average privacy budget per time slot. 
	When publishing new obfuscated data consumes $x$ budget shares ($x>1$), in order to maintain the average value, the following $x-1$ time slots values are approximated by the last publications. These $x-1$ time slots are defined as nullified time slots.
\end{definition}

\begin{figure}[h]
	\centering	\includegraphics[width=0.48\textwidth]{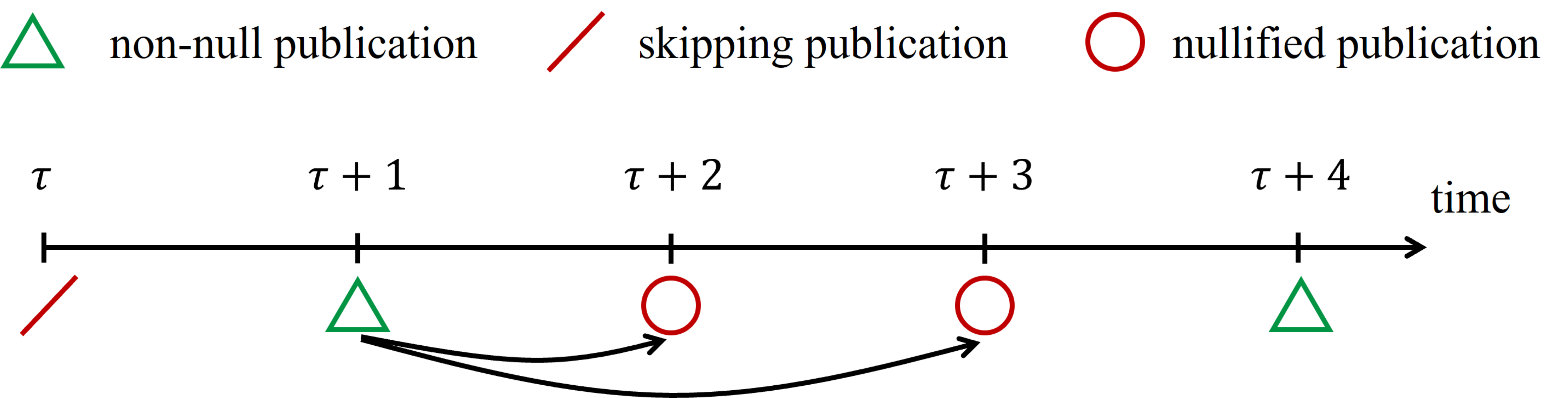}
	\vspace{3pt}
	\caption{A skipped/nullified publication example.}\label{fig:nullified-publication}
\end{figure}

We can see both skipped and nullified publications are non publications. 
Figure~\ref{fig:nullified-publication} illustrates an example for skipped and nullified publications. 
With a privacy budget $\algvar{E}$ of $4$ and a window size of $4$, the budget share $\bar{\epsilon}$ equals $\algvar{E}/w=1$. When time slot $\tau+1$ uses $3$ shares, the publications at time slots $\tau+2$ and $\tau+3$ become nullified publications.

\begin{algorithm}[t!]\small
	\caption{Dissimilarity Calculation (DC)}
	\label{alg:DC}
	\DontPrintSemicolon
	\KwIn{$D_t$, current personalized privacy budget list  $\vectorfont{\epsilon}_t$, historical data publication $(\vectorfont{r}_1,\vectorfont{r}_2,\ldots, \vectorfont{r}_{t-1})$}
	\KwOut{$\vectorfont{r}_t$}
	$\epsilon_{\hbox{\scriptsize\em opt}}\gets \hbox{OBS}(\vectorfont{\epsilon}_{t}$) ;\\ \label{mt1_start2}
	$\tilde{D}_t\gets \hbox{SM}_s(D_t,\vectorfont{\epsilon}_{t},\epsilon_{\hbox{\scriptsize\em opt}})$;\\ \label{pdp_start}
	$\tilde{\vectorfont{c}}_t\gets Q(\tilde{D}_t)$;\\
	Get the last non-null publication $\vectorfont{r}_l$ from $(\vectorfont{r}_1,\vectorfont{r}_2,\ldots , \vectorfont{r}_{t-1})$;\\
	\Return{$\hbox{\em dis}\gets \frac{1}{d}\sum_{j=1}^{d}|\tilde{\vectorfont{c}}_t[j]-\vectorfont{r}_l[j]|+Lap(1/(d\cdot\epsilon_{\hbox{\scriptsize opt}}))$};\\ \label{pdp_end}
\end{algorithm}

\begin{algorithm}[t!]\small
	\caption{\solutionMethodATotalName{} (\solutionMethodA{})}
	\label{alg:PBD}
	\DontPrintSemicolon
	\KwIn{$D_t$, privacy requirement set ($\vectorfont{w}$, $\vectorfont{\algvar{E}}$), historical data publication $(\vectorfont{r}_1,\vectorfont{r}_2,\ldots, \vectorfont{r}_{t-1})$}
	\KwOut{$\vectorfont{r}_t$}
	Calculate the current window average budget $\bar{\epsilon}_i\gets \algvar{E}_i/w_i$ for each $i\in[n]$;\\
	Set $\vectorfont{\epsilon}_{t}^{(1)}\gets (\bar{\epsilon}_1/2,\bar{\epsilon}_2/2,\ldots ,\bar{\epsilon}_n/2)$;\\ \label{mt1_start}
	$\hbox{\em dis}\gets\textrm{DC}\left(D_t, \vectorfont{\epsilon}_{t}^{(1)}, \vectorfont{r}_1,\vectorfont{r}_2,\ldots, \vectorfont{r}_{t-1}\right)$ by \textbf{Algorithm~\ref{alg:DC}};\\ \label{mt1_end}
	Set $\epsilon_{\hbox{\scriptsize\em rm},i}\gets\algvar{E}_i/2-\sum_{k=t-w_i+1}^{t-1}\epsilon_{i,k}^{(2)}$ for each $i\in[n]$;\\ \label{pbd_mt2_start}
	$\vectorfont{\epsilon}_{t}^{(2)}\gets(\epsilon_{\hbox{\scriptsize\em rm},1}/2,\epsilon_{\hbox{\scriptsize\em rm},2}/2,\ldots ,\epsilon_{\hbox{\scriptsize\em rm},n}/2)$;\\
	$\epsilon_{\hbox{\scriptsize\em opt}}^{(2)}$, $\hbox{\em err}_{\hbox{\scriptsize\em opt}}^{(2)}\gets\hbox{OBS}\left(\vectorfont{\epsilon}_{t}^{(2)}\right)$ by \textbf{Algorithm~\ref{alg:OPT_B_C}};\\ \label{pbd_mt2_mid}
	\If{$\hbox{\em dis}>\sqrt{\hbox{\em err}_{\hbox{\scriptsize opt}}^{(2)}}$}{
		$\tilde{D}_t^{(2)}\gets \hbox{SM}_s\left(D_t,\vectorfont{\epsilon}_{t}^{(2)},\epsilon_{\hbox{\scriptsize\em opt}}^{(2)}\right)$;\\ \label{new_pub_start}
		$\tilde{\vectorfont{c}}_t^{(2)}\gets Q\left(\tilde{D}_t^{(2)}\right)$;\\
		\Return{$\vectorfont{r}_t\gets \hbox{SM}_d\left(\tilde{\vectorfont{c}}_t^{(2)},\epsilon_{\hbox{\scriptsize opt}}^{(2)}\right)$}; \label{new_pub_end}
	} \Else {
		$\vectorfont{\epsilon}_{t}^{(2)}\gets (0,0,\ldots ,0)$;\\
		\Return{$\vectorfont{r}_t\gets\vectorfont{r}_{t-1}$};
	}\label{mt2_end}
\end{algorithm}
\noindent\textbf{\solutionMethodATotalName{} (\solutionMethodA{}).} As shown in Algorithm~\ref{alg:PBD}, \solutionMethodA{}  inputs the current user data, all users' fixed privacy requirements, and historical data publication.
The fixed privacy budget requirement $\algvar{E}_i$ of $u_i$ is split into two parts: 1) Part$_{\hbox{\scriptsize DC}}$ for calculating the dissimilarity between the current data and the last publication (Lines~\ref{mt1_start}-\ref{mt1_end}); 2) Part$_{\hbox{\scriptsize NOP}}$ for calculating the new obfuscated publication at the current time slot (Lines~\ref{pbd_mt2_start}-\ref{pbd_mt2_mid} and Lines~\ref{new_pub_start}-\ref{new_pub_end}).

In Part$_{\hbox{\scriptsize DC}}$, we allocate half of the average privacy budget per time slot for dissimilarity calculation (i.e., $\frac{\algvar{E}_i}{2w_i}$ for $u_i$). 
The process then calls the Dissimilarity Calculation (Algorithm~\ref{alg:DC}) to determine the dissimilarity. Within Algorithm~\ref{alg:DC}, the OBS algorithm selects the optimal budget threshold $\epsilon_{\hbox{\scriptsize\em opt}}$. Finally, it uses the SM~\cite{DBLP:conf/icde/JorgensenYC15} to compute the dissimilarity $\hbox{\em dis}$ (Lines \ref{pdp_start}-\ref{pdp_end}).
Notice that the remaining budget calculation in Line~\ref{pbd_mt2_start} is a standard sliding-window sum and can be maintained incrementally.

In Part$_{\hbox{\scriptsize NOP}}$, we first calculate the remaining privacy budget $\epsilon_{\hbox{\scriptsize\em rm},i}$ for each $u_i$.
We then set the publication privacy budget for each $u_i$ to half of $\epsilon_{\hbox{\scriptsize\em rm},i}$.
Similar to dissimilarity calculation, we use the OBS algorithm to determine the optimal privacy budget $\epsilon_{\hbox{\scriptsize\em opt}}^{(2)}$ and its corresponding error $\hbox{\em err}_{\hbox{\scriptsize\em opt}}^{(2)}$.
At this point, we have obtained two measurements: the dissimilarity $\hbox{\em dis}$ and the square root of error $\sqrt{\hbox{\em err}_{\hbox{\scriptsize\em opt}}^{(2)}}$.
We compare these two measurements to determine whether to publish a new obfuscated statistic result or approximate the current result with the last publication.
If the $\hbox{\em dis}$ is greater than $\sqrt{\hbox{\em err}_{\hbox{\scriptsize\em opt}}^{(2)}}$, it indicates that the difference between the current data and the last published data exceeds the error of noise, then we republish a new obfuscated statistic result.
Otherwise, we take the last publication instead.

\begin{figure}[t!]\vspace{-2ex}
\centering
\includegraphics[width=0.35\textwidth]{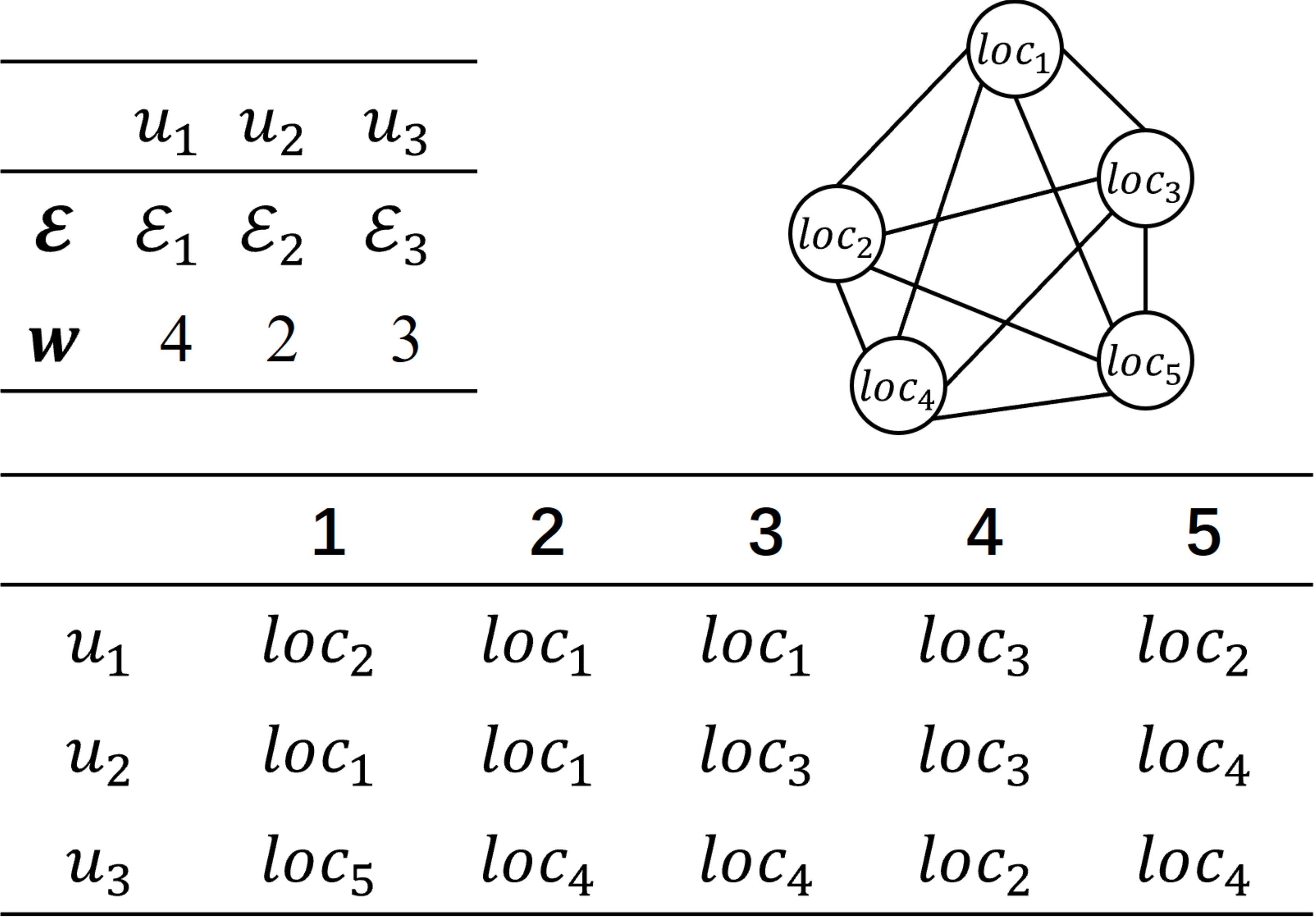}
\vspace{3pt}
\caption{An Information example for \solutionMethodA{}.}\label{data_example}
\end{figure}
\begin{example}\label{example_for_PBD}
Suppose there are $3$ users distributed across $5$ locations, forming a complete graph.
Figure~\ref{data_example} illustrates the fixed personalized privacy requirements and locations for the first three users across time slots $1$ to $5$.
Figure~\ref{PBD_example} demonstrates the estimation process of \solutionMethodA{}.
The total privacy budget for each user $u_i$ is evenly split into two parts, each containing $\algvar{E}_i/2$. 
The first part is allocated for dissimilarity calculation, while the second is for publication noise calculation.
For instance, $\algvar{E}_1$ is divided into $\vectorfont{\epsilon}_1^{(1)}(u_1)=\algvar{E}_1/2$ and $\vectorfont{\epsilon}_1^{(2)}(u_1)=\algvar{E}_1/2$.
We compute the privacy budget usage $\epsilon_{i,t}^{(1)}$ for dissimilarity and $\epsilon_{i,t}^{(2)}$ for obfuscated statistic publication for each user at each time slot. 
These values are recorded in an $n\times 2$  matrix at each time slot in Figure~\ref{PBD_example}.
Using $u_1$ as an example, $\epsilon_{1,t}^{(1)} = \vectorfont{\epsilon}_1^{(1)}(u_1)/w_1=\algvar{E}_1/8$. 
At time slot $1$, $\epsilon_{1,1}^{(2)}=\vectorfont{\epsilon}_1^{(2)}(u_1)/2=\algvar{E}_1/4$.
The algorithm calculates the dissimilarity $\hbox{\em dis}$ at time slot $1$ using all $\epsilon_{i,1}^{(1)}$, and the error $\hbox{\em err}_{\hbox{\scriptsize\em opt}}^{(2)}$ using all $\epsilon_{i,1}^{(2)}$.
Assume $\hbox{\em dis}>\sqrt{\hbox{\em err}_{\hbox{\scriptsize\em opt}}^{(2)}}$, then a new obfuscated statistic $\vectorfont{r}_1$ is published at time slot $1$.
At time slot $2$, assume $\hbox{\em dis}\leq\sqrt{\hbox{\em err}_{\hbox{\scriptsize\em opt}}^{(2)}}$, then $\epsilon_{i,2}^{(2)}$ is not used to publish a new obfuscated statistic result, and its usage is set to zeros for all users.
At time slot $3$,  $\epsilon_{1,3}^{(2)}=(\algvar{E}_{1}/2-\epsilon_{1,1}^{(2)})/2=\algvar{E}_{1}/8$.
The vector below each matrix in Figure~\ref{PBD_example} represents the total privacy budget used at the current time slot for each user. For example, at time slot $1$, the total privacy budget usage for $u_1$ is $\epsilon_{1,1}^{(1)}+\epsilon_{1,1}^{(2)}=3\algvar{E}_1/8$.
\begin{figure}[t!]
	\centering
	\includegraphics[width=0.48\textwidth]{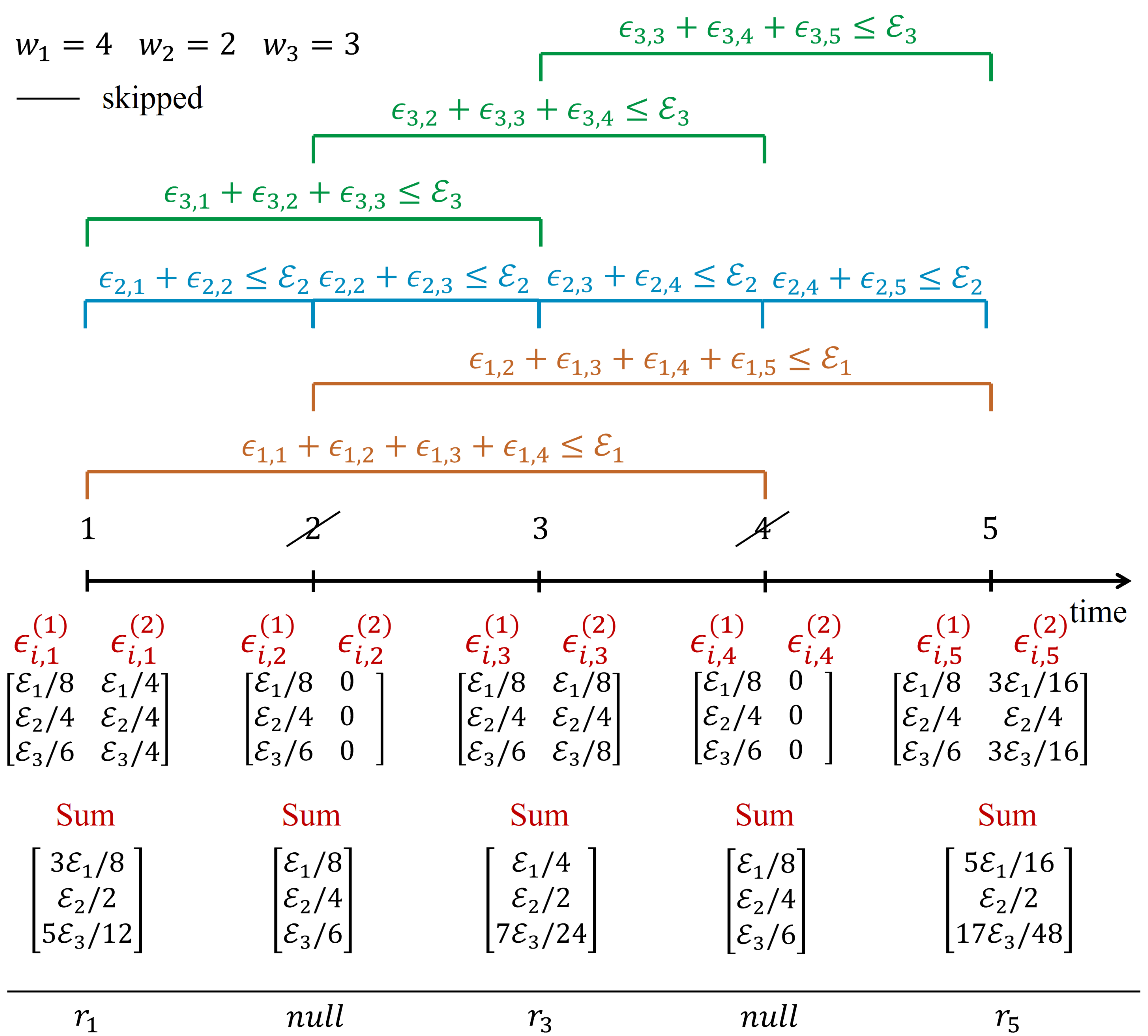}
	\caption{A process example for \solutionMethodA{}.}\label{PBD_example}
\end{figure}
\end{example}

\noindent\textbf{\solutionMethodBTotalName{} (\solutionMethodB{}).} Algorithm~\ref{alg:PBA} outlines the process of \solutionMethodB{}.
The dissimilarity calculation (Part$_{\hbox{\scriptsize DC}}$) in \solutionMethodB{} is identical to that of \solutionMethodA{}. However, \solutionMethodB{} and \solutionMethodA{} differ significantly in their strategies on allocating the publication privacy budget (Part$_{\hbox{\scriptsize NOP}}$).

For Part$_{\hbox{\scriptsize NOP}}$ in \solutionMethodB{}, we assume an average privacy budget of $\frac{\algvar{E}_{i}}{2w_i}$ (one share) for each $u_i$ at each time slot $t$. 
A publication at time slot $t$ can use more than one share by borrowing from its successor time slots. The variable $t_{i,N}$ in Line~\ref{nullified} represents the number of successor time slots occupied by the last publication.
We calculate the maximal $\tilde{t}_{N}$ of all $t_{i,N}$ and determine whether the current time has been occupied ($t-l\leq\tilde{t}_N$). If so, we approximate the publication using the last publication. Otherwise, we calculate the remaining budget shares from the precursor time slots (i.e., $t_{A,i}$ in Line~\ref{absorbed}) and set the current publication budget as the total absorbed shares (Line~\ref{publicationBudget}).
The subsequent steps follow the same process as outlined in Algorithm~\ref{alg:PBD}.
\begin{algorithm}[t!]\small
\caption{\solutionMethodBTotalName{} (\solutionMethodB{})}
\label{alg:PBA}
\DontPrintSemicolon
\KwIn{$D_t$, fixed personalized privacy requirement set ($\vectorfont{w}$, $\vectorfont{\algvar{E}}$), historical data publication $(\vectorfont{r}_1,\vectorfont{r}_2,\ldots, \vectorfont{r}_{t-1})$}
\KwOut{$\vectorfont{r}_t$}
Calculate the current window average budget $\bar{\epsilon}_i=\algvar{E}_i/w_i$ for each $i\in[n]$;\\
$\vectorfont{\epsilon}_{t}^{(1)}\gets(\bar{\epsilon}_1/2,\bar{\epsilon}_2/2,\ldots ,\bar{\epsilon}_n/2)$;\\ \label{mt22_start}
$\hbox{\em dis}\gets\textrm{DC}\left(D_t, \vectorfont{\epsilon}_{t}^{(1)},\vectorfont{r}_1,\vectorfont{r}_2,\ldots, \vectorfont{r}_{t-1}\right)$ by \textbf{Algorithm~\ref{alg:DC}};\\
\For{$i\in[n]$}{
	Initialize nullified time slots $t_{i,N}$ as $0$;\\
	Set $t_{i,N}\gets\frac{\epsilon_{i,l}^{(2)}}{\algvar{E}_i/(2w_i)}-1$ if $l$ exists where $l$ is the last non-null publication time slot;\\ \label{nullified}
}
Set nullified time slot bound $\tilde{t}_N\gets\max_{i\in[n]}t_{i,N}$;\\
\If{$t-l\leq\tilde{t}_{N}$}{
	\Return $\vectorfont{r}_t\gets\vectorfont{r}_{t-1}$;
} \Else{
	\For{$i\in[n]$}{
		Set absorbed time slots $t_{A,i}\gets\max{(t-l-t_{i,N},0)}$;\\	\label{absorbed}
		Set publication budget $\epsilon_{i,t}^{(2)}\gets\frac{\algvar{E}_i}{2w_i}\cdot\min{(t_{A,i}, w_i)}$;\\	\label{publicationBudget}
	}		
	$\vectorfont{\epsilon}_{t}^{(2)}\gets\left(\epsilon_{1,t}^{(2)},\epsilon_{2,t}^{(2)},\ldots ,\epsilon_{n,t}^{(2)}\right)$;\\
	$\epsilon_{\hbox{\scriptsize\em opt}}^{(2)}$, $\hbox{\em err}_{\hbox{\scriptsize\em opt}}^{(2)}\gets\textrm{OBS}\left(\vectorfont{\epsilon}_{t}^{(2)}\right)$;\\ \label{mt222_start}
	\If{$\hbox{\em dis}>\sqrt{\hbox{\em err}_{\hbox{\scriptsize opt}}^{(2)}}$}{
		$\tilde{D}_t^{(2)}\gets \hbox{SM}_s\left(D_t,\vectorfont{\epsilon}_{t}^{(2)},\epsilon_{\hbox{\scriptsize\em opt}}^{(2)}\right)$;\\ 
		$\tilde{\vectorfont{c}}_t^{(2)}\gets Q\left(\tilde{D}_t^{(2)}\right)$;\\
		\Return{$\vectorfont{r}_t\gets \hbox{SM}_d\left(\tilde{\vectorfont{c}}_t^{(2)},\epsilon_{\hbox{\scriptsize opt}}^{(2)}\right)$};
	} \Else {
		$\vectorfont{\epsilon}_{t}^{(2)}\gets (0,0,\ldots ,0)$;\\
		\Return{$\vectorfont{r}_t\gets\vectorfont{r}_{t-1}$};
	}\label{mt222_end}
}	
\end{algorithm}

\begin{example}\label{example_for_PBA}
We continue use the demonstration case shown in Figure~\ref{data_example}. Figure~\ref{PBA_example} illustrates the estimation process of \solutionMethodB{}.
The dissimilarity calculation process in \solutionMethodB{} is identical to that in Example~\ref{example_for_PBD}.
For Part$_{\hbox{\scriptsize NOP}}$, at time slot $1$, with no budget to absorb, all users utilize one share (i.e., $\algvar{E}_i/(2w_i)$) to publish a new obfuscated statistic result.
Assume time slot $2$ is skipped \Big(i.e., \innereqsize{$\hbox{\em dis}\leq\sqrt{\hbox{\em err}_{\hbox{\scriptsize\em opt}}^{(2)}}$}\Big). 
At time slot $3$, $t_{1,N}=t_{2,N}=t_{3,N}=0$. Thus, the nullified bound $\tilde{t}_N$ is $0$. 
Since $t-l=3-1=2>\tilde{t}_{N}$, a new obfuscated statistic result is reported.
The publication budget set is calculated as $\vectorfont{\epsilon}_3^{(2)}=\left(\algvar{E}_1/4,\algvar{E}_2/2,\algvar{E}_3/3\right)$.
At time slot $4$, $t_{1,N}=t_{2,N}=t_{3,N}=1$.
As $t-l=4-3=1\leq\tilde{t}_{N}$, the publication is approximated as the one at time slot $3$.
At time slot $5$, all $t_{i,N}$ remain $1$, and $t-l=5-3=2>\tilde{t}_{N}$. 
The absorbed time slots $t_{A,i}$ all equal $1$.
The publication budget set is $\vectorfont{\epsilon}_5^{(2)}=\left(\algvar{E}_1/8,\algvar{E}_2/4,\algvar{E}_3/6\right)$.
\end{example}

\begin{figure}[t!]\vspace{-2ex}
\centering
\includegraphics[width=0.48\textwidth]{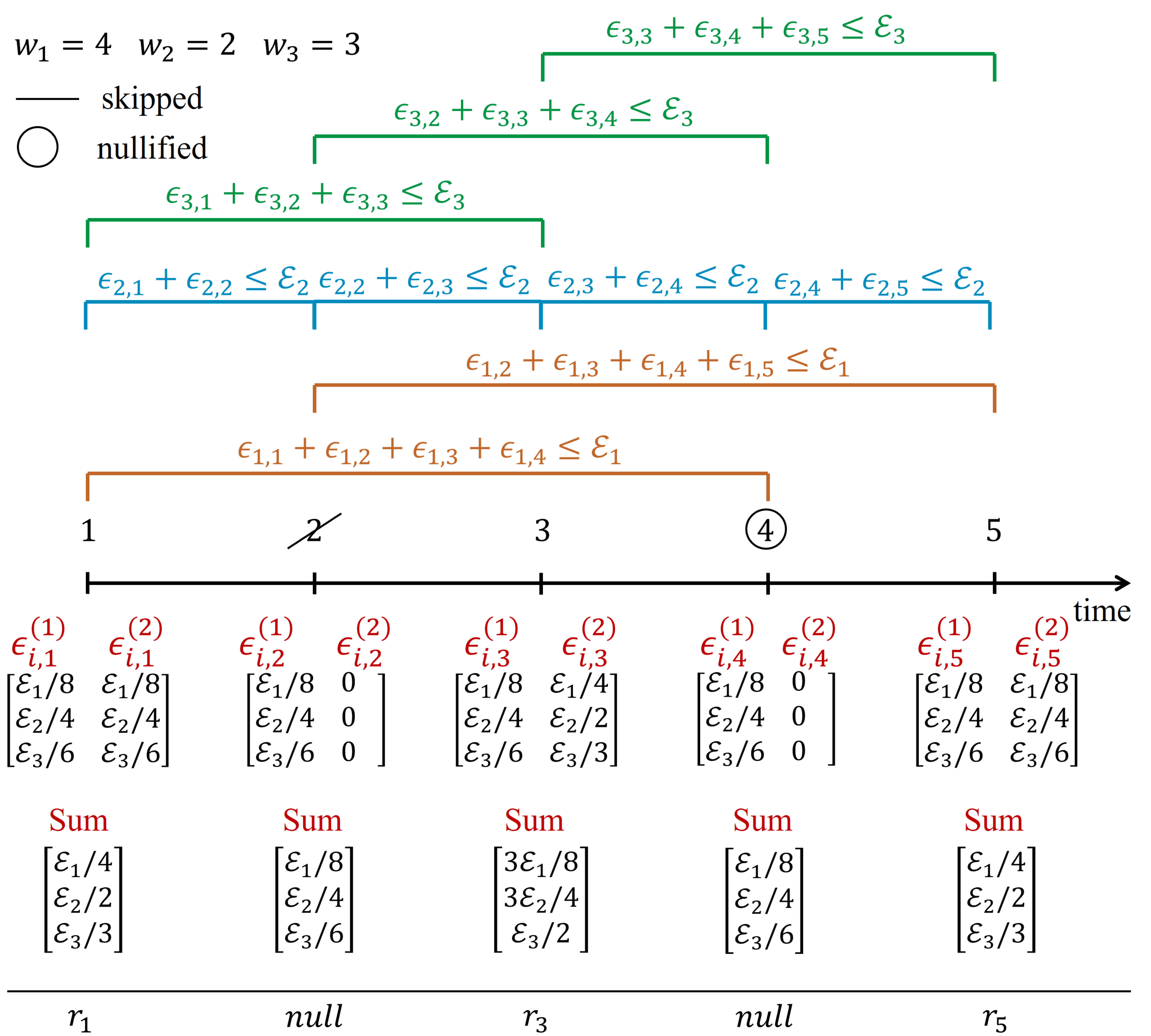}
\caption{A process example for \solutionMethodB{}.}\label{PBA_example}
\end{figure}\vspace{-2ex}

\subsection{Analyses}
\noindent\textbf{Time Cost Analysis.}
Let $m$ be the number of distinct privacy requirements $\left(w_i,\algvar{E}_i\right)$, where $m\leq n$.
Then we have Theorem~\ref{Thm:solutionAB_time_cost_analysis} as follows.
\begin{theorem}\label{Thm:solutionAB_time_cost_analysis}
The time complexities of \solutionMethodA{} and \solutionMethodB{} are both $O(n)$.
\end{theorem}
\begin{proof}
The time complexity of OBS is $O(m)$ for both \solutionMethodA{} and \solutionMethodB{}.
The Sample Mechanism and Query operations each have a time complexity of $O(n)$. Thus, the time complexities of \solutionMethodA{} and \solutionMethodB{} both are  $O(n)$.
\end{proof}

\noindent
\textbf{Memory Complexity Analysis.}
For \solutionMethodA{} and \solutionMethodB{}, we have Theorem~\ref{Thm:solutionAB_memory_cost_analysis} as follows.
\begin{theorem}\label{Thm:solutionAB_memory_cost_analysis}
Both \solutionMethodA{} and \solutionMethodB{} have memory complexity $O(n\cdot w_{\hbox{\scriptsize max}})$.
\end{theorem} 
\begin{proof}
For the process of OBS, the memory complexity is $O(m)$.
For each one of the $n$ users in both \solutionMethodA{} and \solutionMethodB{}, each user requires storing at most $w_{\hbox{\scriptsize\em max}}$ window-related states. 
Thus, the memory complexity is $O(n\cdot w_{\hbox{\scriptsize\em max}})$.
\end{proof}

\noindent
\textbf{Scalability Discussion.}
The above time and memory bounds indicate that \solutionMethodA{} and \solutionMethodB{} are scalable to large user populations. At each time slot, the computation only requires processing the current per-user privacy requirements and maintaining a limited amount of historical state within the relevant window(s), rather than revisiting the entire stream history. Therefore, the per-time-slot cost grows linearly with the number of users $n$, while the memory usage is bounded by the maintained budget/publication states associated with active windows. This makes the method practical for long-running streams with a large user population, provided that the maximum window size remains moderate.

\noindent\textbf{Privacy Analysis.} As for the privacy analysis of \solutionMethodA{} and \solutionMethodB{}, we have Theorem~\ref{Thm:solutionAB_privacy_analysis} as follows.
\begin{theorem}\label{Thm:solutionAB_privacy_analysis}
\solutionMethodA{} and \solutionMethodB{} satisfy $(\vectorfont{w},\vectorfont{\algvar{E}})$-EPDP.
\end{theorem}
\begin{proof}
The complete proof of Theorem~\ref{Thm:solutionAB_privacy_analysis} is provided in Appendix~8.5.1 of the Electronic Supplementary Material (ESM) accompanying this article.
\end{proof}

\noindent
\textbf{Utility Analysis.}
For each user $u_i$ in \solutionMethodA{} and \solutionMethodB{}, we define $w_L$ as the smallest window size among all users.
For each $u_i$, given $(w_i, \algvar{E}_i)$, let $\epsilon_{L}=\min_{i\in[n]}\frac{\algvar{E}_i}{w_i}$ and $\epsilon_{R}=\max_{i\in[n]}\frac{\algvar{E}_i}{w_i}$ be the minimum and maximum values of $\frac{\algvar{E}_i}{w_i}$, respectively.
Let $n_A$ be the number of times $\epsilon_R$ appears among all users.
We assume that at most $\tilde{s}\leq w_L$ non-null publications occur at time slots $q_1$, $q_2$,\ldots, $q_{\tilde{s}}$ in the window of size $w_L$.
We also assume there is no budget absorption from past time slots outside the window.
Furthermore, for each user, each publication approximates the same number of skipped or nullified publications.

We first present a crucial lemma.
\begin{lemma}\label{lemma:sm_utility}
Given $m$ distinct privacy budget-quantity pairs $P=\big\{(\epsilon_{j},n_j)|j\in[m],\sum_{j\in[m]}n_j=n\big\}$ where pair $(\epsilon_{j},n_j)$ indicates that $\epsilon_{j}$ appears $n_j$ times in the user privacy requirement, and a query with sensitivity $I$, the error upper bound $\widetilde{\hbox{err}}_O(P)$ of the SM process with privacy budget chosen from OBS is: 
\begin{equation}\label{eq:lemma:sm_utility}
	\outereqsizelarge{
		\begin{aligned}
			\min{\left({\frac{2I^2}{\min_j\epsilon_j^2}},(n-n_{A})\left(n-n_{A}+\frac{1}{4}\right)+{\frac{2I^2}{\max_j\epsilon_j^2}}\right)},
		\end{aligned}
	}
\end{equation}
where $n_{A}=n_k$ with $k=\arg\max_{j\in[m]}{\epsilon_j}$.
\end{lemma}
\begin{proof}
Let $M_L$ be the SM with privacy budget chosen as $\min_j{\epsilon_j}$.
According to the SM process, all budget types will be selected. 
In this case, the sampling error $err_s$ is $0$ and the noise error $\hbox{\em err}_{\hbox{\scriptsize\em dp}}$ is $2\cdot\left(\frac{I}{\min_j{\epsilon_j}}\right)^2=\frac{2I^2}{\min_j{\epsilon_j^2}}$.
Thus, the total error of $M_L$ is $\hbox{\em err}_{M_L}=\frac{2I^2}{\min_j{\epsilon_j^2}}$.
Let $M_R$ be the SM with privacy budget chosen as $\max_j{\epsilon_j}$. In this case, $(m-1)$ types of privacy budget are chosen with probability $p_k=\frac{e^{\epsilon_k}-1}{e^{\max_j{\epsilon_j}}-1}$ less than $1$ ($k\in[m]$). For the sampling error, we have:
\vspace{-2ex}
\begin{equation}\notag
	\outereqsize{
		\begin{aligned}
			err_s &= \sum_{\epsilon_k<\max_j{\epsilon_j}} n_k p_k(1-p_k) + \left(\sum_{\epsilon_k<\max_j{\epsilon_j}} n_k(1-p_k)\right)^2\\
			&< \sum_{\epsilon_k<\max_j{\epsilon_j}} n_k \left(\frac{p_k+1-p_k}{2}\right)^2 + \left(\sum_{\epsilon_k<\max_j{\epsilon_j}} n_k\right)^2\\
			&=\frac{1}{4}(n-n_{A}) + (n-n_{A})^2\\
			&=(n-n_{A})\left(n-n_{A}+\frac{1}{4}\right).
		\end{aligned}
	}
\end{equation}
The noise error $\hbox{\em err}_{\hbox{\scriptsize\em dp}}$ in this case is $2\cdot\left(\frac{I}{\max_j{\epsilon_j}}\right)^2=\frac{2I^2}{\max_j{\epsilon_j^2}}$.
Thus, the total error of $M_R$ is $\hbox{\em err}_{M_R}=(n-n_A)\left(n-n_A+\frac{1}{4}\right)+\frac{2I^2}{\max_j{\epsilon_j^2}}$.
According to the OBS process, we have $\widetilde{\hbox{\em err}}_{O}(P)\leq \hbox{\em err}_{M_L}$ and $\widetilde{\hbox{\em err}}_{O}(P)\leq \hbox{\em err}_{M_R}$.
Therefore, 
\begin{equation}\notag
	\outereqsize{
		\begin{aligned}
			\widetilde{\hbox{\em err}}_{O}(P)	\leq& \min{(\hbox{\em err}_{M_L},\hbox{\em err}_{M_R})}\\
			=& \min\left(\frac{2I^2}{\min_j{\epsilon_j^2}},(n-n_A)\left(n-n_A+\frac{1}{4}\right)+\frac{2I^2}{\max_j{\epsilon_j^2}}\right).
		\end{aligned}
	}
\end{equation}
\end{proof}

To ensure the robustness of Lemma~\ref{lemma:sm_utility}, we analyze the behavior of the proposed mechanism under extreme conditions. Specifically, we consider two extreme cases:
(1) \textit{Uniform Privacy Budget}. When all users possess the same privacy budget, i.e., $\epsilon_{j}\equiv\epsilon$, the system reduces to a single budget-quantity pair $(\epsilon, n)$. Here, the error equals the standard CDP bound, i.e., $2I^2/\epsilon^2$. Given $n_A=n$, Equation~\eqref{eq:lemma:sm_utility} evaluates to $\min{\left(\frac{2I^2}{\epsilon^2},0+\frac{2I^2}{\epsilon^2}\right)}=\frac{2I^2}{\epsilon^2}$. Thus, the equation consistently recovers the standard CDP error in the uniform setting; (2) \textit{Highly Disparate Privacy Budgets}. When some users have significantly larger privacy budget than others (e.g., $\epsilon_1\ll \epsilon_2\ll\dots\ll\epsilon_n$), the error is no more than $2I^2/\min_j{\epsilon_j^2}$, which is exactly the dominant error term $\hbox{\em err}_{M_L}$ in Equation~\eqref{eq:lemma:sm_utility}. Thus, Lemma~\ref{lemma:sm_utility} still holds in these two extreme cases.

For \solutionMethodA{} we present Theorem~\ref{Thm:solutionMethodA_utility_analysis} to declare its error upper bound as follows.
\begin{theorem}\label{Thm:solutionMethodA_utility_analysis}
The average error per time slot in \solutionMethodA{} is at most \innereqsize{$\min{\left(\frac{8}{d^2\epsilon_L},Z+\frac{8}{d^2\epsilon_R}\right)}+\min{\left(\frac{32\cdot(4^{\tilde{s}}-1)}{3\tilde{s}\epsilon_L},Z+\frac{32\cdot(4^{\tilde{s}}-1)}{3\tilde{s}\epsilon_R}\right)}$} where \innereqsize{$Z=(n-n_A)\left(n-n_A+\frac{1}{4}\right)$}, if at most $\tilde{s}$ non-null publications occur in any window with size $w_L$.
\end{theorem}

\begin{proof}
The complete proof of Theorem~\ref{Thm:solutionMethodA_utility_analysis} is provided in Appendix~8.6.1 of the Electronic Supplementary Material (ESM).
\end{proof}

\solutionMethodA{} achieves low error when the number of non-null publications $\tilde{s}$ per window is small.
However, the error increases exponentially with $\tilde{s}$.
Additionally, the error in Part$_{\hbox{\scriptsize DC}}$ (the first part of the error upper bound in \solutionMethodA{}) rises as $w_L$ increases, however, it diminishes as $d$ increases.
This is because a large $d$ reduces sensitivity leading to smaller noise error.

For \solutionMethodB{}, assume $\alpha$ skipped publications occur before a publication. 
Let $\epsilon_{\tilde{L}}$ and $\epsilon_{\tilde{R}}$ be the minimum and maximum publication privacy budgets among all users at time slots $t=w_{L}$ and  $t=(\alpha+1)$, respectively. 
According to the \solutionMethodB{} process, there will be $\alpha$ nullified publications after the publication.
These nullified publications are set as the last non-null publication without comparison.
Consequently, the nullified publication error depends on the data distribution at nullified time slots.
We denote the average error of each nullified publication in \solutionMethodB{} as $\overline{\hbox{\em err}}_{\hbox{\scriptsize\em nlf}}$.
For \solutionMethodB{}, we have Theorem~\ref{Thm:solutionMethodB_utility_analysis} as follows.
\begin{theorem}\label{Thm:solutionMethodB_utility_analysis}
The average error per time slot in \solutionMethodB{} is at most \innereqsize{$\min\left(\frac{8}{d^2\epsilon_L},Z+\frac{8}{d^2\epsilon_R}\right)+\frac{1}{2\alpha+1}\left(\widetilde{\hbox{\em err}}_{\hbox{\scriptsize\em NOP}}^{\left(\hbox{\scriptsize s,p}\right)}+\alpha\cdot\overline{\hbox{\em err}}_{\hbox{\scriptsize nlf}}\right)$}
where $\widetilde{\hbox{\em err}}_{\hbox{\scriptsize\em NOP}}^{\left(\hbox{\scriptsize s,p}\right)}$ is \innereqsize{$\min\left(\frac{2}{\epsilon_L^2}H^2_{\alpha+1},(\alpha+1)Z+\frac{2}{\epsilon_R^2}H^2_{\alpha+1}\right)$} when $\alpha\leq w_{L}$ and \innereqsize{$\min\left(\frac{2}{\epsilon_L^2}H^2_{w_L},w_L Z+\frac{2}{\epsilon_R^2}H^2_{w_L}\right)+ (\alpha-w_L+1)\min\left(\frac{2}{\epsilon_{\tilde{L}}^2}, Z+\frac{2}{\epsilon_{\tilde{R}}^2}\right)$} when $\alpha>w_{L}$ and
\innereqsize{$Z=(n-n_A)\left(n-n_A+\frac{1}{4}\right)$} and $H^2_{x}$ is the $x$-th square harmonic number, if there are $\alpha$ skipped publications occur in average before each publication.
\end{theorem}
\begin{proof}
The complete proof of Theorem~\ref{Thm:solutionMethodB_utility_analysis} is provided in Appendix~8.6.2 of the Electronic Supplementary Material (ESM).
\end{proof}

\noindent
\textbf{Discussion on Frequent-Update Regimes.}
The bound above shows that the utility of PBA may deteriorate when the number of skipped and nullified publications becomes large. This does not weaken the formal privacy guarantee of PBA, which is still ensured by the budget-composition analysis in Theorem~\ref{Thm:solutionAB_privacy_analysis}, but it may reduce estimation accuracy when the stream changes too frequently. Therefore, PBA is more suitable for relatively smooth streams, whereas PBD is preferable when the stream exhibits persistent rapid changes.

\section{\solutionBTotalName{}}\label{advanced_method}
In this section, we consider the dynamic personalized setting, where each user may specify different backward and forward privacy requirements at different time slots. Unlike the fixed setting, the privacy budget allocation at the current time slot cannot be determined solely from the current requirement. Instead, it must remain consistent with previously consumed privacy budgets and, at the same time, preserve feasibility for future privacy requirements. Therefore, the dynamic setting introduces an online feasibility problem under heterogeneous time-varying personalized privacy requirements.

A naive adaptation that independently allocates each user's budget cannot directly determine a valid release budget, because the system publishes only one aggregate statistic at each time slot. The release budget must simultaneously satisfy all active backward and forward constraints induced by heterogeneous users and historical releases. Therefore, the problem is not per-user budget allocation alone, but online feasibility maintenance for a shared release under dynamic personalized constraints.

To address this problem, we generalize the fixed-setting framework to a dynamic framework called \solutionBTotalName{} (\solutionB{}). 
\solutionB{} focuses on online feasibility maintenance and shared release-budget construction under dynamic personalized privacy requirements, rather than merely substituting personalized parameters into existing $w$-event mechanisms. 
The key idea of \solutionB{} is to compute feasible privacy budget upper bounds at each time slot under both backward and forward privacy requirements, and then use these feasible budgets to make a shared release decision for the current time slot.

\subsection{Feasibility Conditions for Privacy Budget Requirements}

\textbf{Backward Feasibility Condition.} \textit{At time slot $t$, user $u_i$ may declare a desired backward privacy requirement $(w_{B,i,t},\entity{E}_{B,i,t})$. Since the historical privacy budget consumption $\{\epsilon_{i,k}\}_{k<t}$  is maintained internally by the trusted curator, the feasibility of this declared backward budget is checked by the curator rather than by the user. A backward requirement is feasible only if}
\begin{equation}\label{rule_A} \notag
	\outereqsizelarge{
		\begin{aligned}
			0\leq\epsilon_{i,t}\leq\algvar{E}_{B,i,t}-\sum_{k=\max{(t-w_{B,i,t}+1,1})}^{t-1}\epsilon_{i,k}\;\;\; \textrm{for}\;\;\; i\in[n].
		\end{aligned}
	}
\end{equation}
If the declared backward budget is infeasible, the curator either rejects it or projects it to the minimal feasible value before computing the current release budget.

The backward feasibility condition has two parts. 
The first rule requires users to propose valid backward privacy budgets that exceed their historical privacy usage within the current backward window. 
We assume the trusted curator enforces this feasibility condition before budget allocation at each time slot.
The second rule establishes an upper bound for each user's privacy budget usage at the current time slot.

\noindent
\textbf{Forward Feasibility Condition.} \textit{Let $T_{B,i,t}=\{\tau|\tau\leq t\leq\tau+w_{F,i,\tau}-1\}$ be the set of backward time slots whose forward windows cover time slot $t$. The privacy budget usage $\epsilon_{i,t}$ must not exceed the minimal remaining forward privacy budget among all the time slots in $T_{B,i,t}$. Namely,}
\begin{equation}\notag
	\outereqsizelarge{
		\begin{aligned}
			0\leq\epsilon_{i,t}\leq\min_{\tau\in T_{B,i,t}}\left(\algvar{E}_{F,i,\tau}-\sum_{k=\tau}^{t-1}\epsilon_{i,k}\right)\;\;\; \textrm{for}\;\;\; i\in[n].
		\end{aligned}
	}
\end{equation}

The forward feasibility condition ensures that the privacy budget usage at the current time slot does not violate the forward privacy requirements of all historical time slots. 
We illustrate this forward feasibility condition in Example~\ref{expl:forward_regulation}.
\begin{figure}[ht!]
	\centering
	\includegraphics[width=0.45\textwidth]{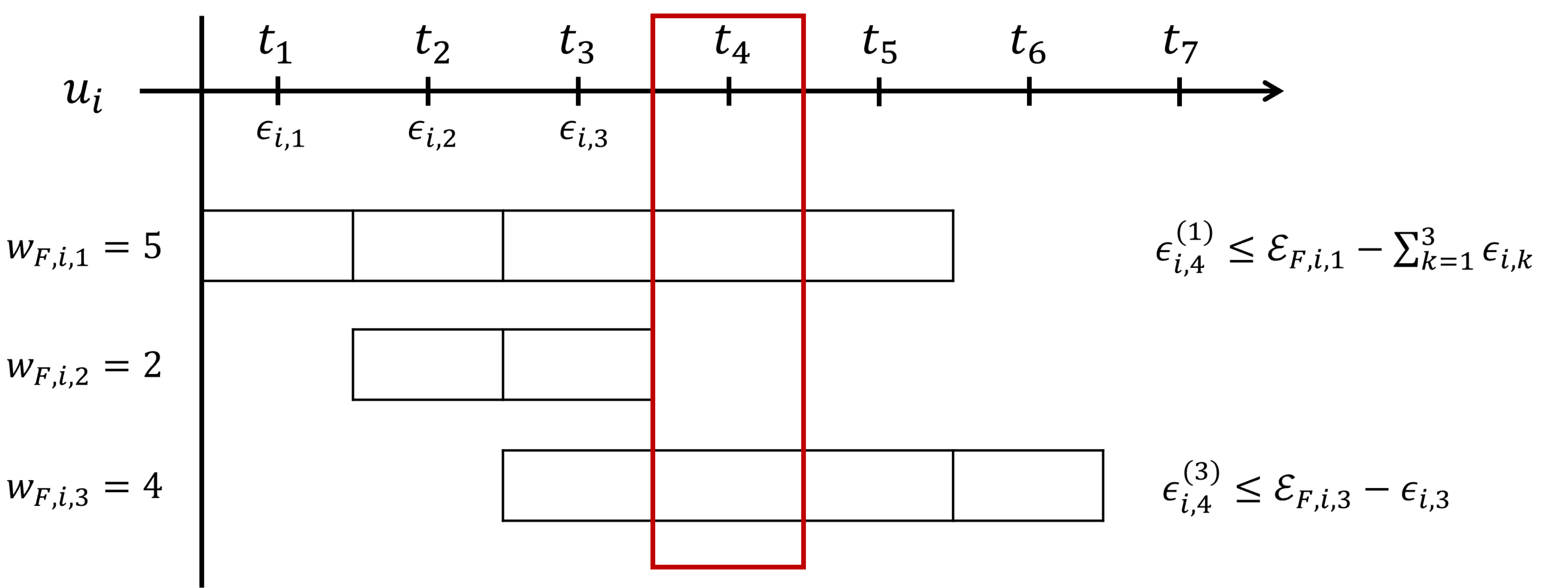}
	\caption{An example for the forward feasibility condition.}\label{FR_example}\figureCaptionMargin{}
\end{figure}

\begin{example}\label{expl:forward_regulation}
	As shown in Figure~\ref{FR_example}, assume $u_i$ spends privacy budgets $\epsilon_{i,1}$, $\epsilon_{i,2}$ and $\epsilon_{i,3}$  at time slots $t_1$, $t_2$ and $t_3$, respectively.
	The forward window sizes of $u_i$ at these time slots are $5$, $2$ and $4$.
	Let $\epsilon_{i,j}^{(k)}$ represent $u_i$'s upper bound of the budget usage at time slot $t_j$ constrained by the forward budget requirement at $t_{k}$.
	Based on the requirement at $t_1$, for the budget usage at $t_4$, we have $\epsilon_{i,4}^{(1)}\leq\algvar{E}_{F,i,1}-\sum_{k=1}^{3}\epsilon_{i,k}$.
	Since the forward window at $t_2$ does not cover $t_4$, its requirement does not affect $\epsilon_{i,4}$.
	Based on the requirement at $t_3$, for the budget usage at $t_4$, we have $\epsilon_{i,4}^{(3)}\leq\algvar{E}_{F,i,3}-\sum_{k=3}^{3}\epsilon_{i,k}=\algvar{E}_{F,i,3}-\epsilon_{i,3}$.
	Therefore, the final budget upper bound is $\epsilon_{i,4}=\min{\left(\epsilon_{i,4}^{(1)},\epsilon_{i,4}^{(3)}\right)}$.
\end{example}

\subsection{Solution for \problemDefineEnhancedSimpleName{}}


In this subsection, we instantiate \solutionB{} with two mechanisms: \solutionMethodDTotalName{} (\solutionMethodD{}) and \solutionMethodETotalName{} (\solutionMethodE{}). Both mechanisms follow the same feasibility principle: the privacy budget used at each time slot must satisfy both backward and forward personalized privacy requirements. They differ in how the feasible publication budget is scheduled across time slots: \solutionMethodD{} follows a distribution-based strategy, whereas \solutionMethodE{} follows an absorption-based strategy.

\noindent
\textbf{\solutionMethodDTotalName{}.} \solutionMethodD{} extends \solutionMethodA{} to satisfy users' variable privacy level demands at different time slots.
Similar to \solutionMethodA{}, the process in \solutionMethodD{} is divided into Part$_{\hbox{\scriptsize DC}}$ for dissimilarity calculation and Part$_{\hbox{\scriptsize NOP}}$ for publication calculation. 
To satisfy both backward and forward feasibility conditions, the privacy budget usages in Part$_{\hbox{\scriptsize DC}}$ and Part$_{\hbox{\scriptsize NOP}}$ are required to comply with these feasibility conditions using half of the backward and forward privacy budgets.
Specifically, for Part$_{\hbox{\scriptsize DC}}$:  $\epsilon_{i,t}^{(1)}\leq \algvar{E}_{B,i,t}/2-\sum_{\substack{k=\max(t-w_{B,i,t}+1,1)}}^{t-1}\epsilon_{i,k}^{(1)}$ (backward feasibility condition),
and $\epsilon_{i,t}^{(1)}\leq\min_{\tau\in T_{B,i,t}}\left(\algvar{E}_{F,i,\tau}/2-\sum_{k=\tau}^{t-1}\epsilon_{i,k}^{(1)}\right)$ (forward feasibility condition).
For Part$_{\hbox{\scriptsize NOP}}$: $\epsilon_{i,t}^{(2)}\leq \algvar{E}_{B,i,t}/2-\sum_{k=\max(t-w_{B,i,t}+1,1)}^{t-1}\epsilon_{i,k}^{(2)}$ (backward feasibility condition),
and $\epsilon_{i,t}^{(2)}\leq\min_{\tau\in T_{B,i,t}}\left(\algvar{E}_{F,i,\tau}/2-\sum_{k=\tau}^{t-1}\epsilon_{i,k}^{(2)}\right)$ (forward feasibility condition).

The process of \solutionMethodD{} is shown in Algorithm~\ref{alg:PDBD}.
For each user $u_i$, we obtain the historical time slot set $T_{B,i,t}$ where each elmement's forward window covers the current time slot (Line~\ref{historical_impact_set}).

In Part$_{\hbox{\scriptsize DC}}$ process, for each historical time slot $\tau$, the total forward privacy budget is set as $\algvar{E}_{F,i,\tau}/2$ and allocated evently among all the time slot within the forward window.
Thus each time slot in the forward window with size $w_{F,i,\tau}$ holds forward privacy budget as $\algvar{E}_{F,i,\tau}/(2w_{F,i,\tau})$.
Therefore, the privacy budget upper bound based on the forward feasibility condition at time slot $t$ is $\min_{\tau\in T_{B,i,t}}{\frac{\algvar{E}_{F,i,\tau}}{2w_{F,i,\tau}}}$ (Line~\ref{forward_impact}).
Based on the backward feasibility condition, the privacy budget upper bound at time slot $t$ is set as $\algvar{E}_{B,i,t}/2-\sum_{\tau=t-w_{B,i,t}+1}^{t-1}\epsilon_{i,\tau}^{(1)}$ (Line~\ref{backward_impact}).
To satisfy both feasibility conditions, each user $u_i$'s budget usage for Part$_{\hbox{\scriptsize DC}}$ is set to the minimum of these two upper bounds: $\epsilon_{i,t}^{(1)}=\min{\left(\epsilon_{F,i,t}^{(1)},\epsilon_{B,i,t}^{(1)}\right)}$ (Line~\ref{backward_minimum}).
The subsequent steps in Part$_{\hbox{\scriptsize DC}}$ follow Algorithm~\ref{alg:PBD} (\solutionMethodA{}).

In Part$_{\hbox{\scriptsize NOP}}$ process, the forward remaining budget $\epsilon_{F,i,t}^{(2)}$ for each $u_i$ is set to half of the minimum forward remaining budgets across all time slots in $T_{B,i,t}$ (Line~\ref{forward_remaining_budget_2}).
The backward remaining budget $\epsilon_{B,i,t}^{(2)}$ for each $u_i$ is set to the remaining budget within the front $w_{B,i,t}$ time slots. 
The final publication budget for each $u_i$ is determined by taking the minimum value between $\epsilon_{F,i,t}^{(2)}$ and $\epsilon_{B,i,t}^{(2)}$. 
The subsequent steps in Part$_{\hbox{\scriptsize NOP}}$ follow those in  Algorithm~\ref{alg:PBD} (\solutionMethodA{}).

\begin{algorithm}[h]
	\caption{\solutionMethodDTotalName{} (\solutionMethodD{})}
	\label{alg:PDBD}
	\DontPrintSemicolon
	\KwIn{$D_t$, dynamic personalized privacy requirement set ($\vectorfont{w}_{B,t}$, $\vectorfont{\algvar{E}}_{B,t}$, $\vectorfont{w}_{F,t}$, $\vectorfont{\algvar{E}}_{F,t}$), historical data publication $(\vectorfont{r}_1,\vectorfont{r}_2,\dots, \vectorfont{r}_{t-1})$ }
	\KwOut{$\vectorfont{r}_t$}
	\For{$i\in[n]$}{
		Calculate $T_{B,i,t}\gets\{\tau|\tau\leq t\leq\tau+w_{F,i,\tau}-1\}$;\\ \label{historical_impact_set}\label{PDBD_start2}
		Calculate $\epsilon_{F,i,t}^{(1)}\gets\min_{\tau\in T_{B,i,t}}{\frac{\algvar{E}_{F,i,\tau}}{2w_{F,i,\tau}}}$;\\ \label{forward_impact} \label{PDBD_start}
		Calculate $\epsilon_{B,i,t}^{(1)}\gets\algvar{E}_{B,i,t}/2-\sum_{\tau=t-w_{B,i,t}+1}^{t-1}\epsilon_{i,\tau}^{(1)}$;\\ \label{backward_impact}
		Set $\epsilon_{i,t}^{(1)}\gets\min{\left(\epsilon_{F,i,t}^{(1)},\epsilon_{B,i,t}^{(1)}\right)}$;\\ \label{backward_minimum}
	}
	$\vectorfont{\epsilon}_{t}^{(1)}\gets\left(\epsilon_{1,t}^{(1)},\epsilon_{2,t}^{(1)},\dots,\epsilon_{n,t}^{(1)}\right)$;\\ \label{PDBD_end}
	Estimate $dis\gets\textrm{DC}\left(D_t,\vectorfont{\epsilon}_{t}^{(1)}, \vectorfont{r}_1,\vectorfont{r}_2,\dots, \vectorfont{r}_{t-1}\right)$ by \textbf{Algorithm~\ref{alg:DC}};\\ \label{PDBD_end2}
	\For{$i\in[n]$}{
		Calculate $\epsilon_{F,i,t}^{(2)}\gets\frac{1}{2}\min_{\tau\in T_{B,i,t}}{\left(\algvar{E}_{F,i,\tau}/2-\sum_{k=\tau}^{t-1}{\epsilon_{i,k}^{(2)}}\right)}$;\\ \label{forward_remaining_budget_2}
		Calculate $\epsilon_{B,i,t}^{(2)}\gets\left(\algvar{E}_{B,i,t}/2-\sum_{\tau=t-w_{B,i,t}+1}^{t-1}\epsilon_{i,\tau}^{(2)}\right)$;\\  \label{PDBD_backward_remaining_budget_2}
		Set $\epsilon_{i,t}^{(2)}\gets\min{\left(\epsilon_{F,i,t}^{(2)},\epsilon_{B,i,t}^{(2)}\right)}$;\\ \label{PDBD_forward_remaining_budget_2_end}
	}
	$\vectorfont{\epsilon}_{t}^{(2)}\gets\left(\epsilon_{1,t}^{(2)},\epsilon_{2,t}^{(2)},\dots,\epsilon_{n,t}^{(2)}\right)$;\\
	Same as Lines~\ref{pbd_mt2_mid}-\ref{mt2_end} in \textbf{Algorithm~\ref{alg:PBD}} \label{PDBD_final}
\end{algorithm}

\begin{table}[t!]
	\caption{Privacy requirements in \solutionMethodD{}, where $B$ and $F$ denote backward privacy requirements and  forward privacy budget requirements.}
	\label{data2_example}
	\scalebox{0.8}{
		\begin{tabular}{|c|c|c|c|c|c|}
			\hline
			Time & 1                                                                 & 2                                                                 & 3                                                                 & 4                                                                 & 5                                                                 \\ \hline
			$u_1$ & \begin{tabular}[c]{@{}c@{}}$B:(1,1.0)$\\ $F:(4,2.4)$\end{tabular} & \begin{tabular}[c]{@{}c@{}}$B:(2,2.4)$\\ $F:(4,3.2)$\end{tabular} & \begin{tabular}[c]{@{}c@{}}$B:(2,2.8)$\\ $F:(3,4.2)$\end{tabular} & \begin{tabular}[c]{@{}c@{}}$B:(2,2.4)$\\ $F:(3,2.4)$\end{tabular} & \begin{tabular}[c]{@{}c@{}}$B:(5,3.0)$\\ $F:(2,0.8)$\end{tabular} \\ \hline
			$u_2$ & \begin{tabular}[c]{@{}c@{}}$B:(1,0.6)$\\ $F:(2,1.6)$\end{tabular} & \begin{tabular}[c]{@{}c@{}}$B:(2,1.6)$\\ $F:(2,2.4)$\end{tabular} & \begin{tabular}[c]{@{}c@{}}$B:(2,3.2)$\\ $F:(2,2.8)$\end{tabular} & \begin{tabular}[c]{@{}c@{}}$B:(3,4.2)$\\ $F:(2,2.8)$\end{tabular} & \begin{tabular}[c]{@{}c@{}}$B:(3,3.6)$\\ $F:(2,2.0)$\end{tabular} \\ \hline
			$u_3$ & \begin{tabular}[c]{@{}c@{}}$B:(1,2.0)$\\ $F:(3,1.2)$\end{tabular} & \begin{tabular}[c]{@{}c@{}}$B:(2,1.2)$\\ $F:(3,3.0)$\end{tabular} & \begin{tabular}[c]{@{}c@{}}$B:(3,1.8)$\\ $F:(2,1.2)$\end{tabular} & \begin{tabular}[c]{@{}c@{}}$B:(2,3.2)$\\ $F:(3,0.6)$\end{tabular} & \begin{tabular}[c]{@{}c@{}}$B:(4,2.4)$\\ $F:(3,1.8)$\end{tabular} \\ \hline
		\end{tabular}
	}
\end{table}
\begin{example} \label{pdbd_example}
	Consider a system with users' privacy requirements shown in Table~\ref{data2_example}.
	Privacy requirements are denoted as $B:(a,b)$ and $F:(c,d)$, where $a$ is the backward window size $w_{B,i,t}$, $b$ is the backward privacy budget $\algvar{E}_{B,i,t}$, $c$ is the forward window size $w_{F,i,t}$, and $d$ is the forward privacy budget $\algvar{E}_{F,i,t}$.
	We analyze the first $5$ time slots with privacy settings shown in Figure~\ref{example_for_pdbd}.
	The status is recorded as $\left[\epsilon_{B,i,t}^{(1)},\epsilon_{F,i,t}^{(1)};\epsilon_{B,i,t}^{(2)},\epsilon_{F,i,t}^{(2)}\right]$ with non-null publications occurring at time slots $t=1$ and $t=3$.
	At each time slot, we first compute the calculation budgets by Lines~\ref{PDBD_start2}-\ref{backward_minimum} of Algorithm~\ref{alg:PDBD}, and then compute the publication budgets by Lines~\ref{forward_remaining_budget_2}-\ref{PDBD_forward_remaining_budget_2_end}. Finally, the release decision follows Line~\ref{PDBD_final} of Algorithm~\ref{alg:PDBD} together with Lines~\ref{pbd_mt2_mid}-\ref{mt2_end} of Algorithm~\ref{alg:PBD}.
	
	At time slot $1$, for $u_1$, the forward and backward privacy budgets in Part$_{\hbox{\scriptsize DC}}$ are $\epsilon_{F,1,1}^{(1)}=\frac{2.4}{2\times 4}=0.3$ and $\epsilon_{B,1,1}^{(1)}=\frac{1.0}{2}=0.5$ according to Lines~\ref{PDBD_start2}-\ref{backward_impact}. 
	The calculation budget is therefore $\epsilon_{1,1}^{(1)}$ $=\min{(0.5,0.3)}$ $=0.3$ in Line~\ref{backward_minimum}. 
	In Part$_{\hbox{\scriptsize NOP}}$, the forward and backward budgets are $\epsilon_{F,1,1}^{(2)}=\frac{2.4}{2}/2=0.6$ and $\epsilon_{B,1,1}^{(2)}=\frac{1.0}{2}=0.5$ according to Lines~\ref{forward_remaining_budget_2} and~\ref{PDBD_backward_remaining_budget_2}.
	Thus, the publication budget is $\epsilon_{1,1}^{(2)}=\min{(0.5,0.6)}=0.5$ in Line~\ref{PDBD_forward_remaining_budget_2_end}. 
	The backward and forward budgets in Part$_{\hbox{\scriptsize DC}}$ and Part$_{\hbox{\scriptsize NOP}}$ are recorded as $[0.5,0.3;0.5,0.6]$, while the calculation and publication budgets are recorded as $[0.3; 0.5]$. 
	The budgets for $u_2$ and $u_3$ are shown below $u_1$'s.
	
	At time slot $2$, take $u_1$ as an example, the forward and backward privacy budgets in Part$_{\hbox{\scriptsize DC}}$ are  $\epsilon_{F,1,2}^{(1)}=\min{\left(\frac{3.2}{2\times 4},0.3\right)}=0.3$ and $\epsilon_{B,1,2}^{(1)}=\frac{2.4}{2}-0.3=0.9$ (according to Lines~\ref{PDBD_start2}-\ref{backward_impact}). 
	The calculation budget is therefore $\epsilon_{1,2}^{(1)}=\min{(0.9,0.3)}=0.3$ in Line~\ref{backward_minimum}. 
	According to the comparison in Line~\ref{PDBD_final} in Algorithm~\ref{alg:PDBD} together with Lines~\ref{pbd_mt2_mid}-\ref{mt2_end} in Algorithm~\ref{alg:PBD}, no new publication occurs at this time slot. Thus, the publication budget usage is $0$.
	
	The budgets for the remaining three time slots are also recorded in Figure~\ref{example_for_pdbd}.
	
\end{example}
\begin{figure}[ht!]\vspace{-4ex}
	\centering
	\includegraphics[width=0.48\textwidth]{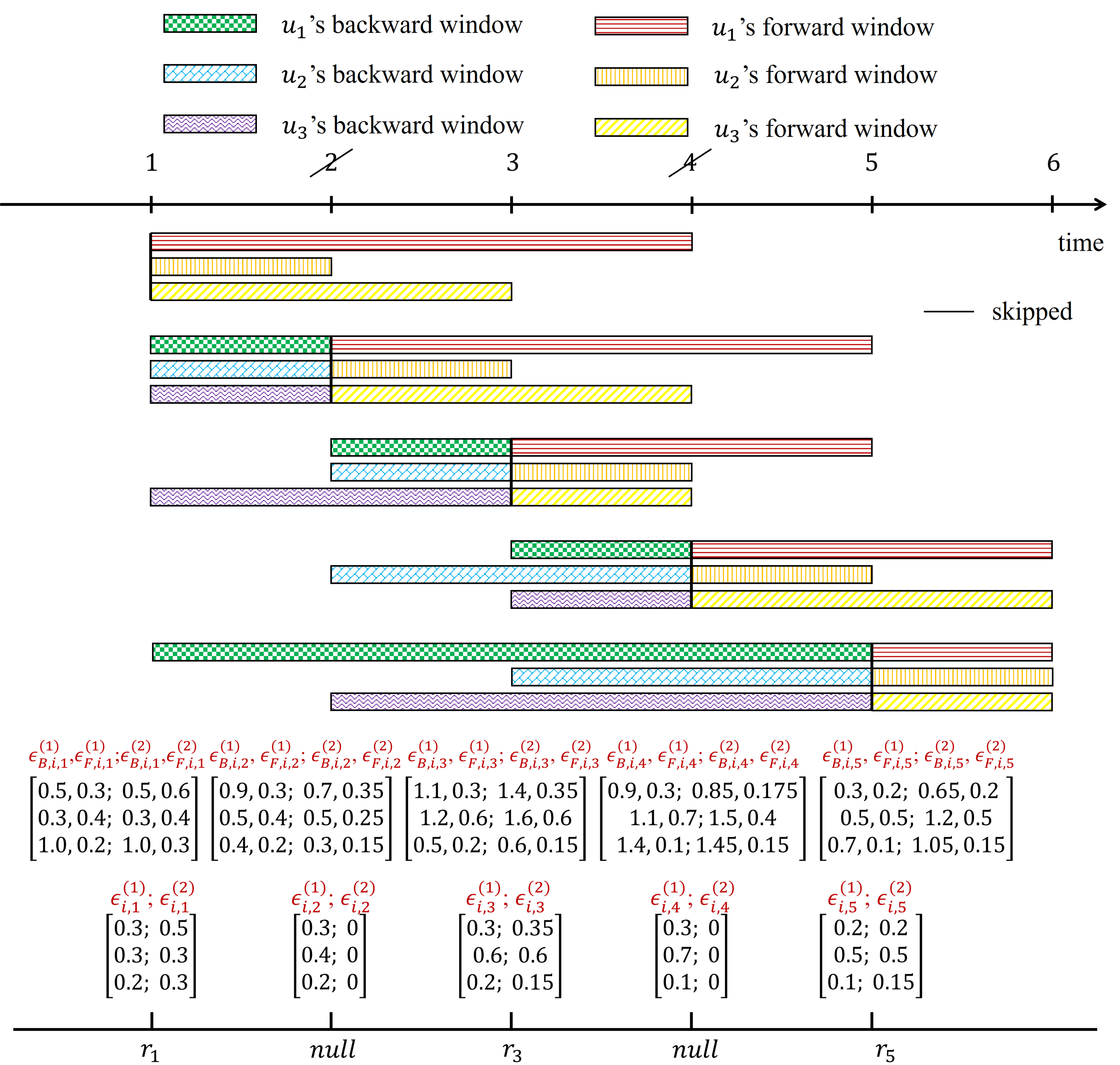}
	\caption{An example for \solutionMethodD{}}\label{example_for_pdbd}\figureCaptionMargin{}
\end{figure}\vspace{-2ex}

\begin{algorithm}[h]
	\caption{\solutionMethodETotalName{} (\solutionMethodE{})}
	\label{alg:PDBA}
	\DontPrintSemicolon
	\KwIn{$D_t$,  dynamic personalized privacy requirement set ($\vectorfont{w}_{B,t}$, $\vectorfont{\algvar{E}}_{B,t}$, $\vectorfont{w}_{F,t}$, $\vectorfont{\algvar{E}}_{F,t}$), historical data publication $(\vectorfont{r}_1,\vectorfont{r}_2,\dots, \vectorfont{r}_{t-1})$}
	\KwOut{$\vectorfont{r}_t$}
	\For{$i\in[n]$}{
		Calculate $T_{B,i,t}\gets\{\tau|\tau\leq t\leq\tau+w_{F,i,\tau}-1\}$;\\ \label{PDBA_start2}
		Calculate $\epsilon_{F,i,t}^{(1)}\gets\min_{\tau\in T_{B,i,t}}{\frac{\algvar{E}_{F,i,\tau}}{2w_{F,i,\tau}}}$;\\ \label{PDBA_start}
		Calculate $\epsilon_{B,i,t}^{(1)}\gets\algvar{E}_{B,i,t}/2-\sum_{\tau=t-w_{B,i,t}+1}^{t-1}\epsilon_{i,\tau}^{(1)}$;\\ 
		Set $\epsilon_{i,t}^{(1)}\gets\min{\left(\epsilon_{F,i,t}^{(1)},\epsilon_{B,i,t}^{(1)}\right)}$;\\ \label{PDBA_back}
		\For{$\tau\in T_{B,i,t}$}{ \label{PDBA_nullified_border_start}
			Calculate $u_i$'s forward nullified time slot right border from $\tau$ as $t_{\hbox{\scriptsize\em FN},i,\tau}\gets\frac{\sum_{\mu=\tau}^{t-1}\epsilon_{i,\mu}^{(2)}}{\algvar{E}_{F,i,\tau}/(2\cdot w_{F,i,\tau})}+\tau-1$\\ \label{nullified_reach}
		}
		Set $u_i$'s forward nullified time slot right border as $R_{\hbox{\scriptsize\em FN},i}\gets\max_{\tau\in T_{B,i,t}}{t_{\hbox{\scriptsize\em FN},i,\tau}}$;\\ \label{nullified_max_reach}
	}
	$\vectorfont{\epsilon}_{t}^{(1)}\gets\left(\epsilon_{1,t}^{(1)},\epsilon_{2,t}^{(1)},\dots,\epsilon_{n,t}^{(1)}\right)$;\\ \label{PDBA_end}
	Estimate $dis\gets\textrm{DC}\left(D_t,\vectorfont{\epsilon}_{t}^{(1)}, \vectorfont{r}_1,\vectorfont{r}_2,\dots, \vectorfont{r}_{t-1}\right)$ by \textbf{Algorithm~\ref{alg:DC}};\\ \label{PDBA_end2}
	Set forward nullified right border $\tilde{R}_{\hbox{\scriptsize\em FN}}\gets\max_{i\in[n]} R_{\hbox{\scriptsize\em FN},i}$;\\ \label{nullified_border}
	\If{$t\leq\tilde{R}_{\hbox{\scriptsize FN}}$}{\label{PDBA_nullified_judge}
		\Return $\vectorfont{r}_t\gets\vectorfont{r}_{t-1}$; \label{PDBA_nullified}
	}\Else{
		\For{$i\in[n]$}{\label{PDBA_budget_start}
			Calculate allocated forward absorption budget $\epsilon_{\hbox{\scriptsize\em AF},i,t}\gets\max\limits_{\tau\in T_{B,i,t}}\left((t-t_{\hbox{\scriptsize\em FN},i,\tau})\cdot\frac{\algvar{E}_{F,i,\tau}}{2\cdot w_{F,i,\tau}}\right)$;\\ \label{absorption_budgets}
			Calculate remaining forward budget upper bound $\epsilon_{\hbox{\scriptsize\em UF},i,t}{\leftarrow}\min\limits_{\tau\in T_{B,i,t}}\left(\frac{\algvar{E}_{F,i,\tau}}{2}-\sum\limits_{\mu=\tau}^{t-1}\epsilon_{i,\mu}^{(2)}\right)$;\\ \label{forward_upper_bound_budgets}
			Set forward absorption budget $\epsilon_{\hbox{\scriptsize\em FA},i,t}\gets\min{(\epsilon_{\hbox{\scriptsize\em AF},i,t},\epsilon_{\hbox{\scriptsize\em UF},i,t})}$;\\ \label{final_absorption_budgets}
			Calculate remaining backward budget upper bound $\epsilon_{\hbox{\scriptsize\em UB},i,t}\gets\algvar{E}_{B,i,t}/2-\sum_{\tau=t-w_{B,i,t}+1}^{t-1}\epsilon_{i,\tau}^{(2)}$;\\  \label{PDBA_backward_remain}
			Set publication budget $\epsilon_{i,t}^{(2)}\gets\min{(\epsilon_{\hbox{\scriptsize\em FA},i,t},\epsilon_{\hbox{\scriptsize\em UB},i,t})}$;\\ \label{PDBA_final_budget}
		}\label{PDBA_budget_end}
		$\vectorfont{\epsilon}_{t}^{(2)}\gets\left(\epsilon_{1,t}^{(2)},\epsilon_{2,t}^{(2)},\dots,\epsilon_{n,t}^{(2)}\right)$;\\
		Same as Lines~\ref{mt222_start}-\ref{mt222_end} in \textbf{Algorithm~\ref{alg:PBA}} \label{PDBA_PBA}
	}	
\end{algorithm}

\noindent
\textbf{\solutionMethodETotalName{}.} \solutionMethodE{} enhances \solutionMethodB{} by supporting dynamic privacy requirements for users. 
Similar to \solutionMethodD{}, \solutionMethodE{} consists of two sub-mechanisms: Part$_{\hbox{\scriptsize DC}}$ and Part$_{\hbox{\scriptsize NOP}}$.
The private dissimilarity calculation in Part$_{\hbox{\scriptsize DC}}$ remains identical to \solutionMethodD{}.
In Part$_{\hbox{\scriptsize NOP}}$, the system determines whether to nullify the current time slot $t$ for each user $u_i$ based on publication budget usage at the relevant historical publication time slots.
For each historical time slot $\tau\in T_{B,i,t}$ influencing time slot $t$, we calculate the nullified right time slot border $t_{\hbox{\scriptsize\em FN},i,\tau}$ with total publication budget shares \Big(where one share equals \innereqsize{$\frac{\algvar{E}_{F,i,\tau}}{2w_{F,i,\tau}}$}\Big) from $\tau$ to $(t-1)$ (Line~\ref{nullified_reach}).
For user $u_i$ at time slot $t$, we determine the final forward nullified time slot right border $R_{\hbox{\scriptsize\em FN},i}$ as the maximum value of these time slot borders (Line~\ref{nullified_max_reach}).
We then obtain the forward nullified right border $\tilde{R}_{\hbox{\scriptsize\em FN}}$ as the maximum value among all $R_{\hbox{\scriptsize\em FN},i}$ (Line~\ref{nullified_border}).
If the current time slot $t$ is no larger than $\tilde{R}_{\hbox{\scriptsize\em FN}}$, $t$ is nullified and skipped.
Otherwise, for each $u_i$, we calculate the budget absorption $\epsilon_{\hbox{\scriptsize\em AF},i,t}$ as the maximum absorption budgets among all historically influencing time slots (Line~\ref{absorption_budgets}).
We also determine the minimum remaining forward budgets $\epsilon_{\hbox{\scriptsize\em UF},i,t}$ across all historical time slots as the remaining forward budget upper bound (Line~\ref{forward_upper_bound_budgets}).
The forward absorption budget $\epsilon_{\hbox{\scriptsize\em FA},i,t}$ is set as the minimum between $\epsilon_{\hbox{\scriptsize\em AF},i,t}$ and $\epsilon_{\hbox{\scriptsize\em UF},i,t}$ (Lines~\ref{absorption_budgets}-\ref{final_absorption_budgets}).
Finally, we calculate the backward budget upper bound $\epsilon_{\hbox{\scriptsize\em UB},i,t}$ (Line~\ref{PDBA_backward_remain}) and set the publication budget as the minimum between $\epsilon_{\hbox{\scriptsize\em FA},i,t}$ and $\epsilon_{\hbox{\scriptsize\em UB},i,t}$ (Line~\ref{PDBA_final_budget}).
The subsequent steps follow those in \solutionMethodA{}.

\begin{example}\label{pdba_example}
	Figure~\ref{example_for_pdba} illustrates the execution of \solutionMethodE{}.
	At each time slot, we first compute the candidate forward window set and the
	calculation budgets by  Lines~\ref{PDBA_start2}-\ref{PDBA_back} of Algorithm~\ref{alg:PDBA}. Then, the forward nullified right border is computed by Lines~\ref{PDBA_nullified_border_start}-\ref{nullified_max_reach}. If $t\leq\tilde{R}_{\hbox{\em \scriptsize FN}}$, the publication
	is nullified according to Line~\ref{PDBA_nullified}; otherwise, we compute the forward absorption
	budget, the remaining backward budget upper bounds, and the final publication
	budget by Lines~\ref{PDBA_budget_start}-\ref{PDBA_budget_end}.
	
	At time slot $1$, $T_{B,i,1}$ for each $u_i$ contains only the current time slot, resulting in $R_{\hbox{\scriptsize\em FN},i}=0$ for all $u_i$. 
	Following Lines~\ref{absorption_budgets}-\ref{final_absorption_budgets} of Algorithm~\ref{alg:PDBA},
	for $u_1$, we calculate $\epsilon_{\hbox{\scriptsize\em AF},1,1}=\frac{\algvar{E}_{F,1,1}}{2w_{F,i,1}}=\frac{2.4}{2\times 4}=0.3$ and $\epsilon_{\hbox{\scriptsize\em UF},1,1}=\frac{\algvar{E}_{F,1,1}}{2}=1.2$, leading to a forward absorption budget of $\epsilon_{\hbox{\scriptsize\em FA},1,1}=0.3$.
	For $u_2$ and $u_3$, we obtain $\epsilon_{\hbox{\scriptsize\em FA},2,1}=0.4$ and $\epsilon_{\hbox{\scriptsize\em FA},3,1}=0.2$.
	With backward remaining budgets $\epsilon_{\hbox{\scriptsize\em UB},i,1}$ of $0.5$, $0.3$ and $1.0$ for $u_1$, $u_2$ and $u_3$ respectively in Line~\ref{PDBA_backward_remain}, their final publication budgets are $0.3$, $0.3$, and $0.2$ according to Line~\ref{PDBA_final_budget}.
	
	At time slot $2$, the publication is skipped,  resulting in $0$ publication budget usage for all users in Line~\ref{PDBA_PBA} together with Lines~\ref{mt222_start}-\ref{mt222_end} in Algorithm~\ref{alg:PBA}.
	
	At time slot $3$, following Line~\ref{PDBA_start2}, $T_{B,1,3}=\{1,2,3\}$ for $u_1$. 
	The forward nullified time slot right border for each time slot in $T_{B,1,3}$ is $\Big\{\frac{0.3}{0.3}+1-1, \frac{0}{0.4}+2-1, \frac{0}{0.7}+3-1\Big\}=\{1,1,2\}$ according to Line~\ref{nullified_reach}. 
	This yields an allocated forward absorption budget of $\epsilon_{\hbox{\scriptsize\em AF},1,3}=\max((3-1)\times 0.3,(3-1)\times 0.4,(3-2)\times 0.7)=0.8$ in Line~\ref{absorption_budgets}.
	For $u_2$ and $u_3$, we calculate $\epsilon_{\hbox{\scriptsize\em AF},2,3}=1.2$ and $\epsilon_{\hbox{\scriptsize\em AF},3,3}=0.6$. 
	The remaining forward budget upper bounds are $\epsilon_{\hbox{\scriptsize\em UF},1,3}=\min\left(\frac{2.4}{2}-0.3,\frac{3.2}{2},\frac{4.2}{2}\right)=0.9$ for $u_1$, $\epsilon_{\hbox{\scriptsize\em UF},2,3}=1.2$ for $u_2$, and $\epsilon_{\hbox{\scriptsize\em UF},3,3}=1.2$ for $u_3$ according to Line~\ref{forward_upper_bound_budgets}.
	This results in forward absorption budgets of $\epsilon_{\hbox{\scriptsize\em FA},1,3}=\min{(0.8,0.9)}=0.8$, $\epsilon_{\hbox{\scriptsize\em FA},2,3}=\min{(1.2,1.2)}=1.2$, and $\epsilon_{\hbox{\scriptsize\em UF},3,3}=\min{(0.6,1.2)}=0.6$ in Line~\ref{final_absorption_budgets}. 
	The remaining backward budget upper bounds are $\epsilon_{\hbox{\scriptsize\em UB},1,3}=\frac{2.8}{2}=1.4$, $\epsilon_{\hbox{\scriptsize\em UB},2,3}=\frac{3.2}{2}=1.6$, and $\epsilon_{\hbox{\scriptsize\em UB},3,3}=\frac{1.8}{2}-0.2=0.7$ in Line~\ref{PDBA_backward_remain}, leading to final publication budgets of $0.8$, $1.2$ and $0.6$ for $u_1$, $u_2$, and $u_3$ respectively in Line~\ref{PDBA_final_budget}.
	
	At time slot $4$, according to Line~\ref{nullified_max_reach}, we calculate the forward nullified time slot right borders as $R_{\hbox{\scriptsize\em FN},1}=3.67$, $R_{\hbox{\scriptsize\em FN},2}=3.71$ and $R_{\hbox{\scriptsize\em FN},3}=4$.
	Since the current time slot $t=4\leq\tilde{R}_{\hbox{\scriptsize\em FN}}=\max(3.67,3.71,4)$, following Lines~\ref{PDBA_nullified_judge}-\ref{PDBA_nullified}, the publication is nullified, resulting in $0$ publication budget usage for all users.
	
	At time slot $5$, following the same process as time slot $3$, we obtain publication budget usage of $0.4$, $0.6$ and $0.3$ for $u_1$, $u_2$, and $u_3$  in Line~\ref{PDBA_final_budget}.
\end{example}

\begin{figure}[ht!]\vspace{-2ex}
	\centering
	\includegraphics[width=0.48\textwidth]{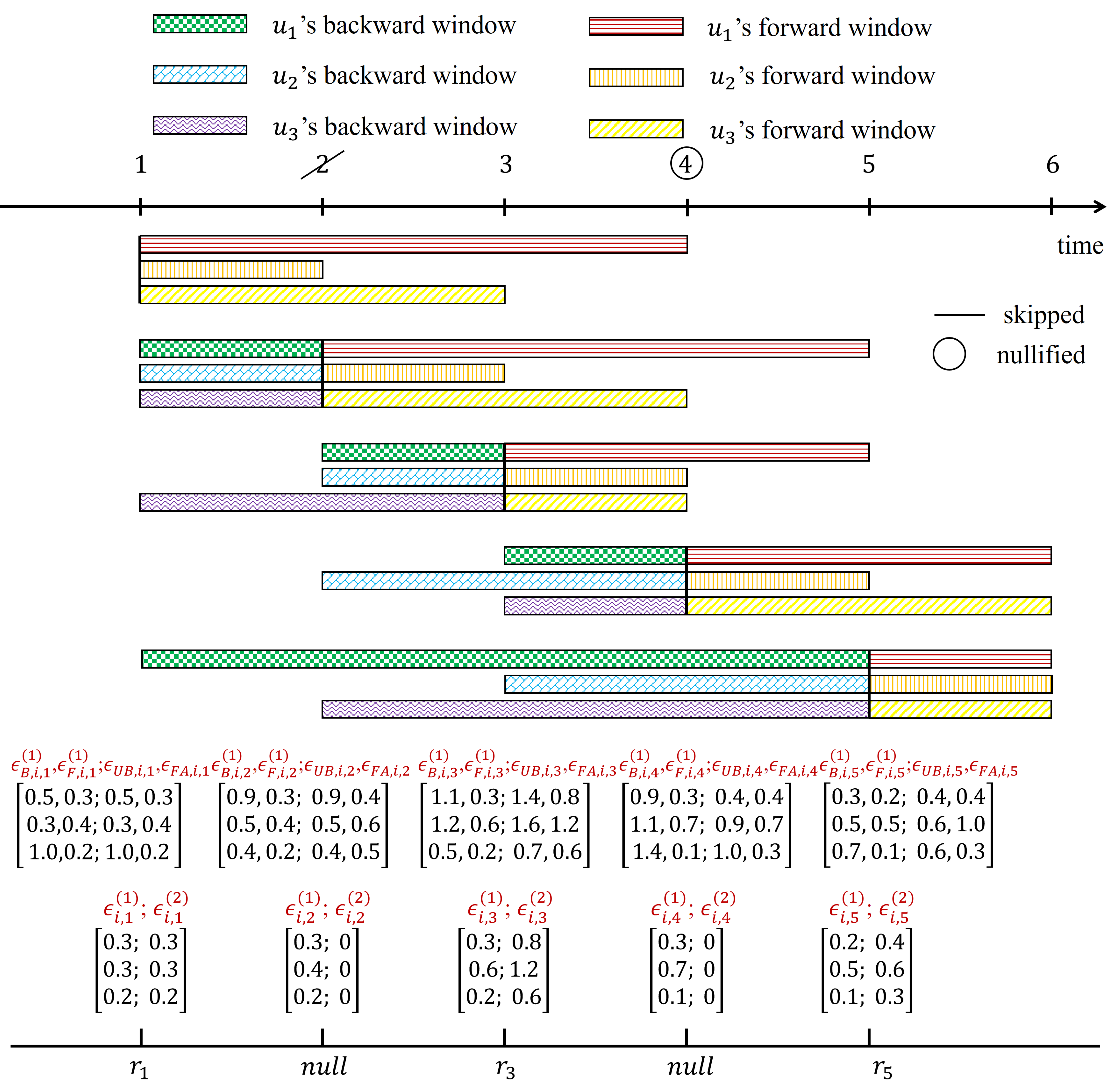}
	\caption{An example for \solutionMethodE{}}\label{example_for_pdba}\figureCaptionMargin{}
\end{figure}\vspace{-3ex}

\subsection{Analysis}
In this subsection, we analyze the time cost and privacy levels of our \solutionMethodD{} and \solutionMethodE{}. 

\noindent
\textbf{Time Cost Analysis.}
For the time cost of \solutionMethodD{} and \solutionMethodE{}, we have Theorem~\ref{Thm:solutionDE_time_cost_analysis} as follows.
\begin{theorem}\label{Thm:solutionDE_time_cost_analysis}
	The time complexities of \solutionMethodD{} and \solutionMethodE{} are both  $O\left(\left(w_{\hbox{\scriptsize max}}^{(F)}+w_{\hbox{\scriptsize max}}^{(B)}\right)\cdot n\right)$.
\end{theorem}
\begin{proof}
	For each $u_i$'s $T_{B,i,t}$ at time slot $t$, we maintain a queue to record $T_{B,i,t}$ and update it by adding the next time slot $t+1$ and removing the time slots that no longer influence time slot $t+1$ to obtain $T_{B,i,t+1}$. The update time cost of all users' queues is $O\left(n\cdot w_{\hbox{\scriptsize\em max}}^{(F)}\right)$, where $w_{\hbox{\scriptsize\em max}}^{(F)}$ represents the maximum window size among all users' forward windows.
	For each user's historical publication budgets, we maintain another queue of size $w_{\hbox{\scriptsize\em max}}^{(B)}$, where $w_{\hbox{\scriptsize\em max}}^{(B)}$ represents the maximum window size among all users' backward windows.
	The time cost of calculating remaining backward budget upper bounds for all users is $O\left(n\cdot w_{\hbox{\scriptsize\em max}}^{(B)}\right)$.
	Thus, the time complexity of calculating $\vectorfont{\epsilon}_t^{(1)}$ and $\vectorfont{\epsilon}_t^{(2)}$ for each time slot is $O\left(\left(w_{\hbox{\scriptsize\em max}}^{(F)}+w_{\hbox{\scriptsize\em max}}^{(B)}\right)\cdot n\right)$.
	The sample mechanism time complexity is $O(n)$.
	Therefore, both \solutionMethodD{} and \solutionMethodE{} have a time complexity of $O\left(\left(w_{\hbox{\scriptsize\em max}}^{(F)}+w_{\hbox{\scriptsize\em max}}^{(B)}\right)\cdot n\right)$.
\end{proof}

The remaining backward/forward budget computation is implemented here using a direct sliding-window summation for clarity. In practice, this step can be maintained incrementally via prefix sums or queue-based sliding-window data structures, reducing the amortized update cost from $O(w_{max})$ to $O(1)$ per user per time slot.

\noindent
\textbf{Memory Complexity Analysis.}
For \solutionMethodD{} and \solutionMethodE{}, we have Theorem~\ref{Thm:solutionDE_memory_cost_analysis} as follows.
\begin{theorem}\label{Thm:solutionDE_memory_cost_analysis}
	Both \solutionMethodD{} and \solutionMethodE{} have complexity $O\left(\left(w_{\hbox{\scriptsize max}}^{(F)}+w_{\hbox{\scriptsize max}}^{(B)}\right)\cdot n\right)$.
\end{theorem} 
\begin{proof}
	For the process of OBS, the memory complexity is $O(m)$.
	For each one of the $n$ users in both \solutionMethodD{} and \solutionMethodE{}, there are at most $w_{\hbox{\scriptsize\em max}}^{(F)}$ forward window size data and $w_{\hbox{\scriptsize\em max}}^{(B)}$ backward window size data need to be stored. 
	Thus, the memory complexity is $O\left(\left(w_{\hbox{\scriptsize\em max}}^{(F)}+w_{\hbox{\scriptsize\em max}}^{(B)}\right)\cdot n\right)$.
\end{proof}

\noindent
\textbf{Scalability Discussion.}
\solutionMethodD{} and \solutionMethodE{} remain scalable for large user populations for the same reason as the fixed case (\solutionMethodA{} and \solutionMethodB{}), while incurring only additional bookkeeping for backward/forward privacy requirements. Thus, their per-time slot cost still depends on the current users and bounded window states rather than the full stream history.

\noindent
\textbf{Privacy Analysis.} The privacy analysis for \solutionMethodD{} and \solutionMethodE{} is presented in Theorem~\ref{Thm:solutionDE_privacy_analysis}.
\begin{theorem}\label{Thm:solutionDE_privacy_analysis}
	\solutionMethodD{} and \solutionMethodE{} satisfy \dynamicPrivacyLevelSimpleNameT{} at each time slot $t$, where $\vectorfont{w}_B=(w_{B,1,t},\dots,w_{B,n,t})$
	and $\vectorfont{w}_F=(w_{F,1,t},\dots,w_{F,n,t})$ represent the requirement sets for all users' backward and forward window sizes at time stamp $t$, and $\vectorfont{\algvar{E}}_B=(\algvar{E}_{B,1,t},\dots,\algvar{E}_{B,n,t})$ 
	and $\vectorfont{\algvar{E}}_F=(\algvar{E}_{F,1,t},\dots,\algvar{E}_{F,n,t})$ represent the requirement sets for all users' backward and forward privacy budgets at time slot $t$.	
\end{theorem}
\begin{proof}
	The complete proof of Theorem~\ref{Thm:solutionDE_privacy_analysis} is provided in Appendix~8.5.2 of the Electronic Supplementary Material (ESM).
\end{proof}

\noindent
\textbf{Utility Analysis under Periodic Dynamic Requirements.}
To obtain a closed-form average error bound, we analyze a common recurrent setting in which each user's privacy requirement sequence is periodic with period $Y$. 
That is for $t>Y$, each user's privacy requirement at time slot $t$ matches that at time slot $t-Y$.
This periodicity assumption is introduced only for the utility analysis; the mechanisms \solutionMethodD{} and \solutionMethodE{}, as well as their privacy guarantees, do not require periodic privacy requirements.
Such periodic requirements arise naturally in applications where user privacy preferences follow daily or weekly routines. For example, a commuter may request stronger privacy protection during regular commuting hours and weaker protection during working hours, leading to a daily repeated requirement pattern. Similarly, drivers or delivery workers may exhibit recurring work/rest schedules that induce periodic changes in privacy requirements.

Besides, assume there are at most $\hat{s}\leq Y$ non-null publications occurring at time slots $t_1$, $t_2$,\dots, $t_{\hat{s}}$. 
Assume each stream approximates the same number ($\rho_{\hbox{\scriptsize\em sk}}$) of skipped publications and the same number ($\rho_{\hbox{\scriptsize\em nu}}$) of nullified publications.
Let $\algvar{E}_{L}^{(F)}(i)=\min_t\algvar{E}_{F,i,t}$ be the minimal proposed forward privacy budget among all time slots for each $u_i$.
We define $\epsilon_{L}^{\left(\hbox{\scriptsize\em B,M}\right)}(i)=\frac{\algvar{E}_{L}^{(F)}(i)}{2^{\beta_i}}$ as the lower bound of $\epsilon_{B,i,t}^{(1)}$ and $\epsilon_{B,i,t}^{(2)}$, where $\beta_i$ is the parameter to be determined. 
We denote $\epsilon_{\hbox{\scriptsize\em BL}}=\min_{i\in[n]}\epsilon_{L}^{\left(\hbox{\scriptsize\em B,M}\right)}(i)$.
Besides, we define $\epsilon_{R}^{\left(\hbox{\scriptsize\em B,M}\right)}(i)=\frac{\algvar{E}_{L}^{(F)}(i)}{2^{\eta_i}}$ as the upper bound of $\epsilon_{B,i,t}^{(1)}$ and $\epsilon_{B,i,t}^{(2)}$, where $\eta_i$ is the parameter to be determined.
We also denote $\epsilon_{\hbox{\scriptsize\em BR}}=\max_{i\in[n]}\epsilon_{R}^{\left(\hbox{\scriptsize\em B,M}\right)}(i)$.
Let $\epsilon_{\hbox{\scriptsize\em FL}}(i)=\min_{t}\frac{\algvar{E}_{F,i,t}}{2w_{F,i,t}}$ be the half of minimal forward privacy budget share for $u_i$ among all time slots.
Let $\epsilon_{\hbox{\scriptsize\em FLL}}=\min_{i\in[n]\epsilon_{\hbox{\scriptsize\em FL}}(i)}$.
Let $\epsilon_{\hbox{\scriptsize\em FLR}}=\max_{i\in[n]\epsilon_{\hbox{\scriptsize\em FL}}(i)}$.
Let $\gamma_L=\min_{i\in[n]}\left(2^{\eta_i-1} - 1\right)$ and $\gamma_R=\max_{i\in[n]}\left(2^{\beta_i-1}-1\right)$.
Let $Z'=(n-n_{B})\left(n-n_{B}+\frac{1}{4}\right)$ be the sampling error upper bound, where $n_{B}$ is the quantity of $\max_{i\in[n],t\in[T]}\frac{\algvar{E}_{F,i,t}}{w_{F,i,t}}$.
For \solutionMethodD{}, we have Theorem~\ref{Thm:solutionD_utility_analysis} as follows.
\begin{theorem}\label{Thm:solutionD_utility_analysis}
	The average error per time slot in \solutionMethodD{} is at most $\min\Big(\frac{2}{d^2(\min(\epsilon_{\hbox{\scriptsize FLL}},\epsilon_{\hbox{\scriptsize BL}}))^2},Z'+\frac{2}{d^2(\max(\epsilon_{\hbox{\scriptsize FLR}},\epsilon_{\hbox{\scriptsize BR}}))^2}\Big)
	+\min\Big(\frac{2(4^{\hat{s}-\gamma_R+1}+3\gamma_R-4)}{3\hat{s}\epsilon_{\hbox{\scriptsize BL}}^2},Z'
	+\frac{2(4^{\hat{s}-\gamma_L+1}+3\gamma_L-4)}{3\hat{s}\epsilon_{\hbox{\scriptsize BR}}^2}\Big)$, where $\gamma_L=\min_{i\in[n]}\left(2^{\eta_i-1} - 1\right)$ and $\gamma_R=\max_{i\in[n]}\left(2^{\beta_i-1}-1\right)$, if at most $\hat{s}$ non-null publications occur in any period $Y$.
\end{theorem}
\begin{proof}
	The complete proof of Theorem~\ref{Thm:solutionD_utility_analysis} is provided in Appendix~8.6.3 in the Electronic Supplementary Material (ESM).
\end{proof}

\noindent
\textbf{Interpretation of Theorem~\ref{Thm:solutionD_utility_analysis}}.
The bound in Theorem~\ref{Thm:solutionD_utility_analysis} consists of the error from Part$_{\hbox{\scriptsize DC}}$ and the accumulated publication error from  Part$_{\hbox{\scriptsize NOP}}$ over one period $Y$. The parameter $\hat{s}$ denotes the maximum number of non-null publications in a period, while $\gamma_L$ and $\gamma_R$ characterize the transition range where the forward publication budget decreases from being above the backward-budget interval to below it. Hence, the bound decomposes the average error into several stages of the \solutionMethodD{} process, rather than treating it as a single opaque expression.

\noindent
\textbf{Discussion on Consistency.}
When the dynamic forward and backward privacy requirements degenerate to fixed personalized privacy requirements,
$\epsilon_{\hbox{\scriptsize\em FL}}(i)=\epsilon_{\hbox{\scriptsize\em L}}^{(B,M)}(i)=\epsilon_{\hbox{\scriptsize\em R}}^{(B,M)}(i)=\algvar{E}_{i}/(2w_{i})$.
Besides, $\gamma_L=\gamma_R=0$, $(\min(\epsilon_{\hbox{\scriptsize\em FLL}},\epsilon_{\hbox{\scriptsize\em BL}}))^2=\min_{i\in[n]}(\algvar{E}_{i}/(2w_{i}))^2$, $(\max(\epsilon_{\hbox{\scriptsize\em FLR}},\epsilon_{\hbox{\scriptsize\em BR}}))^2=\max_{i\in[n]}(\algvar{E}_{i}/(2w_{i}))^2$.
For that matter, the error bound of \solutionMethodD{} reduces to the corresponding bound of the fixed personalized mechanism (\solutionMethodA{}). This confirms that the dynamic analysis is consistent with the fixed-case analysis.

\noindent
\textbf{Discussion on Tightness.}
The bound in Theorem~~\ref{Thm:solutionD_utility_analysis} captures the dominant error source in the main update regimes of \solutionMethodD{}. When the number of non-null publications in each period is small, the Part$_{\hbox{\scriptsize DC}}$ term dominates, which matches the actual behavior of \solutionMethodD{} since most time slots reuse previous releases. When non-null publications occur close to the upper limit $\hat{s}$, the accumulated Part$_{\hbox{\scriptsize NOP}}$ terms dominate, again matching the mechanism behavior because more fresh releases consume more privacy budget and incur more noise. Therefore, the bound explains the dominant error behavior of \solutionMethodD{} in both sparse-update and frequent-update regimes.

For \solutionMethodE{}, we have Theorem~\ref{Thm:solutionE_utility_analysis} as follows.
\begin{figure*}[!t]
	\centering
	\begin{equation}\label{eq:SPValue}
		\outereqsize{
			\begin{aligned}
				\widetilde{\hbox{\em err}}_{\textrm{Part}_{\hbox{\tiny NOP}}}^{(\hbox{\scriptsize\em s,p})}=\left\{
				\begin{array}{ll}
					\min\left(\frac{2H^2_{\rho_{\hbox{\scriptsize\em sk}}+1}}{\epsilon_{\hbox{\scriptsize\em FLL}}^2}, Z'(\rho_{\hbox{\scriptsize\em sk}}+1)+\frac{2H^2_{\rho_{\hbox{\scriptsize\em sk}}+1}}{\epsilon_{\hbox{\scriptsize\em FLR}}^2}\right) & \textrm{if } \lambda_{\hbox{\tiny\em LR}}\geq\lambda_{\hbox{\scriptsize\em RL}};\\ 
					\min\left(\frac{2}{\epsilon_{\hbox{\scriptsize\em FLL}}^2}H^2_{\rho_{\hbox{\scriptsize\em sk}}+1},Z'(\rho_{\hbox{\scriptsize\em sk}}+1)+\frac{2(\rho_{\hbox{\scriptsize\em sk}}+1)}{\epsilon_{\hbox{\scriptsize\em BR}}^2}\right) & \textrm{if } \lambda_{\hbox{\tiny\em LR}}<\lambda_{\hbox{\scriptsize\em RL}} \textrm{ and }\lambda_{L}<\rho_{\hbox{\scriptsize\em sk}}+1\leq\lambda_{R};\\ 
					\min\left(\frac{2}{\epsilon_{\hbox{\scriptsize\em FLL}}^2}H^2_{\lambda_{L}}, Z'\lambda_{L}+\frac{2\lambda_{L}}{\epsilon_{\hbox{\scriptsize\em BR}}^2}\right)+ (\lambda_{R}-\lambda_{L})\min\bigg(\frac{2}{\epsilon_{\hbox{\scriptsize\em BL}}^2}, n\Big(n+\frac{1}{4}\Big)\\ 
					\quad+\frac{2}{\epsilon_{\hbox{\scriptsize\em BR}}^2}\bigg) + \min\left(\frac{2\left(\rho_{\hbox{\scriptsize\em sk}}-\lambda_{R}+1\right)}{\epsilon_{\hbox{\scriptsize\em BL}}^2}, (\rho_{\hbox{\scriptsize\em sk}}-\lambda_{R}+1)Z'+\frac{2}{\epsilon_{\hbox{\scriptsize\em FLR}}^2}H^2_{\rho_{\hbox{\scriptsize\em sk}}-\lambda_{R}+1}\right) & \textrm{otherwise}.
				\end{array}
				\right.
			\end{aligned}
		}
	\end{equation} 
\end{figure*}
\begin{theorem}\label{Thm:solutionE_utility_analysis}
	The average error per time slot in \solutionMethodE{} is at most
	$\min\Big(\frac{2}{d^2(\min(\epsilon_{\hbox{\scriptsize FLL}},\epsilon_{\hbox{\scriptsize BL}}))^2},Z'+\frac{2}{d^2(\max(\epsilon_{\hbox{\scriptsize FLR}},\epsilon_{\hbox{\scriptsize BR}}))^2}\Big)+\frac{\widetilde{\hbox{\scriptsize err}}_{\hbox{\scriptsize\em Part}_{\hbox{\tiny\em NOP}}}^{\left(\hbox{\scriptsize s,p}\right)}+\rho_{\hbox{\scriptsize nu}}\overline{\hbox{\scriptsize err}}_{\hbox{\scriptsize nlf}}}{\rho_{\hbox{\scriptsize sk}}+\rho_{\hbox{\scriptsize nu}}+1}$
	where the value of $\widetilde{\hbox{err}}_{\hbox{\scriptsize\em Part}_{\hbox{\scriptsize\em NOP}}}^{\left(\hbox{\scriptsize s,p}\right)}$ is shown in Equation~\eqref{eq:SPValue}.
\end{theorem}
\begin{proof}
	The complete proof of Theorem~\ref{Thm:solutionE_utility_analysis} is provided in Appendix~8.6.4 of the Electronic Supplementary Material (ESM).
\end{proof}

\noindent
\textbf{Discussion on Frequent-Update Regimes.}
Similar to \solutionMethodB{}, the error bound of \solutionMethodE{} becomes larger when skipped or nullified publications occur frequently. This effect concerns utility rather than privacy: the formal privacy guarantee of \solutionMethodE{} is still ensured by the budget-feasibility and composition analysis, whereas what may deteriorate in frequent-update regimes is estimation accuracy. This does not create unfair privacy treatment across users, because every user's cumulative privacy loss is still bounded by the corresponding personalized window constraint. The trade-off only affects how accurately the shared aggregate release tracks rapidly changing statistics. In practice, \solutionMethodE{} is therefore more suitable for smoother streams in which reuse of previous releases is effective, while \solutionMethodD{} is more appropriate for streams with persistent rapid changes. This interpretation is also consistent with our experimental observations that the absorption-based methods perform better on smoother synthetic streams, whereas the distribution-based methods are more competitive on rapidly changing real datasets. Designing an adaptive switching strategy between these two mechanism families is an interesting direction for future work.

\section{Experiments}\label{experiment}

\subsection{Datasets}\label{subsec:datasets}
We evaluate our solutions on both real and synthetic datasets.

\noindent\textbf{Real datasets.} We use two real-world datasets, \textit{\trajectoryDatasetName{}}~\cite{DBLP:conf/kdd/YuanZXS11,DBLP:conf/gis/YuanZZXXSH10} and \textit{\checkInDatasetName{}}~\cite{DBLP:journals/tist/YangZQ16,DBLP:journals/jnca/YangZCQ15}, to evaluate the performance of our algorithms. 

\begin{figure}[!ht]\centering\vspace{-2ex}
	\subfigure[][{\scriptsize Taxi}]{
		\scalebox{0.26}[0.26]{\includegraphics{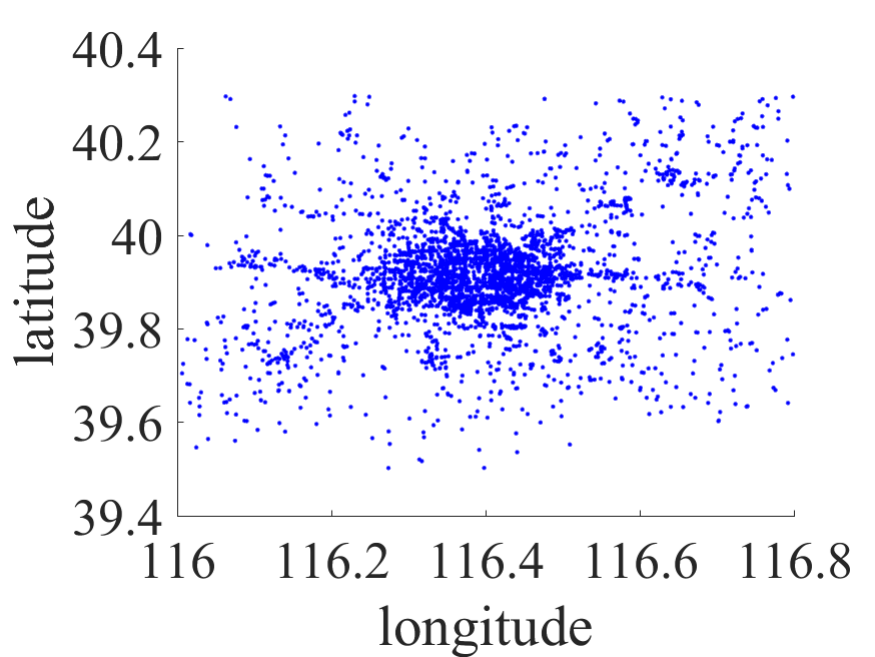}}
		\label{subfig:trajectory}}
	\subfigure[][{\scriptsize Foursquare}]{
		\scalebox{0.26}[0.26]{\includegraphics{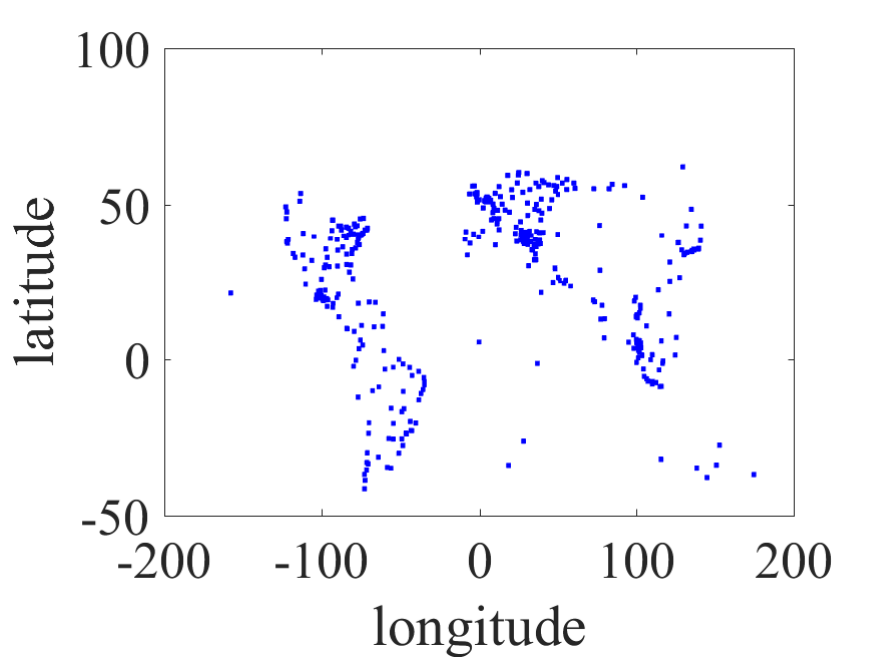}}
		\label{subfig:checkin}}
	\caption{Illustration of Real datasets.}
	\label{fig:real_dataset}
\end{figure}
\textit{\trajectoryDatasetName{}.} It contains real-time trajectories of $10,357$ taxis' in Beijing from February $2$ to February $8$, 2008.
Each taxi has up to $154,699$ records, where each record comprises~\textit{taxi id}, \textit{date time}, \textit{longitude} and \textit{latitude}. 
For the spatial dimension, we first remove all duplicate records, then extract records with longitude between $116$ and $116.8$ and latitude between $39.5$ and $40.3$, resulting in $14,859,377$ records. We denote this area \big($[116,116.8]\times[39.5,40.3]$\big) as $A_{E}$. 
Figure~\ref{subfig:trajectory} shows $50\%$ of uniformly extracted trajectory points in $A_{E}$.
We further divide $A_{E}$ uniformly into a $10\times 10$ grids, designating these $100$ cells as the location space.
For the time dimension, we sample records every minute and get $8,889$ records.

\textit{\checkInDatasetName{}.} It contains $33,278,683$ Foursquare check-ins from $266,909$ users, during April 2012 to September 2013. 
Each record consists of user id, venue id (place), and time. 
We convert the venue id to the country where the venue is located.
After removing invalid records, we uniformly extract $5\%$ of users' check-ins as shown in Figure \ref{subfig:checkin}.
We set the publication time interval to $100$ minutes, thus divide the chick-ins period into $7,649$ time slots.

\noindent\textbf{Synthetic datasets.} We generate three binary stream datasets using different sequence models. 
Let $p_t=f(t)$ be the probability of setting the real value to $1$ at time slot $t$. 
We set the length of each binary stream as $T$ and the number of users as $N$.
For each stream, we first generate a probability sequence $(p_1, p_2,..., p_T)$.
At each time slot $t$, each user's real value is set to $1$ with probability $p_t$ and $0$ otherwise.
Among the three synthetic sequence models, only TLNS involves randomness due to the Gaussian perturbation term, while the Sin and Log sequences are deterministic once their parameters are fixed.
The probability functions we use are as follows:
\begin{itemize}\itemMargin{}
	\item \tlnsDatasetName{} function. In TLNS, $p_t=p_{t-1}+\algvar{N}(0,Q)$, where $\algvar{N}(0,Q)$ is Gaussian noise with standard variance $\sqrt{Q}=0.0025$. We set $p_0=0.05$ as the initial value.
	If $p_t<0$, we set $p_t=0$; If $p_t>1$, we set $p_t=1$.
	For reproducibility, the Gaussian perturbation in TLNS is generated using Java Random with a fixed seed of 1 in the revised implementation. Under the default setting with sequence length 
	$T=10,000$, clipping occurs in $8$ out of $10,000$ time slots ($0.08\%$). All clipping events correspond to values falling below $0$, while no values exceed $1$. This indicates that clipping is very rare and has negligible impact on the temporal trend and statistical properties of the generated sequence.
	\item \sinDatasetName{} function. In Sin, $p_t=A\sin{(\omega t)}+h$, where $A=0.05$, $\omega=0.01$ and $h=0.075$.
	\item \logDatasetName{} function. In Log, $p_t=A/(1+e^{-bt})$, where $A=0.25$ and $b=0.01$.
\end{itemize}\itemMargin{}
\vspace{-4ex}
\subsection{Experiment Setup}
We divide the total time series into two batches for all datasets, with each batch containing at most half of the total time slots.

We compare our \solutionMethodA{}, \solutionMethodB{}, \solutionMethodD{} and \solutionMethodE{} with three non-personalized methods:  \solutionCMPATotalName{} (\solutionCMPA{}), \solutionCMPBTotalName{} (\solutionCMPB{})~\cite{DBLP:journals/pvldb/KellarisPXP14} and \solutionCMPSPAS{}~\cite{DBLP:journals/pacmmod/LiLCGRQW25}. We also compare against a simple personalized LDP method, \solutionCMPPLDPUTotalName{} (\solutionCMPPLDPU{}), which extends \solutionCMPLDPUTotalName{} (\solutionCMPLDPU{})~\cite{DBLP:conf/sigmod/RenSYYZX22} by replacing the inner CDP mechanism with an LDP mechanism.

Let $\algvar{E}$ and $w$ be the privacy budget and window size in non-personalized static methods (\solutionCMPA{}, \solutionCMPB{} and \solutionCMPSPAS{}).
For non-personalized static methods, we set the $\algvar{E}$ to vary from $0.2$ to $1.0$ and $w$ to vary from $40$ to $200$.
To make our \solutionMethodA{} and \solutionMethodB{} comparable with the non-personalized static methods, we set the lower bound of each user's privacy budget as $\algvar{E}$ and the upper bound of each user's window size as $w$ in \solutionMethodA{} and \solutionMethodB{} to match the requirement of privacy level. Similarly, to make \solutionMethodD{} and \solutionMethodE{} comparable with \solutionMethodA{} and \solutionMethodB{}, we set the lower bound of each user's forward privacy budget as $\algvar{E}_i$ and the upper bound of each user's forward window size as $w_i$.
According to the design of \solutionMethodD{} and \solutionMethodE{}, the backward privacy requirement (backward privacy budget and window size) is independent of the forward privacy requirement.
To study how the forward privacy requirement affects accuracy, we set the backward privacy level to a value that does not impact the publication budget decision, namely, $\epsilon_{F,i,t}^{(2)}=\min\left(\epsilon_{F,i,t}^{(2)},\epsilon_{B,i,t}^{(2)}\right)$ in \solutionMethodD{} and $\epsilon_{\hbox{\scriptsize\em FA},i,t}=\min(\epsilon_{\hbox{\scriptsize\em FA},i,t},\epsilon_{\hbox{\scriptsize\em UB},i,t})$ in \solutionMethodE{}.
Therefore, we set the backward privacy budget sufficiently large (i.e., $\algvar{E}_{B,i,t}=10$) and the backward window size sufficiently small (i.e., $w_{B,i,t}=1$).

Given $\tilde{n}$ different privacy budgets $\vectorfont{\tilde{\epsilon}}=\{\epsilon_1,..., \epsilon_{\tilde{n}}\}$, let $N(\epsilon_i)$ be the count of budget value $\epsilon_i$, and  $N(\vectorfont{\tilde{\epsilon}})=\sum_{i=1}^{\tilde{n}}N(\epsilon_i)$ be the total count of all the budgets.
For any $\epsilon_i\in\vectorfont{\tilde{\epsilon}}$, we define the privacy budget ratio of $\epsilon_i$ as $\frac{N(\epsilon_i)}{N(\vectorfont{\tilde{\epsilon}})}$.
Similarly, we define the window size ratio of any $w_i$ in different window sizes $\vectorfont{\tilde{w}}=\{w_1,...,w_{\tilde{n}}\}$ as $\frac{N(w_i)}{N(\tilde{\vectorfont{w}})}$.
We set the privacy domain as $\{0.5, 1.0\}$ and the window size domain as $\{10, 20\}$.
We alter the ratio $o$ of $\algvar{E}_{i}=0.5$ and  $w_{i}=10$ from $0.1$ to $0.9$.

\begin{table}[t!]\vspace{-2ex}
	\begin{center}
		{\small  
			\caption{\small Experimental settings.} \label{tab:settings}
			\begin{tabular}{l|l}\hline
				{\bf \qquad Parameters \qquad \quad } & {\bf \qquad  \qquad Values \qquad } \\ \hline
				static privacy budget $\algvar{E}$    &       $0.2, 0.4, \textbf{0.6}, 0.8, 1.0$ \\
				static window size $w$      &       $40, 80, \textbf{120}, 160, 200$\\
				personalized privacy budget $\algvar{E}_i$ & $\algvar{E},\ldots, 0.8, 1.0$ \\
				personalized window size $w_i$ & $40, 80,\ldots, w$\\
				users' quantity ratio 	$o$	&  0.1, 0.3, \textbf{0.5}, 0.7, 0.9\\
				forward privacy budget $\algvar{E}_{F,i,t}$ & $\algvar{E}_i,\ldots, 0.8, 1.0$ \\
				forward window size $w_{F,i,t}$ & $40, 80,\ldots, w_i$\\
				backward privacy budget $\algvar{E}_{B,i,t}$ & $10$ \\
				backward window size $w_{B,i,t}$ & $1$\\
				\hline
			\end{tabular}
		}\vspace{-3ex}
	\end{center}
\end{table}

The parameters are shown in Table~\ref{tab:settings}, where the default values are in bold font.   
We run the experiments on an Intel(R) Xeon(R) Silver 4210R CPU @ 2.4GHz with 128 RAM in Java.
Each experiment is run 10 times, and we report the average result.

\subsection{Measures}
We evaluate the performance of different mechanisms based on their data utility. 
We measure data utility as \textit{Average Mean Relative Error} ($\hbox{\em AMRE}$) and \textit{Average Jensen-Shannon Divergence} ($\hbox{\em AJSD}$, $\bar{D}_{\hbox{\scriptsize\em JS}}$).

Let $T$ represent the number of time slots and $d$ denote the dimension of data.
$\hbox{\em AMRE}$ is defined as the average value of Mean Relative Error ($\hbox{\em MRE}$), which is 
\begin{equation}\label{dataUtility}
	\outereqsizelarge{
		\begin{aligned}
			\hbox{\em AMRE} = \frac{1}{T}\sum_{\tau=1}^{T}\hbox{\em MRE}_{\tau} = \frac{1}{T}\sum_{\tau=1}^{T}\frac{1}{d}\|\vectorfont{r}_{\tau}-\vectorfont{c}_{\tau}\|_2^2.
		\end{aligned}
	}
\end{equation}
Besides, $\hbox{\em AJSD}$ is defined as the average value of Jensen-Shannon Divergence ($\hbox{\em JSD}$, $D_{\hbox{\scriptsize\em JS}}$)~\cite{DBLP:journals/tit/Lin91}, which is based on Kullback-Leibler Divergence~\cite{kullback1951information}, as 
\begin{equation}\label{jsUtility}
	\outereqsize{
		\begin{aligned}
			&\bar{D}_{\hbox{\scriptsize\em JS}}(\vectorfont{r}\|\vectorfont{c}) \\
			=& \frac{1}{T}\sum_{\tau=1}^{T}D_{\hbox{\scriptsize\em JS}}(\vectorfont{r}\|\vectorfont{c}) \\
			=&\frac{1}{T}\sum_{\tau=1}^{T}\left(\frac{1}{2}D_{\hbox{\scriptsize\em KL}}(\vectorfont{r}\|\vectorfont{v})+\frac{1}{2}D_{\hbox{\scriptsize\em KL}}(\vectorfont{c}\|\vectorfont{v})\right) \\ =&\frac{1}{2T}\sum_{\tau=1}^{T}\sum_{j=1}^{d}\left(\vectorfont{r}_{\tau}(j)\log{\left(\frac{\vectorfont{r}_{\tau}(j)}{\vectorfont{v}_{\tau}(j)}\right)}+\vectorfont{c}_{\tau}(j)\log{\left(\frac{\vectorfont{c}_{\tau}(j)}{\vectorfont{v}_{\tau}(j)}\right)}\right),
		\end{aligned}
	}
\end{equation}
where $\vectorfont{v}$ represents the average distribution of $\vectorfont{r}$ and $\vectorfont{c}$, i.e., $\vectorfont{v}(j)=\frac{1}{2}(\vectorfont{r}(j)+\vectorfont{c}(j))$.
For time slot $\tau$, $r_{\tau}(j)$ and $c_{\tau}(j)$ represent the $j$-th dimensional values in the obfuscated and original data, respectively.
In this subsection, we compare the performance of \solutionCMPA{}, \solutionCMPB{}, \solutionCMPPLDPU{},  \solutionMethodA{} and \solutionMethodB{} using $\hbox{\em AJSD}$ metric.

\begin{table*}[t!]\vspace{-2ex}
	\begin{center}
		\caption{\small Average Mean Relative Error ($\hbox{\em AMRE}$) with $\algvar{E}$ varied.} \label{tab:alter_e}
		\resizebox{\textwidth}{!}{
			{\tiny
				\begin{tabular}{ccrrrrr}
					\hline
					\textbf{Datasets}                    & \textbf{Methods} & $\algvar{E}$=0.2 & $\algvar{E}$=0.4 & $\algvar{E}$=0.6 & $\algvar{E}$=0.8 & $\algvar{E}$=1 \\ \hline
					\multirow{8}{*}{\textbf{\trajectoryDatasetName{}}}&\solutionCMPA{}		&8,459.58		&2,156.19		&990.69		&610.25		&409.43 \\
					&\solutionCMPB{}		&3,050.51		&1,495.79		&961.62		&679.88		&558.63 \\
					&\solutionCMPPLDPU{}		&34,419.70		&34,415.11		&34,417.54		&34,416.61		&34,416.79 \\
					&\solutionCMPSPAS{}		&618,995.29		&154,449.26		&68,648.03		&38,946.65		&25,004.66 \\
					&\solutionMethodA{}		&1,203.40		&\textbf{449.91}		&\textbf{275.14}		&\textbf{204.53}		&\textbf{166.84} \\
					&\solutionMethodB{}		&2,874.89		&1,369.62		&1,041.83		&869.90		&723.00 \\
					&\solutionMethodD{}		&\textbf{613.25}		&\textbf{327.08}		&\textbf{255.70}		&\textbf{223.15}		&\textbf{195.87} \\
					&\solutionMethodE{}		&\textbf{795.28}		&534.17		&461.53		&419.98		&396.67 \\
					\hline
					\multirow{8}{*}{\textbf{\checkInDatasetName{}}}&\solutionCMPA{}		&13,725.14		&3,544.35		&1,704.59		&938.09		&664.04 \\
					&\solutionCMPB{}		&7,162.41		&3,411.68		&2,225.49		&1,332.70		&1,159.01 \\
					&\solutionCMPPLDPU{}		&180,722.61		&180,706.57		&180,700.83		&180,694.79		&180,687.04 \\
					&\solutionCMPSPAS{}		&603,754.52		&149,939.73		&67,711.28		&37,846.80		&24,564.06 \\
					&\solutionMethodA{}		&2,185.22		&\textbf{725.76}		&\textbf{482.02}		&\textbf{347.09}		&\textbf{253.75} \\
					&\solutionMethodB{}		&7,491.81		&3,681.56		&2,672.96		&1,990.94		&1,377.30 \\
					&\solutionMethodD{}		&\textbf{460.58}		&\textbf{167.42}		&\textbf{114.14}		&\textbf{82.10}		&\textbf{66.04} \\
					&\solutionMethodE{}		&\textbf{1,675.06}		&1,020.44		&749.20		&551.41		&407.95 \\
					\hline
					\multirow{8}{*}{\textbf{\tlnsDatasetName{}}}&\solutionCMPA{}		&35,214,891.91		&26,932,213.28		&21,900,899.91		&19,202,732.21		&19,317,360.81 \\
					&\solutionCMPB{}		&72,479.55		&13,456.32		&5,228.94		&2,583.01		&\textbf{1,447.21} \\
					&\solutionCMPPLDPU{}		&9,750,998.15		&9,732,708.88		&9,719,230.14		&9,705,975.87		&9,688,944.63 \\
					&\solutionCMPSPAS{}		&3,276,401.46		&2,774,019.75		&2,506,067.88		&2,359,560.64		&2,246,597.48 \\
					&\solutionMethodA{}		&5,625,685.13		&5,708,081.90		&5,761,339.31		&5,761,668.46		&6,168,585.37 \\
					&\solutionMethodB{}		&\textbf{15,734.77}		&\textbf{7,593.93}		&\textbf{4,031.24}		&\textbf{2,415.89}		&1,449.54 \\
					&\solutionMethodD{}		&4,108,894.40		&4,001,623.91		&4,254,805.65		&4,059,568.08		&4,156,445.29 \\
					&\solutionMethodE{}		&\textbf{3,249.53}		&\textbf{2,096.95}		&\textbf{1,309.70}		&\textbf{1,097.65}		&\textbf{860.35} \\
					\hline
					\multirow{8}{*}{\textbf{\sinDatasetName{}}}&\solutionCMPA{}		&14,809,297.94		&3,868,423.13		&2,822,001.85		&2,166,052.12		&2,875,381.82 \\
					&\solutionCMPB{}		&61,604.76		&21,010.09		&9,193.50		&4,678.60		&\textbf{2,659.74} \\
					&\solutionCMPPLDPU{}		&18,127,404.96		&18,103,836.01		&18,082,416.37		&18,052,150.71		&18,028,552.72 \\
					&\solutionCMPSPAS{}		&234,146.85		&208,385.79		&190,648.17		&185,012.37		&179,534.63 \\
					&\solutionMethodA{}		&1,065,153.09		&891,132.17		&796,008.09		&753,028.43		&689,846.46 \\
					&\solutionMethodB{}		&\textbf{26,390.31}		&\textbf{12,807.53}		&\textbf{7,805.06}		&\textbf{4,123.15}		&2,661.26 \\
					&\solutionMethodD{}		&376,035.19		&331,161.83		&321,423.27		&315,928.35		&312,774.69 \\
					&\solutionMethodE{}		&\textbf{8,258.50}		&\textbf{4,668.58}		&\textbf{3,350.10}		&\textbf{2,269.24}		&\textbf{1,974.88} \\
					\hline
					\multirow{8}{*}{\textbf{\logDatasetName{}}}&\solutionCMPA{}		&12,598,827.10		&3,106,749.11		&1,520,101.37		&1,159,412.40		&912,889.07 \\
					&\solutionCMPB{}		&25,313.21		&7,067.12		&3,701.09		&2,351.14		&\textbf{1,763.29} \\
					&\solutionCMPPLDPU{}		&6,334,895.72		&6,323,757.64		&6,316,533.20		&6,306,368.84		&6,298,869.11 \\
					&\solutionCMPSPAS{}		&23,030.09		&21,985.85		&21,573.75		&21,520.27		&20,983.06 \\
					&\solutionMethodA{}		&580,200.16		&438,397.59		&438,337.35		&362,035.74		&320,952.59 \\
					&\solutionMethodB{}		&\textbf{9,856.66}		&\textbf{4,316.13}		&\textbf{3,365.33}		&\textbf{2,339.70}		&1,837.17 \\
					&\solutionMethodD{}		&75,571.46		&59,450.90		&50,179.75		&51,748.11		&46,332.37 \\
					&\solutionMethodE{}		&\textbf{1,788.03}		&\textbf{1,415.64}		&\textbf{1,405.98}		&\textbf{1,174.66}		&\textbf{1,178.10} \\
					\hline
				\end{tabular}
			}
		}
		\vspace{-4ex}
	\end{center}
\end{table*}\vspace{-2ex}

\subsection{Overall Utility Analysis}\label{exp:utility}

Table~\ref{tab:alter_e} shows the average mean relative error $\hbox{\em AMRE}$ as the privacy budget $\algvar{E}$ varies.
Across most datasets, $\hbox{\em AMRE}$ generally decreases as $\algvar{E}$ increases, since a larger $\algvar{E}$ reduces the variance of the injected noise.
However, the sensitivity of $\hbox{\em AMRE}$ to $\algvar{E}$ differs across mechanisms and datasets. On real datasets, distribution-based methods such as \solutionMethodA{} and \solutionMethodD{} show a more evident reduction as $\algvar{E}$ increases, whereas on synthetic datasets, absorption-based methods such as \solutionMethodB{} and \solutionMethodE{} are more sensitive to the privacy budget and achieve larger reductions. This difference is mainly caused by the temporal characteristics of the streams: real datasets contain more abrupt changes and thus benefit from more responsive budget distribution, while synthetic datasets evolve more smoothly and benefit more from absorbing budgets for fewer but more accurate releases.

We attribute this difference mainly to the \textbf{temporal characteristics of the data streams}. The real datasets exhibit stronger temporal variability and more abrupt changes across consecutive time slots, whereas the synthetic datasets are relatively smoother over time.
When the stream changes rapidly, the dissimilarity between the current statistics and the previous release becomes large, thus, mechanisms that allocate budget more responsively to the current time slot are more effective.
In this case, \solutionMethodA{} publishes more new statistical results than \solutionMethodB{}, because \solutionMethodA{} always reserves part of its privacy budget for the next time slot. Therefore, \solutionMethodA{} achieves lower \hbox{\em AMRE} than \solutionMethodB{} on the real datasets.

In contrast, when the stream evolves more smoothly, the dissimilarity between consecutive time slots remains relatively small, when the density function changes gradually, the dissimilarity at each time slot remains small.
In such cases, concentrating budget on fewer but more accurate releases is more beneficial than publishing more frequently.
Therefore, \solutionMethodB{} performs significantly better than \solutionMethodA{} on the synthetic datasets.

\solutionMethodD{} performs better than \solutionMethodA{}, while \solutionMethodE{} performs better than \solutionMethodB{}.
This improved performance occurs because of the dynamic personalized setting, where users' requirements may vary over time. These results show that the dynamic mechanisms can exploit time-varying feasible budgets while still satisfying the declared per-user constraints.
\solutionCMPPLDPU{} performs worse than other methods across all datasets except for TLNS, since LDP methods achieve lower accuracy than CDP methods under the same privacy budget.

The comparison with \solutionCMPSPAS{} further confirms this observation. On the two real datasets, Taxi and Foursquare, \solutionCMPSPAS{} generally yields much larger \hbox{\em AMRE} than the proposed personalized mechanisms. Although \solutionCMPSPAS{} adaptively allocates privacy budgets under homogeneous $w$-event privacy, it still relies on a single global privacy requirement and cannot exploit heterogeneous user-specific budgets and window sizes. This limitation becomes more evident on real streams with abrupt temporal changes, where more responsive personalized budget distribution is needed. In contrast, on the synthetic datasets, \solutionCMPSPAS{} is more competitive with some distribution-based methods, especially when the stream evolves smoothly. Nevertheless, the absorption-based methods, particularly \solutionMethodB{} and \solutionMethodE{}, still achieve smaller \hbox{\em AMRE} in most cases because they can skip or nullify unnecessary publications, accumulate more budget for informative releases, and use OBS to reduce the reporting error.

Beyond the comparison with \solutionCMPSPAS{}, Table~\ref{tab:alter_e} reveals a more
specific behavior of the absorption-based mechanisms: \solutionMethodB{}
yields higher \hbox{\em AMRE} than homogeneous \solutionCMPB{} on \checkInDatasetName{} for
all tested privacy budgets. This result indicates that relaxing
the per-user privacy constraints does not necessarily guarantee
lower end-to-end error for a concrete personalized mechanism.
Unlike \solutionCMPB{}, \solutionMethodB{} uses OBS and the sampling mechanism to
transform heterogeneous user budgets into a shared release
threshold, thereby introducing additional sampling variance
and bias.

To identify which term mainly contributes to this gap, we further decompose the OBS error in Appendix~8.4 of the Electronic Supplementary Material (ESM).
Figure~18 in the ESM shows that, 
whenever sampling is
performed, the squared sampling bias is consistently much
larger than the sampling variance in both $\textrm{Part}_{\hbox{\scriptsize DC}}$ and
$\textrm{Part}_{\hbox{\scriptsize NOP}}$. In $\textrm{Part}_{\hbox{\scriptsize DC}}$, this bias affects the private
dissimilarity estimate, while in $\textrm{Part}_{\hbox{\scriptsize NOP}}$, it contributes to
the reporting-error threshold used to determine whether a
new result should be published. On the rapidly changing
Foursquare stream, these effects can be further amplified by
skipped or nullified publications and the reuse of historical
values. Therefore, temporal dynamics act as an amplifying
factor rather than the sole explanation for the performance
gap between \solutionMethodB{} and \solutionCMPB{}.

Overall, for the real datasets, our \solutionMethodA{} consistently outperforms non-personalized methods.
The $\hbox{\em AMRE}$ of \solutionMethodA{} is on average $72.6\%$  lower than that of \solutionCMPA{} on
\trajectoryDatasetName{} dataset and $72.0\%$ lower on \checkInDatasetName{} dataset. 
Besides, the $\hbox{\em AMRE}$ of \solutionMethodD{} is on average $73.5\%$ lower than that of \solutionCMPA{} on
\trajectoryDatasetName{} dataset and $93.3\%$ lower on \checkInDatasetName{} dataset. 
We note that this performance gap is not primarily explained by dimensionality alone; the additional dimensionality analysis in Figure~17 of Appendix~8.3 in the Electronic Supplementary Material
(ESM) shows that \solutionMethodA{}/\solutionMethodD{} remain better than \solutionMethodB{}/\solutionMethodE{} across different dimensions on the real datasets.

For synthetic datasets, our \solutionMethodB{} consistently outperforms other non-personalized methods and our \solutionMethodE{} further improves upon \solutionMethodB{}.
Compared to \solutionCMPB{}, the $\hbox{\em AMRE}$ of \solutionMethodB{} is lower on average of $30.2\%$ on \tlnsDatasetName{} , $24.6\%$ on \sinDatasetName{}, and $21.1\%$ on \logDatasetName{}.
Besides, the $\hbox{\em AMRE}$ of \solutionMethodE{} is lower on average by $70.6\%$ on \tlnsDatasetName{}, $61.0\%$ on \sinDatasetName{}, and $63.6\%$ on \logDatasetName{}.
Moreover, our \solutionMethodA{} consistently outperforms \solutionCMPA{}.

\begin{table*}[t!]\vspace{-2ex}
	\begin{center}
		\caption{\small  Average Mean Relative Error (${\hbox{\em AMRE}}$) with $w$ varied.} \label{tab:alter_w}
		\resizebox{\textwidth}{!}{
			{\tiny 
				\begin{tabular}{ccrrrrr}
					\hline
					\textbf{Datasets}                    & \textbf{Methods} & $w$=40        & $w$=80        & $w$=120       & $w$=160       & $w$=200       \\ \hline
					\multirow{8}{*}{\textbf{\trajectoryDatasetName{}}}&\solutionCMPA{}		&163.02		&447.18		&990.69		&1,853.31		&3,060.57 \\
					&\solutionCMPB{}		&294.60		&608.41		&961.62		&1,333.79		&1,664.90 \\
					&\solutionCMPPLDPU{}		&34,410.72		&34,415.09		&34,417.54		&34,415.33		&34,416.53 \\
					&\solutionCMPSPAS{}		&\textbf{29.37}		&39,224.91		&68,648.03		&84,345.03		&69,863.51 \\
					&\solutionMethodA{}		&100.42		&\textbf{206.72}		&\textbf{275.14}		&\textbf{477.71}		&\textbf{673.23} \\
					&\solutionMethodB{}		&303.99		&747.23		&1,041.83		&1,549.01		&1,879.14 \\
					&\solutionMethodD{}		&157.04		&\textbf{177.12}		&\textbf{255.70}		&\textbf{275.78}		&\textbf{297.00} \\
					&\solutionMethodE{}		&\textbf{98.77}		&333.17		&461.53		&642.48		&784.93 \\
					\hline
					\multirow{8}{*}{\textbf{\checkInDatasetName{}}}&\solutionCMPA{}		&241.39		&687.78		&1,704.59		&2,910.86		&4,783.33 \\
					&\solutionCMPB{}		&255.27		&1,240.90		&2,225.49		&3,112.29		&3,758.29 \\
					&\solutionCMPPLDPU{}		&180,654.54		&180,701.20		&180,700.83		&180,713.97		&180,722.91 \\
					&\solutionCMPSPAS{}		&\textbf{36.42}		&37,353.95		&67,711.28		&81,441.64		&68,097.75 \\
					&\solutionMethodA{}		&120.07		&\textbf{297.72}		&\textbf{482.02}		&\textbf{819.29}		&\textbf{1,022.65} \\
					&\solutionMethodB{}		&981.55		&1,988.55		&2,672.96		&3,715.36		&4,636.11 \\
					&\solutionMethodD{}		&\textbf{102.81}		&\textbf{105.90}		&\textbf{114.14}		&\textbf{112.99}		&\textbf{128.03} \\
					&\solutionMethodE{}		&109.47		&388.57		&749.20		&1,245.98		&1,634.49 \\
					\hline
					\multirow{8}{*}{\textbf{\tlnsDatasetName{}}}&\solutionCMPA{}		&10,321,505.61		&16,427,785.34		&21,900,899.91		&22,192,516.74		&22,368,956.59 \\
					&\solutionCMPB{}		&387.96		&1,948.21		&5,228.94		&9,421.72		&16,065.42 \\
					&\solutionCMPPLDPU{}		&9,629,126.83		&9,695,808.71		&9,719,230.14		&9,732,079.75		&9,741,718.28 \\
					&\solutionCMPSPAS{}		&4,002,347.34		&1,766,335.27		&2,506,067.88		&3,060,865.41		&3,527,531.79 \\
					&\solutionMethodA{}		&4,177,406.45		&4,841,586.74		&5,761,339.31		&5,943,316.00		&6,268,467.82 \\
					&\solutionMethodB{}		&\textbf{381.45}		&\textbf{1,835.74}		&\textbf{4,031.24}		&\textbf{6,285.94}		&\textbf{7,870.36} \\
					&\solutionMethodD{}		&3,579,033.68		&3,799,459.12		&4,254,805.65		&4,320,520.18		&4,344,081.28 \\
					&\solutionMethodE{}		&\textbf{225.51}		&\textbf{846.94}		&\textbf{1,309.70}		&\textbf{2,073.61}		&\textbf{2,985.93} \\
					\hline
					\multirow{8}{*}{\textbf{\sinDatasetName{}}}&\solutionCMPA{}		&769,929.84		&1,285,735.11		&2,822,001.85		&3,444,317.33		&6,933,306.34 \\
					&\solutionCMPB{}		&\textbf{611.27}		&3,322.22		&9,193.50		&16,455.49		&25,231.48 \\
					&\solutionCMPPLDPU{}		&17,907,365.61		&18,032,304.84		&18,082,416.37		&18,107,904.18		&18,118,392.15 \\
					&\solutionCMPSPAS{}		&226,588.50		&170,583.07		&190,648.17		&207,628.80		&217,991.62 \\
					&\solutionMethodA{}		&378,448.77		&572,317.51		&796,008.09		&1,144,032.83		&1,339,405.34 \\
					&\solutionMethodB{}		&649.14		&\textbf{3,298.99}		&\textbf{7,805.06}		&\textbf{13,784.27}		&\textbf{17,342.63} \\
					&\solutionMethodD{}		&260,132.70		&300,682.17		&321,423.27		&345,762.45		&365,804.39 \\
					&\solutionMethodE{}		&\textbf{445.32}		&\textbf{1,764.71}		&\textbf{3,350.10}		&\textbf{5,505.09}		&\textbf{7,833.50} \\
					\hline
					\multirow{8}{*}{\textbf{\logDatasetName{}}}&\solutionCMPA{}		&110,158.58		&750,581.32		&1,520,101.37		&3,170,285.91		&8,658,570.30 \\
					&\solutionCMPB{}		&\textbf{770.17}		&2,006.90		&3,701.09		&5,943.62		&7,973.24 \\
					&\solutionCMPPLDPU{}		&6,256,580.11		&6,298,876.91		&6,316,533.20		&6,325,047.93		&6,327,726.59 \\
					&\solutionCMPSPAS{}		&32,948.08		&22,013.63		&21,573.75		&24,442.50		&22,909.65 \\
					&\solutionMethodA{}		&47,909.95		&207,973.60		&438,337.35		&582,733.67		&759,567.65 \\
					&\solutionMethodB{}		&820.44		&\textbf{1,946.36}		&\textbf{3,365.33}		&\textbf{4,270.21}		&\textbf{6,007.43} \\
					&\solutionMethodD{}		&31,028.42		&47,844.35		&50,179.75		&60,542.04		&66,902.55 \\
					&\solutionMethodE{}		&\textbf{633.30}		&\textbf{1,162.11}		&\textbf{1,405.98}		&\textbf{1,617.49}		&\textbf{2,122.58} \\
					\hline
				\end{tabular}
			}
		}\vspace{-2ex}
	\end{center}
\end{table*}

Table~\ref{tab:alter_w} shows the average mean relative error $\hbox{\em AMRE}$ as the window size $w$ varies.
As $w$ increases, $\hbox{\em AMRE}$ generally increases.
This occurs because a large window size results in a small privacy budget at each time slot, leading to increased error.
\solutionCMPPLDPU{} shows lower performance than other methods on most datasets, since LDP methods achieve lower accuracy than CDP methods under equivalent privacy budgets.
\solutionCMPSPAS{} is sensitive to the window size. Although it achieves very small \hbox{\em AMRE} on the real datasets when $w=40$, its error increases sharply when the window size becomes larger. This suggests that the homogeneous adaptive strategy of \solutionCMPSPAS{} becomes less robust when the privacy budget must be allocated over longer windows. In comparison, the proposed personalized mechanisms show more stable performance across different window sizes. On real datasets, \solutionMethodA{} and \solutionMethodD{} achieve lower \hbox{\em AMRE} than \solutionCMPA{}, while on synthetic datasets, \solutionMethodB{} and \solutionMethodE{} remain more effective due to their budget absorption strategy.

Overall, for real datasets, our \solutionMethodD{} generally achieves lower \hbox{\em AMRE} than \solutionCMPA{} and remains stable across different window sizes. Although SPAS obtains very small \hbox{\em AMRE} when $w=40$, its error increases sharply for larger window sizes, indicating that it is sensitive to the window-size setting.
The $\hbox{\em AMRE}$ of \solutionMethodA{} is on average $63.3\%$ lower than that of \solutionCMPA{} on \trajectoryDatasetName{} dataset and $65.8\%$ on \checkInDatasetName{} dataset. 
Besides, the $\hbox{\em AMRE}$ of \solutionMethodD{} is on average $62.7\%$ lower than that of \solutionCMPA{} on \trajectoryDatasetName{} dataset and $85.8\%$ on \checkInDatasetName{} dataset.
For synthetic datasets, our \solutionMethodE{} demonstrates the lowest error among all non-dynamic methods.
Compared to \solutionCMPB{}, the $\hbox{\em AMRE}$ of \solutionMethodB{} is lower by an average of $22.9\%$ for \tlnsDatasetName{}, $11.4\%$ for \sinDatasetName{}, and $11.7\%$ for \logDatasetName{}, respectively. \solutionMethodE{} further reduces the $\hbox{\em AMRE}$ by an average of $66.6\%$ for \tlnsDatasetName{}, $54.6\%$ for \sinDatasetName{}, and $53.6\%$ for \logDatasetName{}, respectively.
Moreover, our \solutionMethodA{} consistently outperforms \solutionCMPA{} across all datasets.

In summary, our \solutionMethodA{} and \solutionMethodD{{} demonstrate superior performance on real datasets, with an $\hbox{\em AMRE}$ at least $63.3\%$ and $62.7\%$ lower than \solutionCMPA{}, respectively. 
	For synthetic datasets, our \solutionMethodB{} and \solutionMethodE{} outperform \solutionCMPB{} with at least $11.4\%$ and $53.6\%$ smaller $\hbox{\em AMRE}$, respectively.
	
	The detailed AJSD results are provided in Appendix~8.2 of the Electronic Supplementary Material (ESM).
	
	\subsection{Impact of User Requirement Type}
	
	\begin{figure*}[t!]\centering
		\subfigure{
			\scalebox{0.27}[0.27]{\includegraphics{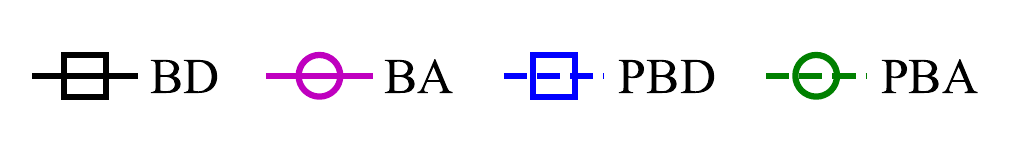}}}\hfill\\
		\addtocounter{subfigure}{-1}\vspace{-3ex}
		\subfigure[][{\small \trajectoryDatasetName{}}]{
			\scalebox{0.177}[0.177]{\includegraphics{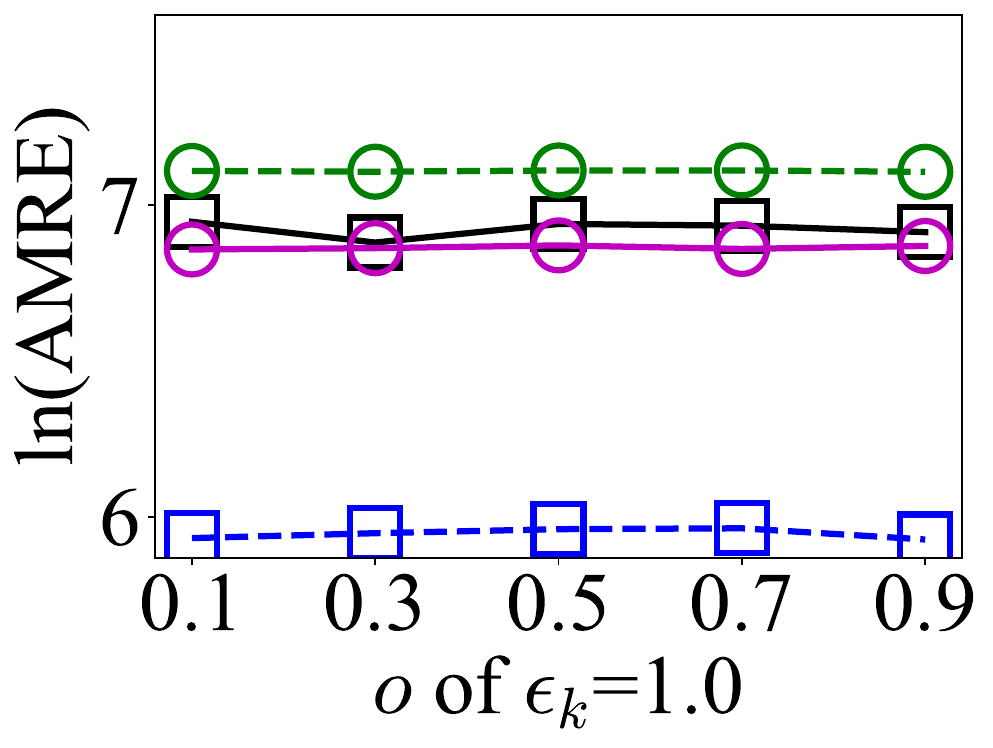}}
			\label{subfig:trajectory_ratio_change_two_budget}}\hfill
		\subfigure[][{\small \checkInDatasetName{}}]{
			\scalebox{0.177}[0.177]{\includegraphics{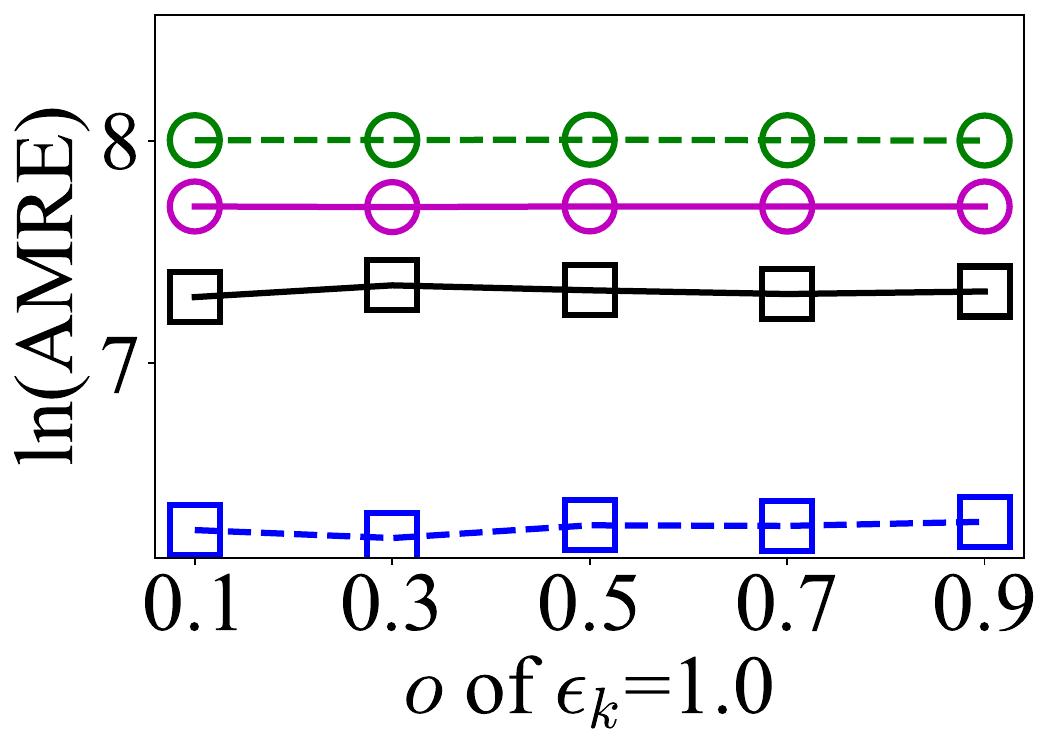}}
			\label{subfig:check_in_ratio_change_two_budget}}\hfill	
		\subfigure[][{\small \tlnsDatasetName{}}]{
			\scalebox{0.177}[0.177]{\includegraphics{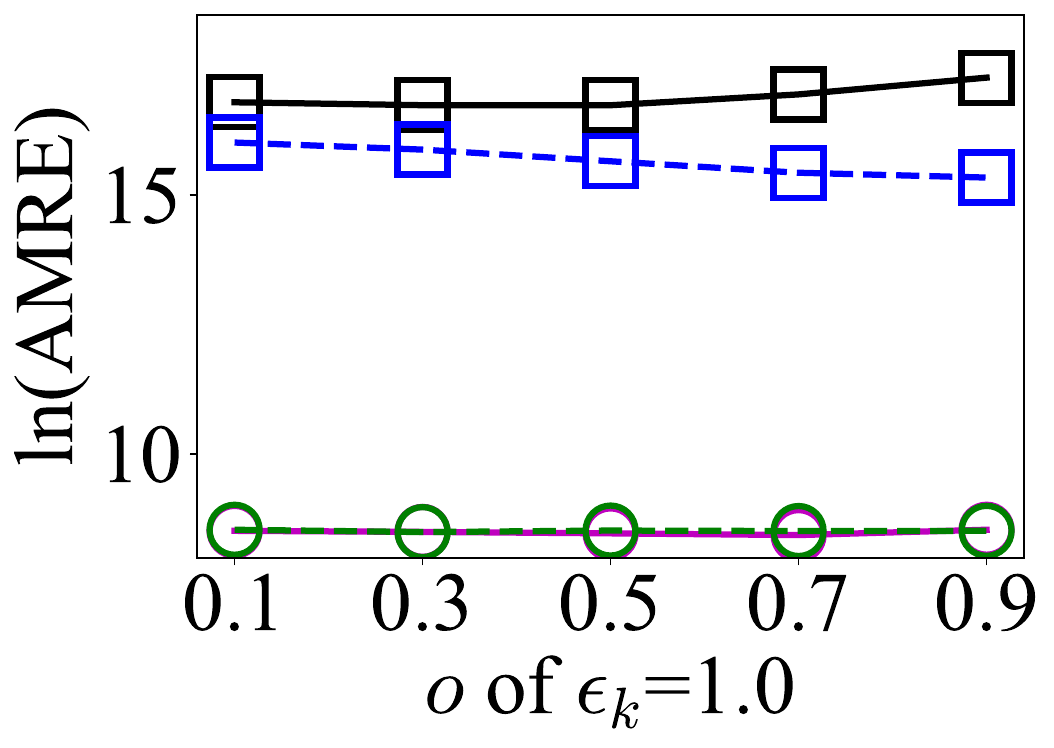}}
			\label{subfig:tlns_ratio_change_two_budget}}\hfill 
		\subfigure[][{\small \sinDatasetName{}}]{
			\scalebox{0.177}[0.177]{\includegraphics{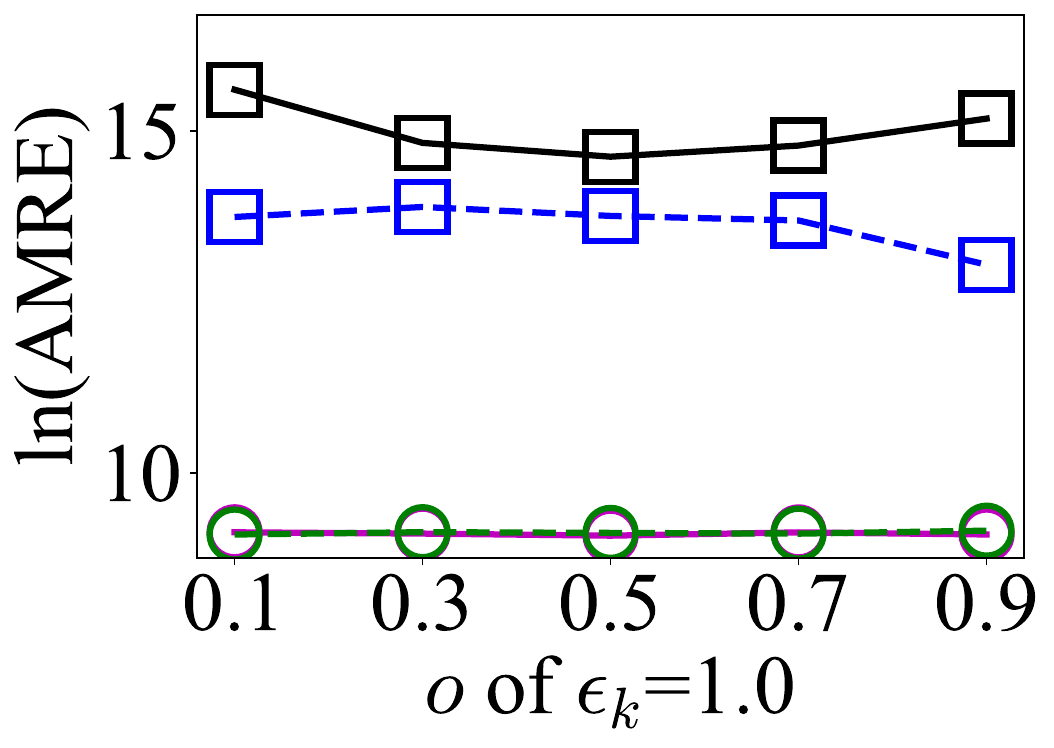}}
			\label{subfig:sin_ratio_change_two_budget}}\hfill 
		\subfigure[][{\small \logDatasetName{}}]{
			\scalebox{0.177}[0.177]{\includegraphics{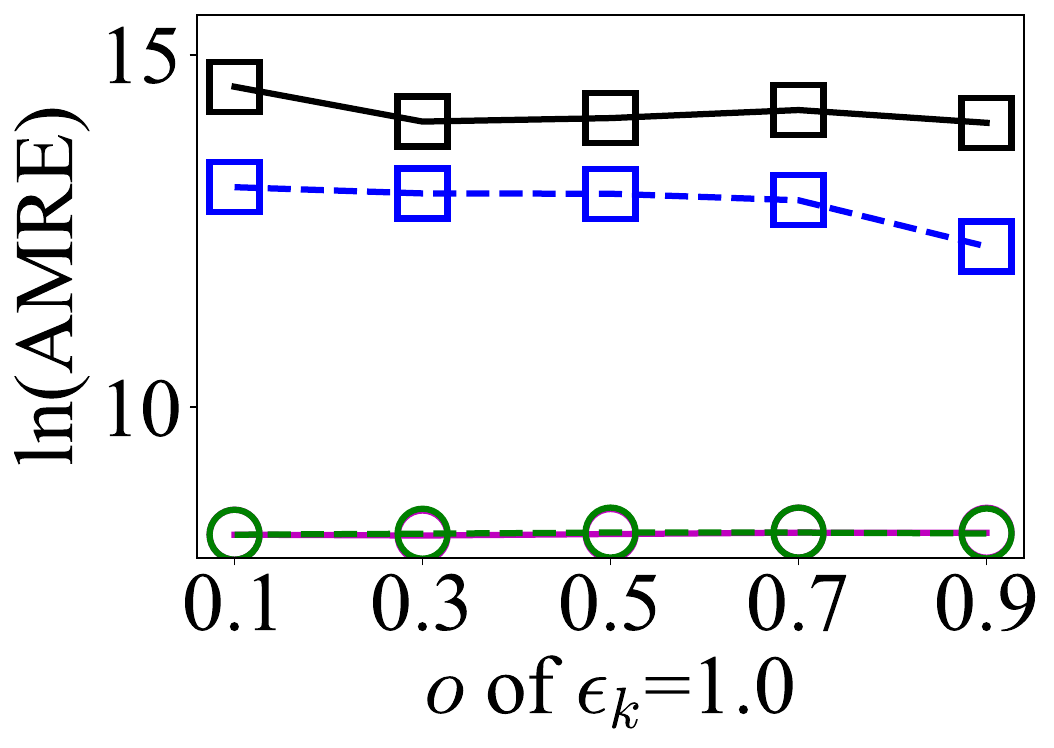}}
			\label{subfig:log_ratio_change_two_budget}}\hfill 
		\caption{\small  Average Mean Relative Error ($\hbox{\em AMRE}$) with  ratio for  privacy budget varied.}\vspace{-2ex}
		\label{fig:alter_ratio_two_budget}
	\end{figure*}
	
	\begin{figure*}[!t]\centering
		\subfigure{
			\scalebox{0.27}[0.27]{\includegraphics{bar_2.pdf}}}\hfill\\
		\addtocounter{subfigure}{-1}\vspace{-2.5ex}
		\subfigure[][{\small \trajectoryDatasetName{}}]{
			\scalebox{0.177}[0.177]{\includegraphics{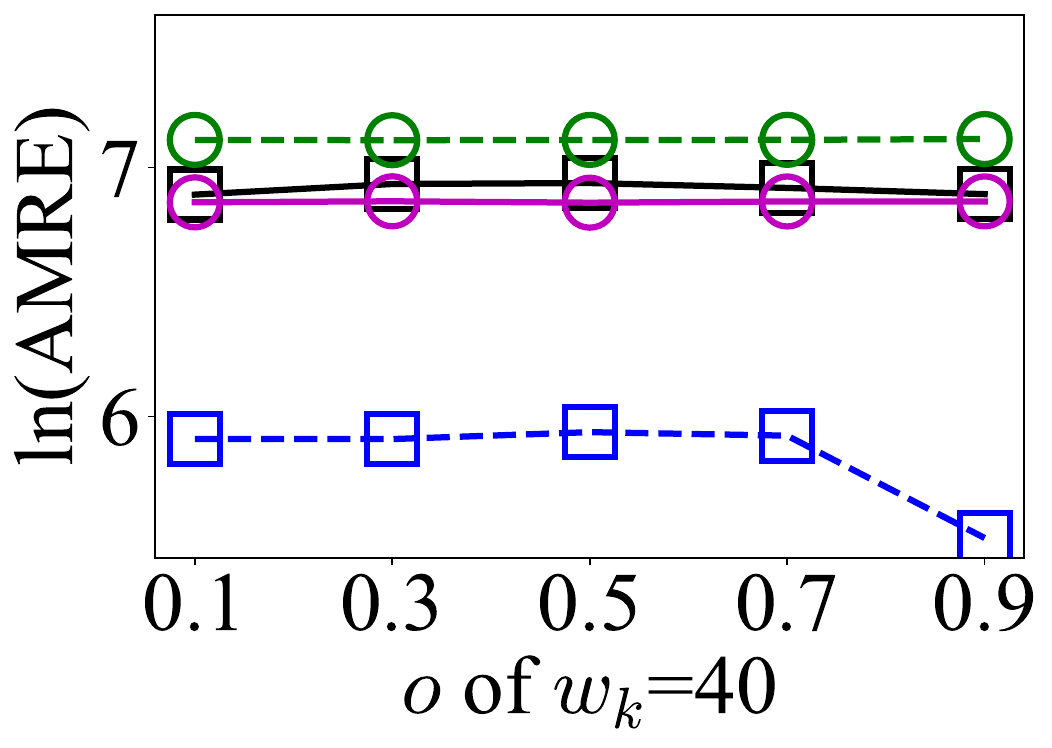}}
			\label{subfig:trajectory_ratio_change_two_window_size}}\hfill
		\subfigure[][{\small \checkInDatasetName{}}]{
			\scalebox{0.177}[0.177]{\includegraphics{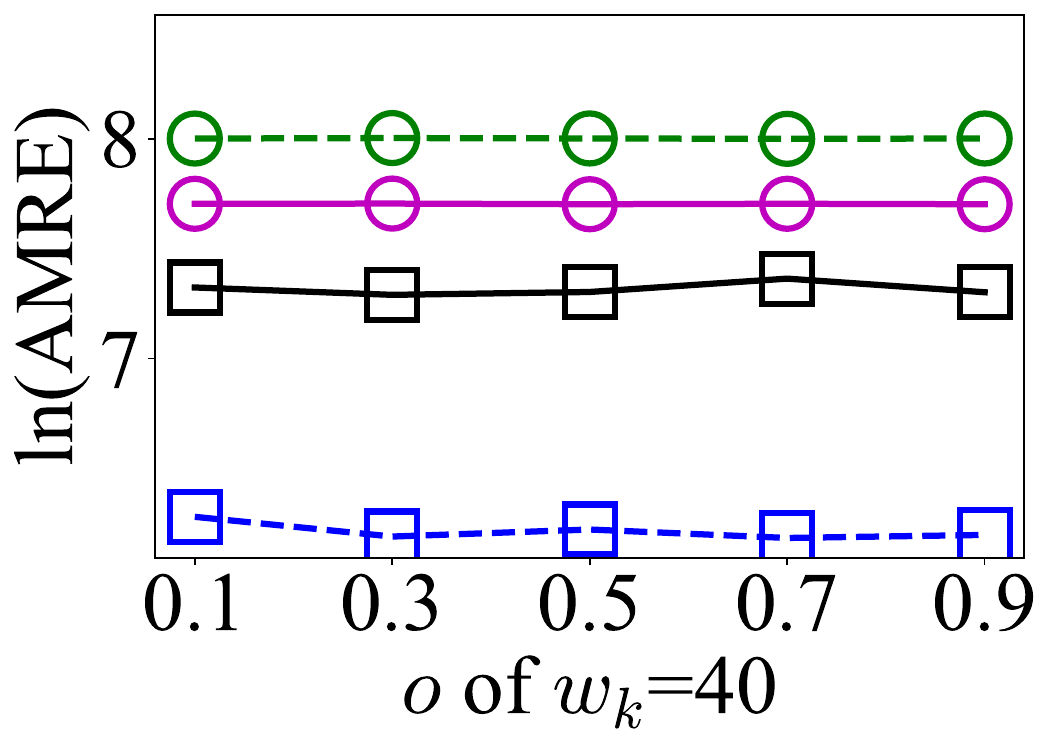}}
			\label{subfig:check_in_ratio_change_two_window_size}}\hfill	
		\subfigure[][{\small \tlnsDatasetName{}}]{
			\scalebox{0.177}[0.177]{\includegraphics{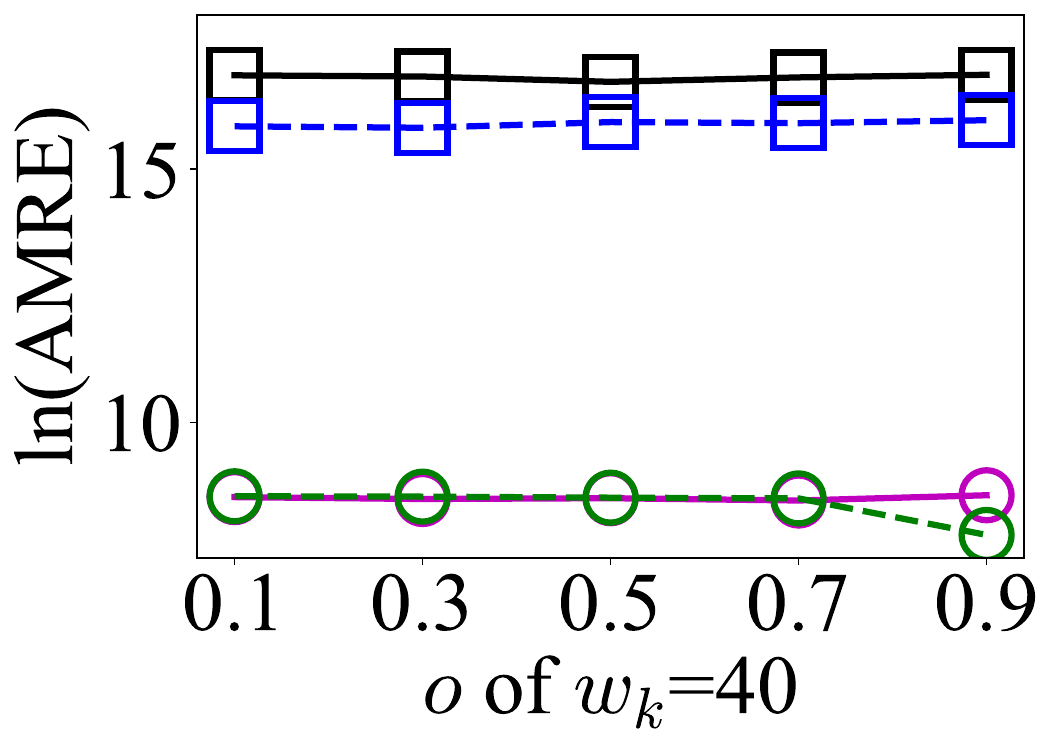}}
			\label{subfig:tlns_ratio_change_two_window_size}}\hfill 
		\subfigure[][{\small \sinDatasetName{}}]{
			\scalebox{0.177}[0.177]{\includegraphics{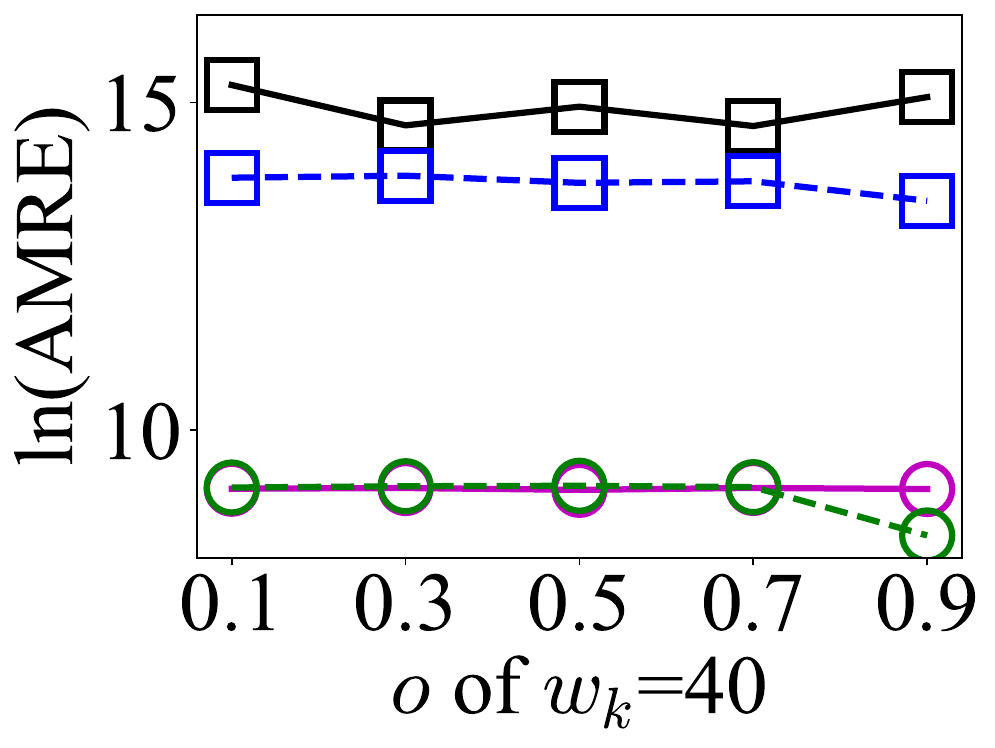}}
			\label{subfig:sin_ratio_change_two_window_size}}\hfill 
		\subfigure[][{\small \logDatasetName{}}]{
			\scalebox{0.177}[0.177]{\includegraphics{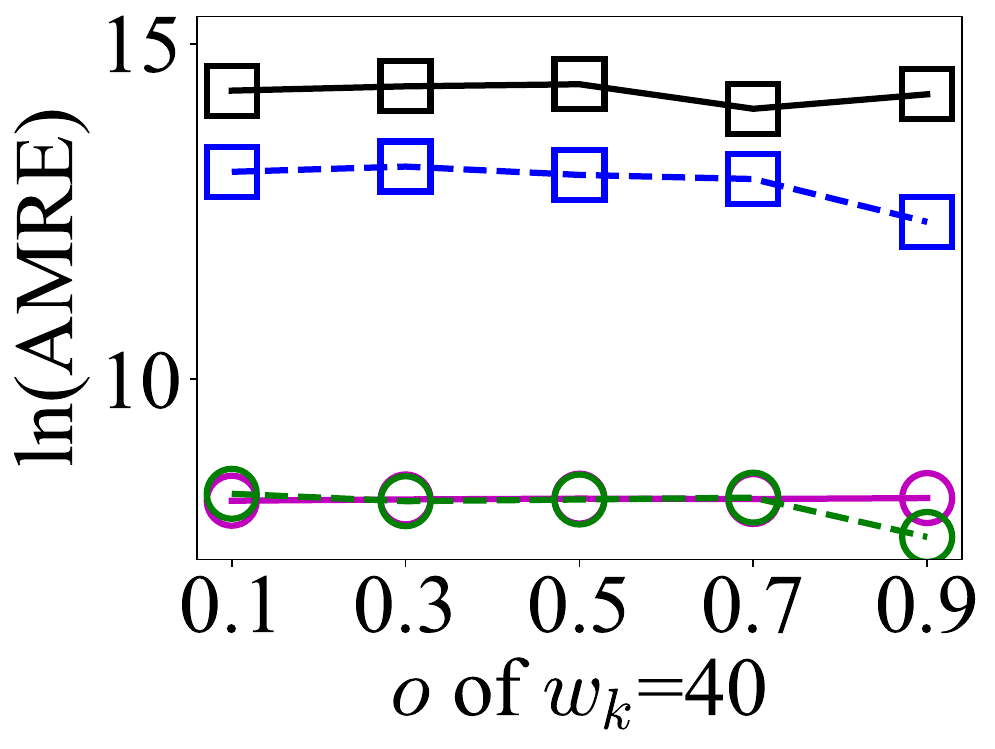}}
			\label{subfig:log_ratio_change_two_window_size}}\hfill 
		\caption{\small Average Mean Relative Error ($\hbox{\em AMRE}$) with ratio for  window size varied.}
		\label{fig:alter_ratio_two_window_size}
	\end{figure*}
	
	We define a set of users with  privacy  requirement as \textit{$(w_k,\algvar{E}_k)$-requirement type}.
	In this subsection, we examine the impact of user type on the utility.
	For our analysis, we set $\algvar{E}_k$ candidate set as $\{0.6,1.0\}$ with a default value of $0.6$, and the $w_k$ candidate set as $\{40,120\}$ with a default value of $120$.
	We first vary the users' quantity ratio of $\algvar{E}_k=1.0$ from $0.1$ to $0.9$ while keeping $w_k=120$, and then
	vary the users' quantity ratio of $w_k=40$ from $0.1$ to $0.9$ while keeping $\algvar{E}_k=0.6$.
	We analyze the impact of these ratio variations on $\hbox{\em AMRE}$.
	
	Figure~\ref{fig:alter_ratio_two_budget} illustrates the change in users' quantity ratio for $\algvar{E}_k=1.0$ from $0.1$ to $0.9$, with a fixed window size of $w_k=120$.
	Figure~\ref{fig:alter_ratio_two_window_size} shows the effect on changing users' quantity for $w_k=40$ from $0.1$ to $0.9$, with a fixed privacy budget of $\algvar{E}_k=0.6$.
	We observe that as the users' quantity ratio increases, the $\hbox{\em AMRE}$ remains relatively stable.
	However, when the users' quantity ratio of $\algvar{E}_k=1.0$ or  $w_k=40$ exceeds $0.8$, we can see a significant decrease in $\hbox{\em AMRE}$ for \solutionMethodA{} and \solutionMethodB{}.
	This occurs because when the ratios surpasses a certain threshold, the optimal budget from  OBS in Algorithm~\ref{alg:OPT_B_C} becomes dominated by a higher $\algvar{E}$, resulting in lower error.

\section{Conclusion}\vspace{-2ex}
In this paper, we address the problem of \problemDefineTotalName{}.
We propose a mechanism called \solutionA{} and two methods called \solutionMethodA{} and \solutionMethodB{} to solve this problem in scenarios with personalized privacy budget and window sizes for each users.
Besides, we propose two dynamic solutions called \solutionMethodD{} and \solutionMethodE{} to solve this problem in scenarios with dynamic personalized privacy budget and window sizes.
We also compare our \solutionMethodA{}, \solutionMethodB{}, \solutionMethodD{} and \solutionMethodE{} with recent solutions to demonstrate their efficiency and effectiveness.

Our future work will focus on developing local differential privacy mechanisms that adapt to users' evolving privacy preferences in local settings while preserving data utility. 
Besides, we plan to optimize our algorithms for high-speed data streams and expand their application to real-world scenarios like mobile crowd sensing and social media analysis.

\vspace{2ex}
\noindent
\textbf{Acknowledgements}
We thank all anonymous reviewers for their constructive feedback that helped us refine our ideas and arguments. This research is funded by the General Program of China Postdoctoral Science Foundation under Grant No. 2025M771500.

\vspace{2ex}
\noindent
\textbf{Data Availability}
The datasets used in this study are publicly available as described in Section~\ref{subsec:datasets}. 
The implementation code of this paper is available at \url{https://github.com/dulei715/DynamicWEventCode}.


%
%


\bgroup\small
\bibliographystyle{acm}    
\let\xxx=\bibitem\def\bibitem{\par\vspace{1mm}\xxx}

%
%
\bibliography{reference}

@inproceedings{DBLP:conf/icalp/Dwork06,
  author    = {Cynthia Dwork},
  editor    = {Michele Bugliesi and
               Bart Preneel and
               Vladimiro Sassone and
               Ingo Wegener},
  title     = {Differential Privacy},
  booktitle = {Automata, Languages and Programming, 33rd International Colloquium,
               {ICALP} 2006, Venice, Italy, July 10-14, 2006, Proceedings, Part {II}},
  series    = {Lecture Notes in Computer Science},
  volume    = {4052},
  pages     = {1--12},
  publisher = {Springer},
  year      = {2006}
}

@article{DBLP:journals/fttcs/DworkR14,
  author    = {Cynthia Dwork and
               Aaron Roth},
  title     = {The Algorithmic Foundations of Differential Privacy},
  journal   = {Found. Trends Theor. Comput. Sci.},
  volume    = {9},
  number    = {3-4},
  pages     = {211--407},
  year      = {2014}
}

@article{DBLP:journals/tmc/WangHLWWYQ19,
  author    = {Zhibo Wang and
               Jiahui Hu and
               Ruizhao Lv and
               Jian Wei and
               Qian Wang and
               Dejun Yang and
               Hairong Qi},
  title     = {Personalized Privacy-Preserving Task Allocation for Mobile Crowdsensing},
  journal   = {{IEEE} Trans. Mob. Comput.},
  volume    = {18},
  number    = {6},
  pages     = {1330--1341},
  year      = {2019}
}

@inproceedings{DBLP:conf/ccs/AndresBCP13,
  author    = {Miguel E. Andr{\'{e}}s and
               Nicol{\'{a}}s Emilio Bordenabe and
               Konstantinos Chatzikokolakis and
               Catuscia Palamidessi},
  editor    = {Ahmad{-}Reza Sadeghi and
               Virgil D. Gligor and
               Moti Yung},
  title     = {Geo-indistinguishability: differential privacy for location-based
               systems},
  booktitle = {2013 {ACM} {SIGSAC} Conference on Computer and Communications Security,
               CCS'13, Berlin, Germany, November 4-8, 2013},
  pages     = {901--914},
  publisher = {{ACM}},
  year      = {2013}
}

@inproceedings{DBLP:conf/stoc/BassilyS15,
  author    = {Raef Bassily and
               Adam D. Smith},
  editor    = {Rocco A. Servedio and
               Ronitt Rubinfeld},
  title     = {Local, Private, Efficient Protocols for Succinct Histograms},
  booktitle = {Proceedings of the Forty-Seventh Annual {ACM} on Symposium on Theory
               of Computing, {STOC} 2015, Portland, OR, USA, June 14-17, 2015},
  pages     = {127--135},
  publisher = {{ACM}},
  year      = {2015}
}

@inproceedings{DBLP:conf/ccs/ErlingssonPK14,
  author    = {{\'{U}}lfar Erlingsson and
               Vasyl Pihur and
               Aleksandra Korolova},
  editor    = {Gail{-}Joon Ahn and
               Moti Yung and
               Ninghui Li},
  title     = {{RAPPOR:} Randomized Aggregatable Privacy-Preserving Ordinal Response},
  booktitle = {Proceedings of the 2014 {ACM} {SIGSAC} Conference on Computer and
               Communications Security, Scottsdale, AZ, USA, November 3-7, 2014},
  pages     = {1054--1067},
  publisher = {{ACM}},
  year      = {2014}
}

@inproceedings{DBLP:conf/infocom/WangZLWQR16,
  author    = {Qian Wang and
               Yan Zhang and
               Xiao Lu and
               Zhibo Wang and
               Zhan Qin and
               Kui Ren},
  title     = {RescueDP: Real-time spatio-temporal crowd-sourced data publishing
               with differential privacy},
  booktitle = {35th Annual {IEEE} International Conference on Computer Communications,
               {INFOCOM} 2016, San Francisco, CA, USA, April 10-14, 2016},
  pages     = {1--9},
  year      = {2016}
}

@inproceedings{DBLP:conf/infocom/WangLPRLC20,
  author    = {Zhibo Wang and
               Wenxin Liu and
               Xiaoyi Pang and
               Ju Ren and
               Zhe Liu and
               Yongle Chen},
  title     = {Towards Pattern-aware Privacy-preserving Real-time Data Collection},
  booktitle = {39th {IEEE} Conference on Computer Communications, {INFOCOM} 2020,
               Toronto, ON, Canada, July 6-9, 2020},
  pages     = {109--118},
  year      = {2020}
}

@inproceedings{DBLP:conf/sigmod/RenSYYZX22,
  author    = {Xuebin Ren and
               Liang Shi and
               Weiren Yu and
               Shusen Yang and
               Cong Zhao and
               Zongben Xu},
  title     = {{LDP-IDS:} Local Differential Privacy for Infinite Data Streams},
  booktitle = {{SIGMOD} '22: International Conference on Management of Data, Philadelphia,
               PA, USA, June 12 - 17, 2022},
  pages     = {1064--1077},
  year      = {2022}
}

@inproceedings{DBLP:conf/stoc/DworkNPR10,
  author       = {Cynthia Dwork and
                  Moni Naor and
                  Toniann Pitassi and
                  Guy N. Rothblum},
  title        = {Differential privacy under continual observation},
  booktitle    = {Proceedings of the 42nd {ACM} Symposium on Theory of Computing, {STOC}
                  2010, Cambridge, Massachusetts, USA, 5-8 June 2010},
  pages        = {715--724},
  year         = {2010}
}

@article{DBLP:journals/tissec/ChanSS11,
  author       = {T.{-}H. Hubert Chan and
                  Elaine Shi and
                  Dawn Song},
  title        = {Private and Continual Release of Statistics},
  journal      = {{ACM} Trans. Inf. Syst. Secur.},
  volume       = {14},
  number       = {3},
  pages        = {26:1--26:24},
  year         = {2011}
}

@article{DBLP:journals/pvldb/KellarisPXP14,
	author       = {Georgios Kellaris and
	Stavros Papadopoulos and
	Xiaokui Xiao and
	Dimitris Papadias},
	title        = {Differentially Private Event Sequences over Infinite Streams},
	journal      = {Proc. {VLDB} Endow.},
	volume       = {7},
	number       = {12},
	pages        = {1155--1166},
	year         = {2014}
}

@inproceedings{DBLP:conf/soda/Dwork10,
	author       = {Cynthia Dwork},
	title        = {Differential Privacy in New Settings},
	booktitle    = {Proceedings of the Twenty-First Annual {ACM-SIAM} Symposium on Discrete
	Algorithms, {SODA} 2010, Austin, Texas, USA, January 17-19, 2010},
	pages        = {174--183},
	year         = {2010}
}

@article{DBLP:journals/tkde/FanX14,
	author       = {Liyue Fan and
	Li Xiong},
	title        = {An Adaptive Approach to Real-Time Aggregate Monitoring With Differential
	Privacy},
	journal      = {{IEEE} Trans. Knowl. Data Eng.},
	volume       = {26},
	number       = {9},
	pages        = {2094--2106},
	year         = {2014}
}

@inproceedings{DBLP:conf/nips/CummingsFMT22,
	author       = {Rachel Cummings and
	Vitaly Feldman and
	Audra McMillan and
	Kunal Talwar},
	title        = {Mean Estimation with User-level Privacy under Data Heterogeneity},
	booktitle    = {Advances in Neural Information Processing Systems 35: Annual Conference
	on Neural Information Processing Systems 2022, NeurIPS 2022, New Orleans,
	LA, USA, November 28 - December 9, 2022},
	year         = {2022}
}

@inproceedings{DBLP:conf/sp/DongLY23,
	author       = {Wei Dong and
	Qiyao Luo and
	Ke Yi},
	title        = {Continual Observation under User-level Differential Privacy},
	booktitle    = {44th {IEEE} Symposium on Security and Privacy, {SP} 2023, San Francisco,
	CA, USA, May 21-25, 2023},
	pages        = {2190--2207},
	year         = {2023}
}

@inproceedings{DBLP:conf/sp/FengMWLQH24,
	author       = {Shuya Feng and
	Meisam Mohammady and
	Han Wang and
	Xiaochen Li and
	Zhan Qin and
	Yuan Hong},
	title        = {{DPI:} Ensuring Strict Differential Privacy for Infinite Data Streaming},
	booktitle    = {{IEEE} Symposium on Security and Privacy, {SP} 2024, San Francisco,
	CA, USA, May 19-23, 2024},
	pages        = {1009--1027},
	year         = {2024}
}

@inproceedings{DBLP:conf/focs/DvijothamMP0T24,
	author       = {Krishnamurthy Dj Dvijotham and
	H. Brendan McMahan and
	Krishna Pillutla and
	Thomas Steinke and
	Abhradeep Thakurta},
	title        = {Efficient and Near-Optimal Noise Generation for Streaming Differential
	Privacy},
	booktitle    = {65th {IEEE} Annual Symposium on Foundations of Computer Science, {FOCS}
	2024, Chicago, IL, USA, October 27-30, 2024},
	pages        = {2306--2317},
	year         = {2024}
}

@article{DBLP:journals/jpc/AlagganGK16,
	author       = {Mohammad Alaggan and
	S{\'{e}}bastien Gambs and
	Anne{-}Marie Kermarrec},
	title        = {Heterogeneous Differential Privacy},
	journal      = {J. Priv. Confidentiality},
	volume       = {7},
	number       = {2},
	year         = {2016}
}

@inproceedings{DBLP:conf/icde/JorgensenYC15,
	author       = {Zach Jorgensen and
	Ting Yu and
	Graham Cormode},
	title        = {Conservative or liberal? Personalized differential privacy},
	booktitle    = {31st {IEEE} International Conference on Data Engineering, {ICDE} 2015,
	Seoul, South Korea, April 13-17, 2015},
	pages        = {1023--1034},
	year         = {2015}
}

@inproceedings{DBLP:conf/uss/Murakami019,
	author       = {Takao Murakami and
	Yusuke Kawamoto},
	title        = {Utility-Optimized Local Differential Privacy Mechanisms for Distribution
	Estimation},
	booktitle    = {28th {USENIX} Security Symposium, {USENIX} Security 2019, Santa Clara,
	CA, USA, August 14-16, 2019},
	pages        = {1877--1894},
	year         = {2019}
}

@inproceedings{DBLP:conf/icde/DuCZX00F23,
	author       = {Leilei Du and
	Peng Cheng and
	Libin Zheng and
	Wei Xi and
	Xuemin Lin and
	Wenjie Zhang and
	Jing Fang},
	title        = {Dynamic Private Task Assignment under Differential Privacy},
	booktitle    = {39th {IEEE} International Conference on Data Engineering, {ICDE} 2023,
	Anaheim, CA, USA, April 3-7, 2023},
	pages        = {2740--2752},
	year         = {2023}
}

@inproceedings{DBLP:conf/icde/KotsogiannisDHM20,
	author       = {Ios Kotsogiannis and
	Stelios Doudalis and
	Samuel Haney and
	Ashwin Machanavajjhala and
	Sharad Mehrotra},
	title        = {One-sided Differential Privacy},
	booktitle    = {36th {IEEE} International Conference on Data Engineering, {ICDE} 2020,
	Dallas, TX, USA, April 20-24, 2020},
	pages        = {493--504},
	year         = {2020}
}

@article{DBLP:journals/pvldb/LiuLXLM21,
	author       = {Junxu Liu and
	Jian Lou and
	Li Xiong and
	Jinfei Liu and
	Xiaofeng Meng},
	title        = {Projected Federated Averaging with Heterogeneous Differential Privacy},
	journal      = {Proc. {VLDB} Endow.},
	volume       = {15},
	number       = {4},
	pages        = {828--840},
	year         = {2021}
}

@inproceedings{DBLP:conf/stoc/BlumLR08,
	author       = {Avrim Blum and
	Katrina Ligett and
	Aaron Roth},
	title        = {A learning theory approach to non-interactive database privacy},
	booktitle    = {Proceedings of the 40th Annual {ACM} Symposium on Theory of Computing,
	Victoria, British Columbia, Canada, May 17-20, 2008},
	pages        = {609--618},
	year         = {2008}
}

@inproceedings{DBLP:conf/icdt/BolotFMNT13,
	author       = {Jean Bolot and
	Nadia Fawaz and
	S. Muthukrishnan and
	Aleksandar Nikolov and
	Nina Taft},
	title        = {Private decayed predicate sums on streams},
	booktitle    = {Joint 2013 {EDBT/ICDT} Conferences, {ICDT} '13 Proceedings, Genoa,
	Italy, March 18-22, 2013},
	pages        = {284--295},
	year         = {2013}
}

@inproceedings{DBLP:conf/ccs/ChenMHM17,
	author       = {Yan Chen and
	Ashwin Machanavajjhala and
	Michael Hay and
	Gerome Miklau},
	title        = {PeGaSus: Data-Adaptive Differentially Private Stream Processing},
	booktitle    = {Proceedings of the 2017 {ACM} {SIGSAC} Conference on Computer and
	Communications Security, {CCS} 2017, Dallas, TX, USA, October 30 -
	November 03, 2017},
	pages        = {1375--1388},
	year         = {2017}
}

@inproceedings{DBLP:conf/ccs/0001C0SC0LJ21,
	author       = {Tianhao Wang and
	Joann Qiongna Chen and
	Zhikun Zhang and
	Dong Su and
	Yueqiang Cheng and
	Zhou Li and
	Ninghui Li and
	Somesh Jha},
	title        = {Continuous Release of Data Streams under both Centralized and Local
	Differential Privacy},
	booktitle    = {{CCS} '21: 2021 {ACM} {SIGSAC} Conference on Computer and Communications
	Security, Virtual Event, Republic of Korea, November 15 - 19, 2021},
	pages        = {1237--1253},
	year         = {2021}
}

@article{DBLP:journals/pvldb/BaoYXD21,
	author       = {Ergute Bao and
	Yin Yang and
	Xiaokui Xiao and
	Bolin Ding},
	title        = {{CGM:} An Enhanced Mechanism for Streaming Data Collection with Local Differential Privacy},
	journal      = {Proc. {VLDB} Endow.},
	volume       = {14},
	number       = {11},
	pages        = {2258--2270},
	year         = {2021}
}

@article{DBLP:journals/tkde/XueYHZW23,
	author       = {Qiao Xue and
	Qingqing Ye and
	Haibo Hu and
	Youwen Zhu and
	Jian Wang},
	title        = {{DDRM:} {A} Continual Frequency Estimation Mechanism With Local Differential
	Privacy},
	journal      = {{IEEE} Trans. Knowl. Data Eng.},
	volume       = {35},
	number       = {7},
	pages        = {6784--6797},
	year         = {2023}
}

@inproceedings{DBLP:conf/infocom/YeHHAX23,
	author       = {Qingqing Ye and
	Haibo Hu and
	Kai Huang and
	Man Ho Au and
	Qiao Xue},
	title        = {Stateful Switch: Optimized Time Series Release with Local Differential
	Privacy},
	booktitle    = {{IEEE} {INFOCOM} 2023 - {IEEE} Conference on Computer Communications,
	New York City, NY, USA, May 17-20, 2023},
	pages        = {1--10},
	year         = {2023}
}

@inproceedings{DBLP:conf/www/XiePMJM23,
	author       = {Yutong Xie and
	Zhaoying Pan and
	Jinge Ma and
	Luo Jie and
	Qiaozhu Mei},
	title        = {A Prompt Log Analysis of Text-to-Image Generation Systems},
	booktitle    = {Proceedings of the {ACM} Web Conference 2023, {WWW} 2023, Austin,
	TX, USA, 30 April 2023 - 4 May 2023},
	pages        = {3892--3902},
	year         = {2023}
}

@article{DBLP:journals/tpds/GuoJZZ12,
	author       = {Peng Guo and
	Tao Jiang and
	Qian Zhang and
	Kui Zhang},
	title        = {Sleep Scheduling for Critical Event Monitoring in Wireless Sensor
	Networks},
	journal      = {{IEEE} Trans. Parallel Distributed Syst.},
	volume       = {23},
	number       = {2},
	pages        = {345--352},
	year         = {2012}
}

@inproceedings{DBLP:conf/cvpr/MoonHPPH23,
	author       = {WonJun Moon and
	Sangeek Hyun and
	Sanguk Park and
	Dongchan Park and
	Jae{-}Pil Heo},
	title        = {Query - Dependent Video Representation for Moment Retrieval and Highlight
	Detection},
	booktitle    = {{IEEE/CVF} Conference on Computer Vision and Pattern Recognition,
	{CVPR} 2023, Vancouver, BC, Canada, June 17-24, 2023},
	pages        = {23023--23033},
	year         = {2023}
}

@inproceedings{DBLP:conf/kdd/YuanZXS11,
	author       = {Jing Yuan and
	Yu Zheng and
	Xing Xie and
	Guangzhong Sun},
	title        = {Driving with knowledge from the physical world},
	booktitle    = {Proceedings of the 17th {ACM} {SIGKDD} International Conference on
	Knowledge Discovery and Data Mining, San Diego, CA, USA, August 21-24,
	2011},
	pages        = {316--324},
	year         = {2011}
}

@inproceedings{DBLP:conf/gis/YuanZZXXSH10,
	author       = {Jing Yuan and
	Yu Zheng and
	Chengyang Zhang and
	Wenlei Xie and
	Xing Xie and
	Guangzhong Sun and
	Yan Huang},
	title        = {T-drive: driving directions based on taxi trajectories},
	booktitle    = {18th {ACM} {SIGSPATIAL} International Symposium on Advances in Geographic
	Information Systems, {ACM-GIS} 2010, November 3-5, 2010, San Jose,
	CA, USA, Proceedings},
	pages        = {99--108},
	year         = {2010}
}

@article{DBLP:journals/tist/YangZQ16,
	author       = {Dingqi Yang and
	Daqing Zhang and
	Bingqing Qu},
	title        = {Participatory Cultural Mapping Based on Collective Behavior Data in
	Location-Based Social Networks},
	journal      = {{ACM} Trans. Intell. Syst. Technol.},
	volume       = {7},
	number       = {3},
	pages        = {30:1--30:23},
	year         = {2016}
}

@article{DBLP:journals/jnca/YangZCQ15,
	author       = {Dingqi Yang and
	Daqing Zhang and
	Longbiao Chen and
	Bingqing Qu},
	title        = {NationTelescope: Monitoring and visualizing large-scale collective
	behavior in LBSNs},
	journal      = {J. Netw. Comput. Appl.},
	volume       = {55},
	pages        = {170--180},
	year         = {2015}
}

@article{kullback1951information,
	title={On information and sufficiency},
	author={Kullback, Solomon and Leibler, Richard A},
	journal={The annals of mathematical statistics},
	volume={22},
	number={1},
	pages={79--86},
	year={1951}
}

@article{DBLP:journals/tit/Lin91,
	author       = {Jianhua Lin},
	title        = {Divergence measures based on the Shannon entropy},
	journal      = {{IEEE} Trans. Inf. Theory},
	volume       = {37},
	number       = {1},
	pages        = {145--151},
	year         = {1991}
}

@article{DBLP:journals/pvldb/Du05,
	author    = {Leilei Du and
	Peng Cheng and
	Lei Chen and
	Heng Tao Shen and
	Xuemin Lin and
	Wei Xi
	},
	title     = {Infinite Stream Estimation under Personalized $w$-Event Privacy},
	journal   = {Proc. {VLDB} Endow.},
	volume    = {18},
	number    = {6},
	pages     = {1111--1123},
	year      = {2025}
}

@article{DBLP:journals/corr/abs-2107-11839,
	author       = {Albert Cheu},
	title        = {Differential Privacy in the Shuffle Model: {A} Survey of Separations},
	journal      = {CoRR},
	volume       = {abs/2107.11839},
	year         = {2021}
}

@inproceedings{DBLP:conf/eurocrypt/CheuSUZZ19,
	author       = {Albert Cheu and
	Adam D. Smith and
	Jonathan R. Ullman and
	David Zeber and
	Maxim Zhilyaev},
	editor       = {Yuval Ishai and
	Vincent Rijmen},
	title        = {Distributed Differential Privacy via Shuffling},
	booktitle    = {Advances in Cryptology - {EUROCRYPT} 2019 - 38th Annual International
	Conference on the Theory and Applications of Cryptographic Techniques,
	Darmstadt, Germany, May 19-23, 2019, Proceedings, Part {I}},
	series       = {Lecture Notes in Computer Science},
	volume       = {11476},
	pages        = {375--403},
	publisher    = {Springer},
	year         = {2019}
}

@inproceedings{DBLP:conf/icml/TenenbaumKMS23,
	author       = {Jay Tenenbaum and
	Haim Kaplan and
	Yishay Mansour and
	Uri Stemmer},
	editor       = {Andreas Krause and
	Emma Brunskill and
	Kyunghyun Cho and
	Barbara Engelhardt and
	Sivan Sabato and
	Jonathan Scarlett},
	title        = {Concurrent Shuffle Differential Privacy Under Continual Observation},
	booktitle    = {International Conference on Machine Learning, {ICML} 2023, 23-29 July
	2023, Honolulu, Hawaii, {USA}},
	series       = {Proceedings of Machine Learning Research},
	volume       = {202},
	pages        = {33961--33982},
	publisher    = {{PMLR}},
	year         = {2023}
}

@inproceedings{DBLP:conf/icde/LiCY23,
	author       = {Xiaoyu Li and
	Yang Cao and
	Masatoshi Yoshikawa},
	title        = {Locally Private Streaming Data Release with Shuffling and Subsampling},
	booktitle    = {39th {IEEE} International Conference on Data Engineering, {ICDE} 2023
	- Workshops, Anaheim, CA, USA, April 3-7, 2023},
	pages        = {125--131},
	publisher    = {{IEEE}},
	year         = {2023}
}

@article{DBLP:journals/tmc/WangLPCYJL25,
	author       = {Shaowei Wang and
	Jin Li and
	Yun Peng and
	Kongyang Chen and
	Wei Yang and
	Hui Jiang and
	Jin Li},
	title        = {Differential Private Data Stream Analytics in the Local and Shuffle
	Models},
	journal      = {{IEEE} Trans. Mob. Comput.},
	volume       = {24},
	number       = {7},
	pages        = {6701--6717},
	year         = {2025}
}

@article{DBLP:journals/pacmmod/LiLCGRQW25,
	author       = {Xiaochen Li and
	Tianyu Li and
	Yitian Cheng and
	Chen Gong and
	Kui Ren and
	Zhan Qin and
	Tianhao Wang},
	title        = {{SPAS:} Continuous Release of Data Streams under w-Event Differential
	Privacy},
	journal      = {Proc. {ACM} Manag. Data},
	volume       = {3},
	number       = {1},
	pages        = {78a:1--78a:27},
	year         = {2025}
}

@article{SunDQY24,
	author    = {Dajun Sun and Wei Dong and Yuan Qiu and Ke Yi},
	title     = {Personalized Truncation for Personalized Privacy},
	journal   = {Proc. ACM Manag. Data},
	volume    = {2},
	number    = {6},
	pages     = {249:1--249:25},
	year      = {2024},
	note      = {SIGMOD 2025}
}

\section{Appendix}
\subsection{Running Time Analysis}\label{exp:running_time}
In this subsection, we compare the running time of \solutionCMPA{}, \solutionCMPB{}, \solutionMethodA{}, \solutionMethodB{}, \solutionMethodD{} and \solutionMethodE{}.

\begin{figure*}[h]\centering
	\subfigure{
		\scalebox{0.27}[0.27]{\includegraphics{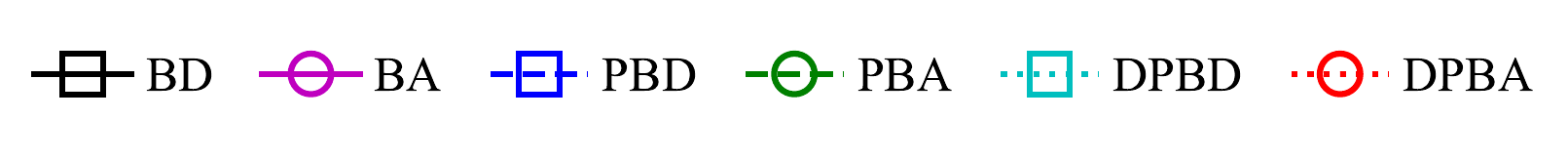}}}\hfill\\
	\addtocounter{subfigure}{-1}\vspace{-2ex}
	\subfigure[][{\small \trajectoryDatasetName{}}]{
		\scalebox{0.177}[0.177]{\includegraphics{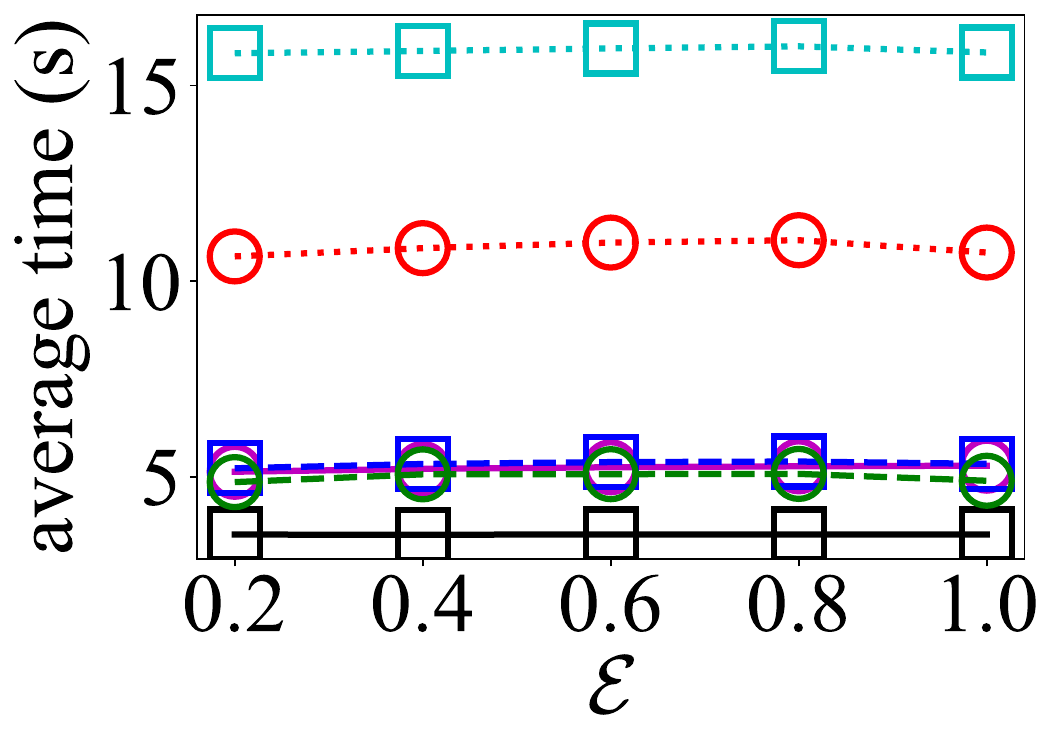}}
		\label{subfig:trajectory_budget_change_running_time}}\hfill
	\subfigure[][{\small \checkInDatasetName{}}]{
		\scalebox{0.177}[0.177]{\includegraphics{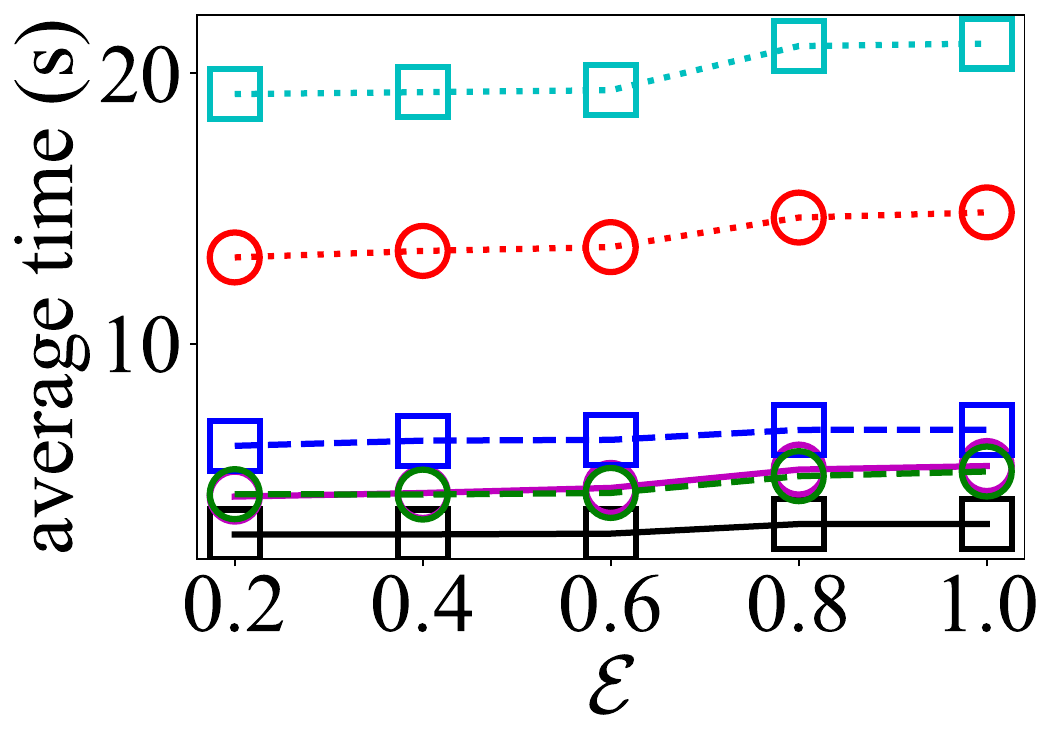}}
		\label{subfig:check_in_budget_change_running_time}}\hfill	
	\subfigure[][{\small \tlnsDatasetName{}}]{
		\scalebox{0.177}[0.177]{\includegraphics{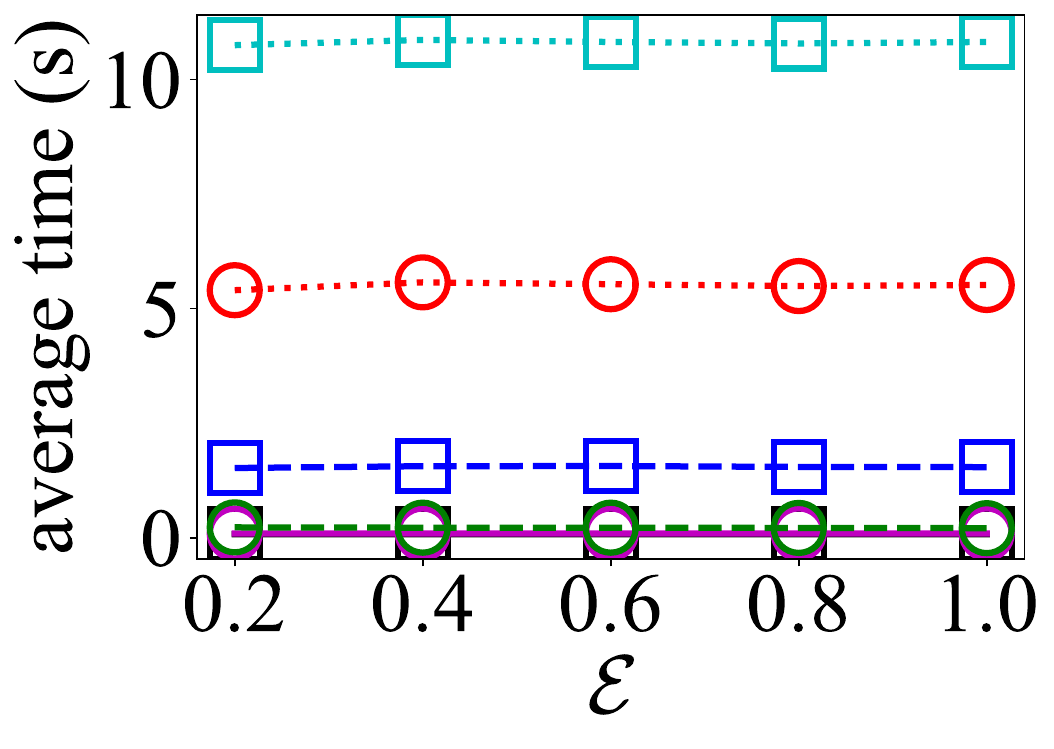}}
		\label{subfig:tlns_budget_change_running_time}}\hfill 
	\subfigure[][{\small \sinDatasetName{}}]{
		\scalebox{0.177}[0.177]{\includegraphics{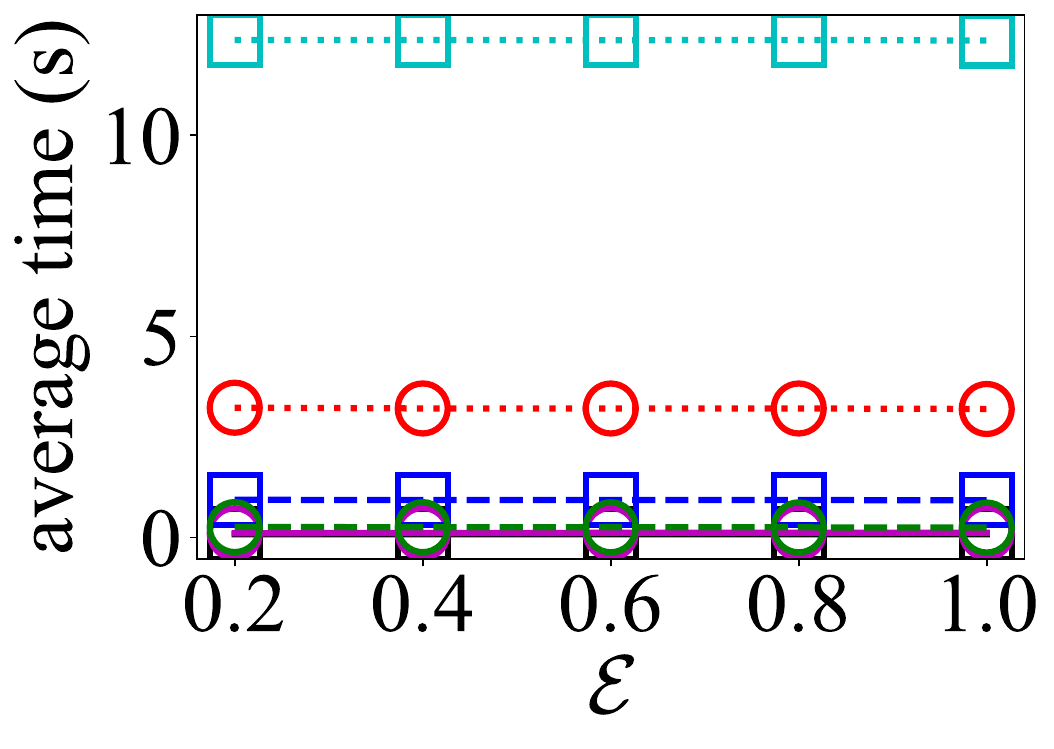}}
		\label{subfig:sin_budget_change_running_time}}\hfill 
	\subfigure[][{\small \logDatasetName{}}]{
		\scalebox{0.177}[0.177]{\includegraphics{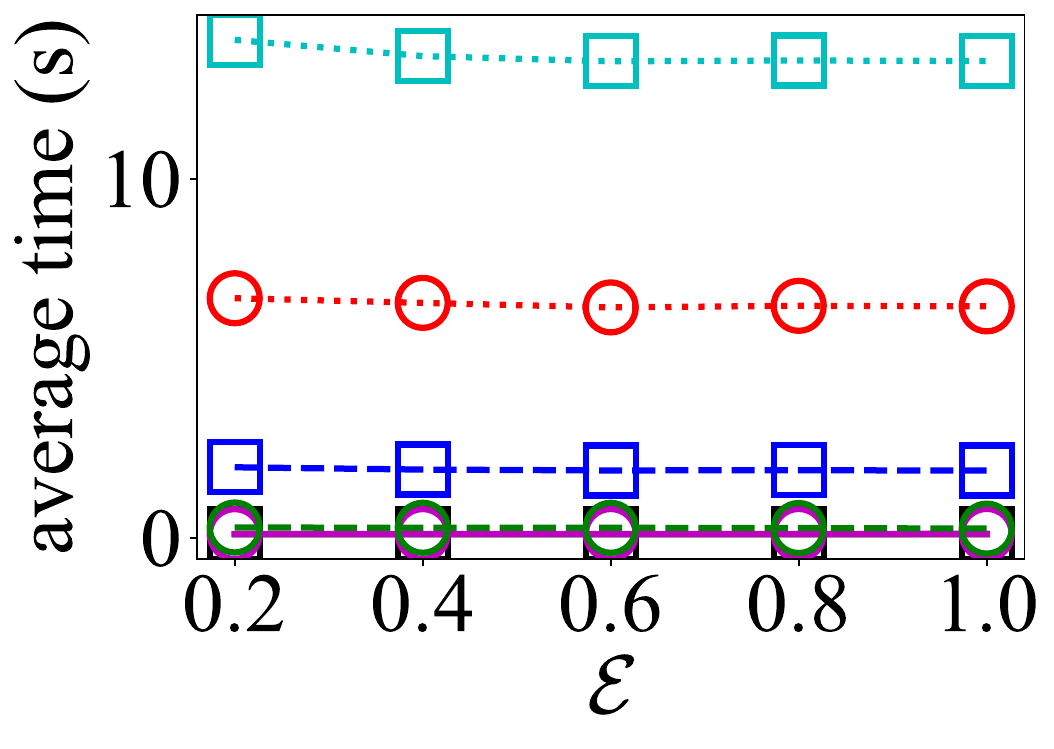}}
		\label{subfig:log_budget_change_running_time}}\hfill 
	\caption{\small The average running time per time slot with $\algvar{E}$ varied.}
	\label{fig:alter_e_running_time}
\end{figure*}

\begin{figure*}[h]\centering
	\subfigure{
		\scalebox{0.27}[0.27]{\includegraphics{bar3_time1.pdf}}}\hfill\\
	\addtocounter{subfigure}{-1}\vspace{-2ex}
	\subfigure[][{\small \trajectoryDatasetName{}}]{
		\scalebox{0.177}[0.177]{\includegraphics{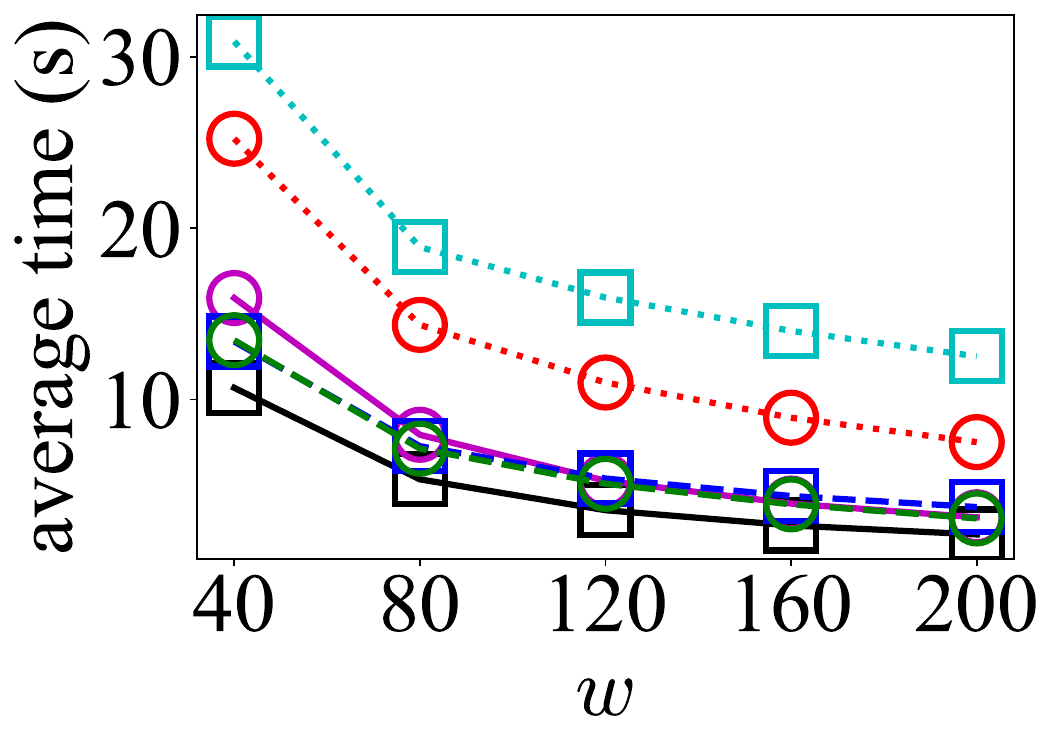}}
		\label{subfig:trajectory_window_size_change_running_time}}\hfill
	\subfigure[][{\small \checkInDatasetName{}}]{
		\scalebox{0.177}[0.177]{\includegraphics{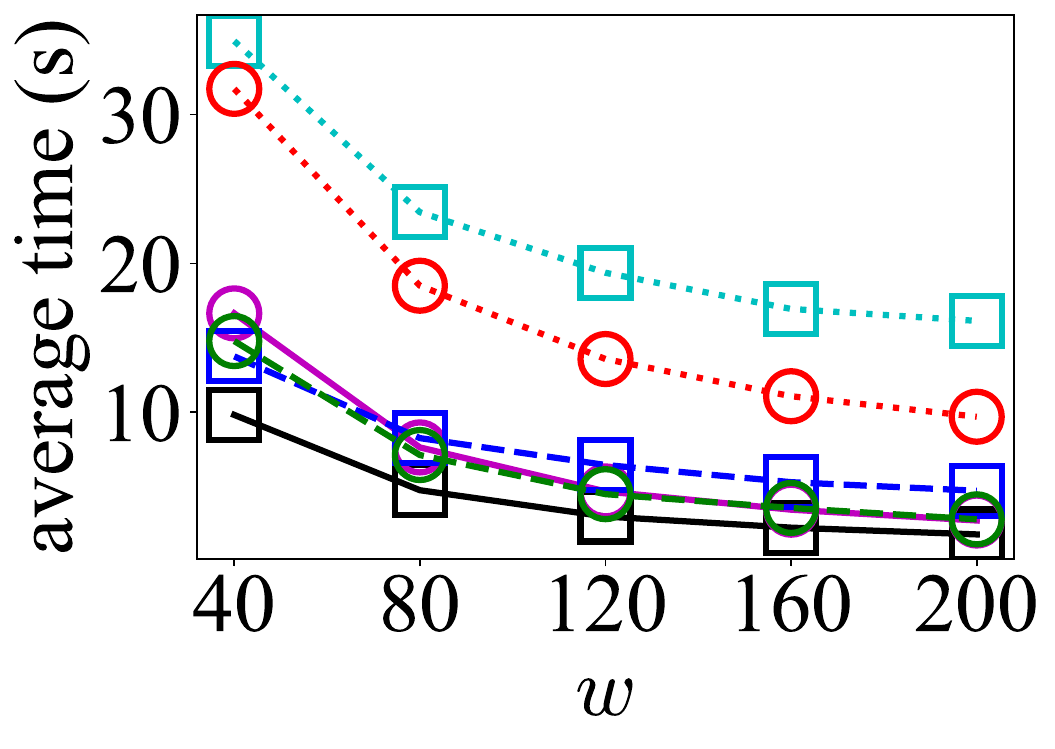}}
		\label{subfig:check_in_window_size_change_running_time}}\hfill	
	\subfigure[][{\small \tlnsDatasetName{}}]{
		\scalebox{0.177}[0.177]{\includegraphics{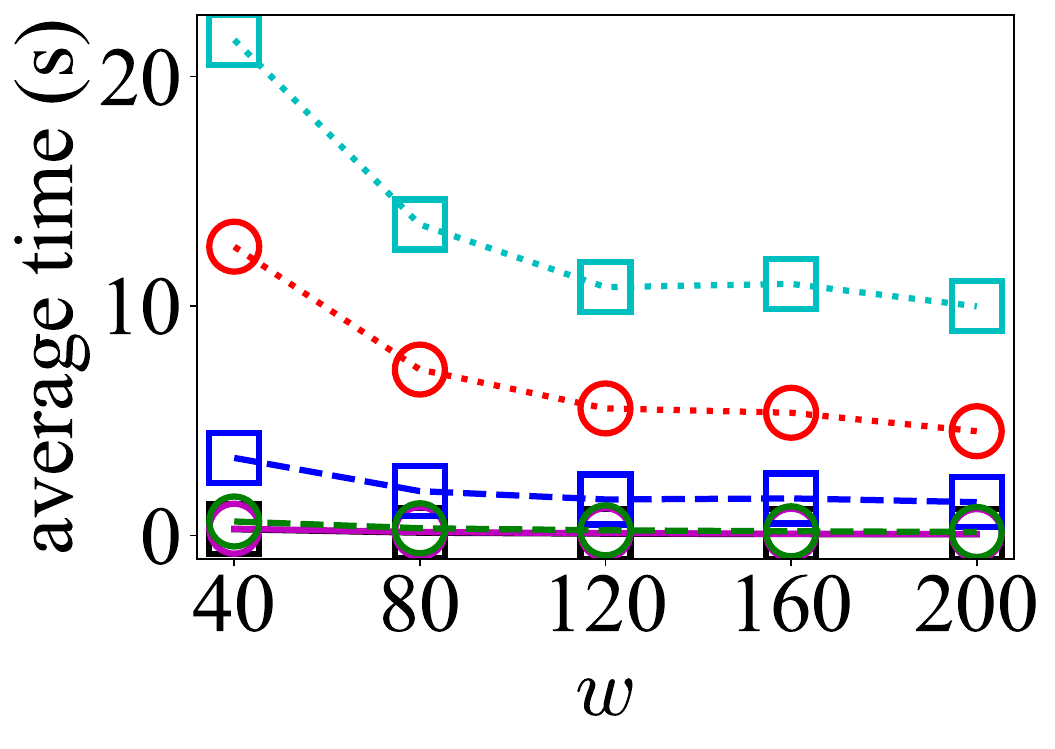}}
		\label{subfig:tlns_window_size_change_running_time}}\hfill 
	\subfigure[][{\small \sinDatasetName{}}]{
		\scalebox{0.177}[0.177]{\includegraphics{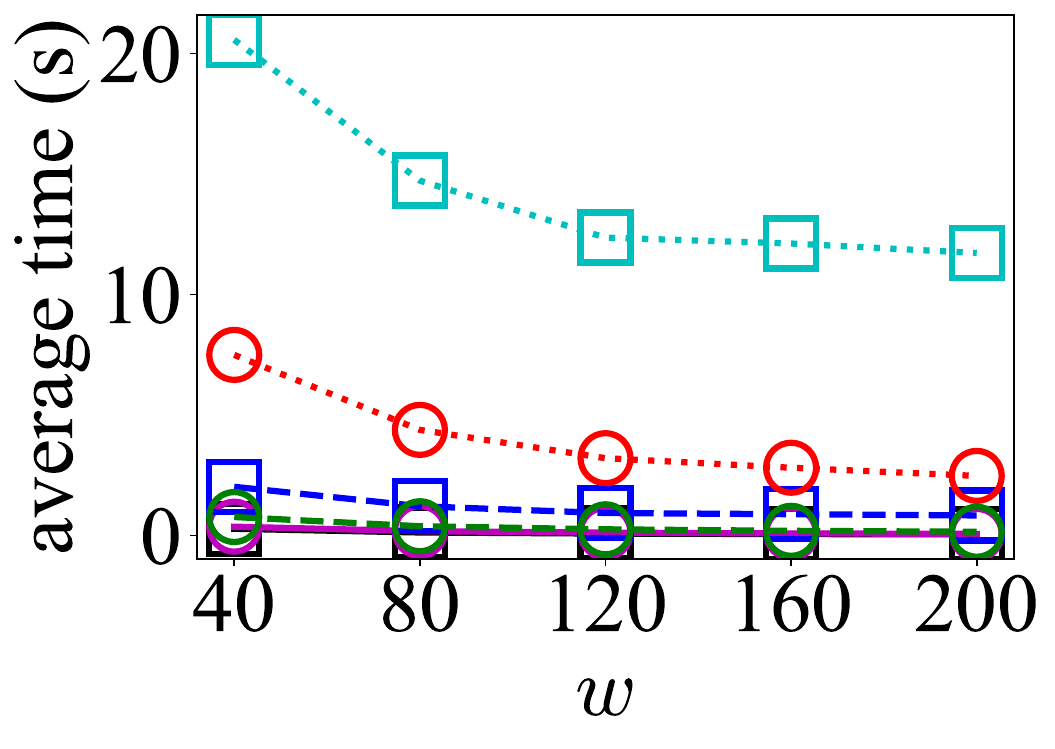}}
		\label{subfig:sin_window_size_change_running_time}}\hfill 
	\subfigure[][{\small \logDatasetName{}}]{
		\scalebox{0.177}[0.177]{\includegraphics{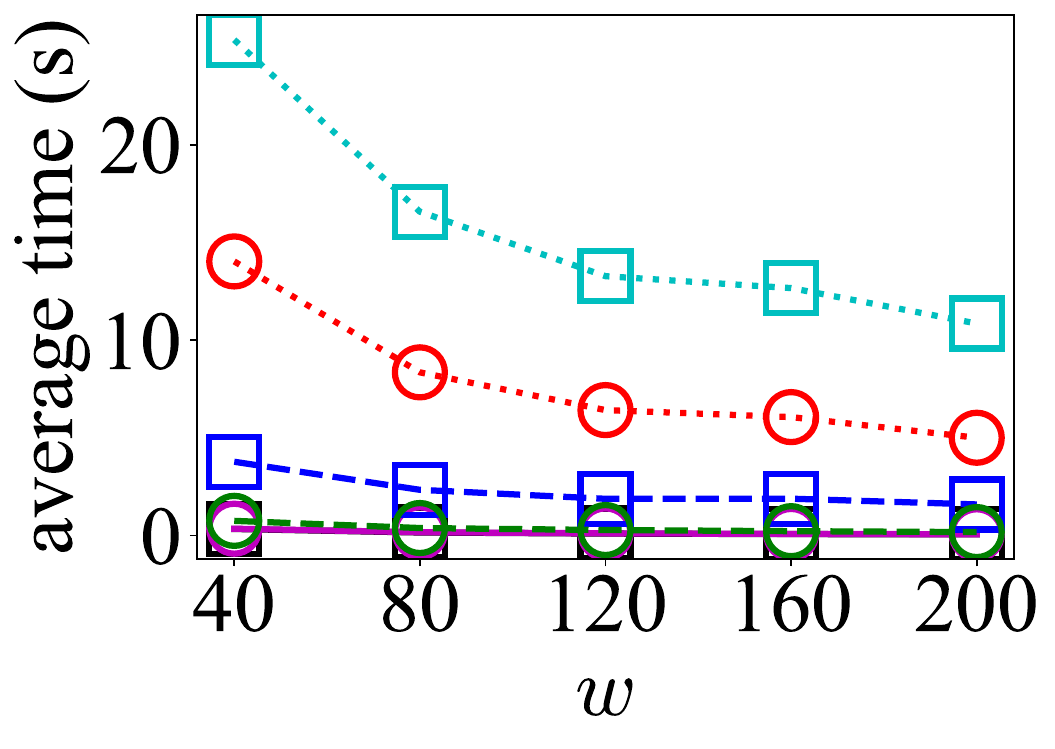}}
		\label{subfig:log_window_size_change_running_time}}\hfill 
	\caption{\small The average running time per time slot with $w$ varied.}
	\label{fig:alter_w_running_time}
\end{figure*}

Figure~\ref{fig:alter_e_running_time} shows the average running time per time slot as the privacy budget varies from $0.2$ to $1$. 
The running time remains stable across different privacy budgets. This stability occurs because the privacy level does not affect the running time.
\solutionMethodD{} requires the highest computation time among all methods, while \solutionCMPA{} requires the least time. Moreover, \solutionMethodA{} and \solutionMethodD{} run slower than their corresponding absorption-based methods, \solutionMethodB{} and \solutionMethodE{}. It is because personalized budget absorption methods are more likely to skip a publication than personalized budget distribution methods, which leads to fewer non-null publication calculations.

Figure~\ref{fig:alter_w_running_time} shows the average running time per time slot as the window size changes from $40$ to $200$. 
The running time decreases as the window size increases. This occurs because large window sizes result in more skipped time slots or nullified time slots, reducing the overall calculation time.
Similar to Figure~\ref{fig:alter_e_running_time}, \solutionCMPA{} requires the least running time among all methods, while \solutionMethodA{} and \solutionMethodD{} incur higher running time. This is because the personalized methods introduce an optimal budget selection step that increases the running time. 
Compared to personalized budget absorption methods, the dissimilarity of personalized budget distribution methods increases more rapidly than the error, resulting in fewer skips or nullifications.

\subsection{Experiments under $\hbox{\em AJSD}$ and $\hbox{\em AWD}$ Metric}\label{appendix:under_AJSD_metric}
\begin{figure*}[h]\centering
	\subfigure{
		\scalebox{0.27}[0.27]{\includegraphics{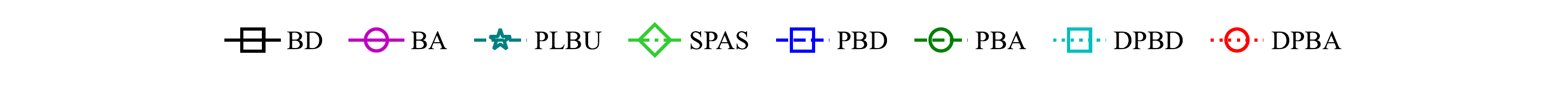}}}\hfill\\
	\addtocounter{subfigure}{-1}\vspace{-2ex}
	\subfigure[][{\small \trajectoryDatasetName{}}]{
		\scalebox{0.177}[0.177]{\includegraphics{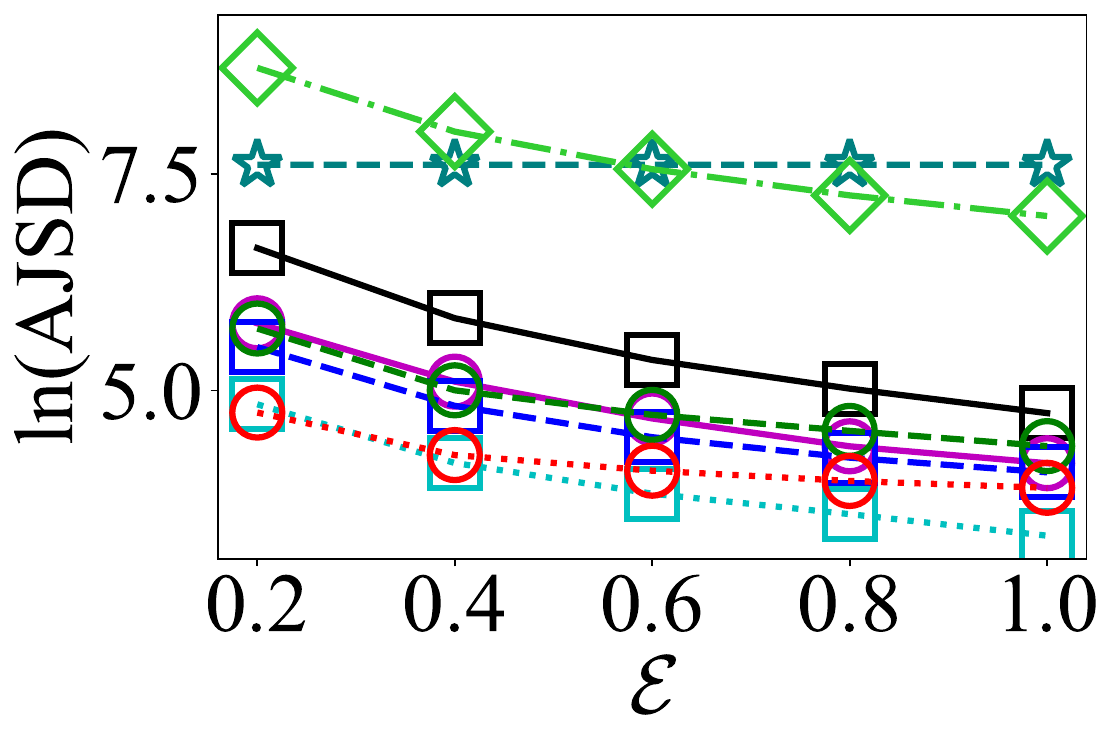}}
		\label{subfig:trajectory_budget_change_AJSD_total}}\hfill
	\subfigure[][{\small \checkInDatasetName{}}]{
		\scalebox{0.177}[0.177]{\includegraphics{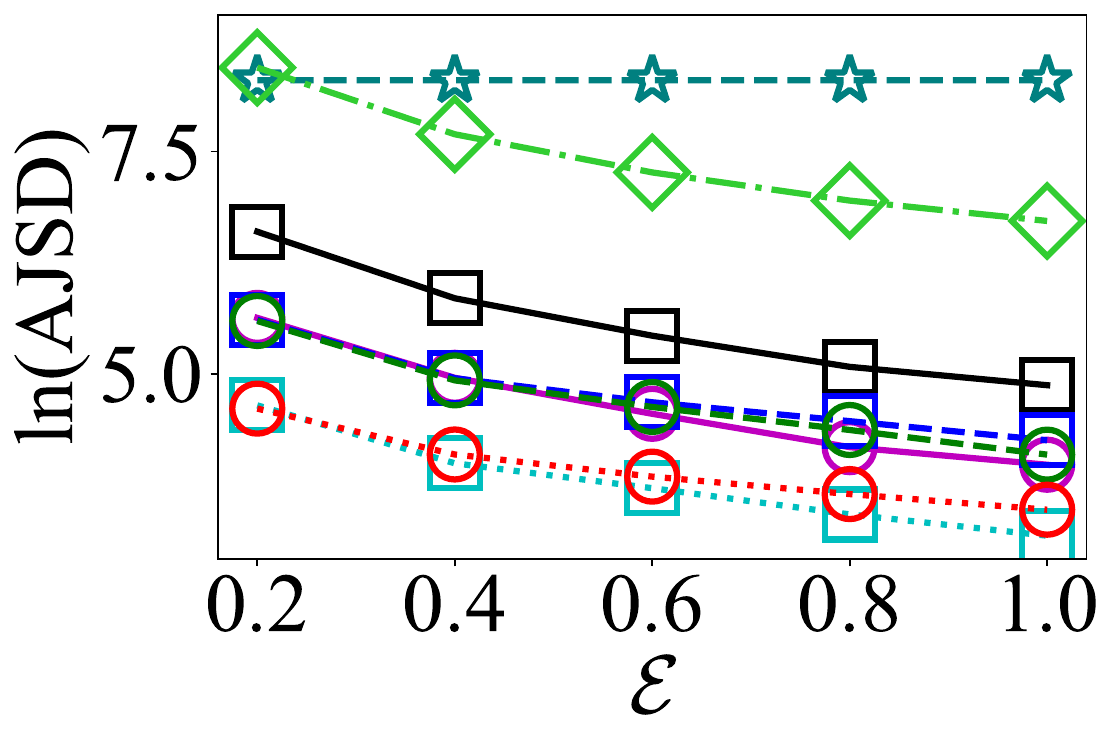}}
		\label{subfig:check_in_budget_change_AJSD_total}}\hfill	
	\subfigure[][{\small \tlnsDatasetName{}}]{
		\scalebox{0.177}[0.177]{\includegraphics{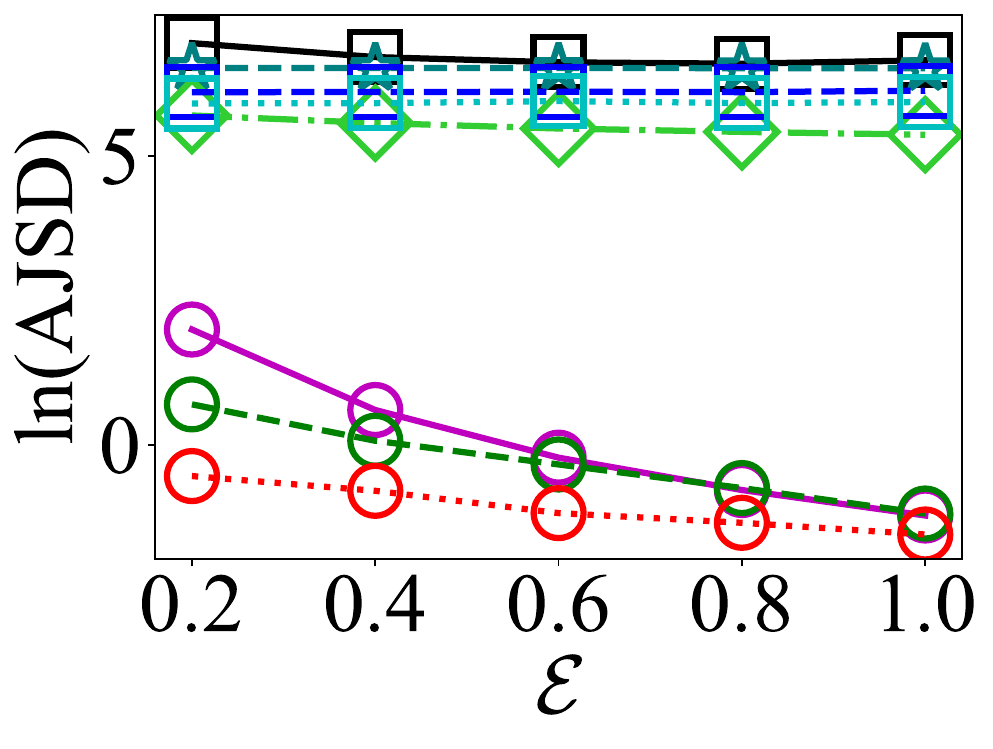}}
		\label{subfig:tlns_budget_change_AJSD_total}}\hfill 
	\subfigure[][{\small \sinDatasetName{}}]{
		\scalebox{0.177}[0.177]{\includegraphics{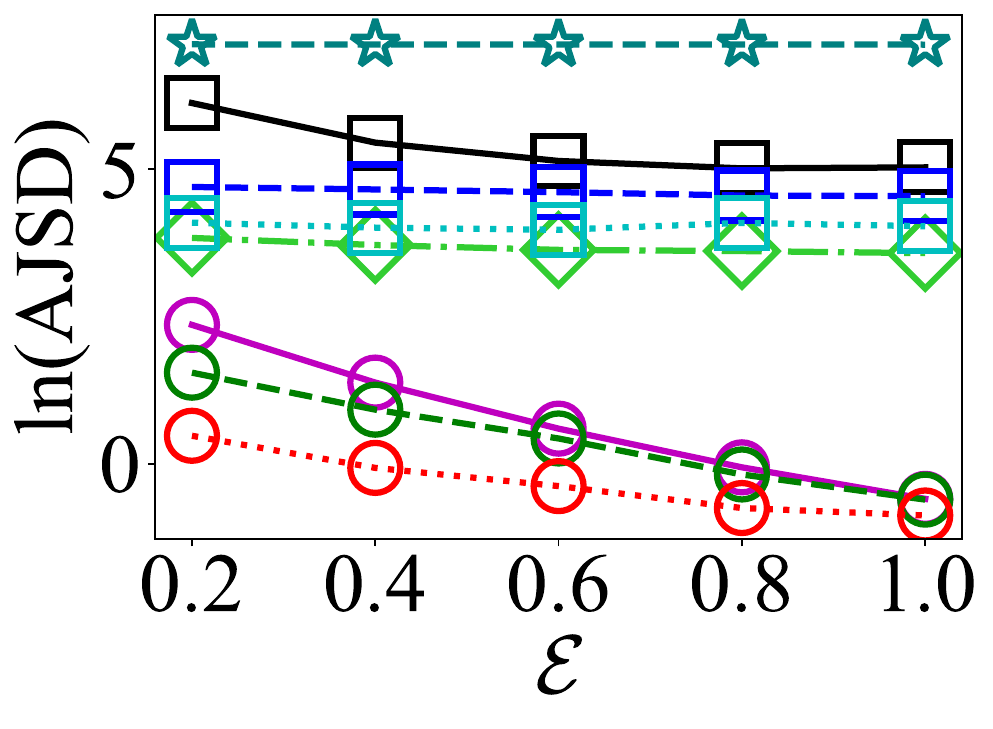}}
		\label{subfig:sin_budget_change_AJSD_total}}\hfill 
	\subfigure[][{\small \logDatasetName{}}]{
		\scalebox{0.177}[0.177]{\includegraphics{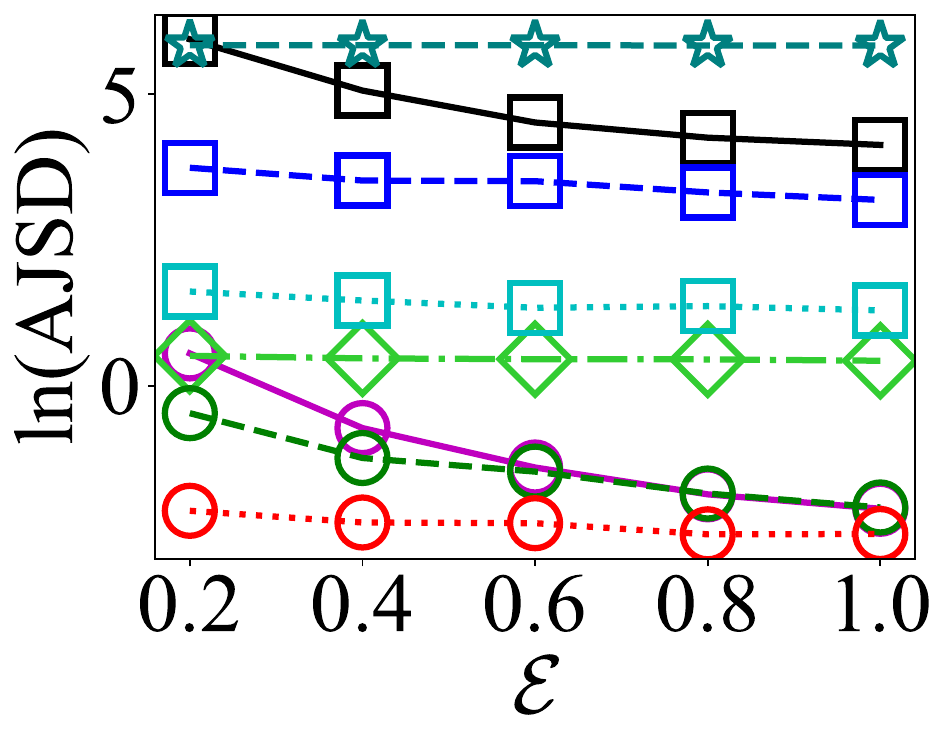}}
		\label{subfig:log_budget_change_AJSD_total}}\hfill 
	\caption{\small Average Jensen-Shannon Divergence ($\hbox{\em AJSD}$) with $\algvar{E}$ varied.}
	\label{fig:alter_e_AJSD_total}
\end{figure*}

\begin{figure*}[h]\centering
	\subfigure{
		\scalebox{0.27}[0.27]{\includegraphics{bar5.pdf}}}\hfill\\
	\addtocounter{subfigure}{-1}\vspace{-2ex}
	\subfigure[][{\small \trajectoryDatasetName{}}]{
		\scalebox{0.177}[0.177]{\includegraphics{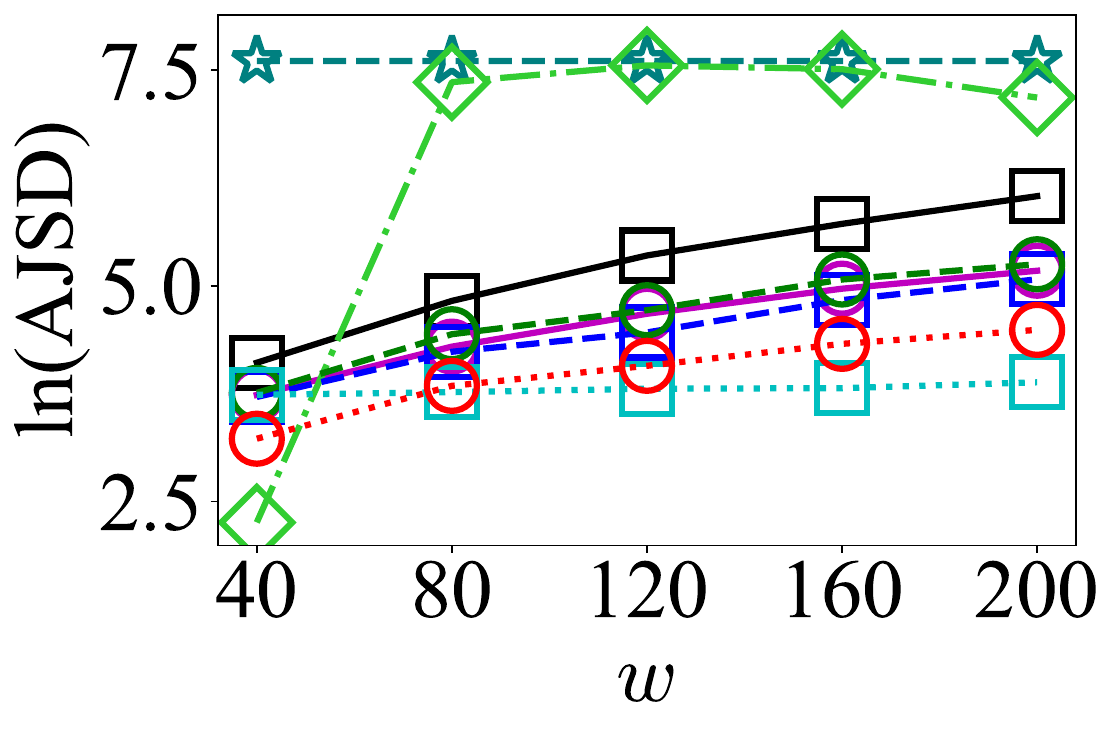}}
		\label{subfig:trajectory_window_size_change_AJSD_total}}\hfill
	\subfigure[][{\small \checkInDatasetName{}}]{
		\scalebox{0.177}[0.177]{\includegraphics{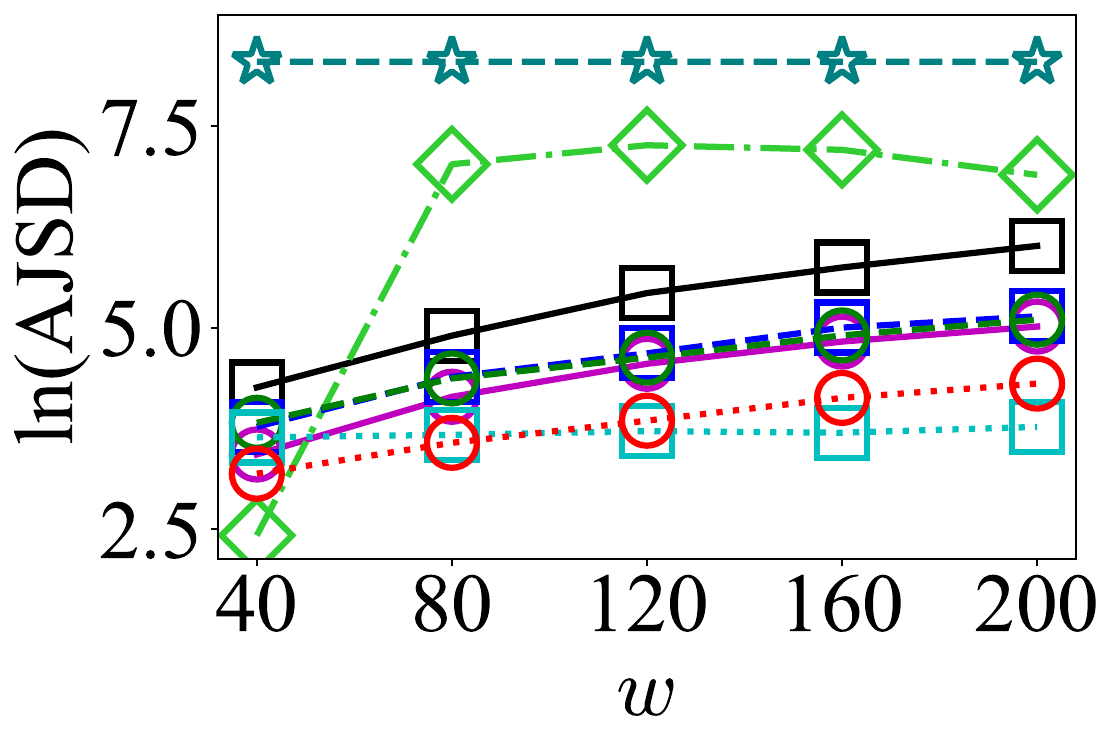}}
		\label{subfig:check_in_window_size_change_AJSD_total}}\hfill	
	\subfigure[][{\small \tlnsDatasetName{}}]{
		\scalebox{0.177}[0.177]{\includegraphics{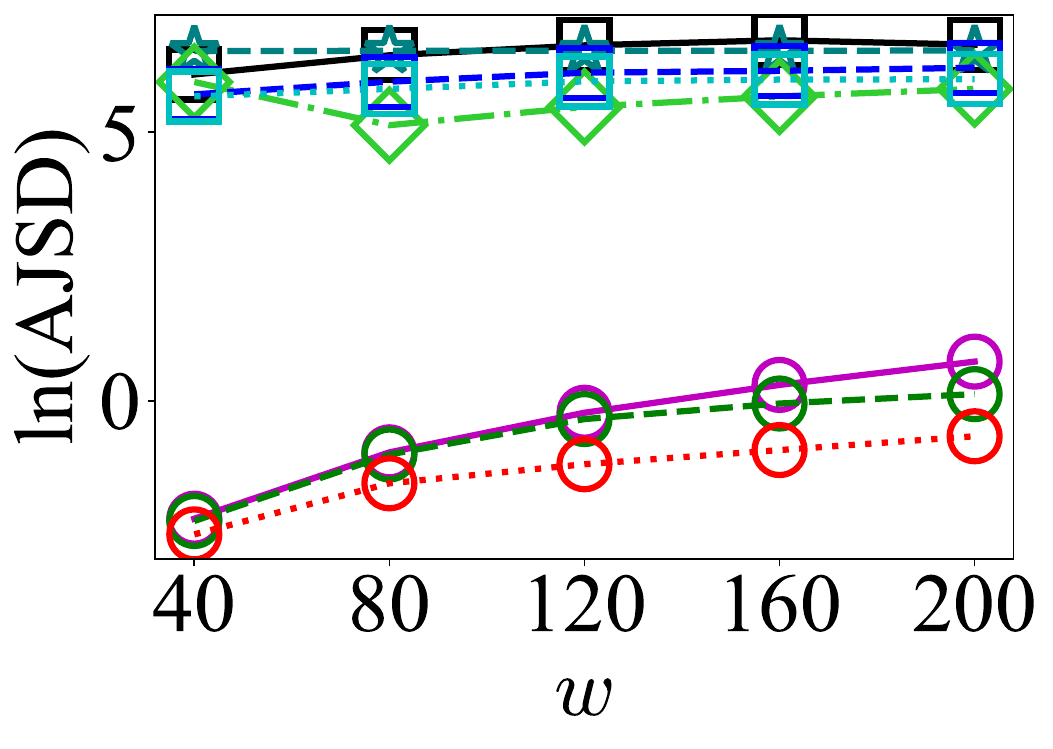}}
		\label{subfig:tlns_window_size_change_AJSD_total}}\hfill 
	\subfigure[][{\small \sinDatasetName{}}]{
		\scalebox{0.177}[0.177]{\includegraphics{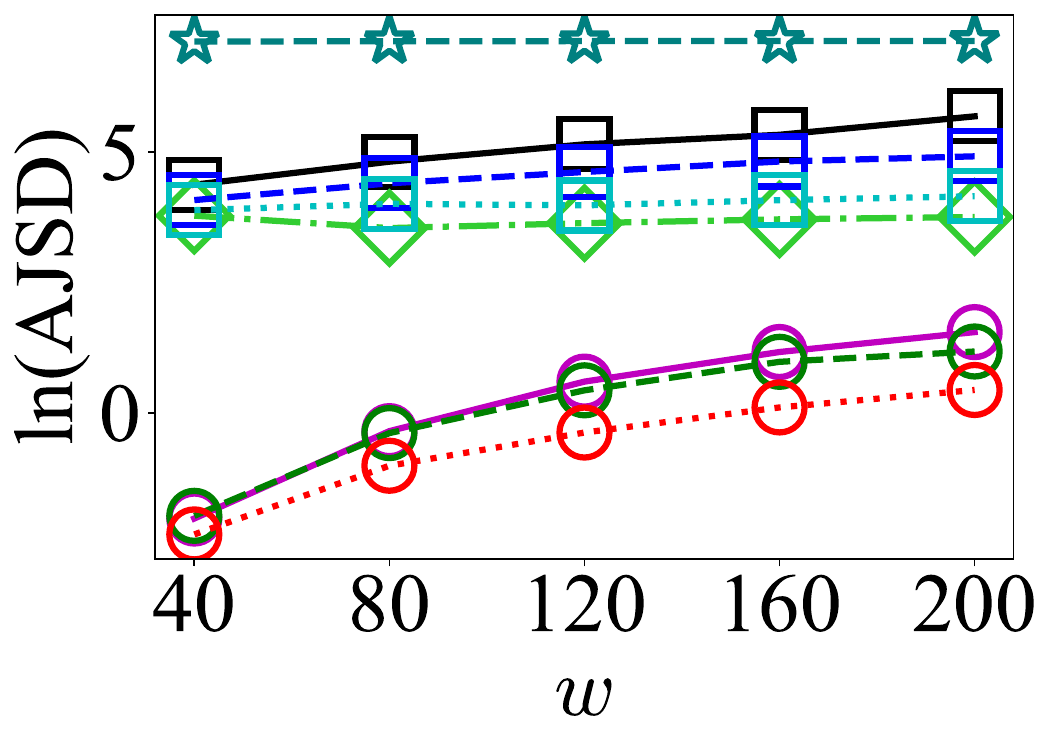}}
		\label{subfig:sin_window_size_change_AJSD_total}}\hfill 
	\subfigure[][{\small \logDatasetName{}}]{
		\scalebox{0.177}[0.177]{\includegraphics{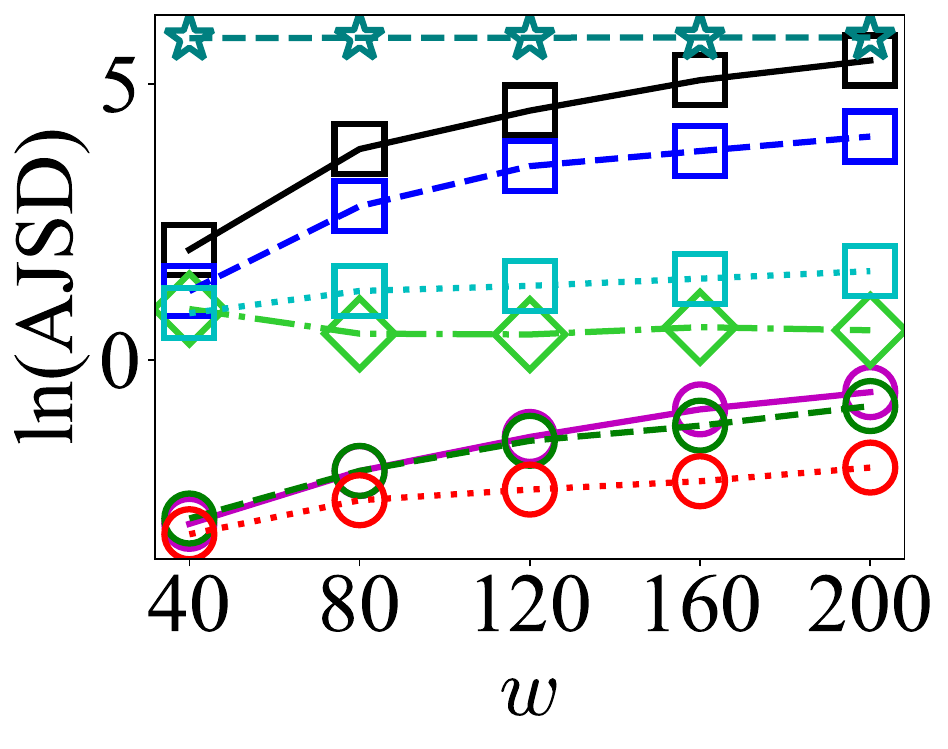}}
		\label{subfig:log_window_size_change_AJSD_total}}\hfill 
	\caption{\small Average Jensen-Shannon Divergence ($\hbox{\em AJSD}$) with $w$ varied.}
	\label{fig:alter_w_AJSD_total}
\end{figure*}

Figure~\ref{fig:alter_e_AJSD_total} shows the results of $\hbox{\em AJSD}$ as the privacy budget $\algvar{E}$ varies from $0.2$ to $1$.
For all methods, $\hbox{\em AJSD}$ decreases as $\algvar{E}$ increases, which is broadly consistent with the $\hbox{\em AMRE}$ trend in Section~\ref{exp:utility}.
\solutionCMPPLDPU{} performs worse than CDP-based methods on most datasets, since LDP methods generally provide lower utility under the same privacy budget.
Both \solutionMethodA{} and \solutionMethodB{} outperform \solutionCMPA{} across all datasets.
\solutionMethodD{} achieves the lowest $\hbox{\em AJSD}$ on the two real datasets, while \solutionMethodE{} performs best with the three synthetic datasets.

Notably, the $\hbox{\em AJSD}$ ranking is not always identical to the $\hbox{\em AMRE}$ ranking. This is expected because $\hbox{\em AMRE}$ measures pointwise estimation error, whereas $\hbox{\em AJSD}$ evaluates similarity between the released and true distributions. Therefore, $\hbox{\em AJSD}$ captures a different aspect of utility and should be viewed as a complementary distribution-level metric, rather than as direct evidence for the same causal explanation used for $\hbox{\em AMRE}$.

Figure~\ref{fig:alter_w_AJSD_total} shows the results of $\hbox{\em AJSD}$ as the window size $w$ varies from $20$ to $200$.
$\hbox{\em AJSD}$ generally increases with larger window sizes for all methods, since a larger window reduces the effective privacy budget available at each time slot.
\solutionCMPPLDPU{} again shows lower utility than the CDP-based methods on most datasets.
Consistent with Figure~\ref{fig:alter_e_AJSD_total}, both \solutionMethodA{} and \solutionMethodB{} outperform \solutionCMPA{}.
\solutionMethodD{} achieves the lowest $\hbox{\em AJSD}$ on the two real datasets, while \solutionMethodE{} leads on the three synthetic datasets. 

Similar to Figure~\ref{fig:alter_e_AJSD_total}, these results should be interpreted together with $\hbox{\em AMRE}$ rather than in isolation. In particular, the discrepancy between $\hbox{\em AMRE}$ and $\hbox{\em AJSD}$ on Foursquare suggests that distribution-level similarity and pointwise estimation accuracy may favor different mechanisms. For this reason, we avoid attributing the $\hbox{\em AJSD}$ behavior to a single factor such as dimensionality or sparsity.

The comparison with \solutionCMPSPAS{} further supports the above observations under the $\hbox{\em AJSD}$ metric. On the real datasets, \solutionCMPSPAS{} usually yields larger $\hbox{\em AJSD}$ than the proposed methods, which indicates that the proposed personalized mechanisms preserve distribution-level similarity more effectively than this homogeneous $w$-event baseline. On the synthetic datasets, \solutionCMPSPAS{} can outperform some distribution-based methods, especially when the stream changes smoothly. However, \solutionMethodB{} and \solutionMethodE{} still achieve the best or near-best $\hbox{\em AJSD}$ in most cases, showing the advantage of budget absorption for smooth streams. Since $\hbox{\em AMRE}$ measures pointwise relative error whereas $\hbox{\em AJSD}$ evaluates the similarity between the released and true distributions, the results under these two metrics provide complementary evidence that the proposed personalized mechanisms improve data utility over a recent homogeneous $w$-event baseline while supporting a more general privacy model.

\begin{figure*}[h]\centering
	\subfigure{
		\scalebox{0.27}[0.27]{\includegraphics{bar5.pdf}}}\hfill\\
	\addtocounter{subfigure}{-1}\vspace{-2ex}
	\subfigure[][{\small \trajectoryDatasetName{}}]{
		\scalebox{0.177}[0.177]{\includegraphics{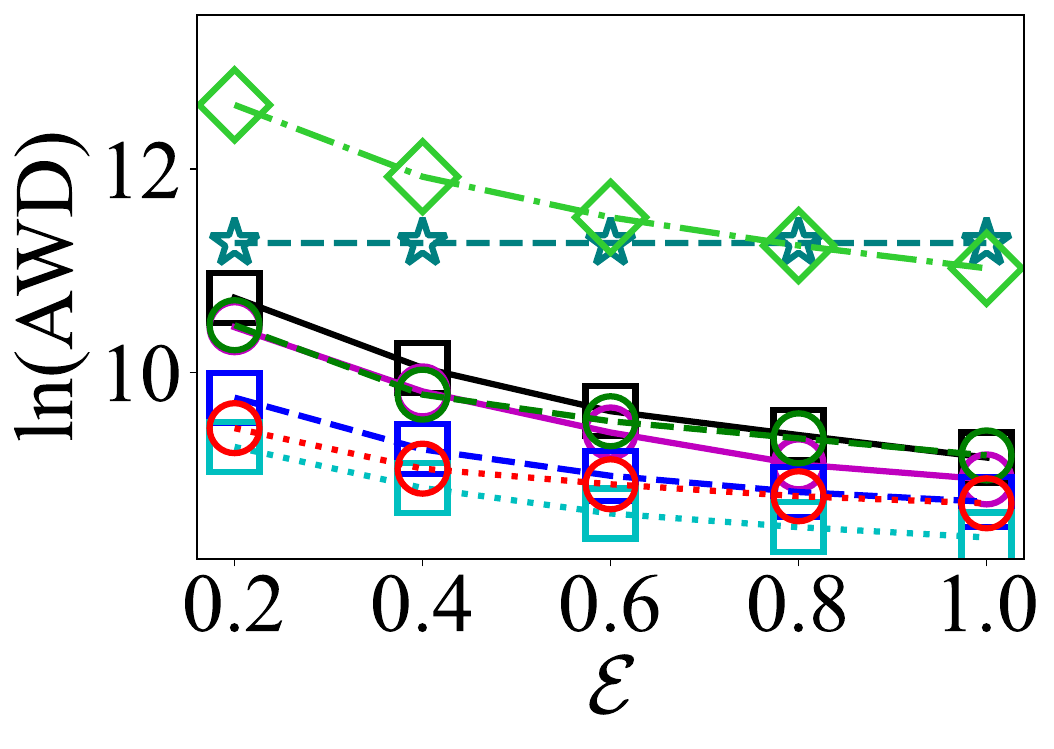}}
		\label{subfig:trajectory_budget_change_AWD_total}}\hfill
	\subfigure[][{\small \checkInDatasetName{}}]{
		\scalebox{0.177}[0.177]{\includegraphics{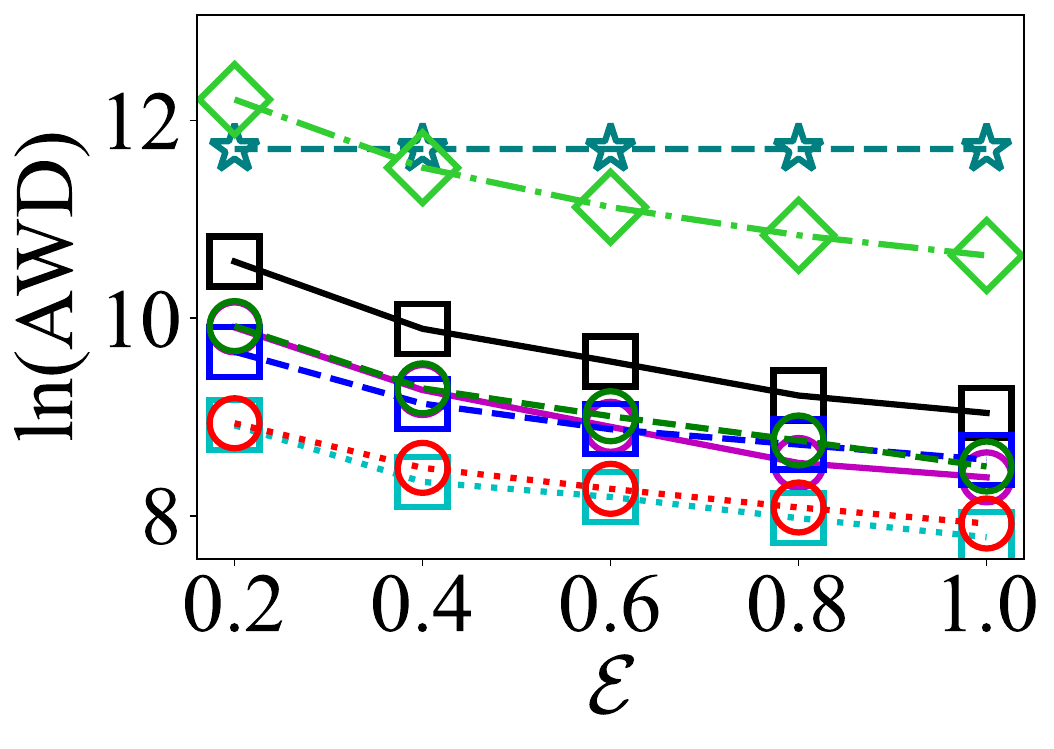}}
		\label{subfig:check_in_budget_change_AWD_total}}\hfill	
	\subfigure[][{\small \tlnsDatasetName{}}]{
		\scalebox{0.177}[0.177]{\includegraphics{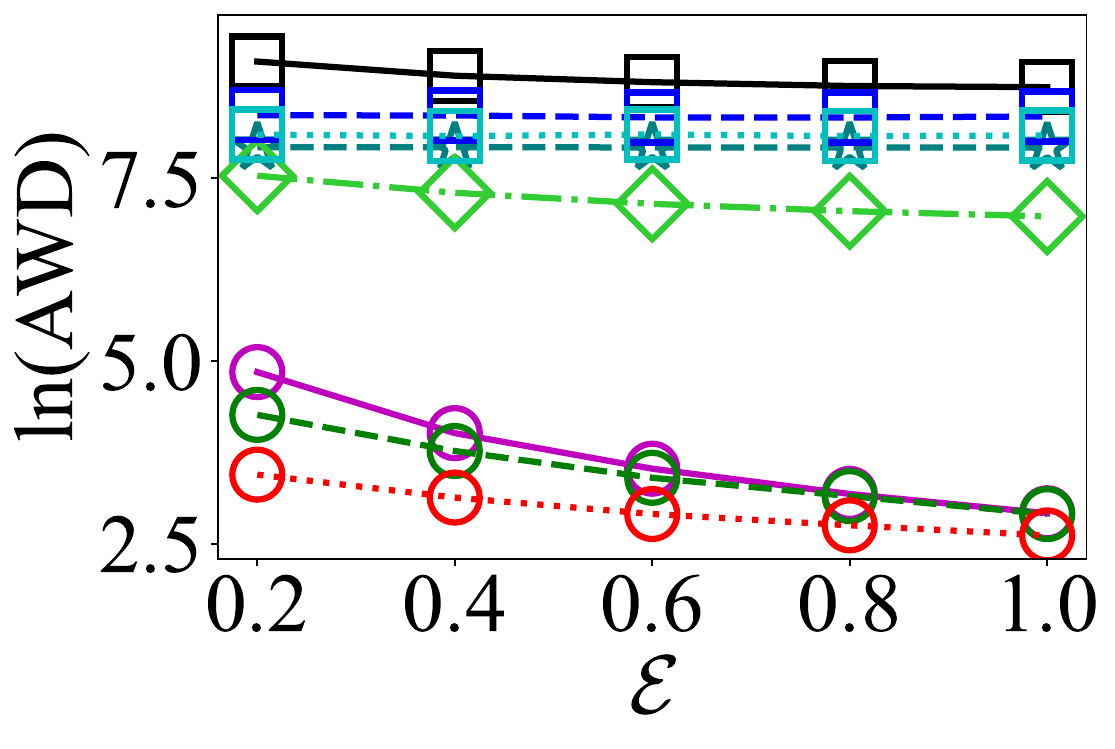}}
		\label{subfig:tlns_budget_change_AWD_total}}\hfill 
	\subfigure[][{\small \sinDatasetName{}}]{
		\scalebox{0.177}[0.177]{\includegraphics{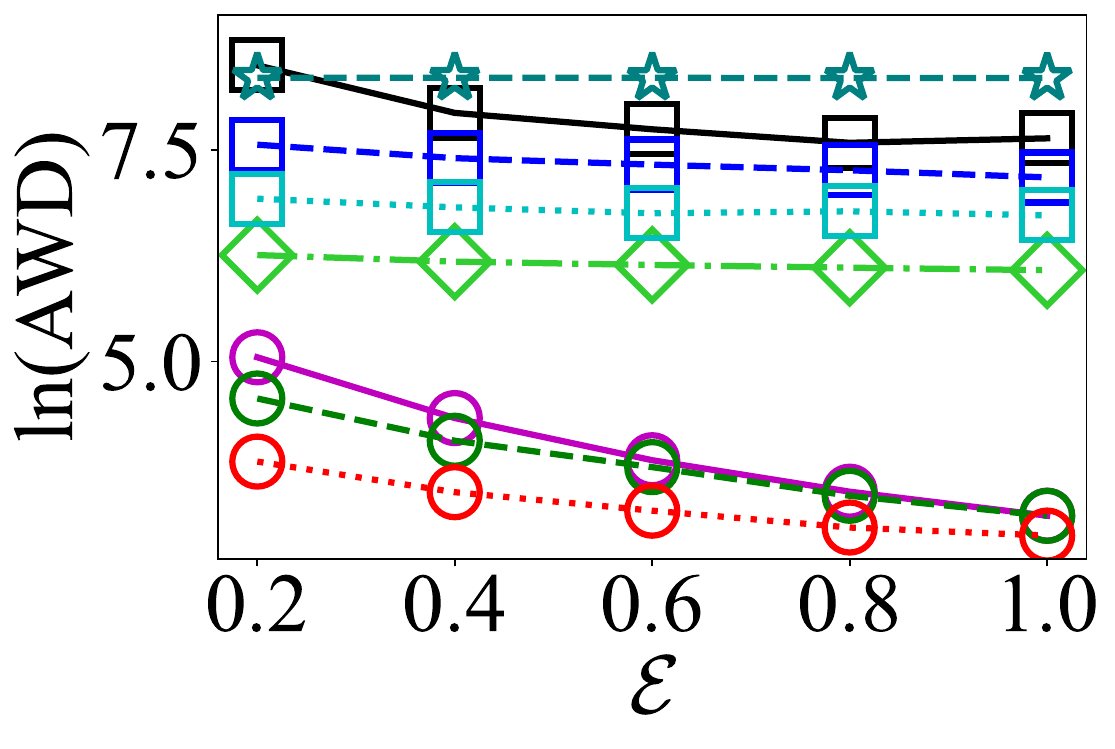}}
		\label{subfig:sin_budget_change_AWD_total}}\hfill 
	\subfigure[][{\small \logDatasetName{}}]{
		\scalebox{0.177}[0.177]{\includegraphics{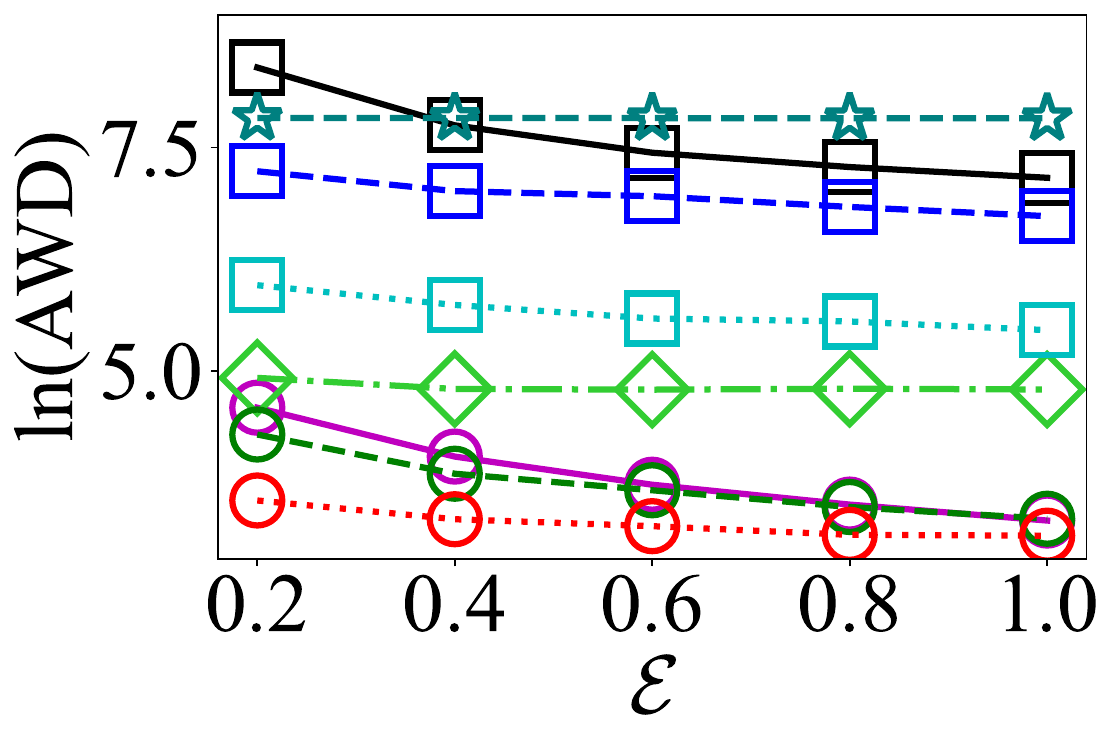}}
		\label{subfig:log_budget_change_AWD_total}}\hfill 
	\caption{\small Average Wasserstein Distance ($\hbox{\em AWD}$) with $\algvar{E}$ varied.}
	\label{fig:alter_e_AWD_total}
\end{figure*}

\begin{figure*}[!h]\centering
	\subfigure{
		\scalebox{0.27}[0.27]{\includegraphics{bar5.pdf}}}\hfill\\
	\addtocounter{subfigure}{-1}\vspace{-2ex}
	\subfigure[][{\small \trajectoryDatasetName{}}]{
		\scalebox{0.177}[0.177]{\includegraphics{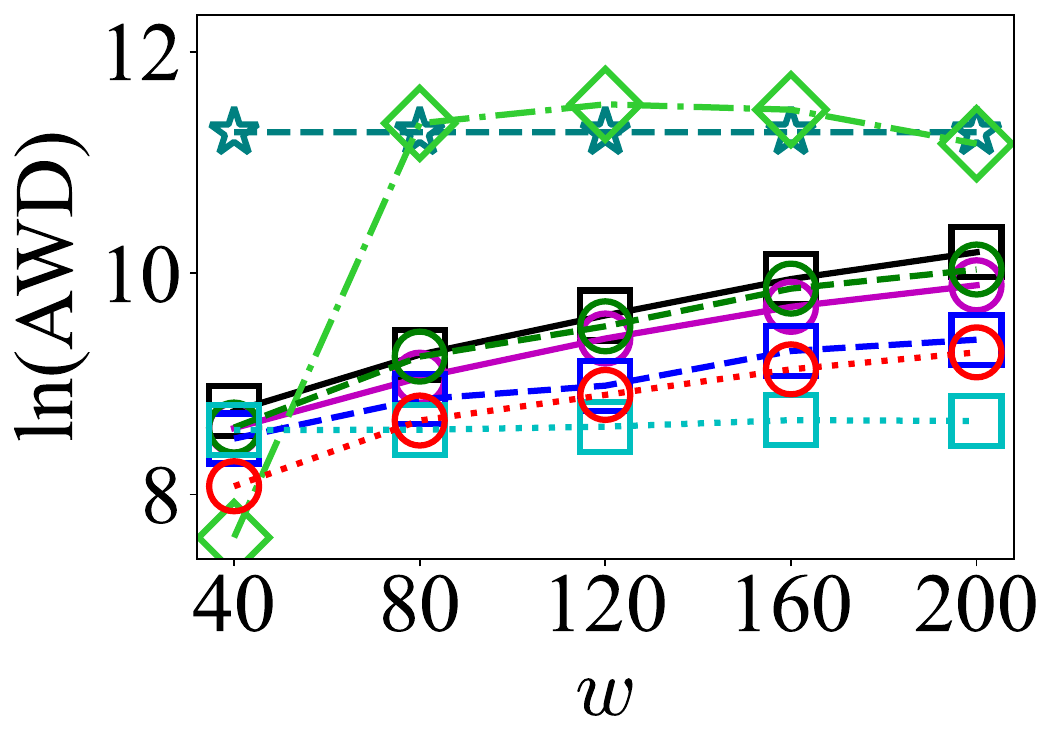}}
		\label{subfig:trajectory_window_size_change_AWD_total}}\hfill
	\subfigure[][{\small \checkInDatasetName{}}]{
		\scalebox{0.177}[0.177]{\includegraphics{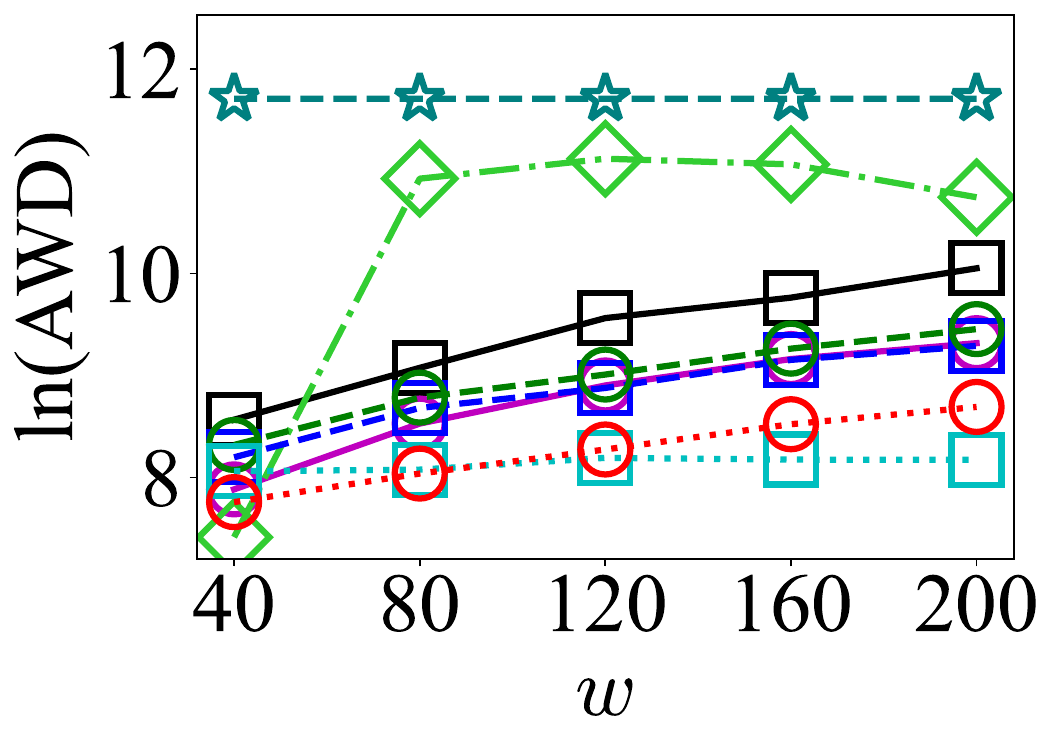}}
		\label{subfig:check_in_window_size_change_AWD_total}}\hfill	
	\subfigure[][{\small \tlnsDatasetName{}}]{
		\scalebox{0.177}[0.177]{\includegraphics{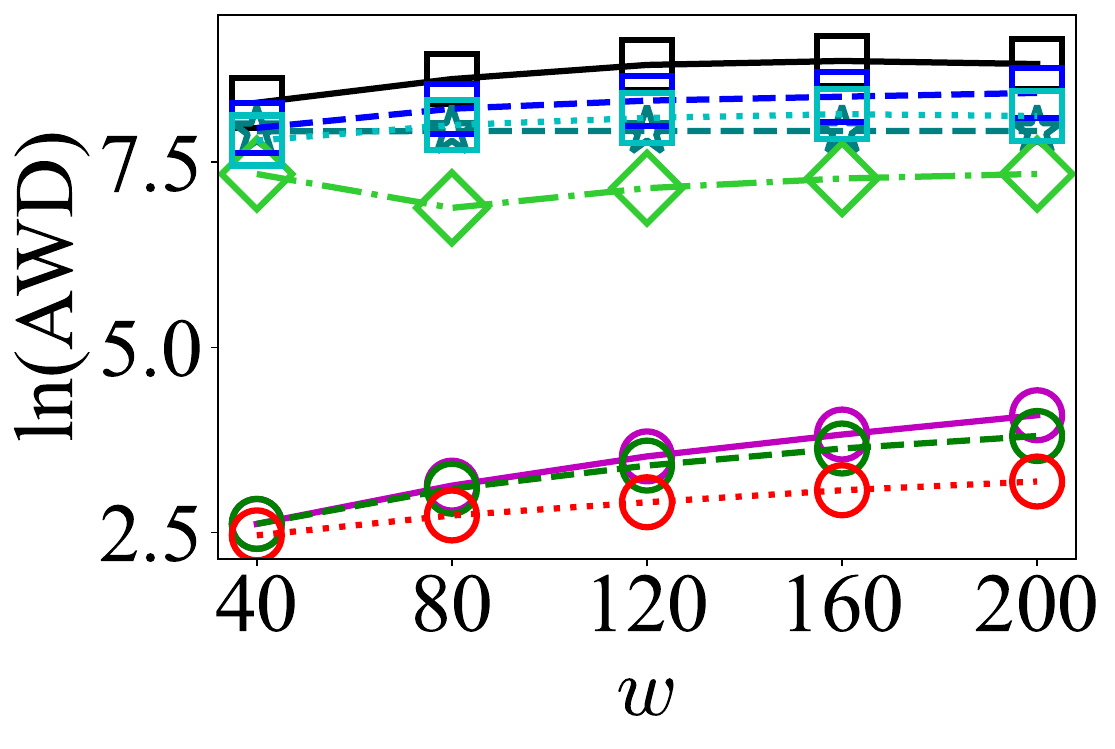}}
		\label{subfig:tlns_window_size_change_AWD_total}}\hfill 
	\subfigure[][{\small \sinDatasetName{}}]{
		\scalebox{0.177}[0.177]{\includegraphics{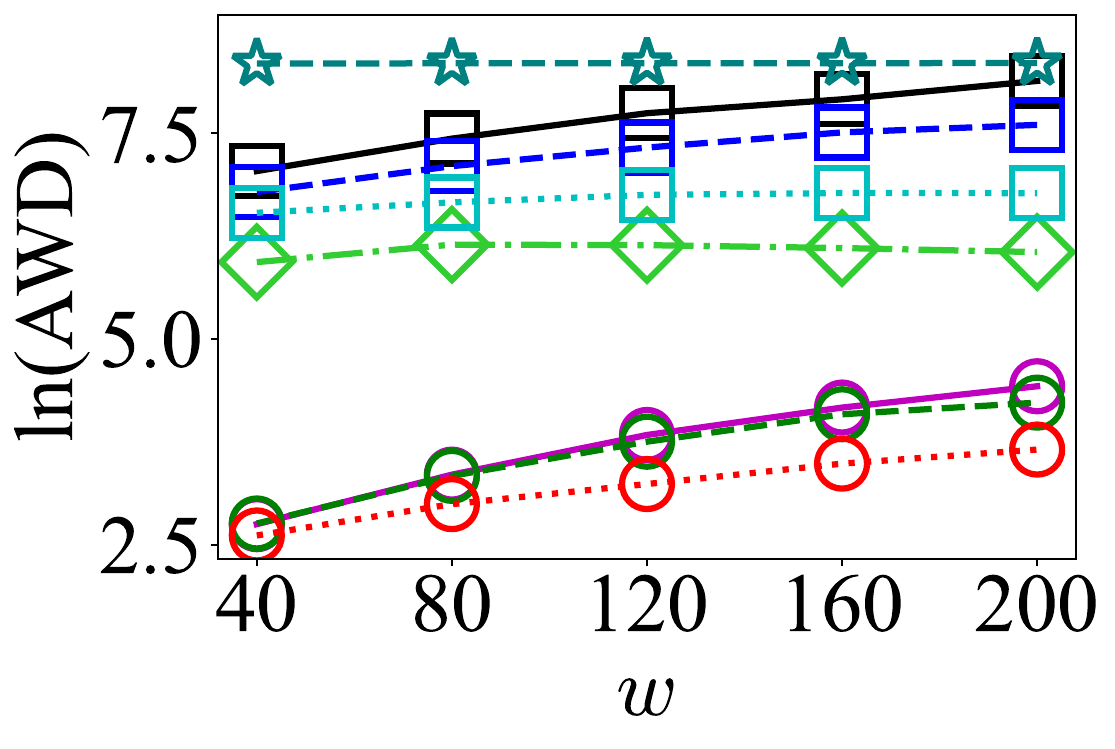}}
		\label{subfig:sin_window_size_change_AWD_total}}\hfill 
	\subfigure[][{\small \logDatasetName{}}]{
		\scalebox{0.177}[0.177]{\includegraphics{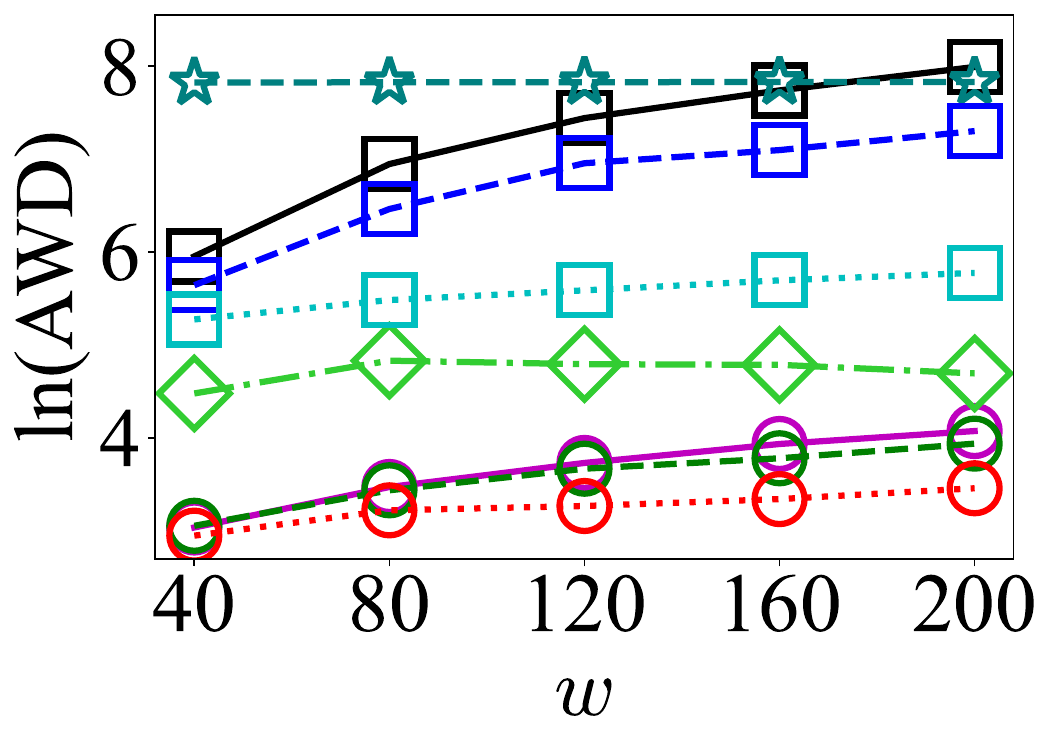}}
		\label{subfig:log_window_size_change_AWD_total}}\hfill 
	\caption{\small Average Wasserstein Distance ($\hbox{\em AWD}$) with $w$ varied.}
	\label{fig:alter_w_AWD_total}
\end{figure*}

Figures~\ref{fig:alter_e_AWD_total} and~\ref{fig:alter_w_AWD_total} show the $\hbox{\em AWD}$ results under varying privacy budget $\algvar{E}$ and window size $w$, respectively. Since the observed trends are similar to those for $\hbox{\em AJSD}$, we omit a detailed discussion for brevity.

\subsection{Experiments for Dimension Change}\label{appendix:exp:dimension}
\begin{figure*}[h]\centering
	\subfigure{
		\scalebox{0.27}[0.27]{\includegraphics{bar3_time1.pdf}}}\hfill\\
	\addtocounter{subfigure}{-1}\vspace{-2ex}
	\subfigure[][{\small \trajectoryDatasetName{}}]{
		\scalebox{0.2}[0.2]{\includegraphics{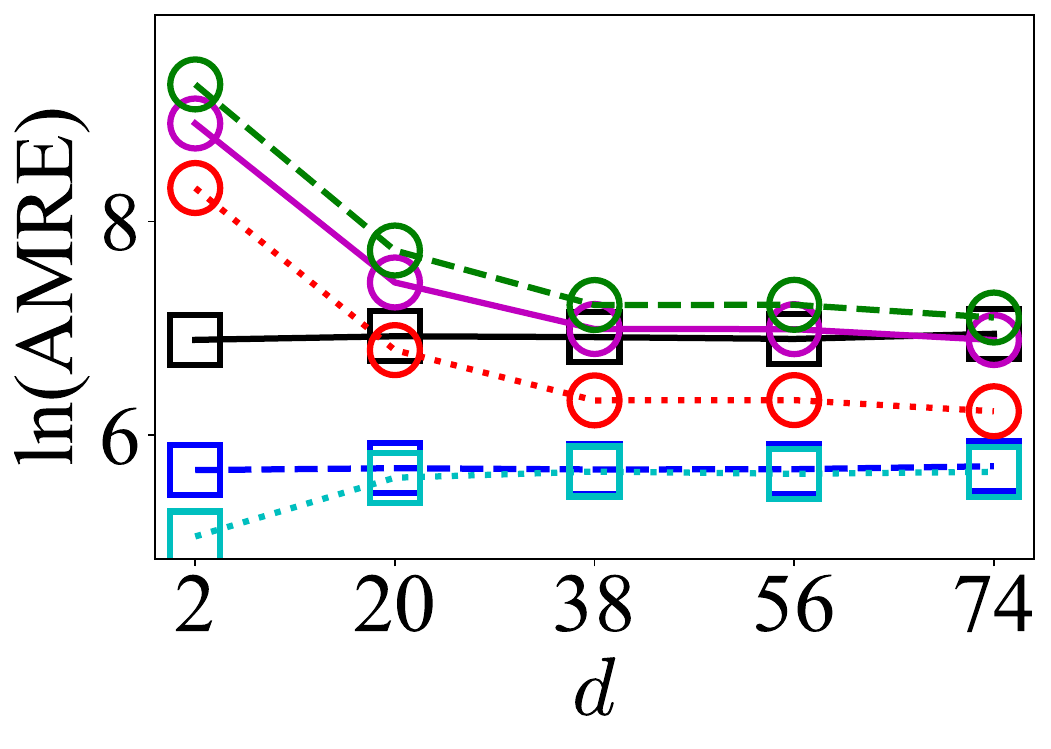}}
		\label{subfig:trajectory_dimension_change_AJSD}}\hfill
	\subfigure[][{\small \checkInDatasetName{}}]{
		\scalebox{0.2}[0.2]{\includegraphics{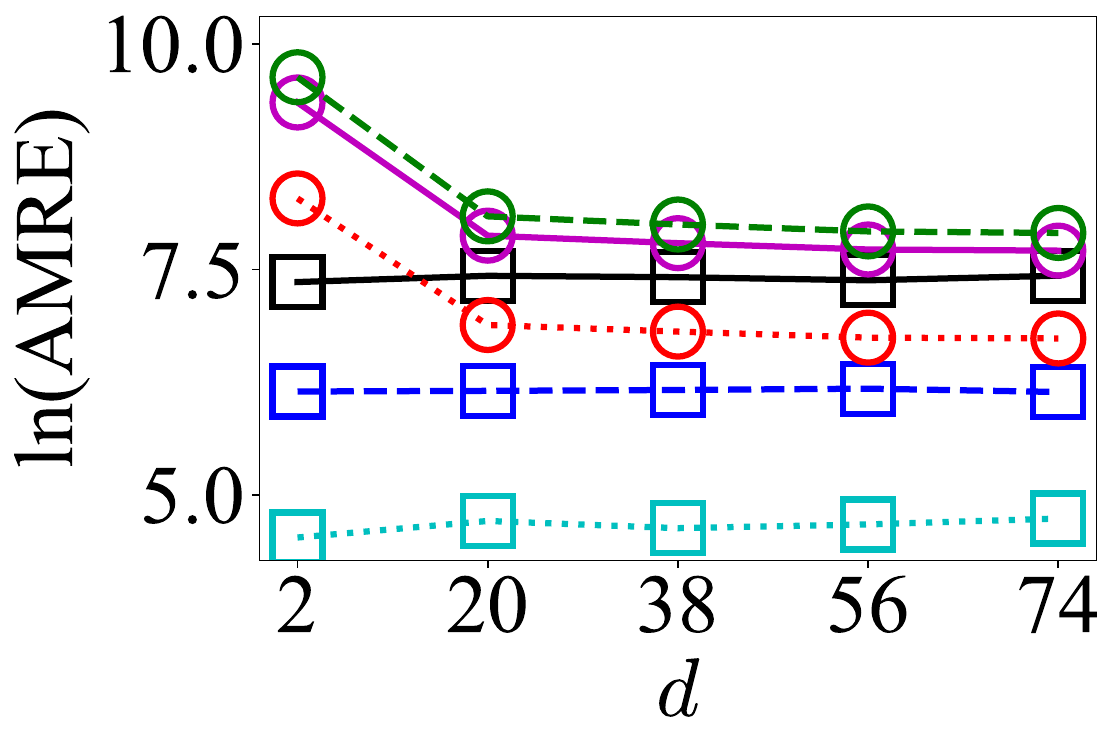}}
		\label{subfig:check_in_dimension_change_AJSD}}\hfill	
	\subfigure[][{\small \trajectoryDatasetName{}}]{
		\scalebox{0.2}[0.2]{\includegraphics{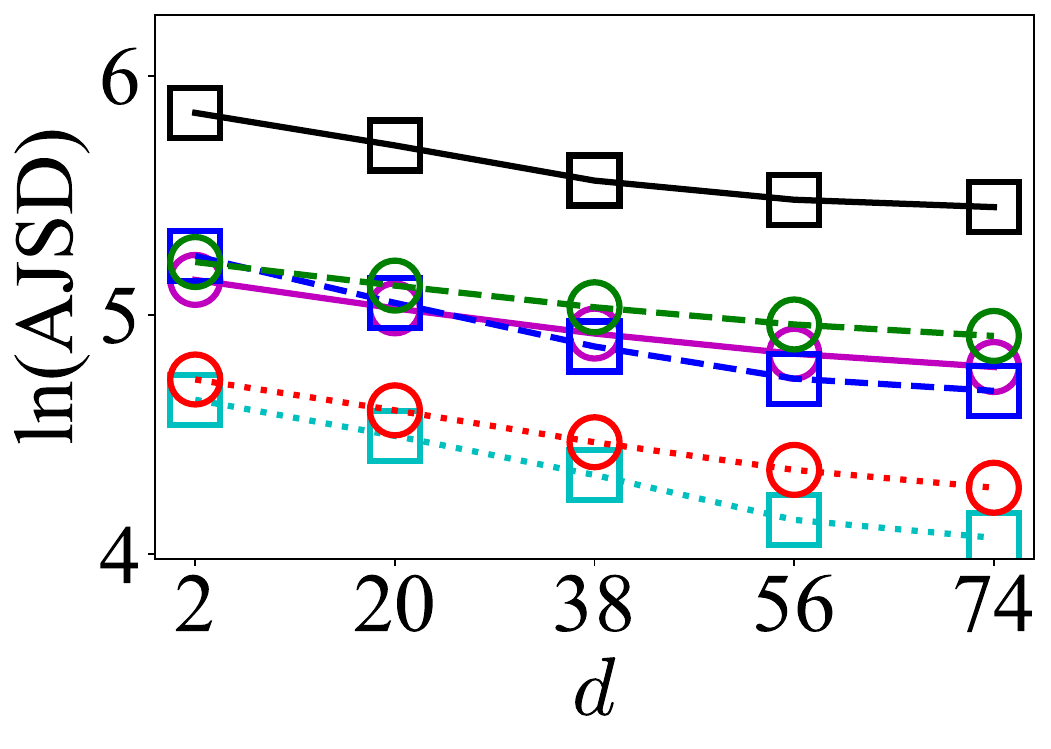}}
		\label{subfig:tlns_dimension_change_AJSD}}\hfill 
	\subfigure[][{\small \checkInDatasetName{}}]{
		\scalebox{0.2}[0.2]{\includegraphics{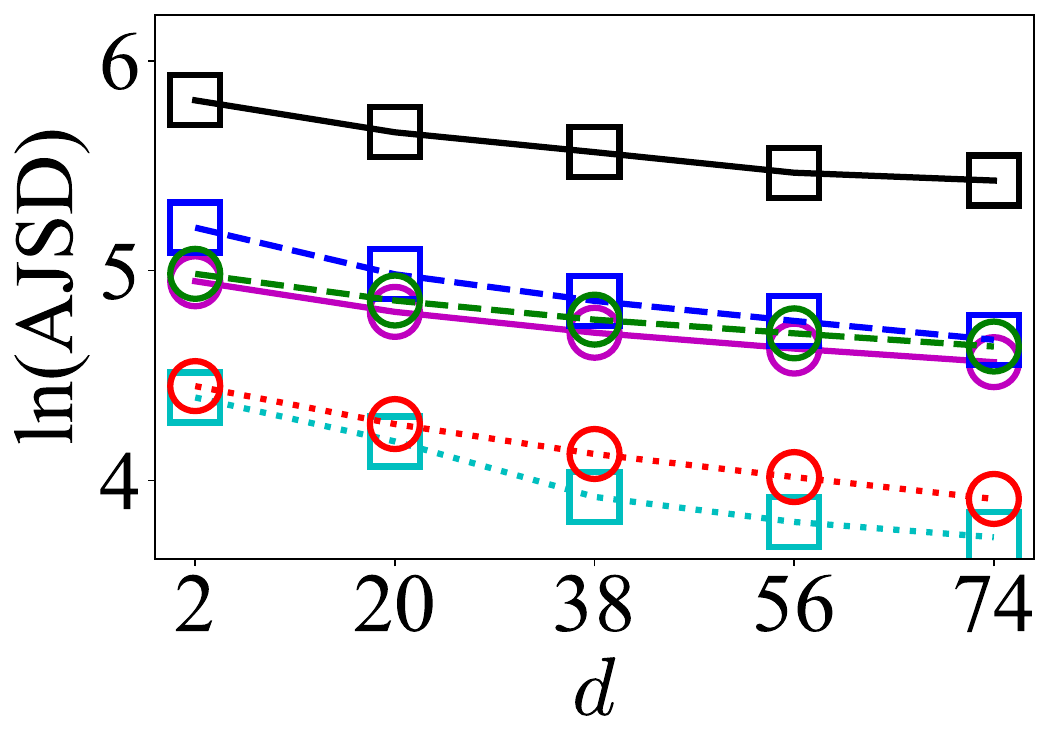}}
		\label{subfig:sin_dimension_change_AJSD}}\hfill 
	\caption{\small AMRE and AJSD with $d$ varied.}
	\label{fig:alter_d_AMRE_AJSD}
\end{figure*}

Figure~\ref{fig:alter_d_AMRE_AJSD} reports the utility under different dimensional settings. For AMRE, we observe that \solutionMethodA{}/\solutionMethodD{} consistently outperform \solutionMethodB{}/\solutionMethodE{} on the real datasets across all tested dimensions. This suggests that the performance gap is not primarily caused by dimensionality. Instead, the main reason is the temporal characteristics of the real data streams. Since real datasets usually exhibit more rapid and irregular changes across consecutive time slots, the absorption-based strategy in \solutionMethodB{}/\solutionMethodE{} becomes less effective, because skipped or nullified updates are more difficult to approximate accurately using previous releases. In contrast, \solutionMethodA{}/\solutionMethodD{} allocate privacy budgets in a more responsive manner, which leads to lower pointwise estimation error.
For AJSD, the trends are not always identical to those observed for AMRE. This is expected because AMRE measures pointwise estimation accuracy, whereas AJSD evaluates similarity between the released and true distributions. Therefore, we treat AJSD as a complementary utility metric that captures a different aspect of data quality. In the revised manuscript, we avoid attributing the AJSD behavior to a single factor such as dimensionality or sparsity, since the current results do not provide sufficient evidence for such a causal conclusion.

\subsection{Experiments for Error Decomposition of OBS}\label{appendix:exp:error_details}
Figure~\ref{fig:alter_error_obs} decomposes the error selected by
OBS into three components: the DP-noise error, the sampling
variance, and the squared sampling bias. We report these
components separately for $\textrm{Part}_{\hbox{\scriptsize DC}}$ and $\textrm{Part}_{\hbox{\scriptsize NOP}}$, denoted by
A and B in the figure, respectively. For $\textrm{Part}_{\hbox{\scriptsize DC}}$, the results
are averaged over all invocations of OBS. For $\textrm{Part}_{\hbox{\scriptsize NOP}}$, the
results are averaged over all non-nullified time slots, since
OBS is invoked before comparing $dis$ with $\sqrt{err}$.
To preserve zero-valued error components while reducing the
difference in magnitude, we report $\ln(1+x)$ for each mean
error component.

\begin{figure*}[h]\centering
	\subfigure{
		\scalebox{0.27}[0.27]{\includegraphics{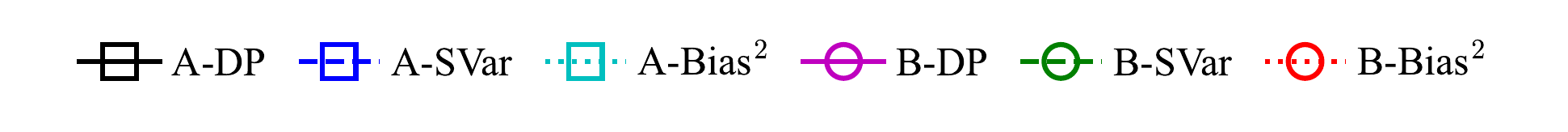}}}\hfill\\
	\addtocounter{subfigure}{-1}\vspace{-2ex}
	\subfigure[][{\small \trajectoryDatasetName{}}]{
		\scalebox{0.177}[0.177]{\includegraphics{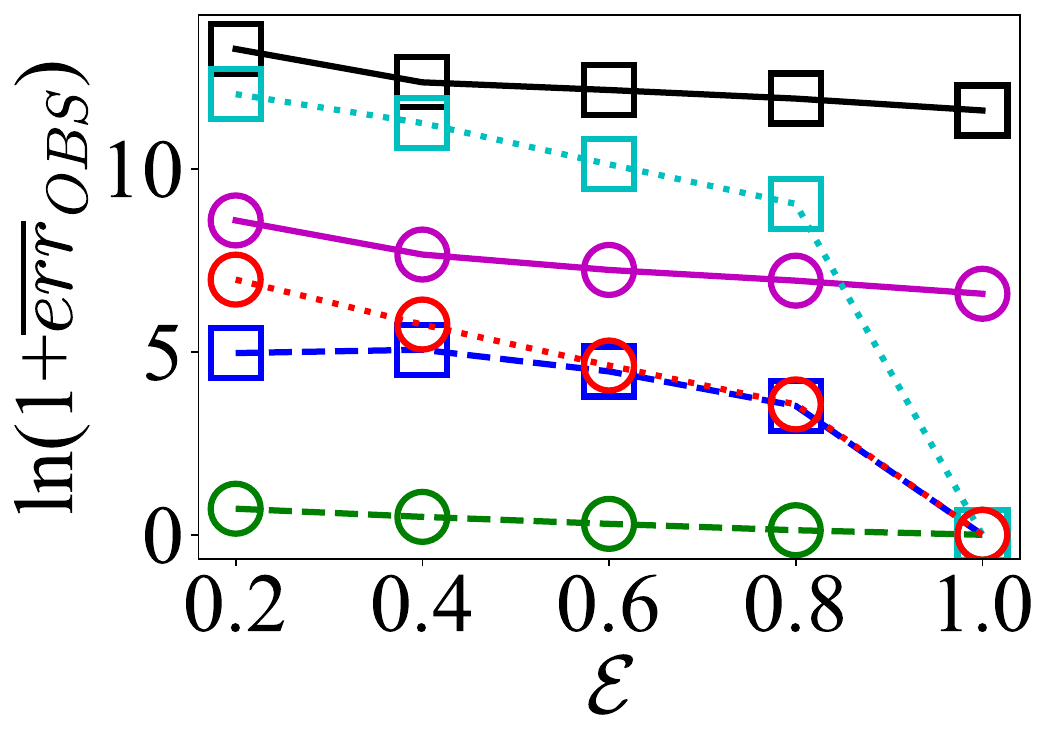}}
		\label{subfig:trajectory_budget_change_error}}\hfill
	\subfigure[][{\small \checkInDatasetName{}}]{
		\scalebox{0.177}[0.177]{\includegraphics{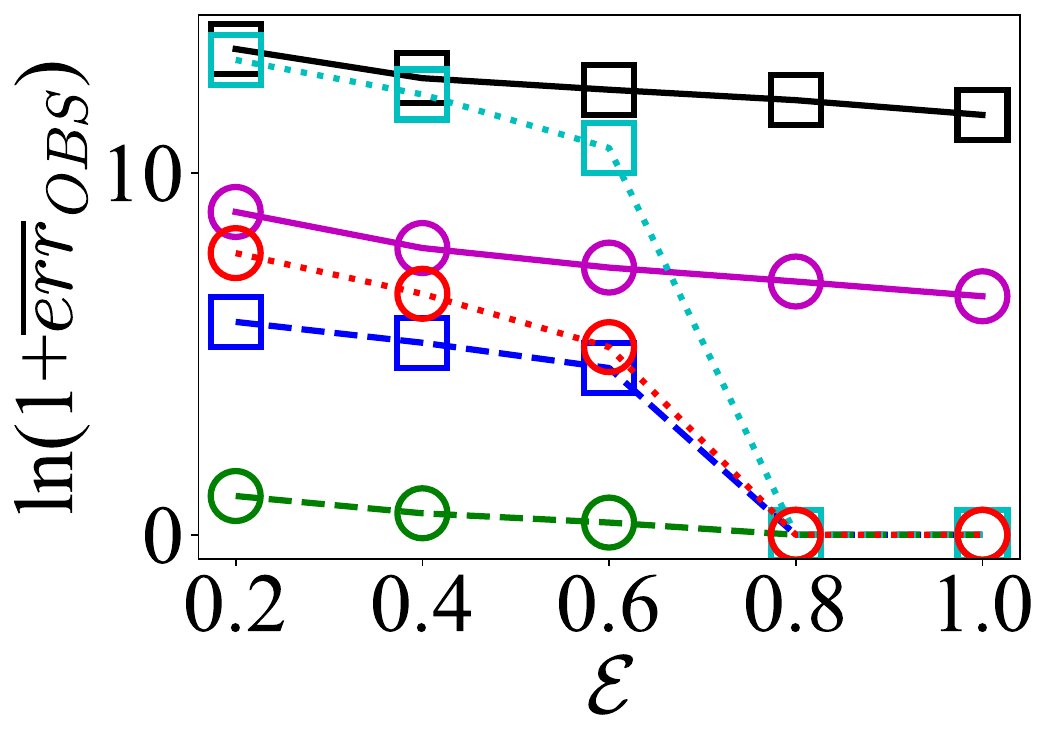}}
		\label{subfig:check_in_budget_change_error}}\hfill	
	\subfigure[][{\small \tlnsDatasetName{}}]{
		\scalebox{0.177}[0.177]{\includegraphics{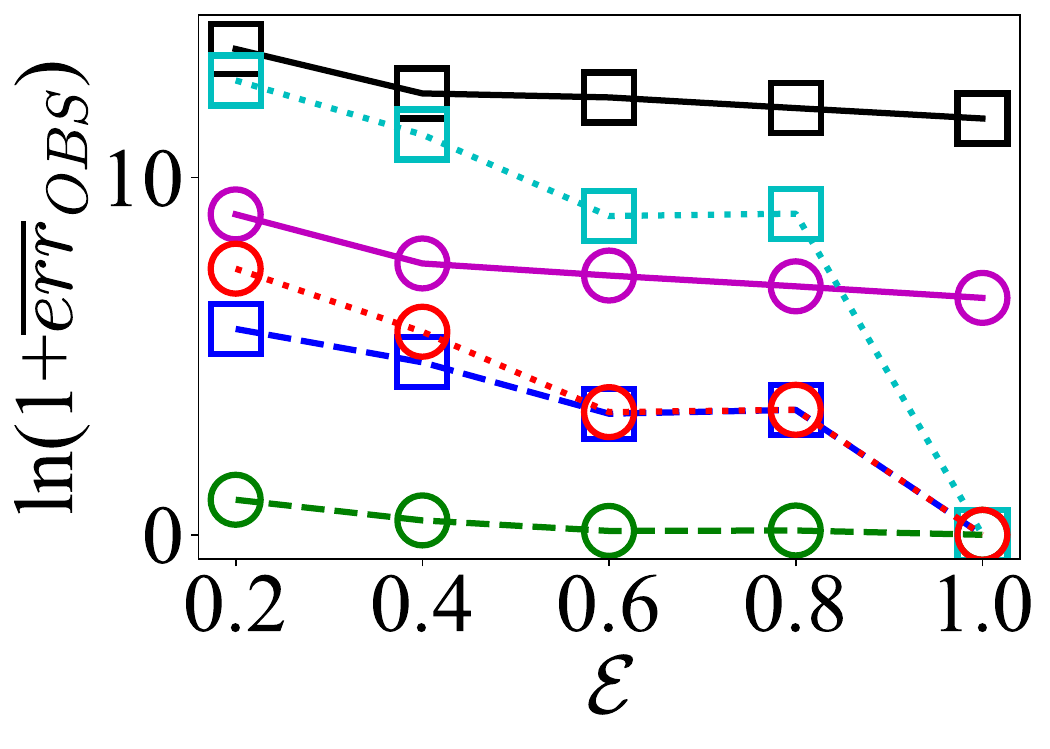}}
		\label{subfig:tlns_budget_change_error}}\hfill 
	\subfigure[][{\small \sinDatasetName{}}]{
		\scalebox{0.177}[0.177]{\includegraphics{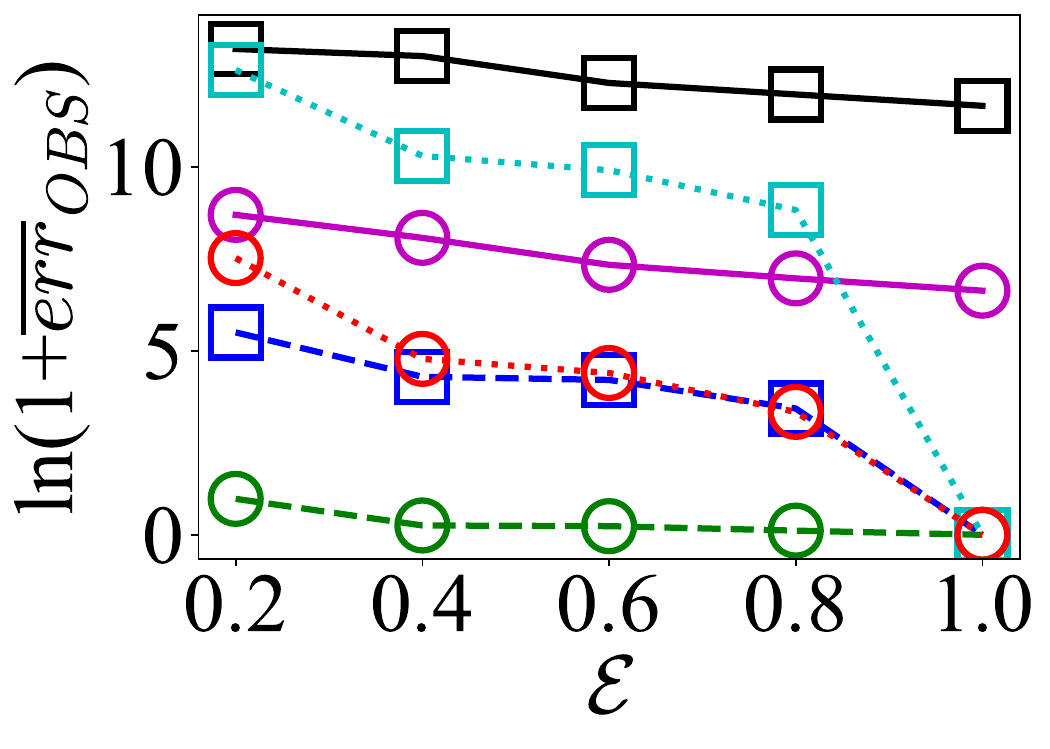}}
		\label{subfig:sin_budget_change_error}}\hfill 
	\subfigure[][{\small \logDatasetName{}}]{
		\scalebox{0.177}[0.177]{\includegraphics{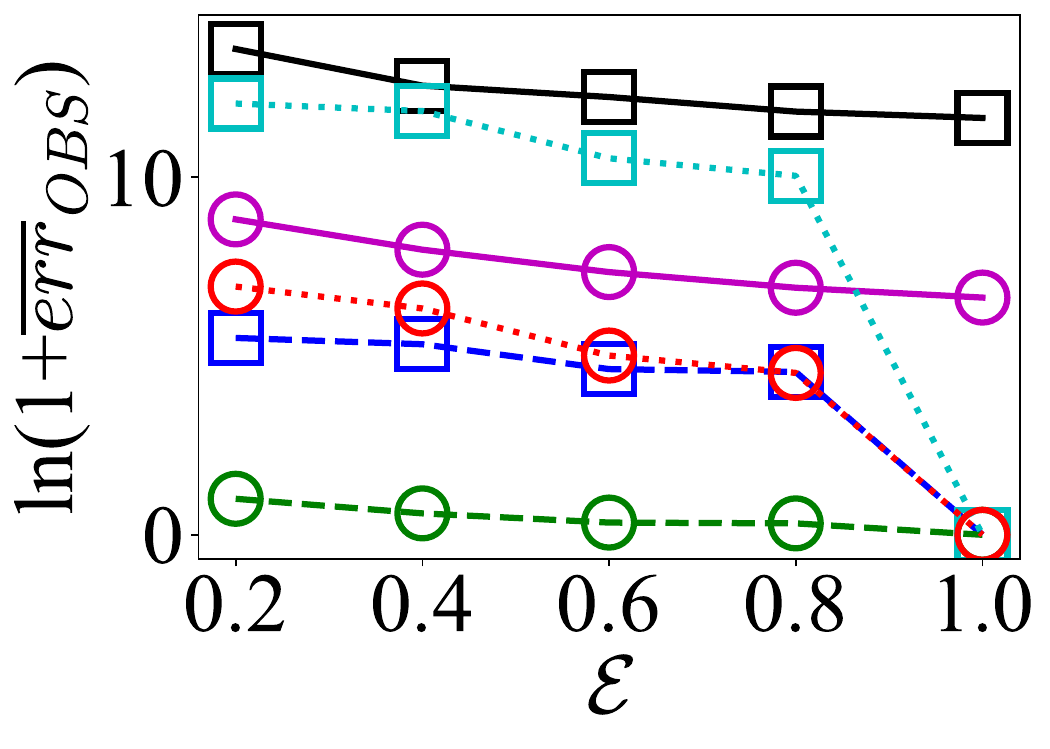}}
		\label{subfig:log_budget_change_error}}\hfill 
	\caption{\small Error decomposition of OBS with $\entity{E}$ varied. A and B denote $\textrm{Part}_{\hbox{\scriptsize DC}}$ and $\textrm{Part}_{\hbox{\scriptsize NOP}}$, respectively; SVar and Bias$^2$ denote sampling variance and squared
		sampling bias, respectively.}
	\label{fig:alter_error_obs}
\end{figure*}

As shown in Figure~\ref{fig:alter_error_obs}, the DP-noise error
is generally the largest component of the total OBS error.
More importantly, whenever sampling is performed, the
squared sampling bias is consistently larger than the
sampling variance in both $\textrm{Part}_{\hbox{\scriptsize DC}}$ and $\textrm{Part}_{\hbox{\scriptsize NOP}}$ across the
datasets. Therefore, the additional sampling error introduced
by OBS is mainly attributable to the squared sampling bias,
rather than to the sampling variance.

In $\textrm{Part}_{\hbox{\scriptsize DC}}$, this bias affects the sampled statistic used to
calculate the private dissimilarity. In $\textrm{Part}_{\hbox{\scriptsize NOP}}$, it contributes
to the reporting error used in the decision threshold
$\sqrt{err}$ and consequently affects whether a new
obfuscated result is published or the previous result is
reused. Thus, the OBS-induced bias can influence utility
through both dissimilarity estimation and adaptive release
decisions.

\subsection{Proof for Privacy Analysis}
\subsubsection{Proof for Theorem~\ref{Thm:solutionAB_privacy_analysis}}\label{Proof:solutionAB_privacy_analysis}
\begin{proof}
	We record $\max(t-w_i+1,1)$ as $t_L$ for short. 
	
	(1) \solutionMethodA{} satisfies $(\vectorfont{w},\vectorfont{\algvar{E}})$-EPDP.
	
	In the process of Part$_{\hbox{\scriptsize DC}}$, for each user $u_i$, the dissimilarity budget at each time slot
	is $\algvar{E}_i/(2w_i)$. Then for each time slot $t$, we have
	\begin{equation}\label{eq:PBD_privacy_DC}
		\outereqsizelarge{
			\begin{aligned}
				\sum_{k=t_L}^{t} \epsilon_{i,k}^{(1)} \leq \frac{\algvar{E}_i}{2}.
			\end{aligned}
		}
	\end{equation}
	
	In Part$_{\hbox{\scriptsize NOP}}$, for each user $u_i$ at time slot $t$, only half of the publication
	budget is used when non-null publication occurs: $\epsilon_{i,t}^{(2)}=\left(\algvar{E}_i/2-\sum_{k=t_L}^{t-1}\epsilon_{i,k}^{(2)}\right)/2$. 
	For any time slot $t\in[1,w_i]$, the summation publication budgets used for $u_i$ is at most $\sum_{k=1}^{w_i}\algvar{E}_i/\left(2\cdot 2^{k}\right)\leq (\algvar{E}_i/2)\cdot \left(1-1/2^{w_i}\right)\leq \algvar{E}_i/2$. 
	
	Assume $\sum_{k=t_L}^{t}\epsilon_{i,k}^{(2)}\leq \algvar{E}_i/2$ for $t=w_i+s$
	\big(i.e., \innereqsize{$\sum_{k=\max(s+1,1)}^{w_i+s} \epsilon_{i,k}^{(2)}$ $\leq \algvar{E}_i/2$}\big).
	Then for $t=w_i+s+1$, we have:
	\begin{equation}\label{divide_equation}
		\outereqsizelarge{
			\begin{aligned}
				\sum_{k=\max(s+2,1)}^{w_i+s+1} \epsilon_{i,k}^{(2)} = \sum_{k=\max(s+2,1)}^{w_i+s} \epsilon_{i,k}^{(2)} + \epsilon_{i,w_i+s+1}^{(2)}. 
			\end{aligned}
		}
	\end{equation}
	Since $\epsilon_{i,w_i+s+1}^{(2)}$ is at most half of the remaining publication budget at time slot $w_i+s$, thus, 
	\begin{equation}\label{distribution_equation}
		\outereqsizelarge{
			\begin{aligned}
				\epsilon_{i,w_i+s+1}^{(2)} \leq \frac{\frac{\algvar{E}_i}{2}-\sum_{k=\max(s+2,1)}^{w_i+s} \epsilon_{i,k}^{(2)}}{2}.
			\end{aligned}
		}
	\end{equation}
	Let $k^*=\max(s+2,1)$. 
	According to Equations~\eqref{divide_equation} and~\eqref{distribution_equation}, we have:
	\begin{equation}\notag
		\outereqsizelarge{
			\begin{aligned}
				\sum_{k=k^*}^{w_i+s+1} \epsilon_{i,k}^{(2)} &\leq \sum_{\substack{k=k^*}}^{w_i+s} \epsilon_{i,k}^{(2)} + \frac{\frac{\algvar{E}_i}{2}-\sum_{\substack{k=k^*}}^{w_i+s} \epsilon_{i,k}^{(2)}}{2} \\
				&=\frac{\algvar{E}_i}{4} + \frac{\sum_{k=k^*}^{w_i+s} \epsilon_{i,k}^{(2)}}{2}\\
				&\leq \frac{\algvar{E}_i}{4} +\frac{\algvar{E}_i}{4}\\
				&=\frac{\algvar{E}_i}{2}.
			\end{aligned}
		}
	\end{equation}
	Therefore, for any $t\geq 1$, we have:
	\begin{equation}\label{eq:PBD_privacy_NOP}
		\outereqsizelarge{
			\begin{aligned}
				\sum_{k=t_L}^{t} \epsilon_{i,k}^{(2)} \leq \frac{\algvar{E}_i}{2}.
			\end{aligned}
		}
	\end{equation}
	According to the Composition Theorems~\cite{DBLP:journals/fttcs/DworkR14} and Equation~\eqref{eq:PBD_privacy_DC} and Equation~\eqref{eq:PBD_privacy_NOP}, we have:
	\begin{equation}\notag
		\outereqsizelarge{
			\begin{aligned}
				\sum_{k=t_L}^{t} \epsilon_{i,k} &= \sum_{\substack{k=t_L}}^{t} \epsilon_{i,k}^{(1)} + \sum_{\substack{k=t_L}}^{t} \epsilon_{i,k}^{(2)} \leq \algvar{E}_i.
			\end{aligned}
		}
	\end{equation}
	
	For any user $u_i$ and any two $w_i$-neighboring stream prefixes $S_t$ and $S_t'$ (i.e., $S_t\sim_{w_i}S_t'$), let $t_s$ be the earliest time slot where $S_t[t_s]\neq S_t'[t_s]$ and $t_e$ be the latest time slot where $S_t[t_e]\neq S_t'[t_e]$. Then we have $t_e-t_s+1\leq w_i$. Denoting the output of our \solutionMethodA{} as $\hbox{\solutionMethodA{}}(S_t[t])=o_{t}\in \entity{O}$, for any $O\subseteq \entity{O}$, we have:
	\begin{equation}\notag
		\outereqsizelarge{
			\begin{aligned}
				\frac{\Pr[\hbox{\hbox{\solutionMethodA{}}}(S_t)\in O]}{\Pr[\hbox{\solutionMethodA{}}(S_t')\in O]} &\leq \Pi_{k=t_s}^{t_e}\frac{\Pr[\hbox{\solutionMethodA{}}(S_t[k])=o_{k}]}{\Pr[\hbox{\solutionMethodA{}}(S_t'[k])=o_{k}]}\\
				&\leq e^{\sum_{k=t_s}^{t_e}\epsilon_{i,k}} \\
				&\leq e^{\sum_{k=\max{(t_e-w_i+1,1)}}^{t_e}\epsilon_{i,k}}\\
				&\leq e^{\algvar{E}_i}.
			\end{aligned}
		}
	\end{equation}
	Therefore, \solutionMethodA{} satisfies $(\vectorfont{w},\vectorfont{\algvar{E}})$-EPDP where $\vectorfont{w}$$=(w_1,w_2,$ $\ldots,w_n)$ and $\vectorfont{\algvar{E}}=(\algvar{E}_1,\algvar{E}_2,\ldots,\algvar{E}_n)$.
	
	(2) \solutionMethodB{} satisfies $(\vectorfont{w},\vectorfont{\algvar{E}})$-EPDP.
	
	Part$_{\hbox{\scriptsize DC}}$ in \solutionMethodB{} is identical to that in \solutionMethodA{}. Consequently, for each time slot $t$, we have:
	\begin{equation}\label{absorb_part1}
		\outereqsizelarge{
			\begin{aligned}
				\sum_{k=t_L}^{t} \epsilon_{i,k}^{(1)} \leq \algvar{E}_i/2.
			\end{aligned}
		}
	\end{equation}
	
	In Part$_{\hbox{\scriptsize NOP}}$, for any user $u_i$ and any window of size $w_i$, there are $s_i$ publication time slots in the window.
	We denote these publication time slots as $(k_1, k_2,\ldots ,k_{s_i})$.
	For any publication time slot $k_j$ \big($j\in[s_i]$\big), the quantity of its absorbing unused budgets is denoted as $\alpha_{i,k_j}$. 
	Figure~\ref{PBD_proof} illustrates an example where $s_i=3$ and $w_i=9$. 
	\begin{figure}[t!]\vspace{-2ex}
		\centering
		\includegraphics[width=0.49\textwidth]{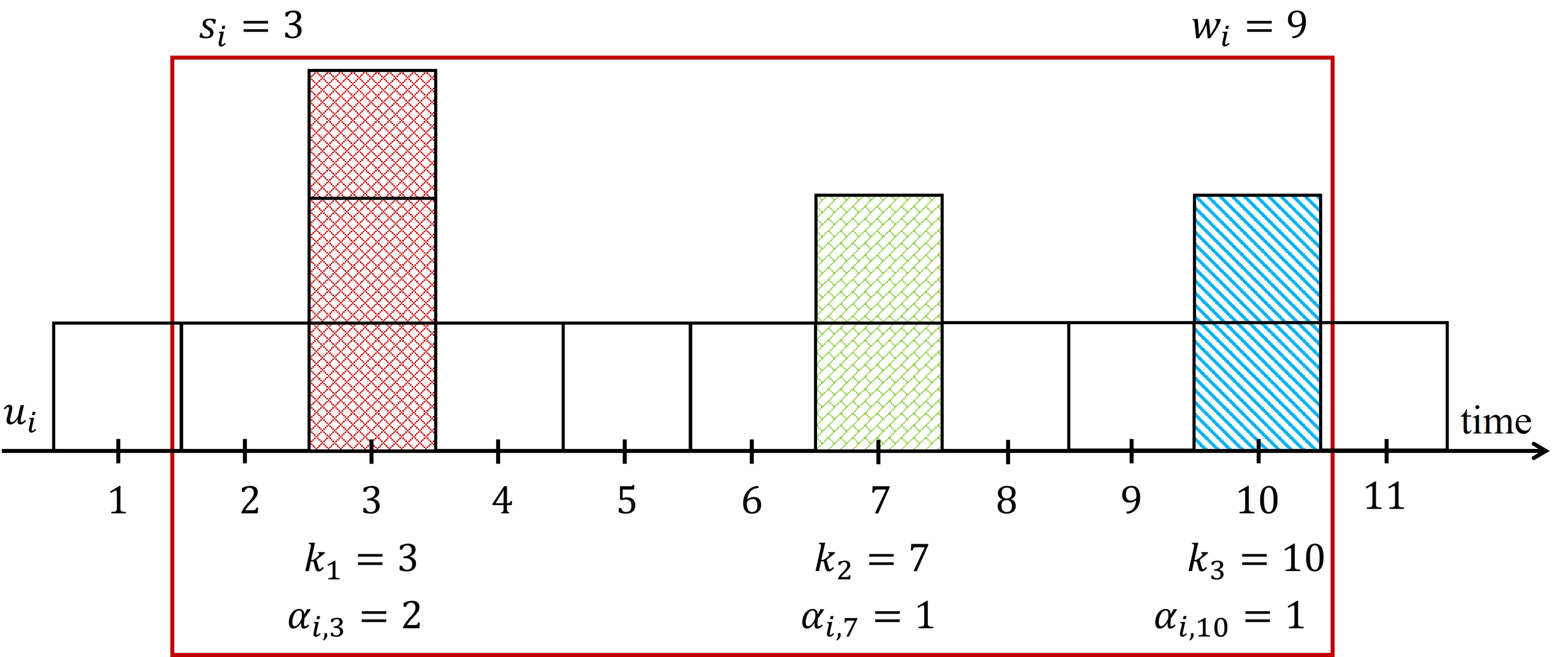}
		\caption{An example for parameters in \solutionMethodB{}.}\label{PBD_proof}
	\end{figure}
	
	Based on Algorithm~\ref{alg:PBA}, we have:
	\begin{equation}\notag
		\outereqsizelarge{
			\begin{aligned}
				w_i\geq\sum_{j=1}^{s_i}(1+2\alpha_{i,k_j})-\alpha_{i,k_1}-\alpha_{i,k_{s_i}}.
			\end{aligned}
		}
	\end{equation}
	Then, for the total publication budgets used in any window, we have
	\begin{equation}\label{absorb_part2}
		\outereqsizelarge{
			\begin{aligned}
				\sum_{k=t_L}^{t} \epsilon_{i,k}^{(2)} &\leq \frac{\algvar{E}_i}{2w_i}\cdot \sum_{j=1}^{s_i}(1+\alpha_{i,k_j})\\
				&\leq \frac{\algvar{E}_i\cdot\sum_{j=1}^{s_i}(1+\alpha_{i,k_j})}{2\sum_{j=1}^{s_i}(1+2\alpha_{i,k_j})-2\alpha_{i,k_1}-2\alpha_{i,k_{s_i}}}\\
				&= \frac{\algvar{E}_i\cdot\sum_{j=1}^{s_i}(1+\alpha_{i,k_j})}{2\sum_{j=1}^{s_i}(1+\alpha_{i,k_j})+2\sum_{j=2}^{s_i-1}\alpha_{i,k_j}}\\
				&\leq \frac{\algvar{E}_i}{2}.
			\end{aligned}
		}
	\end{equation}
	Based on Equations~\eqref{absorb_part1} and~\eqref{absorb_part2}, and applying the Composition Theorems~\cite{DBLP:journals/fttcs/DworkR14}, we obtain:
	\begin{equation}\notag
		\outereqsizelarge{
			\begin{aligned}
				\sum_{k=t_L}^{t} \epsilon_{i,k} &= \sum_{\substack{k=t_L}}^{t} \epsilon_{i,k}^{(1)} + \sum_{\substack{k=t_L}}^{t} \epsilon_{i,k}^{(2)} \leq \algvar{E}_i.
			\end{aligned}
		}
	\end{equation}
	The subsequent proof process follows the same steps as in \solutionMethodA{}.
	Ultimately, we demonstrate that \solutionMethodB{} also satisfies $(\vectorfont{w},\vectorfont{\algvar{E}})$-EPDP.
\end{proof}

\subsubsection{Proof for Theorem~\ref{Thm:solutionDE_privacy_analysis}}\label{Proof:solutionDE_privacy_analysis}
\begin{proof}
	We analyze the privacy guarantees of \solutionMethodD{} and \solutionMethodE{} separately.
	Let $\hat{\algvar{E}}_{B,i,t}=\hat{\algvar{E}}_{B,i,t}^{(1)}+\hat{\algvar{E}}_{B,i,t}^{(2)}$ represent the backward budget usage from $\max(t-w_{B,i,t}+1)$ to $t$ for $u_i$, where $\hat{\algvar{E}}_{B,i,t}^{(1)}$ and $\hat{\algvar{E}}_{B,i,t}^{(2)}$ are the budget usages in sub-mechanism Part$_{\hbox{\scriptsize DC}}$ and Part$_{\hbox{\scriptsize NOP}}$ respectively.
	Let $\hat{\algvar{E}}_{F,i,t}=\hat{\algvar{E}}_{F,i,t}^{(1)}+\hat{\algvar{E}}_{F,i,t}^{(2)}$ represent the forward budget usage from $t$ to $t+w_{F,i,t}-1$ for $u_i$, where $\hat{\algvar{E}}_{F,i,t}^{(1)}$ and $\hat{\algvar{E}}_{F,i,t}^{(2)}$ are the budget usages in sub-mechanism Part$_{\hbox{\scriptsize DC}}$ and Part$_{\hbox{\scriptsize NOP}}$ respectively.
	
	(1) \solutionMethodD{} satisfies \dynamicPrivacyLevelSimpleNameT{} at each time slot $t\in[T]$.
	
	For the backward privacy budget usage of each $u_i$ at time slot $t$, we have
	\begin{equation}\label{eq:m1bb}
		\outereqsize{
			\begin{aligned}
				\hat{\algvar{E}}_{B,i,t}^{(1)} &= \sum_{k=\max(t-w_{B,i,t}+1,1)}^{t} \epsilon_{i,k}^{(1)}\\ 
				&= \epsilon_{i,t}^{(1)} + \sum_{k=\max(t-w_{B,i,t}+1,1)}^{t-1} \epsilon_{i,k}^{(1)}\\
				&\leq \epsilon_{B,i,t}^{(1)} + \frac{\algvar{E}_{B,i,t}}{2}-\epsilon_{B,i,t}^{(1)}\\
				&=\frac{\algvar{E}_{B,i,t}}{2}.
			\end{aligned}
		}
	\end{equation}
	Besides, it holds
	\begin{equation}\label{eq:m2bb}
		\outereqsize{
			\begin{aligned}
				\hat{\algvar{E}}_{B,i,t}^{(2)} &= \sum_{k=\max(t-w_{B,i,t}+1,1)}^{t} \epsilon_{i,k}^{(2)}\\ 
				&= \epsilon_{i,t}^{(2)} + \sum_{k=\max(t-w_{B,i,t}+1,1)}^{t-1} \epsilon_{i,k}^{(2)}\\
				&\leq \epsilon_{B,i,t}^{(2)} + \frac{\algvar{E}_{B,i,t}}{2}-\epsilon_{B,i,t}^{(2)}\\
				&=\frac{\algvar{E}_{B,i,t}}{2}.
			\end{aligned}
		}
	\end{equation}
	From Equations~\eqref{eq:m1bb} and~\eqref{eq:m2bb}, we have $\hat{\algvar{E}}_{B,i,t}=\hat{\algvar{E}}_{B,i,t}^{(1)}+\hat{\algvar{E}}_{B,i,t}^{(2)}\leq\algvar{E}_{B,i,t}/2+\algvar{E}_{B,i,t}/2=\algvar{E}_{B,i,t}$.
	
	For the forward privacy budget usage of each $u_i$ at time slot $t$, in sub-mechanism Part$_{\hbox{\scriptsize DC}}$ we have
	\begin{equation}\label{eq:m1fb}
		\outereqsize{
			\begin{aligned}
				\hat{\algvar{E}}_{F,i,t}^{(1)} &= \sum_{k=t}^{t+w_{F,i,t}-1} \epsilon_{i,k}^{(1)}\\ 
				&\leq \sum_{k=t}^{t+w_{F,i,t}-1} \epsilon_{F,i,k}^{(1)}\\
				&\leq \sum_{k=t}^{t+w_{F,i,t}-1} \frac{\algvar{E}_{F,i,t}}{2w_{F,i,t}}\\
				&=\frac{\algvar{E}_{F,i,t}}{2}.
			\end{aligned}
		}
	\end{equation}
	
	In sub-mechanism Part$_{\hbox{\scriptsize NOP}}$, given two time slot $t$ and $\tau$ with $t\leq \tau$, according to the calculation process of $\epsilon_{i,\tau}^{(2)}$, we have $0\leq\epsilon_{i,\tau}^{(2)}\leq\epsilon_{F,i,\tau}^{(2)}\leq\frac{1}{2}\left(\frac{\algvar{E}_{F,i,t}}{2}-\sum_{k=t}^{\tau-1}\epsilon_{i,k}^{(2)}\right)$, thus
	$\sum_{k=t}^{\tau-1}\epsilon_{i,k}^{(2)}\leq\frac{\algvar{E}_{F,i,t}}{2}-2\epsilon_{i,\tau}^{(2)}$.
	Therefore, it holds
	
	\begin{equation}\label{eq:m2fb}
		\outereqsize{
			\begin{aligned}
				\hat{\algvar{E}}_{F,i,t}^{(2)} &= \sum_{k=t}^{t+w_{F,i,t}-1} \epsilon_{i,k}^{(2)}\\ 
				&\leq \frac{\algvar{E}_{F,i,t}}{2}-2\epsilon_{i,t+w_{F,i,t}}^{(2)}\\
				&\leq \frac{\algvar{E}_{F,i,t}}{2}.
			\end{aligned}
		}
	\end{equation}
	
	From Equations~\eqref{eq:m1fb} and~\eqref{eq:m2fb}, we have 
	$\hat{\algvar{E}}_{F,i,t}=\hat{\algvar{E}}_{F,i,t}^{(1)}+\hat{\algvar{E}}_{F,i,t}^{(2)}\leq\algvar{E}_{F,i,t}/2+\algvar{E}_{F,i,t}/2=\algvar{E}_{F,i,t}$.
	
	For any user $u_i$ at any time slot $t$ and any two stream prefixes $S_t$, $S'_t$ satisfying $S_t$ and $S'_t$ are $w_{B,i,t}$-neighboring and $S_t$ and $S'_t$ are $w_{F,i,t}$-neighboring. 
	Let $t_s$ be the minimal time slot with $S_t[t_s]\neq S'_t[t_s]$ and $t_e$ be the maximal time slot with $S_t[t_s]\neq S'_t[t_s]$. 
	Then we have $t-t_s+1\leq w_{B,i,t}$ and $t_e-t+1\leq w_{F,i,t}$.
	Let $\epsilon_{i,t}$ be the privacy budget usage of $u_i$ at time slot $t$.
	Let the output of our \solutionMethodD{} as $\hbox{\solutionMethodD{}}(S_t[t])=o_t\in \entity{O}$. For any $O\subseteq \entity{O}$ we have
	\begin{equation}\notag
		\outereqsizelarge{
			\begin{aligned}
				\frac{\Pr[\hbox{\solutionMethodD}(S_t)\in O]}{\Pr[\hbox{\solutionMethodD}(S_t')\in O]} &\leq \prod_{k=t_s}^{t_e}\frac{\Pr[\hbox{\solutionMethodD}(S_t[k])=o_{k}]}{\Pr[\hbox{\solutionMethodD}(S_t'[k])=o_{k}]}\\
				&\leq e^{\sum_{k=t_s}^{t_e}\epsilon_{i,k}} \\
				&\leq e^{\sum_{k=t_s}^{t}\epsilon_{i,k}+\sum_{k=t}^{t_e}\epsilon_{i,k}} \\
				&= e^{\hat{\algvar{E}}_{B,i,t}+\hat{\algvar{E}}_{F,i,t}} \\
				&\leq e^{\algvar{E}_{B,i,t}+\algvar{E}_{F,i,t}} \\
			\end{aligned}
		}
	\end{equation}
	Let $\vectorfont{w}_B=(w_{B,1,t},\dots,w_{B,n,t})$, $\vectorfont{w}_F=(w_{F,1,t},\dots,w_{F,n,t})$, $\vectorfont{\algvar{E}}_B=(\algvar{E}_{B,1,t},\dots,\algvar{E}_{B,n,t})$ and $\vectorfont{\algvar{E}}_F=(\algvar{E}_{F,1,t}, \dots, \algvar{E}_{F,n,t})$.
	Then, \solutionMethodD{} satisfies \dynamicPrivacyLevelSimpleNameT{}.

	(2) \solutionMethodE{} satisfies \dynamicPrivacyLevelSimpleNameT{} at each time slot $t\in[T]$.
	
	In \solutionMethodE{}, the backward privacy budget usage for each $u_i$ at time slot $t$ matches that of \solutionMethodD{}, which means $\hat{\algvar{E}}_{B,i,t}\leq \algvar{E}_{B,i,t}$.
	Similarly, the forward privacy budget usage in Part$_{\hbox{\scriptsize DC}}$ at time slot $t$ is identical to \solutionMethodD{}, resulting in $\hat{\algvar{E}}_{F,i,t}^{(1)}\leq \algvar{E}_{F,i,t}/2$.
	
	Next, we consider the forward privacy budget usage of each $u_i$ at time slot $t$ in Part$_{\hbox{\scriptsize NOP}}$.
	For any two time slots $t$ and $\tau$ where $t\leq \tau$, based on the calculation process of $\epsilon_{i,\tau}^{(2)}$, we have $0\leq\epsilon_{i,\tau}^{(2)}\leq\epsilon_{\hbox{\scriptsize\em FA},i,\tau}^{(2)}\leq\epsilon_{\hbox{\scriptsize\em UF},i,\tau}^{(2)}\leq \frac{\algvar{E}_{F,i,t}}{2}-\sum_{k=t}^{\tau-1}\epsilon_{i,k}^{(2)}$. Thus,
	$\sum_{k=t}^{\tau-1}\epsilon_{i,k}^{(2)}\leq\frac{\algvar{E}_{F,i,t}}{2}-\epsilon_{i,\tau}^{(2)}$.
	Therefore, we have
	\begin{equation}\notag
		\outereqsizelarge{
			\begin{aligned}
				\hat{\algvar{E}}_{F,i,t}^{(2)} &= \sum_{k=t}^{t+w_{F,i,t}-1} \epsilon_{i,k}^{(2)}\\ 
				&\leq \frac{\algvar{E}_{F,i,t}}{2}-\epsilon_{i,t+w_{F,i,t}}^{(2)}\\
				&\leq \frac{\algvar{E}_{F,i,t}}{2}.
			\end{aligned}
		}
	\end{equation}
	Thus, $\hat{\algvar{E}}_{F,i,t}=\hat{\algvar{E}}_{F,i,t}^{(1)}+\hat{\algvar{E}}_{F,i,t}^{(2)}\leq\algvar{E}_{F,i,t}$.
	The subsequent steps follow the same proof process as shown in \solutionMethodD{}.
	Therefore, \solutionMethodE{} satisfies \dynamicPrivacyLevelSimpleNameT{}.
\end{proof}

\subsection{Proof for Utility Analysis}
\subsubsection{Proof for Theorem~\ref{Thm:solutionMethodA_utility_analysis}}\label{Proof:solutionMethodA_utility_analysis}
\begin{proof}
	Given a privacy budget-quantity pair set $P$
	and a positive number $\beta$, we define $\beta\cdot P=\{(\beta\cdot\epsilon_j,n_j)|(\epsilon_j,n_j)\in P\}$.
	For each user $u_i$ with fixed personalized privacy requirement $\left(w_i,\algvar{E}_i\right)$, we calculate their average budget per window as $\frac{\algvar{E}_i}{w_i}$. We denote the set of all average budgets as $\overline{\epsilon}=\left\{\frac{\algvar{E}_i}{w_i}|i\in[n]\right\}$.
	We then construct the privacy budget-quantity pair set of each type of average budget as $P_A=\{(\epsilon_j, n_j)|\epsilon_j\in\overline{\epsilon}\}$.
	Let $Z=(n-n_{A})\left(n-n_{A}+\frac{1}{4}\right)$ be the sampling error upper bound, where $n_{A}$ is the quantity of $\max_{i\in[n]}\frac{\algvar{E}_i}{w_i}$ in $\overline{\epsilon}$.
	
	When Part$_{\hbox{\scriptsize DC}}$ is not private, the error stems from Part$_{\hbox{\scriptsize NOP}}$.
	In Part$_{\hbox{\scriptsize NOP}}$, errors arise from both non-null and non publications.
	According to the Part$_{\hbox{\scriptsize NOP}}$, an null publication error does not exceed the non-null publication error at the most recent publication time slot.
	For the average error $\overline{\hbox{\em err}}_{\hbox{\scriptsize NOP}}$ of all time slots within the window of size $w_L$, based on the \solutionMethodA{} process, we have: 
	\begin{equation}\label{equation:pbd_m2}
		\outereqsize{
			\begin{aligned}
				\overline{\hbox{\em err}}_{\hbox{\scriptsize NOP}} &= \frac{1}{w_L}\sum_{k\in[\tilde{s}]}\frac{w_L}{\tilde{s}}\cdot \widetilde{\hbox{\em err}}_O\left(\frac{1}{2^{k+1}}P_A\right)\\
				&< \frac{1}{\tilde{s}}\sum_{k\in[\tilde{s}]}\min{\left(\frac{2}{\left(\frac{\epsilon_L}{2^{k+1}}\right)^2},Z+\frac{2}{\left(\frac{\epsilon_R}{2^{k+1}}\right)^2}\right)}\\
				&< \frac{1}{\tilde{s}}\min\left(\sum_{k\in[\tilde{s}]}\frac{8\cdot4^k}{\epsilon_L^2}, \tilde{s}\cdot Z+\sum_{k\in[\tilde{s}]}\frac{8\cdot 4^k}{\epsilon_R^2}\right)\\
				&= \min\left(\frac{32\cdot(4^{\tilde{s}}-1)}{3\tilde{s}\epsilon_L^2}, Z+\frac{32\cdot (4^{\tilde{s}}-1)}{3\tilde{s}\epsilon_R^2}\right).
			\end{aligned}
		}
	\end{equation}
	
	When Part$_{\hbox{\scriptsize DC}}$ is private, the error from Part$_{\hbox{\scriptsize DC}}$ can lead to two scenarios: (1) falsely skipping a publication or (2) falsely performing a non-null publication.
	Both cases are bounded by the error in Part$_{\hbox{\scriptsize DC}}$.
	In Part$_{\hbox{\scriptsize DC}}$, we execute the SM with OBS. The sensitivity of $\hbox{\em dis}$ is $1/d$.
	For the average error $\overline{\hbox{\em err}}_{\hbox{\scriptsize DC}}$ of each time slot in window size $w_L$,
	according to Lemma~\ref{lemma:sm_utility}, we have
	\begin{equation}\label{equation:pbd_m1}
		\outereqsize{
			\begin{aligned}
				\overline{\hbox{\em err}}_{\hbox{\scriptsize DC}} &< \min{\left({\frac{2}{d^2\min_{i\in[n]}\left(\frac{\algvar{E}_i}{2w_i}\right)^2}},Z+{\frac{2}{d^2\max_{i\in[n]}\left(\frac{\algvar{E}_i}{2w_i}\right)^2}}\right)}\\
				&= \min{\left({\frac{8}{d^2\epsilon_L^2}},Z+{\frac{8}{d^2\epsilon_R^2}}\right)}.
			\end{aligned}
		}
	\end{equation}
	
	Based on Equations~\eqref{equation:pbd_m1} and~\eqref{equation:pbd_m2}, we can get the average error upper bound as $\overline{\hbox{\em err}}_{\hbox{\scriptsize DC}}+\overline{\hbox{\em err}}_{\hbox{\scriptsize NOP}}$.
	
\end{proof}

\subsubsection{Proof for Theorem~\ref{Thm:solutionMethodB_utility_analysis}}\label{Proof:solutionMethodB_utility_analysis}
\begin{proof}
	Similar to \solutionMethodA{}, we first analyze the error of Part$_{\hbox{\scriptsize NOP}}$ in \solutionMethodB{} by assuming Part$_{\hbox{\scriptsize DC}}$ is not private.
	We then add the error of Part$_{\hbox{\scriptsize DC}}$, which is identical to that in \solutionMethodA{}, to obtain the final total error.
	When Part$_{\hbox{\scriptsize DC}}$ is not private, the error stems from Part$_{\hbox{\scriptsize NOP}}$.	
	In Part$_{\hbox{\scriptsize NOP}}$, each non-null publication corresponds to $\alpha$ skipped publications preceding it and $\alpha$ nullified publications succeeding it.
	
	For each $u_i$'s skipped publication, the publication privacy budget lower bound doubles with each time slot increasing until it reaches $\algvar{E}_i/2$ or a non-null publication occurs.
	For example, in Figure~\ref{PBD_publication_lower bound_proof}, assume $\alpha=5$, the non-null publication time slot is $t_6$.
	At time slot $t_1$, each $u_i$'s publication budget lower bound is $\algvar{E}_i/(2w_i)$.
	Take $u_1$ as an example: it reaches $\algvar{E}_1/2$ at time slot $t_4$. 
	The publication lower bound for $u_1$ remains at $\algvar{E}_1/2$ until time slot $t_6$.
	\begin{figure}[ht!]\vspace{-2ex}
		\centering
		\includegraphics[width=0.47\textwidth]{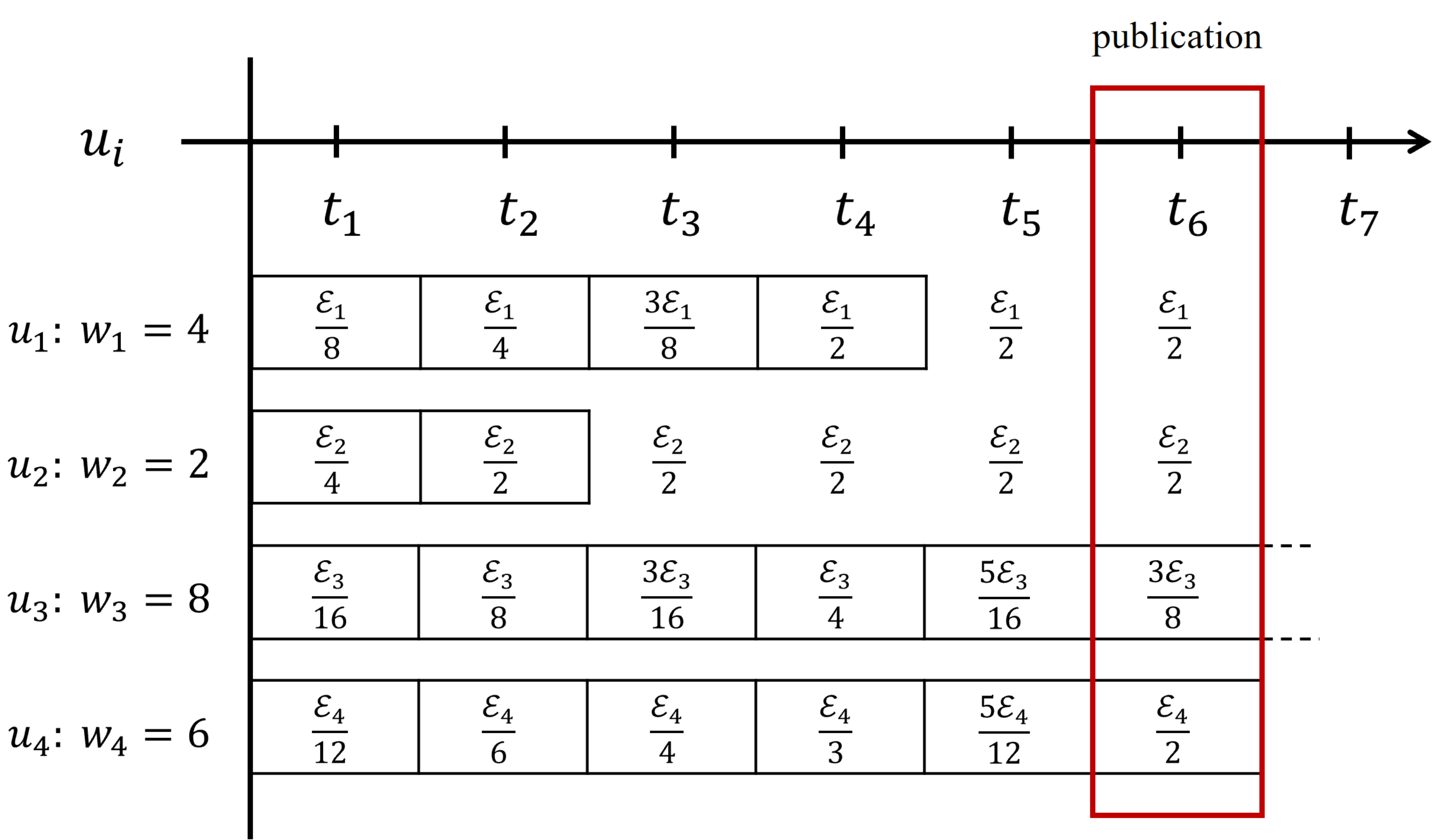}
		\caption{An example of the publication budget lower bound in \solutionMethodB{}.}\label{PBD_publication_lower bound_proof}
	\end{figure}
	Let the publication budget lower bound set for all users at skipped publication time slots (spanning $\alpha$ time slot) be $\hat{\vectorfont{\epsilon}}=\{\vectorfont{\epsilon}_1,\vectorfont{\epsilon}_2,\ldots ,\vectorfont{\epsilon}_{\alpha}\}$.
	Then, the error upper bound of each skipped publication is the error of publishing new data using $\vectorfont{\epsilon}_{k}$ ($k\in[\alpha]$).
	For example in Figure~\ref{PBD_publication_lower bound_proof}, the error upper bound at $t_3$ is the error of publishing a new obfuscated statistic result using $\{\frac{3\algvar{E}_1}{8},\frac{\algvar{E}_2}{2},\frac{3\algvar{E}_3}{16},\frac{\algvar{E}_4}{4}\}$.
	
	Let $Z=(n-n_{A})\left(n-n_{A}+\frac{1}{4}\right)$ be the sampling error upper bound, where $n_{A}$ is the number of users with maximum value of $\frac{\algvar{E}_i}{w_i}$.
	We now consider the following two cases: 
	
	\textbf{Case $\alpha\leq w_{L}$}.
	In this case, the publication budget lower bound doubles with each time slot increasing.
	Let $\hbox{\em err}_{\hbox{\scriptsize NOP}}^{\left(\hbox{\scriptsize\em sk}\right)}(\alpha)$ and $\hbox{\em err}_{\hbox{\scriptsize NOP}}^{\left(\hbox{\scriptsize\em pb}\right)}$ be the total error upper bounds of the $\alpha$ skipped publications and the non-null publication in Part$_{\hbox{\scriptsize NOP}}$, respectively.
	Let $\hbox{\em err}_{\hbox{\scriptsize NOP}}^{\left(\hbox{\scriptsize\em s,p}\right)}$ be the error of all skipped publications and the publication in Part$_{\hbox{\scriptsize NOP}}$.
	According to Lemma~\ref{lemma:sm_utility}, we have
	\begin{equation}\label{skipped_error}\notag
		\outereqsize{
			\begin{aligned}
				\hbox{\em err}_{\hbox{\scriptsize NOP}}^{\left(\hbox{\scriptsize\em sk}\right)}(\alpha) &< \sum_{k\in[\alpha]}\min{\left(\frac{2}{(k\epsilon_L)^2},Z+\frac{2}{(k\epsilon_R)^2}\right)} \\
				&\leq \min{\left(\frac{2}{\epsilon_L^2}H^{2}_{\alpha},\alpha Z+\frac{2}{\epsilon_R^2}H^{2}_{\alpha}\right)}, \\
			\end{aligned}
		}
	\end{equation}
	and 
	\begin{equation}\label{skipped_publication_error}
		\outereqsize{
			\begin{aligned}
				\hbox{\em err}_{\hbox{\scriptsize NOP}}^{\left(\hbox{\scriptsize\em s,p}\right)} &< \hbox{\em err}_{\hbox{\scriptsize NOP}}^{\left(\hbox{\scriptsize\em sk}\right)}(\alpha) + \hbox{\em err}_{\hbox{\scriptsize NOP}}^{\left(\hbox{\scriptsize\em pb}\right)}\\
				&=\hbox{\em err}_{\hbox{\scriptsize NOP}}^{\left(\hbox{\scriptsize\em sk}\right)}(\alpha+1)\\
				&= \min{\left(\frac{2}{\epsilon_L^2}H^{2}_{\alpha+1},(\alpha+1)Z+\frac{2}{\epsilon_R^2}H^{2}_{\alpha+1}\right)}.
			\end{aligned}
		}
	\end{equation}
	Thus, we derive the average error upper bound $\overline{\hbox{\em err}}_{\hbox{\scriptsize NOP}}$ of each time slot in Part$_{\hbox{\scriptsize NOP}}$ as 
	\begin{equation}\label{sp_err_1}
		\outereqsizelarge{
			\begin{aligned}
				\overline{\hbox{\em err}}_{\hbox{\scriptsize NOP}} &< \frac{1}{2\alpha+1}\left(\widetilde{\hbox{\em err}}_{\hbox{\scriptsize NOP}}^{\left(\hbox{\scriptsize\em s,p}\right)}+\alpha\cdot\overline{\hbox{\em err}}_{\hbox{\scriptsize\em nlf}}\right),
			\end{aligned}
		}
	\end{equation}
	where $\widetilde{\hbox{\em err}}_{\hbox{\scriptsize NOP}}^{\left(\hbox{\scriptsize\em s,p}\right)}$ is the final value in Equation~\eqref{skipped_publication_error}.
	
	\textbf{Case $\alpha>w_{L}$}.
	In this case, we have
	\begin{equation}\label{skipped_publication_error2}
		\outereqsize{
			\begin{split}
				& \hbox{\em err}_{\hbox{\scriptsize NOP}}^{\left(\hbox{\scriptsize\em s,p}\right)} \\<&\hbox{\em err}_{\hbox{\scriptsize NOP}}^{\left(\hbox{\scriptsize\em sk}\right)}(w_L) + \sum_{k=w_L+1}^{\alpha+1}\min{\left(\frac{2}{\epsilon_{\tilde{L}}^2},Z+\frac{2}{\epsilon_{\tilde{R}}^2}\right)} \\
				=& \hbox{\em err}_{\hbox{\scriptsize NOP}}^{\left(\hbox{\scriptsize\em sk}\right)}(w_L) + (\alpha-w_L+1)\min{\left(\frac{2}{\epsilon_{\tilde{L}}^2}, Z+\frac{2}{\epsilon_{\tilde{R}}^2}\right)}\\
				<& \min{\left(\frac{2}{\epsilon_L^2}H^{2}_{w_L},w_L Z+\frac{2}{\epsilon_R^2}H^{2}_{w_L}\right)} \\
				&+ (\alpha-w_L+1)\min{\left(\frac{2}{\epsilon_{\tilde{L}}^2}, Z+\frac{2}{\epsilon_{\tilde{R}}^2}\right)}.
			\end{split}
		}
	\end{equation}
	Therefore, we obtain the average error upper bound $\overline{\hbox{\em err}}_{\hbox{\scriptsize NOP}}$ for each time slot in Part$_{\hbox{\scriptsize NOP}}$ as 
	\begin{equation}\label{sp_err_2}
		\outereqsizelarge{
			\begin{aligned}
				\overline{\hbox{\em err}}_{\hbox{\scriptsize NOP}} &< \frac{1}{2\alpha+1}\left(\widetilde{\hbox{\em err}}_{\hbox{\scriptsize NOP}}^{\left(\hbox{\scriptsize\em s,p}\right)}+\alpha\cdot\overline{\hbox{\em err}}_{\hbox{\scriptsize\em nlf}}\right),
			\end{aligned}
		}
	\end{equation}
	where $\widetilde{\hbox{\em err}}_{\hbox{\scriptsize NOP}}^{\left(\hbox{\scriptsize\em s,p}\right)}$ is the value derived in Equation~\eqref{skipped_publication_error2}.
	
	When Part$_{\hbox{\scriptsize DC}}$ is private, its error is identical to that in \solutionMethodA{}:
	\begin{equation}\label{methodB_m_1_error}
		\outereqsizelarge{
			\begin{aligned}
				\overline{\hbox{\em err}}_{\hbox{\scriptsize DC}}<\min{\left({\frac{8}{d^2\epsilon_L^2}},Z+{\frac{8}{d^2\epsilon_R^2}}\right)}.
			\end{aligned}
		}
	\end{equation}
	
	Based on Equations~\eqref{methodB_m_1_error}, \eqref{sp_err_1} and \eqref{sp_err_2}, we can derive the average error upper bound for each time slot in \solutionMethodB{} as:
	\begin{equation}\notag
		\outereqsizelarge{
			\begin{aligned}
				\min{\left({\frac{8}{d^2\epsilon_L^2}},Z+{\frac{8}{d^2\epsilon_R^2}}\right)} + \frac{1}{2\alpha+1}\left(\widetilde{\hbox{\em err}}_{\hbox{\scriptsize NOP}}^{\left(\hbox{\scriptsize\em s,p}\right)}+\alpha\cdot\overline{\hbox{\em err}}_{\hbox{\scriptsize\em nlf}}\right),
			\end{aligned}
		}
	\end{equation}
	where $\widetilde{\hbox{\em err}}_{\hbox{\scriptsize NOP}}^{\left(\hbox{\scriptsize\em s,p}\right)}$ is the final result from Equation~\eqref{skipped_publication_error} when $\alpha\leq w_L$, and from Equation~\eqref{skipped_publication_error2} when $\alpha> w_L$.
\end{proof}

\subsubsection{Proof for Theorem~\ref{Thm:solutionD_utility_analysis}}\label{Proof:solutionD_utility_analysis}
\begin{proof}
	When Part$_{\hbox{\scriptsize DC}}$ is not private, the error of \solutionMethodD{} is from the process of Part$_{\hbox{\scriptsize NOP}}$.
	This error is determined by $\vectorfont{\epsilon}_{t}^{(2)}=\left(\epsilon_{1,t}^{(2)},\dots,\epsilon_{n,t}^{(2)}\right)$ at each time slot $t$, where each element $\epsilon_{i,t}^{(2)}$ depends on $\epsilon_{B,i,t}^{(2)}$ and $\epsilon_{F,i,t}^{(2)}$.
	For any $\epsilon_{B,i,t}^{(2)}$, there are two possible cases: 
	either $\epsilon_{F,i,t}^{(2)}\geq\epsilon_{L}^{\left(\hbox{\scriptsize\em B,M}\right)}(i)$ (as shown by Curve $\hbox{\em AHI}$ in Figure~\ref{backward_forward_bound}) or $\epsilon_{F,i,t}^{(2)}<\epsilon_{L}^{\left(\hbox{\scriptsize\em B,M}\right)}(i)$ (as shown by Curve $\hbox{\em IJ}$ in Figure~\ref{backward_forward_bound}).
	\begin{figure}[t!]\centering\vspace{-2ex}
		\subfigure{
			\scalebox{0.33}[0.33]{\includegraphics{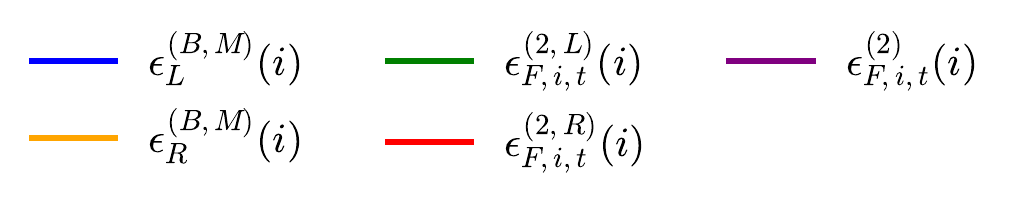}}}\hfill\\
		\addtocounter{subfigure}{-1}\vspace{-2.5ex}
		\subfigure{
			\scalebox{0.33}[0.33]{\includegraphics{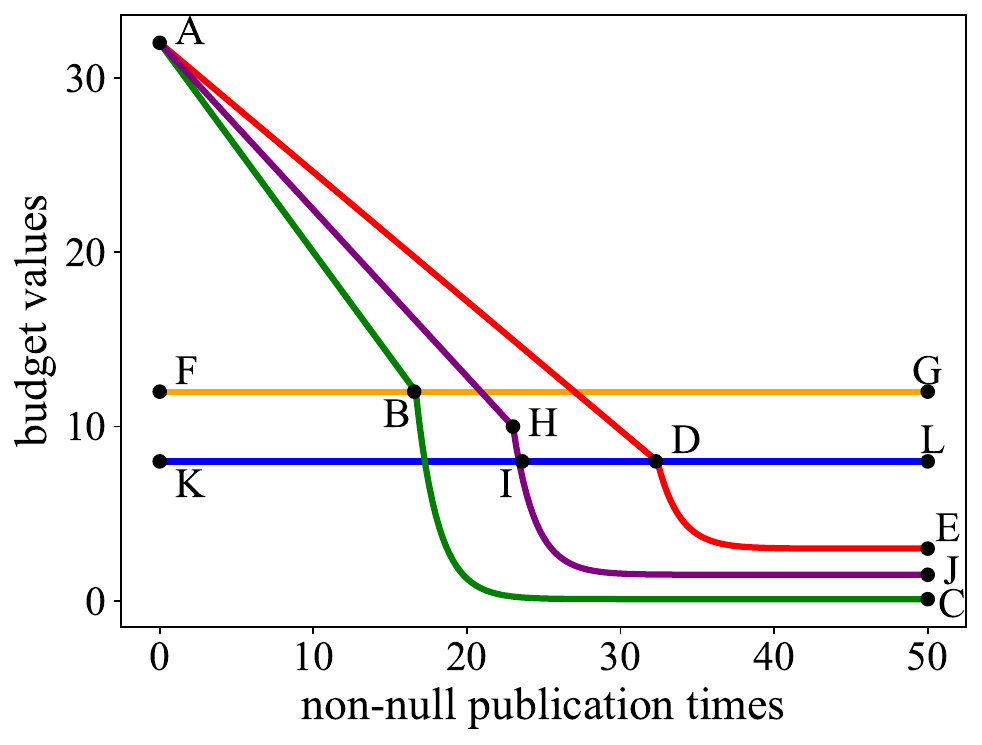}}
			\label{subfig:backward_forward_data}}\hfill
		\caption{\small  An example of backward privacy budget bounds, where $\epsilon_{L}^{\left(\hbox{\scriptsize\em B,M}\right)}(i)$ and $\epsilon_{R}^{\left(\hbox{\scriptsize\em B,M}\right)}(i)$ represent the lower and upper bounds of $\epsilon_{B,i,t}^{(2)}$ respectively, $\epsilon_{F,i,t}^{(\hbox{\scriptsize 2,\em L})}(i)$ and $\epsilon_{F,i,t}^{(\hbox{\scriptsize 2,\em R})}(i)$ represent the lower and upper bounds of budget values, and $\epsilon_{F,i,t}^{(2)}(i)$ shows an example of budget values.}
		\label{backward_forward_bound}
	\end{figure}
	Let $\gamma_i$ be the number of non-null publications where $\epsilon_{F,i,t}^{(2)}\geq\epsilon_{L}^{\left(\hbox{\scriptsize\em B,M}\right)}(i)$.
	Since the current backward privacy budget can be any value within $[\epsilon_{L}^{\left(\hbox{\scriptsize\em B,M}\right)}(i),\epsilon_{R}^{\left(\hbox{\scriptsize\em B,M}\right)}(i)]$, we bound the transition point $\gamma_i$ by considering the two extreme decay cases of the forward budget: the fastest decay with per-publication consumption $\epsilon_{R}^{\left(\hbox{\scriptsize\em B,M}\right)}(i)$, and the slowest decay with per-publication consumption $\epsilon_{L}^{\left(\hbox{\scriptsize\em B,M}\right)}(i)$. Therefore, after $\gamma_i$ non-null publications, the forward budget needs to be no greater than $\epsilon_{L}^{\left(\hbox{\scriptsize\em B,M}\right)}(i)$ in the fastest-decay case, and no smaller than $\epsilon_{R}^{\left(\hbox{\scriptsize\em B,M}\right)}(i)$ in the slowest-decay case.
	Thus, we have
	\begingroup
	\footnotesize
	\setlength{\abovedisplayskip}{3pt plus 1pt minus 1pt}
	\setlength{\belowdisplayskip}{3pt plus 1pt minus 1pt}
	\setlength{\abovedisplayshortskip}{2pt plus 1pt minus 1pt}
	\setlength{\belowdisplayshortskip}{3pt plus 1pt minus 1pt}
	\begin{equation}\notag
		\outereqsizelarge{
			\begin{aligned}
				\begin{array}{l}\vspace{1ex}
					\frac{\algvar{E}_{L}^{(F)}(i)}{2} - \gamma_i\cdot\epsilon_{R}^{\left(\hbox{\scriptsize\em B,M}\right)}(i) \leq \epsilon_{R}^{\left(\hbox{\scriptsize\em B,M}\right)}(i); \\ 
					\frac{\algvar{E}_{L}^{(F)}(i)}{2} - \gamma_i\cdot\epsilon_{L}^{\left(\hbox{\scriptsize\em B,M}\right)}(i) \geq \epsilon_{L}^{\left(\hbox{\scriptsize\em B,M}\right)}(i). \\
				\end{array}
			\end{aligned}
		}
	\end{equation}
	\endgroup
	Thus, 
	\begingroup
	\footnotesize
	\setlength{\abovedisplayskip}{3pt plus 1pt minus 1pt}
	\setlength{\belowdisplayskip}{3pt plus 1pt minus 1pt}
	\setlength{\abovedisplayshortskip}{2pt plus 1pt minus 1pt}
	\setlength{\belowdisplayshortskip}{3pt plus 1pt minus 1pt}
	\begin{equation}\notag
		\outereqsizelarge{
			\begin{aligned}
				2^{\eta_i-1} - 1 \leq \gamma_i \leq 2^{\beta_i-1}-1.
			\end{aligned}
		}
	\end{equation}
	\endgroup
	
	Let $\hat{\vectorfont{\epsilon}}_{k}^{(2)}$ be the privacy budget usage at time slot $k$ in Part$_{\hbox{\scriptsize NOP}}$.
	Then for the average error $\overline{\hbox{\em err}}_{\textrm{Part}_{\hbox{\tiny NOP}}}$ of Part$_{\hbox{\scriptsize NOP}}$ in \solutionMethodD{}, we have
	\begin{equation}\label{eq:methodDm2error}
		\outereqsize{
			\begin{aligned}
				&\mkern2mu\overline{\hbox{\em err}}_{\textrm{Part}_{\hbox{\tiny NOP}}}\\
				=&\frac{1}{Y}\cdot\frac{Y}{\hat{s}}\sum_{k=1}^{\hat{s}}\widetilde{\hbox{\em err}}_O\left(\hat{\vectorfont{\epsilon}}_{k}^{(2)}\right)\\
				=& \frac{1}{\hat{s}}\Bigg(\sum_{k=1}^{\gamma_L}\widetilde{\hbox{\em err}}_O\left(\hat{\vectorfont{\epsilon}}_{k}^{(2)}\right)+\sum_{k=\gamma_L+1}^{\gamma_R}\widetilde{\hbox{\em err}}_O\left(\hat{\vectorfont{\epsilon}}_{k}^{(2)}\right)+\sum_{k=\gamma_R+1}^{\hat{s}}\widetilde{\hbox{\em err}}_O\left(\hat{\vectorfont{\epsilon}}_{k}^{(2)}\right)\Bigg)\\
				\leq& \min\Bigg(\frac{2\gamma_L}{\hat{s}\epsilon_{\hbox{\scriptsize\em BL}}^2},\frac{Z'\gamma_L}{\hat{s}}
				+\frac{2\gamma_L}{\hat{s}\epsilon_{\hbox{\scriptsize\em BR}}^2}\Bigg)\\
				&+\min\Bigg(\frac{2(\gamma_R-\gamma_L)}{\hat{s}\epsilon_{\hbox{\scriptsize\em BL}}^2},\frac{Z'(\gamma_R-\gamma_L)}{\hat{s}}+\frac{32(4^{\gamma_R-\gamma_L}-1)}{3\hat{s}\epsilon_{\hbox{\scriptsize\em BR}}^2}\Bigg)\\
				&+\min\Bigg(\frac{32(4^{\hat{s}-\gamma_R}-1)}{3\hat{s}\epsilon_{\hbox{\scriptsize\em BL}}^2},\frac{Z'(\hat{s}-\gamma_R)}{\hat{s}}
				+\frac{32(4^{\hat{s}-\gamma_L}-4^{\gamma_R-\gamma_L})}{3\hat{s}\epsilon_{\hbox{\scriptsize\em BR}}^2}\Bigg)\\
				\leq& \min\Bigg(\frac{2(4^{\hat{s}-\gamma_R+2}+3\gamma_R-16)}{3\hat{s}\epsilon_{\hbox{\scriptsize\em BL}}^2},Z'
				+\frac{2(4^{\hat{s}-\gamma_L+2}+3\gamma_L-16)}{3\hat{s}\epsilon_{\hbox{\scriptsize\em BR}}^2}\Bigg).
			\end{aligned}
		}
	\end{equation}
	
	When Part$_{\hbox{\scriptsize DC}}$ is private, the error from Part$_{\hbox{\scriptsize DC}}$ will lead to falsely skipping or publishing a  non-null publication. 
	Both cases are bounded by the error in Part$_{\hbox{\scriptsize DC}}$.
	Let $\hat{\vectorfont{\epsilon}}_{k}^{(1)}$ be the privacy budget usage at time slot $k$ in Part$_{\hbox{\scriptsize DC}}$.
	For the error in Part$_{\hbox{\scriptsize DC}}$, we have
	\begin{equation}\label{eq:methodDm1error}
		\outereqsize{
			\begin{aligned}
				\overline{\hbox{\em err}}_{\textrm{Part}_{\hbox{\tiny DC}}} =& \frac{1}{Y}\cdot\frac{Y}{\hat{s}}\sum_{k=1}^{\hat{s}}\widetilde{\hbox{\em err}}_O\left(\hat{\vectorfont{\epsilon}}_{k}^{(1)}\right)\\
				\leq& \min\Bigg(\frac{2}{d^2(\min(\epsilon_{\hbox{\scriptsize\em FLL}},\epsilon_{\hbox{\scriptsize\em BL}}))^2},Z'+\frac{2}{d^2(\max(\epsilon_{\hbox{\scriptsize\em FLR}},\epsilon_{\hbox{\scriptsize\em BR}}))^2}\Bigg).
			\end{aligned}
		}
	\end{equation}
	
	Thus according to Equation~\eqref{eq:methodDm1error} and~\eqref{eq:methodDm2error}, we can get
	the error of \solutionMethodD{} as $\hbox{\em err}_{\hbox{\scriptsize \solutionMethodD{}}}=\overline{\hbox{\em err}}_{\textrm{Part}_{\hbox{\scriptsize DC}}}+\overline{\hbox{\em err}}_{\hbox{\scriptsize Part}_{\hbox{\tiny NOP}}}$. 
	
\end{proof}

\subsubsection{Proof for Theorem~\ref{Thm:solutionE_utility_analysis}}\label{Proof:solutionE_utility_analysis}
\begin{proof}
	Let $\rho_{\hbox{\scriptsize\em sk}}$ be the number of skipped publications before a non-null publication and $\rho_{\hbox{\scriptsize\em nu}}$ be the number of nullified publications after a non-null publication.
	
	When the process of Part$_{\hbox{\scriptsize DC}}$ is not private, the average error is only from Part$_{\hbox{\scriptsize NOP}}$ which can further divided into skipped publication error, non-null publication error and nullified publication error.
	Let $\epsilon_{\hbox{\scriptsize\em FL}}(i)=\min_{t}\frac{\algvar{E}_{F,i,t}}{2w_{F,i,t}}$,  $\epsilon_{\hbox{\scriptsize\em FLL}}=\min_{i\in[n]}\epsilon_{\hbox{\scriptsize\em FL}}(i)$ and  $\epsilon_{\hbox{\scriptsize\em FLR}}=\max_{i\in[n]}\epsilon_{\hbox{\scriptsize\em FL}}(i)$.
	Similar as that in the proof of Theorem~\ref{Thm:solutionMethodB_utility_analysis}, the publication lower bound doubles with the time slot increases.
	For each $u_i$, let $\lambda_{L}(i)=\frac{\epsilon_{L}^{\left(\hbox{\scriptsize\em B,M}\right)}(i)}{\epsilon_{\hbox{\scriptsize\em FL}}(i)/2}$ be the number of skipped time slots whose publication lower bound are no more than $\epsilon_{L}^{\left(\hbox{\scriptsize\em B,M}\right)}(i)$.
	Let $\lambda_{R}(i)=\frac{\epsilon_{R}^{\left(\hbox{\scriptsize\em B,M}\right)}(i)}{\epsilon_{\hbox{\scriptsize\em FL}}(i)/2}$ be the number of skipped time slots whose publication lower bound are no more than $\epsilon_{R}^{\left(\hbox{\scriptsize\em B,M}\right)}(i)$.
	Let $\lambda_{\hbox{\scriptsize\em LR}}=\max_{i\in[n]}\lambda_{L}(i)$ be the maximal value among all $\lambda_{L}(i)$.
	Let $\lambda_{\hbox{\scriptsize\em RL}}=\min_{i\in[n]}\lambda_{R}(i)$ be the minimal value among all $\lambda_{R}(i)$.
	If $\lambda_{\hbox{\scriptsize\em LR}}\geq\lambda_{\hbox{\scriptsize\em RL}}$, then the publication privacy budget lower bounds and upper bounds are determined by the forward publication budgets.
	Thus, for the error $\hbox{\em err}_{\hbox{\scriptsize Part}_{\hbox{\tiny NOP}}}^{(\hbox{\scriptsize\em s,p})}$ in skipped publications and non-null publication, we have 
	\begin{equation}\notag
		\outereqsize{
			\begin{aligned}
				&\hbox{\em err}_{\hbox{\scriptsize Part}_{\hbox{\tiny NOP}}}^{(\hbox{\scriptsize\em s,p})}\\ 
				< &\sum_{k\in[\rho_{\hbox{\scriptsize\em sk}}+1]}\min\bigg(\frac{2}{(k\epsilon_{\hbox{\scriptsize\em FLL}})^2}, Z'+\frac{2}{(k\epsilon_{\hbox{\scriptsize\em FLR}})^2}\bigg)\\
				<& \min\bigg(\frac{2}{\epsilon_{\hbox{\scriptsize\em FLL}}^2}H^2_{\rho_{\hbox{\scriptsize\em sk}}+1}, Z'(\rho_{\hbox{\scriptsize\em sk}}+1)+\frac{2}{\epsilon_{\hbox{\scriptsize\em FLR}}^2}H^2_{\rho_{\hbox{\scriptsize\em sk}}+1}\bigg).
			\end{aligned}
		}
	\end{equation}
	If $\lambda_{\hbox{\scriptsize\em LR}}<\lambda_{\hbox{\scriptsize\em RL}}$, $\hbox{\em err}_{\hbox{\scriptsize Part}_{\hbox{\tiny NOP}}}^{(\hbox{\scriptsize\em s,p})}$ can be classified into three cases that $\rho_{\hbox{\scriptsize\em sk}}+1\leq\lambda_{L}$, $\lambda_{L}<\rho_{\hbox{\scriptsize\em sk}}+1\leq\lambda_{R}$ and $\rho_{\hbox{\scriptsize\em sk}}+1>\lambda_{R}$.
	We denote the $\hbox{\em err}_{\hbox{\scriptsize Part}_{\hbox{\tiny NOP}}}^{(\hbox{\scriptsize\em s,p})}$ as $\hbox{\em err}_{\hbox{\scriptsize Part}_{\hbox{\tiny NOP}}}^{\left(\hbox{\scriptsize\em s,p,\em 1}\right)}$, $\hbox{\em err}_{\hbox{\scriptsize Part}_{\hbox{\tiny NOP}}}^{\left(\hbox{\scriptsize\em s,p,\em 2}\right)}$ and $\hbox{\em err}_{\hbox{\scriptsize Part}_{\hbox{\tiny NOP}}}^{\left(\hbox{\scriptsize\em s,p,\em 3}\right)}$ respectively.
	
	(1) $\rho_{\hbox{\scriptsize\em sk}}+1\leq\lambda_{L}$.
	In this case, the publication privacy budget lower bounds are determined by the forward publication privacy budget lower bounds.
	The publication privacy budget upper bounds are determined by the backward publication privacy budget upper bounds.
	Thus, we have
	\begin{equation}\notag
		\outereqsize{
			\begin{aligned}
				& \hbox{\em err}_{\hbox{\scriptsize Part}_{\hbox{\tiny NOP}}}^{\left(\hbox{\scriptsize\em s,p,\em 1}\right)}\\
				<& \sum_{k\in[\rho_{\hbox{\scriptsize\em sk}}+1]}\min\left(\frac{2}{(k\epsilon_{\hbox{\scriptsize\em FLL}})^2}, Z'+\frac{2}{\epsilon_{\hbox{\scriptsize\em BR}}^2}\right)\\
				<& \min\left(\frac{2}{\epsilon_{\hbox{\scriptsize\em FLL}}^2}H^2_{\rho_{\hbox{\scriptsize\em sk}}+1}, Z'(\rho_{\hbox{\scriptsize\em sk}}+1)+\frac{2(\rho_{\hbox{\scriptsize\em sk}}+1)}{\epsilon_{\hbox{\scriptsize\em BR}}^2}\right).
			\end{aligned}
		}
	\end{equation}
	
	(2) $\lambda_{L}<\rho_{\hbox{\scriptsize\em sk}}+1\leq\lambda_{R}$.
	In this case, the first $\lambda_{L}$ publication privacy budget errors are the same as those in case~(1). For the remaining errors, the publication privacy budget lower bounds and upper bounds are determined by the backward publication privacy budget lower bounds and upper bounds.
	Thus, we have
	\begin{equation}\notag
		\outereqsize{
			\begin{aligned}
				&\hbox{\em err}_{\hbox{\scriptsize Part}_{\hbox{\tiny NOP}}}^{\left(\hbox{\scriptsize\em s,p,\em 2}\right)}\\
				&< \hbox{\em err}_{\hbox{\scriptsize Part}_{\hbox{\tiny NOP}}}^{\left(\hbox{\scriptsize\em s,p,\em 1}\right)}(\lambda_{L})+ \sum_{k=\lambda_{L}+1}^{\rho_{\hbox{\scriptsize\em sk}}+1}\min\left(\frac{2}{\epsilon_{\hbox{\scriptsize\em BL}}^2}, Z'+\frac{2}{\epsilon_{\hbox{\scriptsize\em BR}}^2}\right)\\
				& <\min\left(\frac{2}{\epsilon_{\hbox{\scriptsize\em FLL}}^2}H^2_{\lambda_{L}}, Z'\lambda_{L}+\frac{2\lambda_{L}}{\epsilon_{\hbox{\scriptsize\em BR}}^2}\right) \\
				&+ (\rho_{\hbox{\scriptsize\em sk}}-\lambda_{L}+1)\min\left(\frac{2}{\epsilon_{\hbox{\scriptsize\em BL}}^2}, Z'+\frac{2}{\epsilon_{\hbox{\scriptsize\em BR}}^2}\right).
			\end{aligned}
		}
	\end{equation}

	(3) $\rho_{\hbox{\scriptsize\em sk}}+1>\lambda_{R}$.	
	In this case, the first $\lambda_{R}$ publication privacy budget errors are the same as those in case~(2).
	For the remaining errors, the publication privacy budget lower bounds are determined by the backward publication privacy budget lower bounds.
	The publication privacy budget upper bounds are determined by the forward publication privacy budget upper bounds.
	Thus, we have
	\begin{equation}\notag
		\outereqsize{
			\begin{aligned}
				&\hbox{\em err}_{\hbox{\scriptsize Part}_{\hbox{\tiny NOP}}}^{\left(\hbox{\scriptsize\em s,p,\em 3}\right)}\\
				& <\hbox{\em err}_{\hbox{\scriptsize Part}_{\hbox{\tiny NOP}}}^{\left(\hbox{\scriptsize\em s,p,\em 2}\right)}(\lambda_{R})+\sum_{k=\lambda_{R}+1}^{\rho_{\hbox{\scriptsize\em sk}}+1}\min\left(\frac{2}{\epsilon_{\hbox{\scriptsize\em BL}}^2},Z'+\frac{2}{(k\epsilon_{\hbox{\scriptsize\em FLR}})^2}\right)\\
				&< \min\left(\frac{2}{\epsilon_{\hbox{\scriptsize\em FLL}}^2}H^2_{\lambda_{L}}, Z'\lambda_{L}+\frac{2\lambda_{L}}{\epsilon_{\hbox{\scriptsize\em BR}}^2}\right) \\
				&~~~~+ (\lambda_{R}-\lambda_{L})\min\left(\frac{2}{\epsilon_{\hbox{\scriptsize\em BL}}^2}, Z'+\frac{2}{\epsilon_{\hbox{\scriptsize\em BR}}^2}\right)\\
				&~~~~+ \min\bigg(\frac{2(\rho_{\hbox{\scriptsize\em sk}}-\lambda_{R}+1)}{\epsilon_{\hbox{\scriptsize\em BL}}^2}, (\rho_{\hbox{\scriptsize\em sk}}-\lambda_{R}+1)Z'\\
				&~~~~+\frac{2}{\epsilon_{\hbox{\scriptsize\em FLR}}^2}H^2_{\rho_{\hbox{\scriptsize\em sk}}-\lambda_{R}+1}\bigg).
			\end{aligned}
		}
	\end{equation}
	
	We denote the upper bound of $\hbox{\em err}_{\hbox{\scriptsize Part}_{\hbox{\tiny NOP}}}^{(\hbox{\scriptsize\em s,p})}$ as $\widetilde{\hbox{\em err}}_{\hbox{\scriptsize Part}_{\hbox{\tiny NOP}}}^{(\hbox{\scriptsize\em s,p})}$.
	Therefore, we can get the average error upper bound $\overline{\hbox{\em err}}_{\hbox{\scriptsize Part}_{\hbox{\tiny NOP}}}$ of each time slot in Part$_{\hbox{\scriptsize NOP}}$ as 
	\begin{equation}\label{eq:methodEm2error}
		\outereqsizelarge{
			\begin{aligned}
				\overline{\hbox{\em err}}_{\hbox{\scriptsize Part}_{\hbox{\tiny NOP}}} <& \frac{1}{\rho_{\hbox{\scriptsize\em sk}}+\rho_{\hbox{\scriptsize\em nu}}+1}\left(\widetilde{\hbox{\em err}}_{\hbox{\scriptsize Part}_{\hbox{\tiny NOP}}}^{(\hbox{\scriptsize\em s,p})}+\rho_{\hbox{\scriptsize\em nu}}\overline{\hbox{\em err}}_{\hbox{\scriptsize\em nlf}}\right).
			\end{aligned}
		}
	\end{equation}
	
	When Part$_{\hbox{\scriptsize DC}}$ is private, the average error $\overline{\hbox{\em err}}_{\textrm{Part}_{\hbox{\scriptsize DC}}}$ from Part$_{\hbox{\scriptsize DC}}$ is the same as that in Theorem~\ref{Thm:solutionD_utility_analysis} (same as Equation~\eqref{eq:methodDm1error}).
	From Equations~\eqref{eq:methodDm1error} and~\eqref{eq:methodEm2error}, we can get the average error of \solutionMethodE{} as
	the sum of Equations~\eqref{eq:methodDm1error} and~\eqref{eq:methodEm2error}.
	
\end{proof}

\end{document}